\documentclass[a4paper,UKenglish,leqno]{article}

\usepackage[utf8]{inputenc}
\usepackage{authblk}
\usepackage[letterpaper, top=25.4mm, bottom=25.4mm, left=25.4mm, right=25.4mm, includefoot]{geometry}
\usepackage{amsmath,amssymb,mathtools,stmaryrd}
\usepackage{enumerate}
\usepackage{xcolor}

\usepackage{imakeidx}
\makeindex
\makeindex[name=Symbol,title=Symbols]

\newcommand{\Symbol}[1]{\index[Symbol]{#1}}
\newcommand{\Index}[1]{\index{#1}}
\newcommand{\mIndex}[2][]{\ifthenelse{\equal{#1}{}}{\index{#2}}{\index{#1@#2}}\marginpar{#2}}

\newcommand{\mSymbol}[2][]{\ifthenelse{\equal{#1}{}}{\index[Symbol]{#2}}{\index[Symbol]{#1@#2}}}%

\usepackage{hyperref}
\definecolor{dark-blue}{rgb}{0.05,0.25,0.85}
\hypersetup{colorlinks=true,linkcolor=dark-blue,citecolor=dark-blue}

\usepackage{caption,subcaption}

\newcommand{\N}{\mathbb{N}}

\newcommand{\AAA}{\mathcal{A}}

\newcommand{\CCC}{\mathcal{C}}

\newcommand{\EEE}{\mathcal{E}}

\newcommand{\FFF}{\mathcal{F}}
\newcommand{\GGG}{\mathcal{G}}
\newcommand{\HHH}{\mathcal{H}}

\newcommand{\LLL}{\mathcal{L}}

\newcommand{\PPP}{\mathcal{P}}
\newcommand{\QQQ}{\mathcal{Q}}
\newcommand{\RRR}{\mathcal{R}}
\newcommand{\SSS}{\mathcal{S}}

\newcommand{\WWW}{\mathcal{W}}

\renewcommand{\SS}{\mathfrak{S}}

\newcommand{\bd}{\textup{bd}}

\newcommand*{\Abs}[1]{| #1 |}

\newcommand{\cut}{\operatorname{Cut}}

\newcommand{\Chubby}{\operatorname{Chubby}}
\newcommand{\m}{\mathfrak{m}}
\newcommand{\n}{\mathfrak{n}}

\newcommand{\ebw}[1]{\mathrm{cw}(#1)}

\newcommand{\leaves}[1]{\mathrm{leaves}(#1)}
\newcommand{\loops}{\mathrm{loops}}
\newcommand{\ol}{\mathrm{outl}}

\newcommand{\tw}[1]{\mathrm{tw}(#1)}

\newcommand{\w}{\operatorname{width}}

\newcommand{\head}{\operatorname{head}}
\newcommand{\tail}{\operatorname{tail}}

\newcommand{\rt}{\operatorname{r}}
\newcommand{\dom}{\operatorname{dom}}
\newcommand{\stitch}{\operatorname{stitch}}
\newcommand{\knit}{\operatorname{knit}}
\newcommand{\sth}{\mathrel : }
\renewcommand{\mid}{\sth}
\newcommand{\lenks}[1]{\mathrm{left}(#1)}
\newcommand{\riets}[1]{\mathrm{right}(#1)}

\newcommand{\spl}{\operatorname{split}}

\newcommand\restr[2]{{%
  \left.\kern-\nulldelimiterspace %
  #1 %
  \vphantom{|} %
  \right|_{\text{\scriptsize $#2$}} %
  }}
\newcommand{\fw}{\textup{rep}}
\usepackage{xcolor}
\usepackage{refcount}

\makeatletter
\renewcommand\footnotesize{%
   \@setfontsize\footnotesize\@ixpt{11}%
   \abovedisplayskip 8\p@ \@plus2\p@ \@minus4\p@
   \abovedisplayshortskip \z@ \@plus\p@
   \belowdisplayshortskip 4\p@ \@plus2\p@ \@minus2\p@
   \def\@listi{\leftmargin\leftmargini
               \topsep 4\p@ \@plus2\p@ \@minus2\p@
               \parsep 2\p@ \@plus\p@ \@minus\p@
               \itemsep \parsep}%
   \belowdisplayskip \abovedisplayskip
}
\makeatother

\usepackage[amsmath,thmmarks,hyperref]{ntheorem}
\usepackage{cleveref}
\crefformat{page}{#2page~#1#3}%
\Crefformat{page}{#2Page~#1#3}%
\crefformat{equation}{#2(#1)#3}%
\Crefformat{equation}{#2(#1)#3}%
\crefformat{figure}{#2Figure~#1#3}%
\Crefformat{figure}{#2Figure~#1#3}%
\crefformat{section}{#2Section~#1#3}
\Crefformat{section}{#2Section~#1#3}
\crefformat{chapter}{#2Chapter~#1#3}
\Crefformat{chapter}{#2Chapter~#1#3}
\crefformat{chapter*}{#2Chapter~#1#3}
\Crefformat{chapter*}{#2Chapter~#1#3}
\crefformat{part}{#2Part~#1#3}
\Crefformat{part}{#2Part~#1#3}
\crefformat{enumi}{#2(#1)#3}
\Crefformat{enumi}{#2(#1)#3}

\theoremnumbering{arabic}
\theoremstyle{plain}
\theoremsymbol{}
\theorembodyfont{\itshape}
\theoremheaderfont{\bfseries\scshape}
\theoremseparator{.}

\newtheorem{theorem}{Theorem}[section]
\crefformat{theorem}{#2Theorem~#1#3}
\Crefformat{theorem}{#2Theorem~#1#3}

\newcommand{\newtheoremwithcrefformat}[2]{%
  \newtheorem{#1}[theorem]{\textsc{#2}}%
  \crefformat{#1}{##2\MakeUppercase#1~##1##3}%
  \Crefformat{#1}{##2\MakeUppercase#1~##1##3}%
}

\newtheoremwithcrefformat{proposition}{Proposition}
\newtheoremwithcrefformat{observation}{Observation}
\newtheoremwithcrefformat{lemma}{Lemma}
\newtheoremwithcrefformat{conjecture}{Conjecture}
\newtheoremwithcrefformat{corollary}{Corollary}

\theorembodyfont{\upshape}
\newtheoremwithcrefformat{example}{Example}
\newtheoremwithcrefformat{definition}{Definition}
\newtheorem*{remark}{Remark}

\theoremsymbol{\ensuremath{\square}}
\theoremheaderfont{\scshape}
\theorembodyfont{\normalfont}
\theoremsymbol{\ensuremath{\square}}
\newtheorem*{proof}{Proof}

\theoremheaderfont{\slshape}
\theoremsymbol{}
\theorembodyfont{\itshape}
\newtheorem{claim}{Claim}[theorem]
\theorembodyfont{\normalfont}
\theoremsymbol{\ensuremath{\diamond}}
\newtheorem*{claimproof}{Proof}

\usepackage{tikz}
\usetikzlibrary{arrows}
\usetikzlibrary{shapes.geometric}
\usetikzlibrary{positioning,arrows.meta}

\title{Well-Quasi-Ordering Eulerian Digraphs Embeddable in Surfaces by Strong Immersion}

\DeclareRobustCommand{\authorthing}{
	\begin{center}
		Dario Cavallaro\thanks{\texttt{d.cavallaro@tu-berlin.de}} \\
		{\small Technical University Berlin, Germany} \\

		  \medskip
		Ken-ichi Kawarabayashi \thanks{\texttt{k\_keniti@nii.ac.jp} }\\
		{\small National Institute of Informatics, Japan}\\
            {\small The University of Tokyo, Japan}\\
		\medskip
		Stephan Kreutzer \thanks{\texttt{stephan.kreutzer@tu-berlin.de}} \\
		{\small Technical University Berlin, Germany} \\

\end{center}}
\author{\authorthing}

 \date{}
\begin{document}
\maketitle

\begin{abstract}
We prove that for every surface $\Sigma$, the class of Eulerian directed graphs that are Eulerian embeddable into $\Sigma$ (in particular they have degree at most $4$) is well-quasi-ordered by strong immersion. This result marks one of the most versatile directed graph classes (besides tournaments) for which we are aware of a positive well-quasi-ordering result regarding a well-studied graph relation. 

Our result implies that the class of bipartite circle graphs is well-quasi-ordered under the pivot-minor relation. Furthermore, 
this also yields two other interesting applications, namely, a polynomial-time algorithm for testing immersion closed properties of Eulerian-embeddable graphs into a fixed surface, and a characterisation of the Erd\H{o}s-P\'osa property for Eulerian digraphs of maximum degree four.

Further, in order to prove the mentioned result, we prove that Eulerian digraphs of carving width bounded by some constant $k$ (which correspond to Eulerian digraphs with bounded treewidth and additionally bounded degree) are well-quasi-ordered by strong immersion. We actually prove a stronger result where we allow for vertices of the Eulerian digraphs to be labeled by elements of some well-quasi-order $\Omega$. We complement these results with a proof that the class of Eulerian planar digraphs of tree-width at most $3$ is not well-quasi-ordered by strong immersion, noting that any antichain of bounded treewidth cannot have bounded degree.

\end{abstract}
\section{Introduction}

Well-quasi-orderings of mathematical objects have been studied for many decades in various mathematical as well as computer science research areas \cite{kruskal72,wqo_lang,Higman1952,wqo_survey}. In general, given a set of objects~$V$ and a \emph{partial order}~$\preceq$ which is a subset of $V \times V$, that is, a binary relation that is reflexive and transitive, we call~$(V,\preceq)$ a \emph{partially ordered set}. Then~$V$ is called \emph{well-quasi-ordered} if for every infinite sequence~$e_1,e_2,\ldots$ of objects~$e_i \in V$, there exist~$j_1 < j_2 $ such that~$e_{j_1} \preceq e_{j_2}$. 

In 1960, Kruskal \cite{Kru60} proved one of the first impactful graph-theoretic results regarding well-quasi-ordering, namely that trees are well-quasi-ordered by topological containment. (The exact statement is stronger and more complex, and so is its proof.) Later, Nash-Williams \cite{nash63} provided a simplified proof that finite trees are well-quasi-ordered by topological containment. 

In the late $20$th century, Robertson and Seymour finally gave a proof of Wagner's Conjecture \cite{GMXX}, stating that undirected graphs are well-quasi-ordered by the minor relation---$H$ is a \emph{minor} of $G$ if it can be obtained from a subgraph of $G$ by contracting edges---marking a major breakthrough in the area. The proof of this result culminated in the creation of a whole field of research: graph minor theory, laying the groundwork for more general structural graph theory. 

Besides the graph minor relation, another prominent example of a relation between graphs is the \emph{immersion} relation. We say that a graph~$H$ \emph{immerses} in a graph~$G$ if there exists a map~$\gamma : H \to \{P \mid P \subseteq G\}$ such that~$\restr{\gamma}{V(H)}\colon V(H) \to V(G)$ is injective and~$\restr{\gamma}{E(H)} \colon E(H) \to \{P \mid P \text{ is a path in } G\}$ such that~$P \coloneqq \gamma(\{u,v\})$ is a path starting in~$\gamma(u)$ and ending in~$\gamma(v)$, and such that for two distinct edges $e_1,e_2 \in E(H)$, the paths $\gamma(e),\gamma(e')$ are edge-disjoint. We call $\gamma$ an \emph{immersion}. 

The immersion $\gamma$ is \emph{strong} if no path $\gamma(e)$ contains a vertex of $\gamma(V(G))$ as an internal vertex. We write $\gamma\colon H \hookrightarrow G$ to mean that $H$ strongly immerses in $G$. 
An immersion that is not strong is called \emph{weak}.
Immersions and strong immersions for directed graphs are defined analogously, mapping directed edges to directed paths.  

Later on, still part of their seminal work on graph minors, Robertson and Seymour derived Nash-William's Conjecture from their work. That is, they  proved that undirected graphs are well-quasi-ordered by \emph{immersion} \cite{GMXXIII}, marking another breakthrough in the area. 

Their proof uses weak immersions and, as Robertson and Seymour discussed in \cite{GMXXIII}, their framework cannot easily be adapted to prove well-quasi-ordering of undirected graphs by \emph{strong} immersion. In fact they claim that a proof of the strong immersion result is way more complicated: ``\textit{It seemed to us at one time that we had a proof of the stronger [immersion conjecture], but even if it was correct it was very much more complicated, and it is unlikely that we will write it down}'' \cite[p2)]{GMXXIII}. As far as we are aware, the strong immersion conjecture is still open, and we hope that the framework, ideas, and results presented in this paper will help in proving it or at least reignite the research in the area. 

\smallskip

\noindent\textit{Directed graphs. } After proving Wagner's and Nash-Williams' Conjecture, Robertson and Seymour turned their interest to directed graphs. For, it is natural to ask `what are good relations suited for well-quasi-ordering directed graphs?'. Together with Johnson and Thomas, they introduced a notion of \emph{directed treewidth} \cite{dtw_def} seemingly suited to their needs---but unfortunately the natural extension of minors to directed graphs, which is \emph{butterfly minors}, does \emph{not} give a well-quasi-ordering on directed graphs, not even on directed graphs of low directed treewidth\footnote{Consider a growing sequence of alternating paths.}. 

It quickly turned out that directed graphs behave very differently from undirected graphs: for most well-studied containment relations on directed graphs, counterexamples to well-quasi-orderings have been found. 
As a positive result, Liu and Muzi \cite{alternatingpaths} prove that digraphs that do not contain long alternating paths are indeed well-quasi-ordered by the \emph{strong minor} relation. Unfortunately, not containing long alternating paths is a very strict restriction for directed graphs, in particular such graphs must have low directed and even low undirected treewidth \cite{alternatingpaths,KreutzerT2012}. 

Besides this, only a few positive results on well-quasi-ordering directed graphs seem to be known. One of the most prominent ones, due to Chudnovsky and Seymour \cite{tournaments}, states that \emph{tournaments}---graphs obtained from directing edges of undirected cliques---are well-quasi-ordered by strong immersion. This result was later extended to semi-complete digraphs by Barbero, Paul, and Pilipczuk.

\subsection{Our contributions}
A class of directed graphs that lies somewhat between general directed and undirected graphs is the class of \emph{Eulerian digraphs}. A digraph $G$ is called \emph{Eulerian}, if  for each vertex its in-degree is exactly the same as its out-degree. In other words, $G$ is the union of a set of pairwise edge-disjoint directed cycles.  Eulerian digraphs have many nice properties making them a particularly interesting class of digraphs to study. For instance, whereas it is known that the $2$-edge-disjoint-paths problem is NP-complete on general digraphs, the problem was
recently shown to be fixed-parameter tractable for Eulerian digraphs \cite{EDP_Euler}. 
See \cite[Chapter 4]{Bang-JensenG2018} for an introduction to Eulerian digraphs.

\subsubsection{Main results}

In this paper we are mostly concerned with the strong immersion relation. 
Our first result is negative in the sense that it yields a counter example for a well-quasi-ordering on general Eulerian digraphs. 

\begin{theorem}\label{thm:antichain}
    The class of Eulerian digraphs is not well-quasi-ordered under the strong immersion relation.
\end{theorem}
\begin{proof}
    For each $k\geq 1$ let $G_k$ be defined as follows. See \cref{fig:antichain} for an illustration.
    
    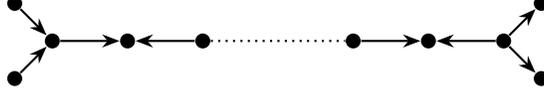
\begin{figure}
        \centering
        \begin{tikzpicture}[>=Stealth]
            \tikzstyle{vertex}=[circle, fill=black, inner sep=2pt]
            \node[vertex] (l) at (-3,0) {};
            \node[vertex] (r) at (3,0) {};
            \node[vertex] (v1) at (-2,0) {};
            \node[vertex] (v2) at (-1,0) {};
            \node[vertex] (v3) at (1,0) {};
            \node[vertex] (v4) at (2,0) {};
            \node[vertex] (l1) at (-3.5,0.5) {};
            \node[vertex] (l2) at (-3.5,-0.5) {};
            \node[vertex] (r1) at (3.5,0.5) {};
            \node[vertex] (r2) at (3.5,-0.5) {};
            \foreach \s/\t in {l2/l,l1/l,r/r1,r/r2,l/v1,v2/v1,r/v4,v3/v4}
            \draw[thick,->] (\s) to (\t) ;
            \draw[thick,dotted] (v2) to (v3);
        \end{tikzpicture}
        \caption{An infinite antichain of Eulerian digraphs with respect to strong immersions.}
        \label{fig:antichain}
    \end{figure}
    We set $V(G_k) := \{ l_1, l_2, v_0, \dots, v_{2k}, r_1, r_2 \}$ and edges from $l_1$ and $l_2$ to $v_0$, from $v_{2k}$ to $r_1, r_2$ and from $v_0$ to $v_1$ and from $v_{2k}$ to $v_{2k-1}$.
    Finally, for all $1 \leq i < k$ we add edges $(v_{2i}, v_{2i-1})$ and $(v_{2i}, v_{2i+1})$. 

    Thus, the graph $G_k$ consists of two special vertices, $v_0$ and $v_{2k}$, which are the only vertices of degree $3$ and between them a long path with edges in alternating directions. 

    We claim that if $i \not= j$ then $G_i \not\hookrightarrow G_j$. This is obvious if $i > j$, thus we assume $i < j$. But then, the two vertices $v^i_0$ and $v^i_{2i}$ of degree $3$ in $G_i$ must be mapped to the corresponding vertices $v^j_0$ and $v^j_{2j}$ of degree $3$ in $G_j$, as no other vertex in $G_j$ has degree $3$. This implies that the vertices $v^i_{1}, \dots, v^i_{2i-1}$ on the alternating path in $G_i$ must be mapped to vertices on the alternating path in $G_j$. As $j>i$ there must be two adjacent vertices $v^i_{l}$ and $v^i_{l+1}$ in $G_i$ that are mapped to two vertices $v^j_{l_1}$ and $v^j_{l_2}$ which are not adjacent in $G_j$. But then there is no directed path in $G_j$ between $v^j_{l_1}$ and $v^j_{l_2}$ and thus the edge between $v^i_{l}$ and $v^i_{l+1}$ cannot be immersed into $G_j$.

    So far the example only shows that the class of directed graphs is not well-quasi-ordered under the strong immersion relation. But the digraphs are not yet Eulerian. But we can easily make them Eulerian using a standard construction from the theory of Eulerian digraphs. 

    For $i>1$ let $G_i'$ be the digraph obtained from $G_i$ by adding a fresh vertex $b_i$ and for all  $v \in V(G_i)$ that have more incoming than outgoing edges we add sufficiently many edges from $v$ to $b_i$ so that the indegree of $v$ equals its outdegree in $G_i'$ and likewise we add edges from $b_i$ to $v \in V(G)$ if $v$ has more outgoing than incoming edges. (To avoid parallel edges one may subdivide parallel edges if need be.)

    It is easily seen that the resulting graphs $G'_i$ are Eulerian. But it is still the case that for $i < j$ there is no strong immersion of $G_i'$ in $G'_j$. For, any such immersion must map the high degree vertex $b_i$ of $G'_i$ to the corresponding vertex $b_j$ in $G'_j$ and so the remaining vertices must be mapped to each other. Thus the same argument as above shows that this is impossible.
\end{proof}

\begin{remark}
    For later reference we remark that the class $\CCC := \{ G_i' \sth i > 1 \}$ has unbounded maximum degree. Furthermore, the class $\CCC$ only is an antichain for the strong immersion relation, but not for weak immersions.  
\end{remark}

Finally we observe that the (undirected) treewidth of the graphs $G'_i$ in the proof of \cref{thm:antichain} is $\leq 3$ and the graphs admit a planar embedding.

\begin{corollary} \label{cor:antichain_tw}
    Let $k\geq3$ and let $\CCC$ be the class of Eulerian planar digraphs of treewidth $\leq k$. Then $\CCC$ is not well-quasi-ordered by strong immersion.
\end{corollary}

The previous example shows that to obtain a well-quasi-ordering of Eulerian digraphs under strong immersions we must restrict ourselves to classes of bounded maximal degree. Note here that one easily adapts the construction in \cref{thm:antichain} to provide an antichain for strong immersion of unbounded \emph{undirected} treewidth, by replacing the alternating paths with a version of ``grids'', where each path is alternating.

In this paper we focus on Eulerian digraphs with maximal degree $4$.\footnote{There are strong indications that the case of maximum degree four is fundamental, see the discussion in  \cite{Johnson2002,EDP_Euler} and later in the paper.} We call an Eulerian digraph of maximum degree four \emph{Eulerian embeddable in a surface $\Sigma$}, if there is an embedding of $G$ into  $\Sigma$, such that every vertex $v \in V(G)$ can be embedded so that the embedding of its incident edges alternates between in- and out-edges (given some orientation of $\Sigma$ locally around $v$). We call $v$ \emph{Eulerian embedded in $\Sigma$} and if $G$ admits an Eulerian embedding we call it \emph{Euler-embeddable}. If a vertex $v$ is \emph{not} Eulerian embedded---i.e., both in-edges are followed by both out-edges in the embedding---we call it \emph{strongly planar}. See \cref{fig:Eulerian-embedding} for an illustration.

\begin{figure}
    \centering
\begin{tikzpicture}[x=1cm,y=0.5cm,>=Stealth]
    \tikzstyle{vertex}=[black, fill=black, circle, inner sep=2pt]
    \begin{scope}
    \node[vertex] (u) at (0,0) {};
    \node[vertex] (i1) at (-1,1) {};
    \node[vertex] (i2) at (1,-1) {};
    \node[vertex] (o1) at (1,1) {};
    \node[vertex] (o2) at (-1,-1) {};
    \node[anchor=south] at (u) { $u$ };
    \node[anchor=east] at (i1) { $i_1$ };
    \node[anchor=west] at (i2) { $i_2$ };
    \node[anchor=west] at (o1) { $o_1$ };
    \node[anchor=east] at (o2) { $o_1$ };
    \draw[->] (i1) to (u) ;
    \draw[->] (i2) to (u) ;
    \draw[->] (u) to (o1) ;
    \draw[->] (u) to (o2) ;
    \node at (0,-2) { Eulerian embedding};
    \end{scope}
    \begin{scope}[xshift=5cm]
    \node[vertex] (u) at (0,0) {};
    \node[vertex] (i1) at (-1,1) {};
    \node[vertex] (o2) at (1,-1) {};
    \node[vertex] (o1) at (1,1) {};
    \node[vertex] (i2) at (-1,-1) {};
    \node[anchor=south] at (u) { $u$ };
    \node[anchor=east] at (i1) { $i_1$ };
    \node[anchor=east] at (i2) { $i_2$ };
    \node[anchor=west] at (o1) { $o_1$ };
    \node[anchor=west] at (o2) { $o_2$ };
    \draw[->] (i1) to (u) ;
    \draw[->] (i2) to (u) ;
    \draw[->] (u) to (o1) ;
    \draw[->] (u) to (o2) ;
    \node at (0,-2) { strongly planar embedding};
    \end{scope}
\end{tikzpicture}
    \caption{Euerlian and strongly planar embeddings. If the embedding of $u$ is strongly planar, two edge-disjoint paths can cross at $u$, i.e., there are two edge-disjoint paths going from $i_1$ to $o_2$ and from $i_2$ to $i_1$. This is impossible in a Eulerian embedding.}
    \label{fig:Eulerian-embedding}
\end{figure}
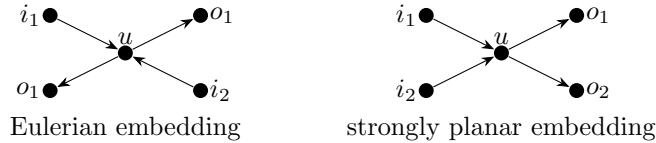

It is easy to see that the class of Eulerian embeddable 
graphs in a fixed surface $\Sigma$ is closed under (strong) 
immersion, and obviously if a graph $H$ immerses into an 
Eulerian embeddable graph $G$ then $H$ is Eulerian embeddable 
(see \cref{def:Euler-embedding} and \cref{sec:surface_defs} 
for a rigorous definition and properties of Euler-embeddable 
digraphs). 

The following is our first main result.

\begin{theorem}\label{thm:main1}
 Let $\Sigma$ be a surface (possibly with boundary). The class of all Eulerian digraphs of maximum degree four that can be Eulerian embedded in $\Sigma$ is well-quasi-ordered by the strong immersion relation.
\end{theorem}

To the best of our knowledge, this result marks one of the most versatile directed graph classes (besides tournaments) for which we are aware of a positive well-quasi-ordering result regarding a well-studied graph relation. Note that the assumption on the graphs being \emph{Euler-embeddable} is crucial, in particular simple planarity does not suffice as demonstrated by \cref{cor:antichain_tw}. Note here that regarding strong immersion of Eulerian digraphs, Eulerian embeddings sort of correspond to planar embeddings regarding the minor relation on undirected graphs. In particular, Eulerian embeddings are the ``cross-free'' embeddings for Eulerian digraphs with respect to immersion, and there is an immersion-based structure theorem for Eulerian digraphs of maximum degree four \cite{Johnson2002,EDP_Euler} very similar to the minor-based structure theorem for undirected graphs due to Robertson and Seymour. In particular, the relevant embeddings for the Eulerian immersion-based structure theorem are Eulerian embeddings. Finally, note that a vertex of degree at least six can be used to ``cross'' paths edge-disjointly, and thus in a sense higher degree vertices help in finding strong immersions. Similarly, degree four vertices in undirected graphs allow to cross paths edge-disjointly, whence Euler-embeddable digraphs of degree at most four behave very similar to embeddable undirected cubic graphs. The \cref{thm:main1} also has several interesting applications which we explain below.

The second main result of this paper is the following theorem, which essentially proves that if in addition to the degree we bound the maximum treewidth of our class, then it is well-quasi-ordered by strong immersion.

\begin{theorem}\label{thm:main2}
     For each fixed $k\geq 0$, the class of Eulerian digraphs of carving width at most $k$, or equivalently of maximum degree and (directed) treewidth at most $k$, is well-quasi-ordered by strong immersion.
\end{theorem}

Carving width---denoted by $\ebw{G}$---is a width parameter similar to the notion of treewidth---denoted by $\tw{G}$---but tailored towards edge cuts rather than vertex separations. 

The class of Eulerian digraphs in \cref{thm:main2} is more general than in our first main result in that we no longer require the maximum degree to be four. Note, however, that bounded carving width implies bounded degree. In a sense, \cref{thm:main2} is best possible: To see this note that it is well-known that for an Eulerian digraph $G$ it holds $\frac{2}{3} \tw{G} \leq \ebw{G} \leq \Delta(G) \cdot \tw{G}$, where $\Delta(G)$ is the maximum degree and $\tw{G}$ is the treewidth of the underlying undirected graph. Thus, restricting to a class of Eulerian digraphs of bounded degree, carving width is qualitatively equivalent to undirected treewidth. By \cref{cor:antichain_tw} there is an antichain for strong immersion well-quasi-ordering of Eulerian digraphs admitting bounded treewidth. \cref{thm:main2} implies that such an antichain must have unbounded degree. 

\smallskip
We believe that these results mark an important step towards proving the more general statement that Eulerian directed graphs (of maximum degree four or more generally bounded degree) are well-quasi-ordered by strong immersion. This is in line with proofs by other researchers for many different width parameters on different graph classes or related objects. For example, Geelen, Gerards, and Whittle proved that binary matroids of bounded branch-width are well-quasi-ordered by the matroid minor relation \cite{Gee02}. Oum proved that graphs of bounded rank-width are well-quasi-ordered by the vertex-minor relation \cite{wqo_rankwidth}. Kanté generalized the result due to Oum, which as a byproduct yields that directed graphs of bounded rankwidth are well-quasi-ordered by pivot-minors \cite{directed_pivotminor}. And of course the bounded treewidth case was an important first step towards a proof of Wagner's Conjecture \cite{GMIV}.

Before we discuss applications of our results, let us mention the following conjectures.
Our current results only apply to Eulerian digraphs Euler-embeddable into a fixed surface. However, we believe that, using the immersion-based structure theorem for Eulerian digraphs \cite{Johnson2002,EDP_Euler}, we can extend this to the class of all $4$-regular Eulerian digraphs---where \cref{thm:main1} marks an important base case---and therefore conjecture the following.

\begin{conjecture}\label{conj:wqo_4reg}
    The class of~$4$-regular Eulerian digraphs is well-quasi-ordered by (strong) immersion.
\end{conjecture}

We also conjecture that this can be extended to any fixed maximum degree $d \in 2\N$.

\begin{conjecture}\label{conj:wqo_dreg}
  For every $d \geq 1$, the class of Eulerian digraphs of maximum degree $2d$ is well-quasi-ordered by strong immersion.
\end{conjecture}

\Cref{thm:antichain} shows that there is no hope to extend these conjectures to the class of all Eulerian digraphs. However, we believe that this is an artefact of the \emph{strong} immersion relation. As the construction in \cref{thm:antichain} forces us to map the high degree vertices to each other, and thus make them unusable for routing other edges in a strong immersion, strongly immersing the remaining digraphs is equivalent to strong immersion in general digraphs. For weak immersions this is not the case and the high degree vertex in the target digraph can still be used to route further edges. We believe that this is not just an artefact of our construction but is a fundamental difference and we therefore conjecture the following. 

\begin{conjecture}\label{conj:wqo_gen}
    The class of Eulerian digraphs is well-quasi-ordered by weak immersion.
\end{conjecture}

\subsubsection{Applications}

Our two main results \Cref{thm:main1,thm:main2} yield several interesting applications that we discuss now.

\subsubsection*{A polynomial time algorithm for immersion closed families of bounded genus Eulerian digraphs of maximum degree four}

Let $\Sigma$ be a surface and let $\EEE(\Sigma)$ denote the class of all Eulerian digraphs of maximum degree $4$ which are Eulerian embeddable into $\Sigma$. 
A direct consequence of \cref{thm:main1} together with the proof in \cite{EDP_Euler} is that any class $\CCC \subseteq \EEE(\Sigma)$ that is closed under immersions can be decided in polynomial time. For, by \cref{thm:main1},  there is a finite set $\{ H_1, \dots, H_k \} \subseteq \EEE(\Sigma)$ such that a digraph $G \in \EEE(\Sigma)$ is contained in $\CCC$ if, and only if, $H_i$ does not strongly immerse into $G$, for all $1 \leq i \leq k$. As shown in \cite{EDP_Euler}, the edge-disjoint paths problem is fixed-parameter tractable. Thus for each $H_i$ we can test in polynomial time whether $H_i \hookrightarrow G$. Thus, we have shown the following theorem. 

\begin{theorem}\label{thm:polytime}
    Every immersion closed property of bounded genus Eulerian-embeddable digraphs of maximum degree four can be decided in polynomial time. 
\end{theorem}

This is similar to the way the graph minor algorithm for finding a fixed minor~$H$ by Robertson and Seymour, together with the well-quasi-ordering theorem by the standard minor-relation (Wagner's Conjecture), implies that every minor closed property of undirected graphs can be tested in polynomial time \cite{GMXIII}\footnote{The run-time of their algorithm has first been improved to a quadratic run-time by Kawarabayashi, Kobayashi, and Reed \cite{Ken_minor_testing, KAWARABAYASHI2012424}, and more recently to almost linear time by Korhonen, Pilipczuk, and Stamoulis \cite{Giannos_almost_lin})}. 

However, unlike the minor testing algorithm, our algorithm for deciding immersion closed classes is currently only in XP, i.e. the degree of the polynomial bounding the runtime depends on the set $H_1, \dots, H_k$. 
It is conceivable that this can be improved to fixed-parameter tractability, i.e. to a runtime of the form $f(H_1, \dots, H_k) \cdot n^c$, for some fixed constant $c$ and some function $f$. We leave this for future research.

\subsubsection*{Erd\H{o}s-P\'osa property for Eulerian digraphs of maximum degree four}

A family $\FFF$ of graphs is said to satisfy the Erd\H{o}s-P\'osa property if there exists a function $f$ such that for every positive integer $k$, every graph $G$ either contains  $k$ vertex (or edge)-disjoint subgraphs in $\FFF$ or a set of at most $f(k)$ vertices (or edges) intersecting every subgraph of $\FFF$ in $G$. This property does not always hold; 
Lov\'asz and Schrijver observed that the family of odd cycles does not
satisfy the Erd\H{o}s-P\'osa property, due to certain projective planar grids. 

Robertson and Seymour \cite{GMV} characterised the classes of undirected graphs admitting the Erd\H{o}s-P\'osa property in a very general form: they show that, for a fixed graph $H$, the family of graphs $\FFF_H$ that contain $H$ as a minor satisfies the Erd\H{o}s-P\'osa property if and only $H$ is planar. 

Using \cref{thm:main1}, we can show the following general result, which can be regarded as an analogue of Robertson and Seymour's result. Informally, a \emph{$k$-swirl} $\SSS_k$ is a $4$-regular Eulerian digraph that is obtained from an undirected $k\times k$-grid by replacing each vertex by a $4$-circle such that the circles share a common vertex if and only if the respective vertices are adjacent in the grid. Then, given the canonical plane drawing, direct all circles in the same direction, resulting in an Eulerian embedding (see \cref{def:swirl} for a formal definition and \cref{fig:swirl} for an illustration).  

\begin{theorem} 
For an Eulerian digraph $H$ of maximum degree four, 
the family $\EEE_H$ of Eulerian digraphs of maximum degree four that contain $H$ as an immersion satisfies the Erd\H{o}s-P\'osa property if and only if $H$ can be immersed into some swirl $S$, or equivalently, $H$ can be Eulerian embedded in the plane. 
\end{theorem}
\begin{proof}[Sketch]
The proof is very similar to that in \cite{GMV}, and very standard in this area. So we just give a sketch of the proof. 

Assume that a swirl $S_{h}$ contains $H$ as an immersion, for some $h \geq 1$. In this case, if we take a swirl $S_{k\cdot h}$, then clearly there are $k$ edge-disjoint immersions of $H$ in $S_{k\cdot h}$. By \cref{thm:swirl_embedds_Eulerembeddable_graphs}, if the (undirected) treewidth of $G$ is at least $f(k, h)$, then $G$ admits an immersion of a swirl $S_{k\cdot h}$, and hence there are $k$ edge-disjoint immersions of $H$ in $G$. 
On the other hand, if the treewidth of $G$ is at most $f(k, h)$,  then the Erd\H{o}s-P\'osa property also holds, see \cite{GMV}. 

Assume now that there is no $h \in \N$ such that $S_{h}$  contains $H$ as an immersion (in other words, using \cref{thm:swirl_embedds_Eulerembeddable_graphs}, $H$ is not Eulerian embeddable in the sphere).
Let $\Sigma$ be a surface of minimum genus into which $H$ has a Eulerian embedding. By assumption, $\Sigma$ is not the sphere. In the Eulerian embedding of $H$ into $\Sigma$ we now replace every vertex by a large swirl and duplicate the edges of $H$ accordingly. Let $G$ be the resulting Eulerian digraph which is Eulerian embedded into $\Sigma$. Then it is not hard to see that $G$ does not immerse two disjoint copies of $H$. For, any immersion of $H$ in $G$ induces a drawing of the image of $H$ on $\Sigma$ that contains a genus-reducing cycle $C$. As the embedding of $G$ is Eulerian, if $H'$ is the image of an immersion of $H$ in $G$ which is edge-disjoint from the first immersion of $H$, then no path in $H'$ can cross the cycle $C$. Thus implies that $H'$ would be a Eulerian embedding of $H$ into a surface of lower genus which is impossible by assumption. Thus, $G$ has no two edge-disjoint immersions of $H$. But clearly any hitting set for the immersions of $H$ in $G$ has to be as big as the size of the swirls we used to replace the vertices of $H$, as a hitting set must separate the sides of at least one swirl.
\end{proof}

This is interesting in its own right, since the Eulerian digraphs admitting the Erd\H{o}s-P\'osa property have a rather strong restriction on their embedding: not only do they need to be planar but every vertex needs to be Eulerian embedded. Thus, being planar is not enough for a $4$-regular Eulerian digraph to have the Erd\H{o}s-P\'osa property, it needs to admit a Eulerian embedding into the plane.

\subsubsection*{Well-Quasi-Ordering Circle Graphs by the Pivot-Minor Relation}

A highly active research area is the study of vertex-minor and pivot-minor closed classes of undirected graphs. In this context the class of circle graphs is of special importance.
Given a circle graph one can construct an undirected Eulerian graph such that the vertex-minor operation corresponds to the immersion relation on the Eulerian graph (see \cite{rose_vertex_minors} for details). 

A similar analogy can be drawn for the pivot-minor relation. Here, with every circle graph a \emph{directed} Eulerian $4$-regular graph is constructed such that the pivot-minor relation on the circle graph corresponds to the immersion relation on the Eulerian digraph. 
See \cite{Oum09} or \cite{circlegraph_survey,rose_vertex_minors} for a discussion of this connection.

A direct consequence of this connection is that a proof of \cref{conj:wqo_4reg} would imply that circle graphs are well-quasi-ordered by pivot-minors, which, to the best of our knowledge, is a longstanding open problem. 

A well-studied subclass of circle graphs are bipartite circle graphs. It is well-known in the community working on pivot-minors that in the translation from circle graphs to Eulerian digraphs the class of bipartite circle graphs corresponds to the class of planar Eulerian $4$-regular digraphs with a Eulerian embedding into the plane. Thus we can apply \cref{thm:main1} to obtain  the following result.

\begin{corollary}\label{cor:pivot}
    The class of bipartite circle graphs is well-quasi-ordered under the pivot-minor relation. 
\end{corollary}
Note that, as discussed in \cite[Section 8]{circlegraph_survey}, \cref{cor:pivot} is a consequence of the announced result that ``binary
matroids are well-quasi-ordered by the minor relation'', but as stated by the authors of \cite{circlegraph_survey} the proof is still being written up.

We emphasize that \cref{thm:main1} implies a stronger result than \cref{cor:pivot}: For any fixed surface $\Sigma$, the class of circle graphs with a respective Eulerian digraph Eulerian embeddable in $\Sigma$ is well-quasi-ordered by the pivot-minor relation. It is however not clear to us apart from $\Sigma$ being the sphere which class of circle graphs this exactly encodes but it may be of worth some further investigation.

\subsubsection{Strong immersions in undirected graphs}

As explained above, the question whether the class of undirected graphs is well-quasi-ordered by the \textit{strong} immersion relation is a long-standing open problem. 

While there are several simple reductions from undirected graphs to Eulerian digraphs, \cref{thm:antichain} shows that, unless a construction can be found that also bounds the maximal degree of the resulting digraphs, the strong immersion problem for undirected graphs cannot be reduced to the strong immersion problem for Eulerian digraphs (and we are not aware of a way to to prove the converse).

However, several of the tools we develop in this paper and also in our work towards proving \cref{conj:wqo_dreg} can also be applied or adapted to work for strong immersions in undirected graphs.
This makes us optimistic that eventually it can be shown that the class of undirected graphs is well-quasi-ordered under the strong immersion relation. But this is a massive project that will require much more work and we leave it to future research.

\subsection{Related Work}

In his thesis \cite{Johnson2002}, Thor Johnson studied Eulerian digraphs and their Eulerian embeddings into surfaces. The goal of his thesis is to develop a structure theorem for internally $6$-connected Eulerian digraphs of maximum degree four with respect to immersions\footnote{Immersions are defined slightly different in his thesis, namely the immersion-map need not be injective on vertices.}. Naturally, this work is closely related to our work, and some of the theorems we prove are very similar to theorems in \cite{Johnson2002}. For example, Johnson proved that Eulerian digraphs Euler-embeddable into a surface $\Sigma$ can be immersed into any Eulerian digraph Euler-embeddable in $\Sigma$ with sufficiently high representativity. Our theorem \cref{thm:high-rep} is essentially the same result, but as discussed in the paragraph before \cref{thm:high-rep}, we need a refined version that applies to \emph{rooted} Eulerian digraphs Eulerian embedded in surfaces possibly with boundary. Note however, that this is also the ``easiest'' case, as it discusses Eulerian graphs that can be embedded into a surface such that the drawing admits ``high representativity'': this case can be reduced to the undirected case using results from \cite{GMVII} (see \cref{sec:high-rep}). As far as we are aware, none of the ``low representativity'' cases have been dealt with (such as the bounded treewidth case for example), which are the cases for which we need to deploy a lot of machinery and establish a new framework suited for strong immersions. We will point out and discuss the relation to Johnson's thesis whenever we prove a related result. 

Furthermore, the results of his thesis have unfortunately never been published in a peer-reviewed journal or conference, and his thesis can be difficult to obtain. We therefore occasionally give proofs for results that also appear in his thesis or can be derived from results in his thesis, so that the paper is self-contained.
At the end of his thesis, Johnson explains that he hopes to use his results to prove well-quasi-ordering for $4$-regular Eulerian digraphs in the future, but no further steps in this direction have appeared since then. 
In this paper, we prove Johnson's conjecture for Eulerian digraphs Euler-embeddable on a fixed surface.

\section{High level overview of the proof}

At a high level, our line of argumentation follows the ideas developed by Robertson and Seymour in their seminal work on graph minors proving Wagner's Conjecture \cite{GMIV,GMVII,GMVIII,GMXVIII,GMXIX,GMXX}.
However, the details vary drastically in many ways. For one, we work with strong immersions which, as we discussed in the introduction, have very different properties than minors. Moreover, directed graphs behave very differently in many ways. As explained in the introduction, many ``natural'' results for undirected graphs---such as structure theorems and well-quasi-ordering by a suitable minor concept---fail when dealing with directed graphs. This is the case even if they admit low treewidth or low directed treewidth. Fortunately, Eulerian digraphs (of bounded degree) turn out to be a ``tameable'' and yet  structurally rich class of directed graphs, opening up a world between undirected and directed graphs. 

\smallskip

Let us briefly describe the high-level line of reasoning before we dive into a short summary of the respective sections. 
Given a sequence $(G_i)_{i \in \N}$ of (Eulerian-embeddable directed) graphs of maximum degree $4$ into a fixed surface $\Sigma$, we filter the sequence for an infinite subsequence of digraphs of ``similar type''---which we clarify as we go---and prove for each of the types that they are well-quasi-ordered under the strong immersion relation by induction on the Euler genus of $\Sigma$ and the structure of the graphs of the sequence. The first type are graphs with embeddings of ``high representativity''; intuitively this means that  their drawings use up the whole surface with strictly growing treewidth, see \cref{def:representativity} for details. This is the ``easiest'' case (at least in our setting): If $\Sigma$ is the sphere, this boils down to admitting high carving width (recall that on bounded degree graphs carving width is qualitatively equivalent to treewidth). This means that, fixing $G_1$ and $n \coloneqq \Abs{V(G_1)}$, there is $j \in \N$ such that $G_j$ has large carving width with respect to $n$. In that case, large carving width certifies the existence of a large strongly immersed swirl in $G_j$ using \cref{thm:4_reg_swirl}. It is not hard to see, and can be derived from the equivalent undirected theorem for planar graphs, that in this case $G_1 \hookrightarrow G_j$ (see \cref{thm:swirl_embedds_Eulerembeddable_graphs} for details).

Now, let us consider the case where $\Sigma$ is not the sphere. Roughly speaking, a drawing on a surface $\Sigma$ without boundary that is \emph{not} the sphere has representativity $f(n)$, for some function $f\colon\N \to \N$, if every curve $\gamma$ drawn on $\Sigma$ that only intersects edges\footnote{In the undirected setting, representativity is usually defined via curves intersecting vertices; see \cite{MoharT2001}.} either bounds a disc or must intersect at least $f(n)$ edges. That is, there are no ``short genus-reducing curves''. It turns out (as we prove as a special case of \cref{thm:high-rep} in \cref{sec:high-rep}) that if there is $j \in \N$ such that $G_j$ has representativity much larger than $\Abs{V(G_1)}$, then $G_1 \hookrightarrow G_j$. We derive this result from the analogous  result for undirected graphs in \cite{GMVII} by analysing the homotopy types of the respective paths of the underlying undirected graph: for a Eulerian embedded graph, if two undirected vertex-disjoint paths with common endpoints $u,v$ of the underlying undirected graph bound a disc $\Delta \subseteq \Sigma$ in the surface, then there exist two (not necessarily edge-disjoint) directed paths in the same disc $\Delta$ with endpoints $u,v$, one starting in $u$ and ending in $v$ and one starting in $v$ and ending in $u$. As explained above, this specific high representativity case was already proven by Johnson in his thesis \cite{Johnson2002}.  

Thus, it remains to prove \Cref{thm:main1,thm:main2} in the case that there is a function $f:\N \to \N$ (solely depending on $\Sigma$) such that for every graph in the sequence, the drawing has representativity bounded by $f(n)$, where $n = |V(G_1)|$. This is by far the most difficult case of our proof. 

\paragraph{Cutting the surface along edges and rooting Graphs.}

Since the general proof is by induction on the genus of the surface, we will extensively cut the surface and the respective drawings. Note that cutting surfaces results in new surfaces with non-empty boundary, which is why we will work with Eulerian embeddings in surfaces possibly with boundary. Recall that (after a suitable homeomorphism) the boundary $\bd(\Sigma)$ of a surface $\Sigma$ falls into disjoint components  $\zeta_1,\ldots,\zeta_k$ called \emph{cuffs}, each homeomorphic to a circle. Since our graphs and embeddings are of a very specific type---the graphs and embeddings are Eulerian---we need to be careful about the way we choose to cut the surface and embeddings, for we need the resulting graphs (and embeddings) to remain Eulerian and Eulerian embedded (see \Cref{fig:cuff_conn_cut,fig:double-torus-cut,fig:cuff-based} for examples). Since we are analysing strong immersions, it seems counter-intuitive to cut our drawings at vertices, for: 
\begin{itemize}
    \item[(1)]  it is not guaranteed that the graph remains Eulerian (nor Eulerian embedded) after splitting vertices accordingly, which would result in further restrictions on how the curves may interact with vertices, and
    \item[(2)] if we want to ``knit'' strong immersions of cut-out pieces back together to conclude the induction, it would be much harder to keep track of which vertices have already been ``used'' when duplicating vertices (recall that for strong immersion no vertex $\gamma(v)$ is allowed to be an internal vertex of a path $\gamma(e)$).
\end{itemize}

These are the two most obvious of many reasons why we choose to cut Eulerian embedded graphs along edges instead of vertices as usual in the literature. Note here that while this seems conceptually equivalent to ``switching to the line-graph'', line-graphs of Eulerian digraphs of maximum degree $4$ do in general not result in Eulerian digraphs (due to $2$-cuts, degree-$2$ vertices and loops) and moreover the respective drawings would not be Eulerian embedded (and it gets only worse for higher degree, which we hope to tackle in the future using the presented framework). There are further issues when trying to switch to the line-graph, we will, however, not discuss these here. 

Thus, instead of switching to the conceptually similar line-graphs, we simply work with edges as independent objects: graphs in this paper will be viewed as incidence structures $G=(V,E,\operatorname{inc})$ where every directed edge has exactly two incidences---$\tail$ and $\head$---and when we \emph{cut} an edge $e \in E(G)$ we mean that we replace it by two newly introduced edges $e_1,e_2$ one keeping the incidence with $\head(e)$ and the other with $\tail(e)$. While this may seem counter-intuitive at first, let us highlight some beautiful aspects of it. Given a graph $G=(V,E,\operatorname{inc})$, let $X\subseteq V$. Using the above notation $e \in \rho(X)$ if and only if $\head(e) \in X \iff \tail(e) \in \bar{X}\coloneqq V(G) \setminus X$; note that $\rho(X) = \rho(\bar X)$. It turns out that Eulerian digraphs admit \emph{Eulerian cuts}, i.e., $\Abs{\rho(X)} \in 2\N$ and $\rho(X)$ contains equally many edges with head in $X$ as with tail in $X$. Thus, if we find a genus reducing curve $\eta:[0,1] \to \Sigma$, separating the surface into two smaller surfaces $\Sigma_1,\Sigma_2$ say, and only cutting ``nicely'' through edges---not cutting an edge several times or simply touching it---then $\Sigma_1,\Sigma_2\subset \Sigma$ yield a natural partition of $V(G)$ into $X,\bar X$ respectively, where $X$ is the set of vertices drawn on $\Sigma_1$ say. In particular, $\eta$ \emph{traces} a natural induced cut in $G$, and $\eta([0,1])$ represents a new cuff in $\Sigma_1$ and $\Sigma_2$ after cutting $\Sigma$.

\paragraph{Stitching.} Given such a set $X\subset V(G)$ we define the \emph{up- and down-stitches} of $X$ as the graphs $G^X$ and $G_X$ where $G^X$ is the graph obtained by contracting $X$ into a single vertex $x^*$ and similarly $G_X$ is the graph obtained from $G$ by contracting $\bar{X}$ into a single vertex $x_*$ (we call these new vertices \emph{beads}). The resulting graphs are again Eulerian (since the cuts are Eulerian) and they satisfy $G^X \cap G_X = \rho(X)\subset E(G)$, and in particular $V(G^X) \cap V(G_X) = \emptyset$. In this way, we may decompose our graphs along cuts traced by separating curves by stitching the induced cuts that partition the vertex sets. The idea is then to guarantee that given graphs $G, G'$ with induced cuts $\rho(X),\rho(Y)$ respectively, with strong immersions $\gamma_1:G_X \hookrightarrow G_Y$ and $\gamma_2:G^X \hookrightarrow G^Y$, we can ``knit'' the respective strong immersions back together, yielding a strong immersion $\gamma: G \hookrightarrow G'$ (where, since the vertices are partitioned, it is easier to keep it strong). Note, however, that $x^*,x_*,y^*,y_*$ may be of much larger degree than $4$---in fact of degree $f(\Abs{V(G_1)})$ for example. This comes with up- and down-sides: on one hand, if all other vertices are of degree $\leq 4$, we are guaranteed that $\gamma_1(x^*)=y^*$ and $\gamma_2(x_*) = y_*$ by definition of strong immersions, which greatly helps with knitting the immersions as this guarantees that $\rho(X)$ is mapped to paths ending in $\rho(Y)$. However, on the other hand, we leave the class of Eulerian digraphs of maximum degree $4$, and, most crucially, the resulting graphs are no longer Eulerian embeddable on \emph{any} surface. Fortunately, given a surface of finite genus, there is only a finite number of times one can cut along separating genus-reducing curves until one ends up in a surface homeomorphic to a disc (see \cref{fig:double-torus-cut} for a schematic illustration of an exemplary cut). Hence there are not too many non-homeomorphic ``short'' separating genus reducing curves, whence we may ``sacrifice'' a few vertices of higher degree, and switch to Eulerian digraphs \emph{rooted in a bounded number of beads}, essentially \emph{rooting the graphs in the cuts $\rho(X),\rho(Y)$}. It is a natural idea to embed $G_X$ and $G^X$ in $\Sigma_1$ and $\Sigma_2$ respectively, where we may keep the embedding of $G_X\cap G$ in $\Sigma_1$---every vertex is Eulerian embedded---and simply assign the bead $x_*$ to the new cuff $\eta([0,1])$, drawing its edges on the cuff, defining an \emph{Eulerian embedding of $G_X$ up to beads}. More generally, think of beads as placeholder vertices representing cuts we made, then an Eulerian digraph rooted in $t$ beads is Eulerian embeddable in a surface $\Sigma$ if we can assign each bead to a cuff of $\Sigma$ and Eulerian embedd all non-bead vertices in $\Sigma$ as well as their incidences with bead vertices by drawing edges adjacent to bead-vertices on the cuffs. 
Unfortunately, as tempting as this sounds, this comes with many obstacles, corner cases, and technicalities that require a lot of work to overcome. 
\smallskip

First note that given the above sequence, it may happen that the only curves contradicting high representativity are \emph{not} separating; whence cutting along such a curve we do not end up with natural partitions $\Sigma_1,\Sigma_2$---which come with an induced cut $X$ as described above---but rather a new connected surface $\Sigma'$ of lower genus and an additional cuff $\zeta$. In that case, we cannot use the above stitching technique and need to carefully duplicate the edges to keep the graph Eulerian and introduce a new bead ``drawn in'' $\zeta$ in a different way, while keeping the graph Eulerian embeddable up to beads (see \cref{fig:cuff_conn_cut} for a schematic illustration of an exemplary construction). 

Secondly, note that while it is in theory true that the immersions of the parts introduced above satisfy $\gamma_1(x_*) = y_*$ and $\gamma_2(x^*)=y^*$, thus mapping $\rho(X)$ to paths ending in $\rho(Y)$, it may nevertheless happen that we cannot ``knit'' both immersions together because the two immersions $\gamma_1,\gamma_2$ may map the same edge $e \in \rho(X) \cap \rho(\bar X)$ to two paths $\gamma_1(e), \gamma_2(e)$ starting or ending in different edges of $\rho(F)$. Thus, in addition to the cut-edges, we need to fix an ordering on the cuts, and make sure that the immersions keep track of that order. 

Thirdly---and now the above idea unfortunately starts to show its flaws---since we cut our drawings along curves inductively, we may need to cut it several times. While it is still fairly easy to derive the high representativity result for embeddings with root-edges drawn on the boundary (see \cref{thm:high-rep} and note that here our result differs from the one provided by Johnson in \cite{Johnson2002} as the embeddings considered there are in surfaces without boundaries), we need to deal with a different notion of ``high representativity''. In essence, we may find short curves that ``cut off cuffs'', whence cutting $\Sigma$ along such a curve, we produce a cylinder and a curve $\Sigma' \cong \Sigma$ making no progress on the genus (but fortunately on the number of root edges); this is why the cylinder case needs to be proved separately. Worse, however, we may need to cut along a curve $\iota$ that starts and ends in the same cuff $\zeta$, say. Note that a priori it is not clear why cutting along such a curve results in an Eulerian digraph and, more importantly, in an Euler-embeddable graph, especially not if we still have beads hidden in holes that are connected to the cuffs. Fortunately, this can be guaranteed as we discuss in \cref{sec:surface_defs} and \cref{sec:cylinder}; see \cref{fig:cuff-based} for a schematic illustration of the respective case.

But even if it does result in an Eulerian digraph, and even if the $\iota$ is otherwise nice and separating in $\Sigma$ producing surfaces $\Sigma_1,\Sigma_2$ both of lower genus with respective vertex sets $X,\bar{X}$ where $X$ is the the set of vertices drawn on $\Sigma_1$, the following may happen: Since old beads of the cuff $\zeta$ are not drawn in $\Sigma$, they may not naturally be part of $X$ or $\bar{X}$ for they may have ends in \emph{both} surfaces $\Sigma_1,\Sigma_2$ and thus edges to vertices in $X$ and $\bar X$. Thus, when cutting along $\iota$, the inductive reasoning starts to get harder as we do not produce ``fully disconnected'' pieces (there are old beads connecting both surfaces). We will circumvent these problems by loosening the assumptions in \cref{sec:edge-rooted} and switching to digraphs that are ``Eulerian up to a few root edges'', i.e., we allow for $k$ distinguished degree-one vertices in $G$. Since we will need the bead-rooted versions for future work, and since it is ``conceptually clearer'', we will provide the first part of the paper in that setting, especially since it helps with the low treewidth case as we discuss later. 

\paragraph{The Disc and Cylinder Case.}
After some skilled cuttings, we will end up with Eulerian digraphs rooted in a bounded (in the genus of the original $\Sigma$) number of edges with Eulerian embeddings in the disc or in the cylinder, where the root edges are drawn on cuffs (think of beads being ``drawn'' in the respective holes to a cuff, not part of the rest of the drawing). Let us briefly discuss both cases, starting with the cylinder.
\smallskip

\textit{The Cylinder.} Let $\Theta$ denote a cylinder with two distinct cuffs $\zeta_1,\zeta_2$. Given a sequence $(G^i,E_1^i,E_2^i)$ of Eulerian digraphs $G_i$, and sets of root edges $E_1^i, E_2^i$ such that $G_i$ can be Eulerian embedded in $\Theta$ such that the edges $E_1^i$ are drawn on $\zeta_1$ and $E_2^i$ on $\zeta_2$ with one end in $\Theta$ and the other drawn somewhere in the hole, for every $i\in \N$; call the drawings $\Gamma_i$. By filtering for a respective subsequence we may assume that $ \Abs{E_p^i} =k_p = \Abs{E_p^j}$ for respective $k_p \in 2\N$  and $p=1,2$ and every $i,j \in \N$. (Note here that $k_p \in 2\N$ comes from the fact that our graphs are Eulerian embedded, and the embedding restrictions are crucial for this to be true). If there exists some $j>1$ such that $\Gamma_j$ has ``representativity'' (we will make this precise shortly) much larger than $\Abs{V(G_1)}$, then we prove in \cref{lem:cyclinder_high_rep} that $G_i \hookrightarrow G_j$ by mapping roots to roots respecting some priorly fixed ordering of the roots (we will tacitly omit the orderings to not overcomplicate this short summary). Here, by high representativity, we mean that: 
\begin{itemize}
    \item[1.] there is no curve cutting $\leq \theta(n)$ edges in $\Gamma_j$ that connects both cuffs, and 
    \item[2.] there is no simple closed curve winding around $\zeta_p$ that disconnects it from the other cuff and cuts less than $k_p$ edges, and 
    \item[3.] $\Gamma_j$ has treewidth at least $\theta(n)$, and immerses a large swirl $\SSS_{f(n)}$ which cannot be cut off from $E_p^j$ by a cut smaller than $k_p$.
\end{itemize}   The high representativity case can be reduced to the case where $k_1=k_2= 0$ by introducing respective gadgets for the root edges, and using the fact that every Euler-embeddable graph embeds in a large enough swirl, i.e., $\SSS_{f(n)}$, and finally routing the root edges (which are part of the gadgets) accordingly using 3 above. Henceforth, we may again assume that the embeddings in our sequence admit low representativity. 

Given embeddings of low representativity the idea is to ``chop up'' the cylinder (and embedding) along curves $F_i$ that wind around $\zeta_1$---call these \emph{cut-cycles}---resulting in two cylinders $\Theta_i^1,\Theta_i^2 \subset \Theta$, the former admitting $\zeta_1,F_i$ as new cuffs, and the latter admitting $F_i,\zeta_2$ as new cuffs. We choose the cut-cycles such that for $i<j$, $\Theta_i^1 \subset \Theta_j^1$ and $\Theta_j^2 \subset \Theta_i^2$, producing a \emph{laminar} sequence of cut-cycles. Using the fact that $G$ is Eulerian embedded we derive a form of Mengers \cref{lem:boundary_linked_Menger_for_embeddings_in_cylinder} for the the cylinder, proving that either there are $t$ paths in the cylinder $\Theta_j^2 \setminus \Theta_i^1$ that connect the new root edges lying on the cuffs $F_i,F_j$, or we find a cut-cycle $F^*$ cutting $<t$ edges in the induced drawing embedded in the part bounded by the cylinder $\Theta_j^2 \setminus \Theta_i^1$; in particular $F^*$ is a cut-cycle in the original cylinder $\Theta$ cutting the same edges. This way we get a ``linked'' decomposition of the drawing in the cylinder. It is crucial to note that such a ``Menger-type Theorem'' does \emph{not} exist for general directed graphs using standard embeddings, in particular, its existence highly relies on our graphs being Eulerian (embedded). The existence of this Menger Theorem (and the fact that, more generally, there is a Menger Theorem for Eulerian digraphs that yields linkages that are edge-disjoint in both directions, as opposed to the general directed Menger Theorem) as easy as it may be, is one of the key reasons why everything ends up working. 

Finally, we prove that if drawings in the cylinder do not admit high representativity, we find a nice ``linked'' decomposition as described above where every piece is \emph{$\theta$-short}. An Eulerian embedding in the cylinder is $\theta$-short if there is a curve connecting both cuffs $\zeta_1,\zeta_2$ that cuts at most $\theta$ edges. By cutting open along a $\theta$-short curve (and keeping the graph Eulerian embedded) we end up in a disc with $k_1 + k_2 + 2\ell \in 2\N$ many root edges, for some $\ell \leq \theta$, whence again solely bounded by a function of the genus of the original $\Sigma$, for $\theta$ was a constant. We conclude the cylinder case by induction, proving that given a sequence of embeddings in a cylinder with a ``well behaved'' decomposition into pieces, each of which comes from a well-quasi-order---for the base case these pieces are Eulerian embeddable in the disc case---the original sequence is already well-quasi-ordered.
\smallskip

\textit{The Disc.} Let $\Delta$ be a disc and let $k\in 2\N$ be the number of root edges. We are left with a sequence of Eulerian digraphs with $k$ root edges Eulerian embedded in $\Delta$ with the root edges embedded on the unique cuff of the disc. This case is unfortunately the most requiring of all: again, we want to either be done by ``high representativity'' (this works similarly as for the cylinder and we omit a discussion) or decompose our drawings into smaller parts that are well-quasi-ordered. Since we cannot lower the genus any further, to apply induction we need to either reduce the number of roots $k$ or some other parameter depending on the structure of the graph, that is, its treewidth (technically carving width, but they are qualitatively the same in our case). Note that we already proved the disc case without roots (that one is actually the easiest of all cases), and thus we can start the induction on the number of root edges. The idea is now as follows: If the treewidth of the graphs is bounded, then it follows by the last step of this summary as discussed below; thus assume for now that the treewidth grows infinitely with the sequence. In particular $\tw{G_i} < \tw{G_j}$ for every $i<j$ after some suited filtering of the sequence. By \cref{thm:4_reg_swirl} due to \cite{EDP_Euler, Johnson2002}, there is a large immersed swirl in $\GGG_i$ and we may assume that we cannot link the swirl to all of the roots via edge-disjoint paths (else we have high representativity and the case is again fairly easy). In particular, for every large swirl $\SSS^*$ in $G_i$ we can find a cut-cycle $F^*$---a closed curve cutting the drawing only in edges---that cuts off the swirl from the boundary of $\Delta$ and hence its root edges. Again, by a version of Menger's theorem for the disc (see \cref{cor:Menger_on_a_disc}), it is guaranteed that $F^*$ cuts at most $(k-2)$ edges. Thus, switching to the graph ``inside'' the disc $\Delta(F^*)$ bounded by $F^*$ we know that it comes from a well-quasi-order by induction on $k$. Note here that due to the embedding restrictions, the graph ``inside'' $\Delta(F^*)$ can again be shown to be Eulerian embedded in the disc, with the obvious embedding; in particular, the vertices $X(F^*)$ drawn in $\Delta(F^*)$ induce a cut.  We now use the above-introduced concepts of stitching to replace these cuts with beads, where we label the beads by the respective type of the graph in the well-quasi-order guaranteed by induction. Using the Eulerianness of the digraph one can show that there exists a collection of such curves $\{F_1,\ldots,F_t\}$ that together ``get rid'' of all the large swirls and have pairwise disjoint interiors $\Delta(F_i) \cap \Delta(F_j) = \emptyset$; thus the respective stitchings do not depend on each other. Note, however, that $t$ may be unbounded in $k$, and thus we may produce an ``unbounded number'' of beads (for graphs further down the sequence), each of degree $<k$ (see \cref{sec:disc} for the details).

The resulting graph is a special case of what we will call an \emph{$\Omega$-knitwork} (where $\Omega$ is the well-quasi-order for graphs Eulerian embeddable in the disc with $\leq k-2$ roots by induction), and in particular, since we got rid of all the swirls, it has bounded treewidth and bounded degree ($<k$), hence bounded carving width.

\paragraph{Bounded Carvings and $\Omega$-Knitworks.}

In the last step, we deal with $\Omega$-knitworks of bounded carving width. Let us describe what an $\Omega$-knitwork encodes:  given a well-quasi-order $\Omega=(V(\Omega),\preceq)$ and an Eulerian digraph $G$ of degree at most $k \in 2\N$ an $\Omega$-knitwork $\GGG=( (G,\pi(E)), \mu, \m ,\Phi)$ for a graph $G$ consists of a rooting of $G$ in edges $E$ (where $\pi(E)$ is an ordering on $E$), a labelling function $\Phi: V(G) \to V(\Omega)$, an interface map $\mu$ that assigns to a vertex of $v$ an ordering $\pi(v)$ of the edges of the cut $\rho(\{v\})$ that $v$ induces, and a routing handler $\m$ that assigns to each vertex a list of ``feasible routings''. That is, given $v \in V(G)$, $\m(v)$ is a list of matchings of $\rho(v)$, where each element of the list represents a choice of possible ways how we can connect edges through $v$. The idea is that when we replace a part $X \subset V(H)$ of a graph $H$ by a bead vertex, we need to remember \emph{how} the immersion is allowed to route inside $H[X]$, since we will need to lift the $\Omega$-knitwork immersions back to an immersion in the original graph, whence we need to reroute paths passing through beads in $H[X]$ accordingly. The concepts are fairly technical, and so are the proofs, which is why we leave out the details and refer the reader to \cref{subsec:knitworks}. We prove in \cref{thm:wqo_bounded_carvingwidth_knitworks} that given a special type of $\Omega$-knitworks---that is, \emph{reliable and well-linked} ones where the routing handler allows ``a lot'' of routings---the class of well-linked $\Omega$-knitworks of bounded carving width is well-quasi-ordered by strong immersion. Fortunately, we can prove that the pieces of the disc that we replace with beads (see the previous step) end up being ``reliable'' and ``well-linked'' (which again heavily relies on the graph being Eulerian embedded) as needed for the \cref{thm:wqo_bounded_carvingwidth_knitworks}, thus concluding the proof by induction.

\section{Preliminaries}

\subsection{Graph Theory}

We denote directed graphs (or digraphs) by~$G=(V,E,\operatorname{inc})$ where~$\operatorname{inc}\subset (V\times E)\cup (E\times V)$ is a binary relation, and for every~$e \in E$ there are exactly two incidences in~$\operatorname{inc}$, i.e., there are two (possibly equal) unique vertices~$v_i,v_o \in V$ such that~$(v_i,e),(e,v_o) \in \operatorname{inc}$. We define functions~$\tail,\head: E \to V$ via~$\tail(e)=v_i$ for~$(v_i,e) \in \operatorname{inc}$ and~$\head(e) = v_o$ for~$(e,v_o) \in \operatorname{inc}$. Since~$\tail,\head$ are uniquely defined we may equally well write~$G=(V,E,\tail,\head)$ for directed graphs meaning the obvious. 

We will crucially use edges as ``stand-alone'' objects as it simplifies many of the proofs, but for simplicity and ease of readability we may write~$e=(u,v)$ for~$e\in E$ to mean~$\tail(e) = u$ and~$\head(e)=v$; i.e., we may view~$E \subseteq V \times V$ as a multiset (we allow for parallel edges). Similarly we may write~$G=(V,E)$ for some~$E\subseteq V \times V$ for convenience, meaning the obvious. In general we will try to ``hide'' the fact that we work with incidence structures whenever possible for ease of readability.
Note further that the above definition allows for \emph{loops}, where an edge~$(u,v) \in E(G)$ is called a \emph{loop} if and only if~$u = v$; we write~$\loops(v)$ for the set of loops adjacent to~$v$ and~$\loops(G)$ for the set of all loops in~$G$. 

Given~$X \subseteq V(G)$ we write~$\bar{X} \coloneqq V(G) \setminus X$. Given~$U\subseteq V(G)$ we write~$N_G^+(U) \coloneqq \{v \mid (u,v) \in E(G), \text{ for some }u\in U\}$, and~$N_G^-(U) \coloneqq \{v \mid (v,u) \in E(G), \text{ for some }u\in U\}$, and~$N_G(U) \coloneqq N_G^+(U) \cup N_G^-(U)$. If~$G$ is clear from the context we may omit the subscript and if~$U = \{u\}$ we write~$N^+(u)$ instead of~$N^+(\{u\})$ and analogously for the other notions.

We refer to the graph obtained from~$G$ by omitting the directions of its edges as the \emph{underlying undirected graph}. Whenever we talk about parameters defined for undirected graphs, we implicitly assume that the parameter uses the underlying undirected graph, e.g.,~$\tw{G}$ is well-defined for directed graphs, where $\operatorname{tw}(\cdot)$ is the \emph{treewidth} of an undirected graph (see \cite{Diestel2017} for definitions). 

Given two graphs~$H=(V,E_H)$ and~$G=(V,E_G)$ on the same set of vertices we write~$G-H \coloneqq (V,E_G \setminus E_H)$ and similarly~$G+H \coloneqq (V,E_G \cup E_H)$ where multi-edges are allowed.

Given a directed graph~$G$ and~$v \in V(G)$ we call~$v$ \emph{Eulerian} if the number of edges~$e \in E(G)$ with $\tail(e) = v$ equals the number of edges~$e' \in E(G)$ with~$\head(e') = v$, i.e., the in- and out-degree at~$v$ are equal. We call~$G$ \emph{Eulerian} if every vertex in~$V(G)$ is Eulerian.

In the remainder of this paper we will mostly work with edge-disjoint paths and directed walks that do not repeat edges but may repeat vertices, because it is more general. Walks without repeated vertices are only needed very rarely. We therefore adopt the following simplifying but non-standard notation. 
\begin{definition}[Paths, Cycles and Circles]\label{def:paths}
   Let~$G=(V,E,\operatorname{inc})$ be a directed Eulerian graph. Let~$P =(e_1,\ldots,e_m)_G$ be a sequence of distinct edges~$e_1,\ldots,e_m \in E$ for some~$ m\in \N$ such that for every~$1 \leq i < m$ there is~$v_i \in V(G)$ with~$\head(e_i) = v_i = \tail(e_{i+1})$. Then we call~$P$ a \emph{path} in~$G$ and let~$m$ denote its \emph{length}. We write~$V(P) = \{v_1,\ldots,v_{m-1}\}$ and~$E(P) =\{e_1,\ldots,e_m\}$. The vertices~$v_0 = \tail(e_1)$ and~$v_{m}=\head(e_m)$ are called the \emph{endpoints} of~$P$ whereas~$V(P)$ are the \emph{internal vertices} of~$P$. If the graph is clear from context we omit the subscript and write~$P = (e_1,\ldots,e_m)$. We call~$e_1,e_m$ the \emph{ends} of~$P$ and call~$e_1$ the \emph{first} edge of~$P$ and say that~$P$ \emph{starts in} $e_1$, and we call~$e_m$ the \emph{last} edge of~$P$ and say that~$P$ \emph{ends in}~$e_m$.

    If in addition for every~$1 \leq i ,j < m$ with~$i \neq j$ we have~$v_i \neq v_j$ then we call~$P$ \emph{linear}. Let~$1 \leq j < \ell \leq m$, then we call~$(e_j,\ldots,e_\ell) \subseteq P$ a \emph{subpath} of~$P$. In particular when writing~$(f_1,\ldots,f_t) \subseteq P$ we mean a subpath of~$P$ of length~$t$, i.e.,~$f_1 = e_j$ and~$f_t= e_{j+t}$ for some~$1 \leq j \leq j+t \leq m$.

   If in addition~$P'=(e_2,\ldots,e_m,e_1)$ is a path in~$G$, then we call~$C \coloneqq (e_1,\ldots,e_m,e_1)_G$ \emph{a cycle} in~$G$ (omitting the subscript if clear from context). If $P$ and $P'$ are linear we call~$C$ a \emph{circle}, that is~$v_m=v_0$ and for every~$1 \leq i,j \leq m$ with~$i\neq j$ we have~$v_i \neq v_j$. 

   Given two edges~$e,e' \in E(G)$ and a path~$P \subset G$ with ends~$e,e'$, we call~$P$ an \emph{$\{e,e'\}$-path}, and if~$P$ starts in~$e$ and ends in~$e'$ we call it an~$(e,e')$-path and define~$\tau(P) = (e,e')$ to be its \emph{type}. Similarly we refer to~$P$ as a \emph{$(v_0,v_m)$-path} if~$v_0 = \tail(e)$ and~$v_m = \head(e')$.
\end{definition}
\begin{remark}
     Note that a path~$(e_1,\ldots,e_m)$ may have a single endpoint~$\tail(e_1) = v_1 = v_m = \head(e_m)$.
    Note that if we have two graphs~$H,G$ and a set of edges~$\{e_1,\ldots,e_m\} \subset E(G) \cap E(H)$ then we may write~$(e_1,\ldots,e_m)_H,(e_1,\ldots,e_m)_G$ to mean the respective paths in the respective graphs (if they are paths).

     Note that a linear path~$P=(e_1,\ldots,e_m)_G$ induces a linear order on~$V(P)$ given by~$\tail(e_i) \prec \head(e_i)$ for all~$1 < i < m$. Further note that every strict subsequence of a circle is a linear subpath.
\end{remark}
Thus, a ``path'' in this paper corresponds to a directed walk without repeated edges in the general literature, or equivalently, a standard path in the ``line graph''. The usual paths without repeated vertices are called linear. Similarly, a cycle may repeat vertices but not edges. It turns out that the above definitions simplify argumentation by getting rid of overhead and case distinctions in what is to follow. Especially the ability to omit the~$\head$ and $\tail$ relations for the first and last edges of paths. 
\smallskip

A directed graph $G$ is \emph{weakly connected} if the underlying undirected graph is connected, and it is \emph{strongly connected} if for every pair of distinct vertices $u,v \in V(G)$, there is a $(u,v)$-path and a $(v,u)$-path in $G$.

Note that every strongly connected  Eulerian digraph can be traced by a single cycle, and thus ``cycles'' are just Eulerian digraphs with a fixed linear ordering on its edges.
\smallskip

For notational convenience, we may write paths and cycles as ordered tuples of vertices when needed. That is~$P=(e_1,e_2,\ldots,e_m)$ may be rewritten as~$P=(v_1,v_2,v_3,\ldots,v_{m+1})$ where~$e_i=(v_i,v_{i+1})$ for~$1 \leq i \leq m$ (here we need to include both incidences for each edge in order to be unambiguous). Similarly we may write~$(e_1,v_1,e_2,v_2,\ldots,e_m,v_{m+1})$ for paths to highlight the adjacencies of vertices and edges. 

 Further, to avoid corner cases we allow for paths of length~$0$ which consist of isolated vertices, i.e.,~$(v)$ is considered a path with both endpoints being~$v \in V(G)$.

\begin{definition}[Disjoint paths and Linkages]
    Let~$G$ be a digraph and let~$k,\ell \in \N$. Two paths $P_1=(e_1,\ldots,e_k),P_2=(f_1,\ldots,f_\ell)$ in~$G$ are \emph{edge-disjoint} or \emph{vertex-disjoint} if~$E(P_1) \cap E(P_2) = \emptyset$ or~$V(P_1)\cup\{\tail(e_1),\head(e_k)\} \cap V(P_2)\cup\{\tail(f_1),\head(f_\ell)\} = \emptyset$ respectively

    The paths~$P_1,P_2$ are \emph{internally edge-disjoint} if for~$e \in E(P_1) \cap E(P_2)$ it holds~$e\notin \{e_2,\ldots,e_{k-1}\}\cup\{f_2,\ldots,f_{\ell-1}\}$ and \emph{internally vertex-disjoint} if~$V(P_1) \cap V(P_2) = \emptyset$.

    Given a set of edge-disjoint paths~$\LLL$ we call~$\LLL$ a \emph{linkage (in $G$)}. We call the linkage \emph{linear} if each path in~$\LLL$ is linear. We call ~$\Abs{\LLL}$ the \emph{order} of the linkage. Finally we define~$\tau(\LLL) \coloneqq \{\tau(P) \mid P \in \LLL\}$ to be the \emph{type} of the linkage.
\end{definition}

The definition of edge-disjointness and vertex-disjointness lifts to cycles and circles in a straightforward way. The following is obvious by re-routing.
\begin{observation}\label{obs:linkage_gives_linear_linkage}
    Let~$G$ be a digraph and~$\LLL$ a linkage in~$G$ of order~$t \in \N$. Then there exists a linear linkage~$\LLL'$ in $G$ of order~$t$ with~$\tau(\LLL) = \tau(\LLL')$.
\end{observation}

Refining on the notion of types of linkages we define the following.
\begin{definition}
    Let~$G$ be a digraph and let~$A,B \subseteq V(G)$ be distinct. Let~$\LLL$ be a (linear) linkage in~$G$ such that every~$P \in \LLL$ has one endpoint in~$A$ and one endpoint in~$B$, then we call~$\LLL$ a (linear) \emph{$\{A,B\}$-linkage}.

    Similarly, let~$E,F \subset E(G)$ be distinct. Let~$\LLL$ be a (linear) linkage in~$G$ such that every~$P \in \LLL$ has one end in~$E$ and one end in~$F$ and is otherwise edge-disjoint from~$E\cup F$, then we call~$\LLL$ a (linear) \emph{$\{E,F\}$-linkage}. 
\end{definition}

\begin{definition}[Concatenation of Paths]\label{def:concatenation_of_paths}
      Let $P_1=(e_1,\ldots,e_k),P_2=(f_1,\ldots,f_\ell)$ be internally edge-disjoint paths such that~$e_k = f_1$. Then we define~$P = P_1 \circ P_2 \coloneqq (e_1,\ldots,e_k,f_2,\ldots,f_\ell)$ and call~$P_1,P_2$ the \emph{summands} of~$P$.
      Similarly if~$P_1,P_2$ are edge-disjoint and there exists~$v \in V(G)$ with~$ (e_k,v) \in \operatorname{inc}$ and~$(v,f_1) \in \operatorname{inc}$ we define~$P = P_1 \circ P_2 = (e_1,\ldots,e_k,f_1,\ldots,f_\ell)$.
\end{definition}
It is straightforward to verify that~$P$ is a path in our setting.

\begin{observation}\label{obs:concat_paths}
    Let $P_1=(e_1,\ldots,e_k),P_2=(f_1,\ldots,f_\ell)$ be internally edge-disjoint paths in a graph~$G$ such that~$e_k = f_1$. Then~$P_1 \circ P_2$ is a path in~$G$. If~$P_1$ and~$P_2$ are edge-disjoint and there exists~$v \in V(G)$ with~$(e_k,v) \in \operatorname{inc}$ and~$ (v,f_1) \in \operatorname{inc}$, then~$P_1 \circ P_2$ is a path. If additionally to the edge-disjointness~$P_1,P_2$ are internally vertex-disjoint and linear then $P_1 \circ P_2$ is linear.
\end{observation}

We abuse notation and denote undirected graphs by~$G=(V,E)$ where~$E \subseteq \{ \{v,w\} \mid v,w \in V(G)\}$ may be a multiset, noting that it will be clear from the context whenever we talk about undirected graphs.

A \emph{tree}~$T=(V,E)$ is a connected acyclic undirected graph. We call vertices of degree~$1$ in~$V(T)$ \emph{leaves}. If $T$ is a tree, we denote the set of its leaves by $\leaves{T} \subseteq V$. We call a tree \emph{cubic} if apart from its leaves all its vertices are of degree three.

A \emph{clique}~$K_k = (V,E)$ is an undirected graph on~$k\in \N$ vertices such that~$E=\{\{v,w\} \mid v,w \in V, v\neq w\}$.

An~$n\times m$-grid~$W_{n,m} = (V,E)$ is an undirected graph on~$nm$ vertices admitting an enumeration of~$V$ such that
\begin{align*}
    V &= \{v_{i,j} \mid 1 \leq i \leq n, 1 \leq j \leq m\},\\
    E &= \{\{v_{i,j},v_{i+1,j}\} \mid 1 \leq i \leq n-1, 1 \leq j \leq m\} \cup\\ & \{\{v_{i,j},v_{i,j+1}\} \mid 1 \leq i \leq n, 1 \leq j \leq m-1\}. 
\end{align*}

\subsection{Immersions}
The following is the minor relation of interest in this paper.

\begin{definition}[Immersion]
\label{def:immersion}
    Let~$G,H$ be directed graphs. Then \emph{$H$ immerses in $G$}, or \emph{$G$ immerses $H$}, if there exists a map $\gamma$ defined on~$V(H) \cup E(H)$ satisfying the following:
    \begin{enumerate}
        \item $\restr{\gamma}{V(H)}: V(H) \to V(G)$ is injective,
        \item for $e \in E(H)$, $\gamma(e)$ is a path in $G$ and for distinct~$e_1,e_2 \in E(H)$ the paths $\gamma(e_1),\gamma(e_2)$ are edge-disjoint.
        \item let~$e=(u,v)\in E(H)$ for some~$u,v \in V(H)$, then~$\gamma(e)$ starts in~$\gamma(u)$ and ends in~$\gamma(v)$.
    \end{enumerate}

We call an immersion \emph{strong} if for every~$e\in E(H)$ and~$v \in \gamma(V(H)) \cap V(\gamma(e))$ the vertex $v$ is no internal vertex of~$\gamma(e)$. We write $H \hookrightarrow G$ if $G$ strongly immerses $H$ and $\gamma\colon H \hookrightarrow G$ if $\gamma$ is a map witnessing it.
\end{definition}
\begin{remark}
    As we allow for loops~$e=(v,v) \in E(G)$, the above definition is only correct since we allow paths to start and end at the same vertex, i.e.~$\gamma(e)$ starts and ends in~$\gamma(v)$. 
\end{remark}

Note here that (strong) immersion induces a \textit{partial order} on Eulerian digraphs, that is, it is a transitive, reflexive, and `antisymmetric' relation. In particular if~$G \hookrightarrow H$ and~$H \hookrightarrow G$ then~$H \cong G$.

\begin{observation}
(Strong) Immersion defines a partial order~$\preceq$ on Eulerian digraphs, where we say that~$H \preceq G$ if and only if $G$ (strongly) immerses $H$.

\end{observation}

In light of \cref{def:immersion} we define \emph{immersion models}, or simply \emph{models}, as follows.

\begin{definition}[Immersion model]
    Let~$H,G$ be Eulerian digraphs and~$\gamma: E(H) \cup V(H) \to  G$ a (strong) immersion. Then~$\big(\gamma(V(H)),\gamma(E(H))\big)$ is a \emph{(strong) immersion model} or simply \emph{(strong) model} of~$H$ in~$G$. 
\end{definition}

There are different ways to define minors for undirected graphs, one of which is via graph-modifying operations, namely taking subgraphs and contracting edges. Regarding immersions we have the following operation.

\begin{definition}[Splitting off] \label{def:splitting_edgepairs}
    Let~$G=(V,E,\operatorname{inc})$ be a directed graph, let~$v \in V(G)$, and let~$e,e' \in E(G)$ such that~$(e,v),(v,e') \in \operatorname{inc}$. Let~$u\coloneqq\tail(e)$ and~$w \coloneqq \head(e')$. \emph{Splitting off~$(e,e')$ (at~$v$) in~$G$} results in a new graph~$G'$ by deleting~$e,e' \in E(G)$ and~$(u,e),(e,v),(v,e'),(e',w) \in \operatorname{inc}$ and subsequently adding a new edge~$f \notin V(G) \cup E(G)$ together with the incidences~$(u,f),(f,w)$. We remove~$v$ from the vertex set if and only if it is an isolated vertex (it has no incidences left) after splitting off; we define~$(G',f) \coloneqq \spl(G;(e,e'))$ or simply write~$G' \coloneqq \spl(G;(e;e))$ if we do not specifically need~$f$ for simplicity.

    Let $M=\{(e_i,e_i')\mid 1 \leq i \leq m\}$ be a set such that for every~$i \in \{1,\ldots,m\}$~$(e_i,e_i')$ is a two-path in~$G$. Let~$G_i$  be obtained from~$G$ by inductively splitting off edge-pairs as follows:
    \begin{enumerate}
        \item $(G_1,e_1^2) \coloneqq \spl(G;(e_1,e_1'))$ and rename~$e_i,e_j'$ to~$e_1^2$ if~$e_i,e_j' \in \{e_1,e_1'\}$ for any~$i,j > 1$
        \item $(G_{i+1},e_1^{i+1})) \coloneqq \spl(G_{i},(e_{i+1},e_{i+1}'))$ for~$i\leq m-1$ making obvious identifications as above.
    \end{enumerate}
    Then we call~$G_m$ the graph obtained from \emph{splitting off~$G$ along~$M$} and write~$G_m \coloneqq \spl(G; M)$.
\end{definition}
\begin{remark}
    In particular we may split off~$G$ along paths.
    Note that splitting off may result in loops and multi-edges.
\end{remark}

The following are well-known and straightforward observations.
\begin{observation}
    Let~$G$ be an Eulerian digraph and let~$P=(e_1,\ldots,e_m),Q \subseteq G$ be edge-disjoint paths. 
    
    \begin{itemize}
        \item Let~$(G_1,e_1^2) \coloneqq \spl(G;(e_1,e_2))$ and~$P_1 \coloneqq (e_1^2,e_3,\ldots,e_m)$. Then~$P_1,Q \subseteq G_1$ are edge-disjoint paths.
        \item Let~$G_P^Q$ be the graph resulting from splitting off~$G$ along~$P$ and subsequently along~$Q$, and let~$G_Q^P$ be the graph resulting from splitting off~$G$ along~$Q$ and subsequently~$P$ respectively. Then~$G_P^Q \cong G_Q^P$.
    \end{itemize}
\end{observation}
In particular, the graph resulting from~$\spl(G,M)$ does not depend on an order in~$M$ (up to renaming).

The following is a direct consequence to the above.

\begin{observation}\label{obs:splitting_off_linkages_gives_linkages}
    Let~$G$ an Eulerian digraph, let~$A,B \subset V(G)$ be distinct and let~$\LLL$ be an~$\{A,B\}$-linkage. Let~$P=(e_1,\ldots,e_i,e_{i+1},\ldots,e_m) \in \LLL$ be some path and let~$(e_i,e_{i+1}) \subset P$ be a~$2$-edge subpath. Let~$(G^*,e^*) \coloneqq \spl(G;(e_i,e_{i+1})$ and let~$P^* = (e_1,\ldots,e_{i-1},e^*,e_{i+2},\ldots,e_m)$ be the respective path in~$G^*$. Let~$\LLL^* \coloneqq (\LLL \setminus \{P\})\cup\{P^*\}$, then~$\LLL^*$ is an~$\{A,B\}$-linkage in~$G^*$.

    Similarly, let~$(G',f) \coloneqq \spl(G,(f_1,f_2))$ for some two-path~$(f_1,f_2) \subset G$, and let~$\LLL'$ be an~$\{A,B\}$-linkage in~$G'$, then there exists an~$\{A,B\}$-linkage~$\LLL$ of the same order in~$G$ where~$\LLL$ and~$\LLL'$ agree up to possibly one path in~$\LLL'$ using the edge~$f$ which can be replaced by~$(f_1,f_2)$ in~$\LLL$.
\end{observation}

The following is an extension of the above, well known, and easy to verify.

\begin{observation}\label{obs:immersion_robust_under_splitting_off}
    Let~$H,G$ be graphs such that $G$ immerses $H$. Then there exists an Eulerian subgraph~$G'$ of~$G$ and a series of splitting off edge-pairs in~$G'$ that results in a graph isomorphic to~$H$. 
    Conversely, if~$H$ can be obtained from~$G$ by taking an Eulerian subgraph~$G' \subseteq G$ and then splitting off edge-pairs at vertices, it holds that $H$ is Eulerian and $G$ immerses $H$.
\end{observation}
\begin{proof}
    Let~$\gamma:V(H) \cup E(H) \to G$ be an immersion and let~$(\gamma(V(H)),\gamma(E(H)))$ be the respective model, then~$G' = \gamma(V(H)) + \bigcup_{P \in \gamma(E(H)}P$ is an Eulerian subgraph of~$G$. One easily verifies that fixing an order~$(P_1,\ldots,P_m)$ on~$\gamma(E(H))$ and iteratively splitting off~$G$ along~$P_1,\ldots,P_m$ results in~$H$.

    Similar given~$H$ and the sequence of splitting off edge-pairs we can reconstruct a model for~$H$ in~$G'$ by undoing the splittings that resulted in the edges of~$H$ and keeping track of the paths along which the splitting off resulted in~$H$; this is straightforward to verify.
\end{proof}

Finally, we have the following easy result regarding immersions.

\begin{observation}\label{obs:immersion_maps_path_to_path}
    Let~$\gamma:E(H) \cup V(H) \to G$ be an immersion of directed graphs~$H,G$. Let~$P$ be a path in~$H$. Then~$\gamma(P)$ is a path in~$G$. 
    Further if~$\gamma$ is a strong immersion and~$P$ is a linear path in~$H$, then~$\gamma(P)$ is a linear path in~$G$.
\end{observation}
\begin{remark}
    Note that if the immersion is not strong, then linear paths may be mapped to non-linear paths, i.e., they may self-intersect.
    
\end{remark}

\smallskip

\subsection{Induced Cuts}
We start by defining edge-separations and induced cuts. The way we define them is purely undirected and is thus not equivalent to the standard directed definitions of these notions. 
Observe that a Eulerian digraph is strongly connected if and only if it is weakly connected. We therefore simply say that $G$ is \emph{connected} if it is weakly connected.

\begin{definition}[Edge- and Vertex-Separations]
    Let~$G$ be a connected digraph. An \emph{edge-separation}, or a \emph{cut}, is a set~$F \subset E(G)$ such that the digraph~$G-F$ is disconnected. We refer to~$\Abs{F}$ as the \emph{order} of the cut.

    Let~$X_1,X_2 \subset V(G)$. We call~$F$ an \emph{$(X_1,X_2)$-cut} if every path in the underlying undirected graph of~$G$ with one endpoint in~$X_1$ and one endpoint in~$X_2$ uses an edge of~$F$. We define~$\delta(X_1,X_2)$ to be the minimal order of any~$(X_1,X_2)$-cut, and infinite if it does not exist.%

    Let~$G_1,G_2 \subset G$ be such that~$G= G_1 \cup G_2$ and~$E(G_1) \cap E(G_2) = \emptyset$. Then we call~$(G_1,G_2)$ a \emph{vertex-separation of~$G$} and define its \emph{order} as~$\Abs{V(G_1 \cap G_2)}$.
\end{definition}
\begin{remark}
    By definition if~$X_1 \cap X_2 \neq \emptyset$ then there exists no~$(X_1,X_2)$-cut. 
\end{remark}
    The definitions can be extended to non-connected directed graphs by extending the definition component-wise.

\begin{definition}[Induced Cuts]
    Let~$G$ be a directed graph and let~$X \subseteq V(G)$. 
    
    We define~$\rho(X,\bar{X}) \coloneqq \{(u,v) \in E(G) \mid u \in X \text{ and } v\in \bar{X}\}$ and $\rho^{+}(X) \coloneqq \rho(X, \bar X)$ and $\rho^{-}(X) \coloneqq \rho(\bar{X},X)$. %
  Finally we let~$\rho(X) \coloneqq \rho^-(X) \cup \rho^+(X)$ and refer to it as \emph{the cut induced by~$X$}. We refer to~$\delta(X) \coloneqq \Abs{\rho(X)}$ as the \emph{order of the cut}.
\end{definition}
\begin{remark}
    Clearly~$\delta(X) = \delta(\bar{X})$ by definition, and similarly~$\rho^-(X) = \rho^+(\bar{X})$.

    Note that we have abused notation, using~$\rho$ and~$\delta$ on sets and pairs of sets of vertices for the sake of readability. It may never cause confusion which one we are talking about since they are defined on disjoint domains.
\end{remark}

If~$X=\{v\}$ we may write~$\rho(v)$ instead of~$\rho(\{v\})$ for simplicity. Let~$E(v) \coloneqq \{e \in E(G) \mid v \in \tail(e)\cup \head(e)\}$ denote the set of edges incident to~$v$. The following is clear from the definition.
\begin{observation}\label{obs:cut_at_vertex}
    Let~$G$ be a digraph and~$v \in V(G)$, then~$\rho(v) = E(v) \setminus \loops(v)$.
\end{observation}

It is well-known that on induced cuts the function $\delta$ is symmetric and sub-modular. That is~$\delta(X) = \delta(\bar{X})$ and~$\delta(X) + \delta(Y) \geq \delta(X \cap Y) + \delta(X \cup Y)$ for any pair~$X,Y \subseteq V(G)$. And, of course, the notions of cut and induced cuts are closely related: The following is well-known and straightforward to verify (for it is essentially defined on the undirected underlying graph).

\begin{observation}\label{lem:relating_cut_to_induced_cuts}
    Let~$G$ be a directed graph and let~$X_1,X_2 \subset V(G)$. Then~$\delta(X_1,X_2)$ is given by the minimum~$\delta(X)$ over all~$ X \subseteq V(G)$ with~$X_1 \subseteq X$ and~$X \cap X_2 = \emptyset$; and is infinite if no such~$X$ exists.%
\end{observation}

The way~$\delta$ is defined, it does not come with a nice \emph{directed} Menger property in the context of general digraphs. That is, it is not true that~$\delta(X_1,X_2)$ is equal to the maximal number of \emph{directed} edge-disjoint paths connecting~$X_1$ and~$X_2$: take for example a simple alternating path of length at least two and let~$X_1,X_2$ be its extremities. These pathologies disappear when restricted to Eulerian digraphs. 

The following is a well-known and easy fact.
\begin{observation}
    Let~$G$ be an Eulerian digraph and let $X \subseteq V(G)$. Then $\Abs{\rho(X,\bar{X})} = \Abs{\rho(\bar{X},X)}$. In particular~$\Abs{\rho(\cdot,\cdot)}$ is symmetric and~$\delta(X) \in 2\N$.
\end{observation}

Together with \cref{lem:relating_cut_to_induced_cuts} this implies a directed Menger-type property for Eulerian digraphs. This is an easy consequence of the known directed Menger theorem (for directed cuts) together with the fact that if~$G$ is Eulerian and there is a directed path from~$u$ to~$v$ then there exists another path---edge-disjoint from the first---that connects~$v$ to~$u$. See also \cite{frank95,EDP_Euler} for further connectivity properties of Eulerian digraphs.

\begin{lemma}\label{lem:Menger_for_Euler}
  Let~$G$ be an Eulerian digraph and~$X_1, X_2 \subseteq V(G)$ be disjoint.
  \begin{enumerate}
  \item For every~$k \in \N$, either there is a set~$\LLL$ of edge-disjoint paths in~$G$ containing~$k$ paths from~$X_1$ to~$X_2$ and~$k$ paths from~$X_2$ to~$X_1$ or there is a partition~$(X, \bar{X})$ of~$V(G)$ such that~$X_1 \subseteq X$,~$X_2 \subseteq \bar{X}$ and~$\delta(X) < 2k$.
  \item The maximum number of edge-disjoint paths from~$X_1$ to~$X_2$ is the same as the maximum number of edge-disjoint paths from~$X_2$ to~$X_1$. 
  \end{enumerate}
\end{lemma}

Thus note that it suffices to look at undirected induced cuts, for the Eulerianness imposes that these cuts must be of even order, and when looking at their directions half of them go from one side to the other of the partition and vice-versa. This can then be used together with Menger's theorem to either find paths witnessing that the cut is minimum, or we can find a smaller cut which in turn means we can find a smaller undirected cut and vice-versa. So the Eulerianness guarantees that the smallest undirected cut suffices to witness the maximum number of directed edge-disjoint paths between the resulting components. 

\subsection{Graph Drawings and Surfaces} 
A surface~$\Sigma$ in our setting is a compact~$2$-manifold possibly with boundary. The boundary~$\bd(\Sigma)$ of a surface is a collection of \emph{cuffs} where a cuff is a connected component of~$\bd(\Sigma)$ that can be traced by a simple closed curve. We denote the set of cuffs of a surface via~$c(\Sigma)$. Let~$c(\Sigma) =\{C_1,\ldots,C_k\}$ be the cuffs of~$\Sigma$ then $\bigcup_{ 1 \leq i \leq k} C_i = \bd(\Sigma)$. We write~$\hat{\Sigma}$ for the surface obtained from~$\Sigma$ by adding~$k$ disjoint discs $\Delta_1,\ldots,\Delta_k$ to the surface and gluing the boundary of~$\Delta_i$ to~$C_i$ in the obvious way; in particular~$\hat{\Sigma}$ has no boundary. It is well-known that every surface without boundary is homeomorphic to a surface obtained from a sphere by adding \emph{handles} and \emph{cross-caps}. Further every surface with boundary can be obtained (up to a homeomorphism) from a surface without boundary by punching finitely many holes into it, in particular each of these holes represents a cuff. We will omit the details regarding the definition of surfaces unless specifically needed and redirect the reader to \cite{Armstrong2010} or \cite{MoharT2001} for more information.

Given a set~$U \subseteq \Sigma$ we write~$\bar{U} \subseteq \Sigma$ to denote its topological closure in~$\Sigma$ and $U^\circ$ for its topological interior.

\begin{definition}[Embedding of a digraph]
    Let~$G=(V,E,\operatorname{inc})$ be a digraph, then a \emph{drawing} of~$G$ or an \emph{embedding} of~$G$ on a surface~$\Sigma$ is a pair~$(U,\nu)$ satisfying the following.
    \begin{itemize}
    \item $U \subset \Sigma$ is a closed subset of~$\Sigma$,
    \item $\nu: V \cup E \to U$ is injective,
    \item every component of~$U \setminus \nu(V \cup E)$ is homeomorphic to an open interval, and the components are in bijection with the elements of~$\operatorname{inc}$. 
    \item let~$u_1,u_2 \in V$, $e \in E$ such that~$u_1 = \tail(e)$ and~$u_2 = \head(e)$. Then there exist components~$\operatorname{inc}_1,\operatorname{inc}_2 \in U \setminus \nu(\{u_1,u_2,e\})$, each homeomorphic to an open interval such that the closure~$\overline{\operatorname{inc}_i} \subset \Sigma$ consists of~$\operatorname{inc}_i \cup \{\nu(e),\nu(u_i)\}$ for~$i=1,2$. Further, if~$e$ is a loop, then~$\nu(u_1) = \nu(u_2)$---by injectivity---and~$\overline{\operatorname{inc}_1}$ is homeomorphic to a circle.
    \end{itemize}
    We may add the respective directions of the edges to the embedding by keeping track of~$\operatorname{inc}$ in an obvious way (for example by drawing arrows on the components). Further we may write $\nu((u_1,e)),\nu((e,u_2))$ for $(e,u_2),(u_1,e) \in \operatorname{inc}$ to mean the components $\operatorname{inc}_1,\operatorname{inc}_2 \subset U$ defined as above. If a graph can be drawn on a surface in this way, and given some drawing~$(U,\nu)$, we say that the graph is \emph{embedded in~$\Sigma$}. 

     We refer to an edge~$e\in E$ with~$\nu(e) \notin \bd(\Sigma)$ as \emph{internal edge (of the embedding)}; otherwise we refer to it as a \emph{boundary edge (of the embedding)} and we define internal and boundary vertices analogously.
\end{definition}
\begin{remark}
    With the above extension of $\nu$ to $\operatorname{inc}$ we may write $\nu(H)$ for $H\subseteq G$ to mean the set $U' \subseteq U$ consisting of $\nu\big(V(H) \cup E(H) \cup \restr{\operatorname{inc}}{H\times H}\big)$.
\end{remark}

We will mostly work with \emph{Eulerian embeddings}.

\begin{definition}[Eulerian Embedding]\label{def:embedding}
    Let~$\Sigma$ be a surface. Let~$G$ be an digraph of maximum degree four and~$(U,\nu)$ an embedding of~$G$ in~$\Sigma$ without boundary vertices. Let~$v \in V(G)$ be an Eulerian vertex of degree four and let~$\gamma_v:[0,1]\to \Sigma$ be a simple closed curve winding around~$\nu(v)$ in~$\Sigma$ such that~$\gamma[0,1]$ intersects~$U$ exactly in edges~$e \in E(G)$ incident to~$v$ and in each such edge exactly once. Let~$(e_1,e_2,e_3,e_4)$ be the edges visited in that order by~$\gamma_v$ (up to a cyclic rotation). Then we call~$\nu(v)$ \emph{strongly planar} if~$e_{1},e_{2}$ are edges with~$v$ as a head and~$e_{3}$ and~$e_{4}$ are edges with~$v$ as a tail (up to a cyclic rotation). Otherwise we call it \emph{Eulerian embedded}.

    If every Eulerian vertex of~$G$ is Eulerian embedded given the embedding~$(U,\nu)$, then we call it an \emph{Eulerian embedding}. And if there exists an Eulerian embedding~$(U',\nu')$ for~$G$, then we call~$G$ \emph{Euler-embeddable}.
    \label{def:Euler-embedding}
\end{definition}
\begin{remark}
    More concisely~$G$ admits an Eulerian embedding if and only if there is an embedding without boundary vertices for which, when choosing an orientation of $\Sigma$ locally at each vertex, the in- and out-edges at each vertex are drawn alternately. 
\end{remark}

The following is obvious and well-known.
\begin{observation}\label{obs:eulerian_emb_closed_under_splitting_and_immersion}
    Let $H,G$ be Eulerian digraphs and let $G$ be Eulerian embeddable. If $H \hookrightarrow G$ then $H$ is Eulerian embeddable.
    In particular, given an Eulerian embeddable digraph $G$ and a two-path $(e,e')$ in $G$, $\spl(G,(e,e'))$ is Eulerian embeddable.
\end{observation}

We will discuss Eulerian embeddings in more detail in \cref{sec:surface_defs} as well as in \cref{sec:la-grande-inductione}; for further details on Eulerian embeddings see \cite{Johnson2002,EDP_Euler,Frank_2path}.

\subsection{Carvings and Tree Decompositions}

We assume the reader to be familiar with the definition of tree and branch decompositions of undirected graphs as well as treewidth and branchwidth; see \cite{Diestel2017} for details.
In this subsection we discuss the concept of \emph{carving width} introduced and analysed in \cite{carvingwidth}\footnote{Here the parameter is referred to as \emph{congestion}.} as well as \cite{RobertsonST1994}, a notion similar to branch width tailored towards induced cuts of a graph. 
\begin{definition}[Carving]
    \label{def:carving}
    Let~$G$ be a digraph. A~\emph{carving} of~$G$ is given by a pair~$(T,\ell)$ where
    \begin{itemize}
        \item[(i)] $T$ is a cubic undirected tree and 
        \item[(ii)] $\ell: V(G) \to \leaves{T}$ is an injective map labelling part of the leaves.
    \end{itemize}

    Let~$\{t_1,t_2\} \in E(T)$ and let~$T_1,T_2$ be the components of~$T-\{t_1,t_2\}$ containing~$t_1$ and~$t_2$ respectively. We define~$\epsilon_{(T,\ell)}:E(T) \to 2^{E(G)}$ via~$\epsilon_{(T,\ell)}(\{t_1,t_2\})\coloneqq \rho(\ell^{-1}(\leaves{T_1}))$, where for a leaf~$t$ not in the image of~$\ell$ we define~$\ell^{-1}(t) \coloneqq \emptyset$. We may simply write~$\epsilon$ if the carving is clear from context.
    
    We define the \emph{width} of~$\{t_1,t_2\}$ via~$\w(\{t_1,t_2\}) \coloneqq \Abs{\epsilon(\{t_1,t_2\})}$. Finally we define the width of the decomposition via~$\w(T,\ell) \coloneqq \max\{\w(e) \mid e \in E(T)\}$. 

    The \emph{carving width of~$G$} is defined via~$$\ebw{G} \coloneqq \min\{\operatorname{w}(T,\ell) \sth (T,\ell) \text{ is a carving of } G\}.$$
\end{definition}
\begin{remark}
    We allow for leafs~$t$ with~$\ell^{-1}(t) = \emptyset$ so that a branch decomposition always exists (since~$T$ needs to be cubic this is needed for~$\Abs{V(G)} \leq 2$. It may be neglected otherwise.
\end{remark}

The next lemma guarantees that for any $d\geq 1$, the class of Eulerian digraphs of carving width at most $d$ is closed under immersions.

\begin{lemma}[Homogeneity for Immersions]\label{lem:ebw_closed_under_immersion}
    Let~$G$ and~$H$ be Eulerian digraphs. If $G$ immerses $H$ then $\ebw{H} \leq \ebw{G}$.
\end{lemma}
\begin{proof}
    
    Let~$H,G$ be Eulerian digraphs such that~$G$ immerses $H$. Let~$\tau=(T,\ell)$ be a carving of~$G$ witnessing its carving width. Since~$G$ contains~$H$ as an immersion, there exists a Eulerian subgraph~$G' \subseteq G$ and a sequence of splitting off operations at vertices of~$G'$ which produces $H$ by \cref{obs:immersion_robust_under_splitting_off}.  Thus it suffices to prove that carving width is closed under taking Eulerian subgraphs and splitting off at vertices. 
    
    \begin{claim}
        Let~$G' \subseteq G$ be Eulerian. Then~$\ebw{G'} \leq \ebw{G}$.
    \end{claim}
    \begin{claimproof}
        It suffices to prove this for~$G' \coloneqq G-E(C)$ for some fixed cycle~$C$ in $G$. Clearly~$(T,\ell)$ is still a valid~carving of~$G'$. Since we only deleted edges, the width of any edge in~$E(T)$ cannot increase.
    \end{claimproof}

    \begin{claim}
        Let~$G'$ be obtained from~$G$ by splitting off at~$v \in V(G)$ with respect to some two-path~$(v_l,v,v_r)$ in~$G$. Then~$\ebw{G'} \leq \ebw{G}$.
    \end{claim}
    \begin{claimproof}
        Let~$e_1=(v_l,v),e_2=(v,v_r) \in E(G)$ be the two edges adjacent to~$v$ along which we split off and let~$e'$ be the resulting edge. We keep~$(T,\ell)$ and claim the~carving witnesses that~$G'$ satisfies~$\ebw{G'} \leq \ebw{G}$. To see this note that the resulting edge~$e'$ appears in~$\epsilon(\{t_1,t_2\})$ if and only if $v_l$ and~$v_r$ lie on opposite sides of~$T \setminus \{t_1,t_2\}$. But then~$\epsilon(\{t_1,t_2\}) \cap \{e_1,e_2\} \neq \emptyset$ and thus the width of the edge cannot increase.
    \end{claimproof}
This concludes the proof.
\end{proof}
\begin{remark}
    Note that if a directed graph~$G$ contains a vertex~$v$ of degree~$k \in \N$ (not counting loops), then~$\ebw{G} \geq k$ simply because the edge adjacent to the leaf representing~$v$ in the cubic tree must have width~$k$. 
\end{remark}

Finally, there is a nice qualitative relation between carving width and treewidthdue to Bienstock \cite[Theorem 1 and preceding discussion]{carvingwidth}. 

\begin{lemma}[$ ${\cite[Theorem 1]{carvingwidth}}]
\label{thm:qualitative_equivalence_of_tw_and_ebw}
    Let~$\Delta \in \N$ and~$G$ be an Eulerian digraph of maximum degree~$\Delta$. Then~$\frac{2}{3}\tw{G} \leq \ebw{G} \leq \Delta \cdot \tw{G}$. 
\end{lemma}

\subsection{Well-Quasi-Ordering }
 A \emph{quasi-order}~$\Omega = (V,\preceq)$ is a tuple consisting of a set of elements~$V$ and a relation~$\preceq \subseteq V\times V$ satisfying the following. Let~$u,v,w \in V$ then~$\preceq$ is \textit{reflexiv}, i.e., $v \preceq v$, and \emph{transitive}, i.e.,~$u \preceq v$ and~$v \preceq w$ implies~$u \preceq w$. If in addition the relation \textit{antisymmetric}, i.e.,~$v \preceq w$ and~$w \preceq v$ implies~$ v = w$, we call~$\preceq$  a \emph{partial order}.

We say that~$\Omega=(V,\preceq)$ is a \emph{well-quasi-order} if for every infinite sequence~$(v_\ell)_{\ell \in \N}$ of elements~$v_\ell \in V$, there exist~$1 \leq i < j $ such that~$v_i \preceq v_j$. Given~$I \subseteq \N$ and a sequence~$(v_\ell)_{\ell \in I}$ such that for all~$i,j \in I$ with~$i<j$ it holds~$v_i \not \preceq v_j$ then we call it a \emph{bad sequence (with respect to $\preceq$)}. If for all~$i,j \in I$ it holds~$v_i \not \preceq v_j$, then we call it an \emph{antichain}, and if~$v_i \succ v_j$ for all~$i,j \in I$ with~$i < j$ we call it a \emph{strictly decreasing sequence}. Completing the above, if~$v_i \leq v_j$ for every~$i,j \in I$ with~$i<j$ we call the sequence a \emph{chain with respect to~$\preceq$}. Clearly ~$\Omega$ is a well-quasi-order if and only if there exist no infinite bad sequences with respect to $\preceq$.

We gather some well-known results regarding well-quasi-orderings.

\begin{observation}\label{obs:wqo_yields_infinite_chain}
    Let~$\Omega$ be a well-quasi-order and~$(v_\ell)_{\ell \in \N}$ a sequence with~$v_\ell \in V(\Omega)$ for every~$\ell \in \N$. Then there exists an infinite~$I\subset \N$ such that~$(v_i)_{i \in I}$ is a chain with respect to~$\preceq$.
\end{observation}

\begin{observation}\label{obs:wqo_of_tuples}
    Let~$\Omega_1,\ldots,\Omega_k$ be well-quasi-orders with~$\Omega_i = (V_i,\preceq_i)$ for~every~$1 \leq i \leq k$ and some~$k \in \N$. Then~$\Omega=(V(\Omega_1)\times \ldots \times V(\Omega_k), \preceq)$ where~$(v_1,\ldots,v_k)\preceq(u_1,\ldots,u_k)$ if and only if~$v_i \preceq_i u_i$ for all~$1 \leq i \leq k$.
\end{observation}

And finally we will need the following well-known result by Higman \cite{Higman1952}.

\begin{theorem}[Higman's Theorem]\label{thm:higman}
   Let~$\Omega$ be a well-quasi-order and~$(X_i)_{i \in \N}$ a sequence of finite sets~$X_i \subseteq V(\Omega)$. Then there exists an infinite subset~$I\subseteq \N$ such that for every~$1 \leq i < j$ with~$i,j \in I$ there is an injective map~$\alpha_i^j:X_i \to X_j$ with~$x \leq \alpha_i^j(x)$ for every~$x \in X_i$.
\end{theorem}

\section{The general framework}

We start with defining the basic framework and the tools needed throughout the paper and future work. This is stretched over three sections. The first section introduces terminology and known results regarding Eulerian digraphs. The second section introduces a new data-structure generalising Eulerian digraphs that we will call \emph{knitworks}. The third section gathers known results and techniques about well-quasi-ordering graphs by some containment relation and transfers them to our setting. Throughout the rest of this paper we will tacitly assume our graphs to be connected unless stated otherwise.

\subsection{The Structure of Eulerian Digraphs}

\paragraph{Euler-Grids:} It turns out that, similar to the undirected notion of treewidth and branch width, carving width comes with a dual obstruction bringing to light a richer structure of the graph: \emph{Euler-Grids}. The class of graphs admitting immersed Euler-Grids contains the class of graphs admitting immersed \emph{Swirls} and \emph{Routers}. 

\begin{definition}[Covers]
    Let~$G$ be an Eulerian digraph. Let~$\CCC\coloneqq \{C_1,\ldots,C_n\}$ be a collection of edge-disjoint cycles~$C_i \subseteq G$ for~$1 \leq i \leq n$ and some~$n \in \N$ such that~$G=C_1\cup \ldots \cup C_n$. Then we call~$\CCC$ a \emph{cycle cover (of order $n$) of~$G$} and we call~$(G,\CCC)$ a \emph{cycle-covered Eulerian graph}, or simply a \emph{covered graph}. If all the cycles are circles we call it a \emph{circle-cover} respectively. 
\end{definition}

\begin{definition}[Euler-Grid and Swirl]\label{def:swirl}
    Let~$k,n \in \N$ and let~$W_{k,n}$ be an undirected~$k \times n$ grid with vertices~$\{v_{i,j} \mid 1 \leq i \leq k, 1 \leq j \leq n\}$. Replace every vertex~$v_{i,j}$ by an undirected~$4$-cycle~$C_{i,j}=(l_{i,j},t_{i,j},r_{i,j},b_{i,j})$. Then for any edge~$\{v_{i,j},v_{i+1,j}\} \in E(W_{k,n})$ identify~$r_{i,j}$ with~$l_{i+1,j}$, and similary for every edge~$\{v_{i,j},v_{i,j+1}\}$ identify~$b_{i,j}$ with~$t_{i,j+1}$. Then this undirected graph is called the \emph{~$k \times n$ medial grid}.
    Assign to every cycle~$C_{i,j}$ some orientation by orienting the edges in the same direction forming a directed circle. Then the resulting graph is an Eulerian digraph and we refer to it as a \emph{$k \times n$ Euler-Grid} and denote it by~$\mathcal{E}_{k,n}$.

    If all the cycles have the same orientation with respect to some fixed orientation of the plane, we call the graph a~\emph{$k\times n$-swirl} denoted by~$\mathcal{S}_{k,n}$. If $k=n$ we may simply call it a $k$-swirl.\

\end{definition}
\begin{remark}
    To readers familiar with \cite{EDP_Euler}, the notion of swirl in this paper is qualitatively the same as the notion of induced swirls introduced in \cite{EDP_Euler}. It is for convenience that we slightly change its definition here.
\end{remark}
See \Cref{fig:euler-grid} for an illustration of an Euler-Grid and \Cref{fig:swirl} for an illustration of an eulerian embedded swirl.

\begin{figure}
\centering
\begin{subfigure}{0.4\textwidth}
\centering
    \begin{tikzpicture}[x=0.5cm,y=0.5cm, every node/.style={inner sep=1.5pt, fill=black, circle}]
      \tikzstyle{e}=[->,>=Stealth, thick, black]
        \foreach \i in {0, 2,..., 10}
        {
        \draw[orange!30!white] (1,\i) to (11,\i) ;
        \draw[orange!30!white] (\i+1,0) to (\i+1, 10) ;
        }
        \foreach \i in {0, 1, ..., 5}
        {
        \foreach \j in {0, 1, ..., 5}
        {
        \begin{scope}[xshift=\i cm, yshift=\j cm]
            \node[fill=orange!50!white] (g-\i-\j) at (1,0) {} ;
            \node (l-\i-\j) at (0,0) {  } ;
            \node (r-\i-\j) at (2,0) {  } ;
            \node (t-\i-\j) at (1,1) {  } ;
            \node (b-\i-\j) at (1,-1) {  } ;
            \ifodd\i
               \ifodd\j
                \draw[e] (l-\i-\j) to (t-\i-\j) ;
                \draw[e] (t-\i-\j) to (r-\i-\j) ;
                \draw[e] (r-\i-\j) to (b-\i-\j) ;
                \draw[e] (b-\i-\j) to (l-\i-\j) ;
              \else
                \draw[e] (l-\i-\j) to (b-\i-\j) ;
                \draw[e] (b-\i-\j) to (r-\i-\j) ;
                \draw[e] (r-\i-\j) to (t-\i-\j) ;
                \draw[e] (t-\i-\j) to (l-\i-\j) ;
              \fi
            \else
               \ifodd\j
                \draw[e] (l-\i-\j) to (b-\i-\j) ;
                \draw[e] (b-\i-\j) to (r-\i-\j) ;
                \draw[e] (r-\i-\j) to (t-\i-\j) ;
                \draw[e] (t-\i-\j) to (l-\i-\j) ;
              \else
                \draw[e] (l-\i-\j) to (t-\i-\j) ;
                \draw[e] (t-\i-\j) to (r-\i-\j) ;
                \draw[e] (r-\i-\j) to (b-\i-\j) ;
                \draw[e] (b-\i-\j) to (l-\i-\j) ;
              \fi
            \fi
        \end{scope}
        }
        }
    \end{tikzpicture}
    \caption{Every degree $4$ vertex is strongly planar.} 
    \label{fig:euler-grid}
\end{subfigure}
\begin{subfigure}{0.4\textwidth}
\centering
    \begin{tikzpicture}[x=0.5cm,y=0.5cm, every node/.style={inner sep=1.5pt, fill=black, circle}]
      \tikzstyle{e}=[->,>=Stealth, thick, black]
        \foreach \i in {0, 2,..., 10}
        {
        \draw[orange!30!white] (1,\i) to (11,\i) ;
        \draw[orange!30!white] (\i+1,0) to (\i+1, 10) ;
        }
        \foreach \i in {0, 1, ..., 5}
        {
        \foreach \j in {0, 1, ..., 5}
        {
        \begin{scope}[xshift=\i cm, yshift=\j cm]
            \node[fill=orange!50!white] (g-\i-\j) at (1,0) {} ;
            \node (l-\i-\j) at (0,0) {  } ;
            \node (r-\i-\j) at (2,0) {  } ;
            \node (t-\i-\j) at (1,1) {  } ;
            \node (b-\i-\j) at (1,-1) {  } ;
               \ifodd\j
                \draw[e] (l-\i-\j) to (t-\i-\j) ;
                \draw[e] (t-\i-\j) to (r-\i-\j) ;
                \draw[e] (r-\i-\j) to (b-\i-\j) ;
                \draw[e] (b-\i-\j) to (l-\i-\j) ;
              \else
                \draw[e] (l-\i-\j) to (t-\i-\j) ;
                \draw[e] (t-\i-\j) to (r-\i-\j) ;
                \draw[e] (r-\i-\j) to (b-\i-\j) ;
                \draw[e] (b-\i-\j) to (l-\i-\j) ;
              \fi
             
        \end{scope}
        }
        }
    \end{tikzpicture}
    \caption{Every degree $4$ vertex is eulerian embedded.} 
    \label{fig:swirl}
\end{subfigure}
\caption{Two $6 \times 6$ Euler-grids depicted in black with underlying undirected grids depicted in orange. The right-hand side depicts a $6 \times 6$-swirl. }
\end{figure}
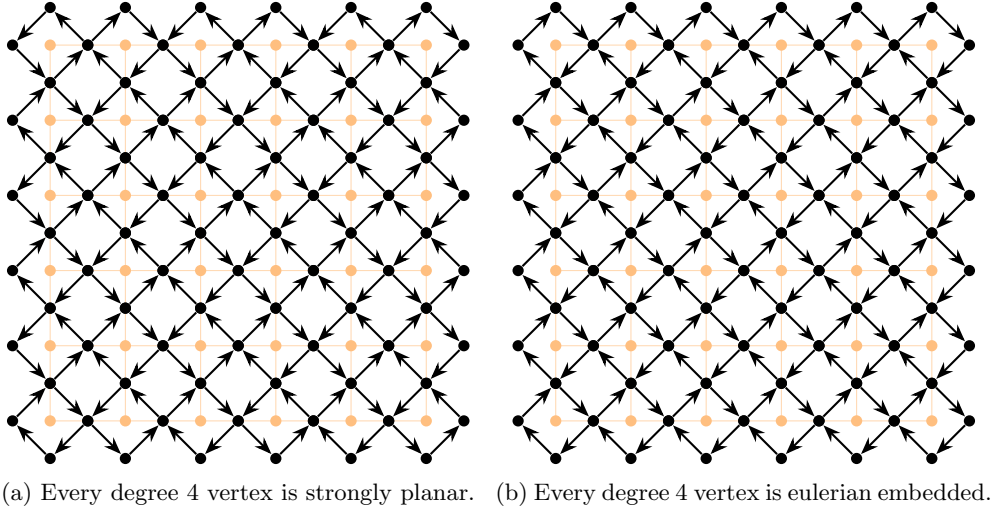

Using simple but tedious combinatorial arguments one easily derives that a large Euler-Grid contains a large swirl as an immersion.

\begin{lemma}
\label{lem:swirl_immerses_in_eulergrid}
For all $k \geq 1$, $\SSS_{k,k} \hookrightarrow \EEE_{2k^2,2k^2}$.
\end{lemma}
\begin{proof}
 We construct an immersion of $\SSS_{k,k}$ into $\EEE_{2k^2,2k^2}$ as follows.
 For $1 \leq i,j \leq 2k^2$ let $C_{i,j}$ be the $4$-circle of the canonical cycle cover introduced to replace the vertex $v_{i,j}$ in the construction of the Euler grid.
 We partition $\EEE_{2k^2,2k^2}$ into $2k \cdot 2k$ ``$k \times k$-squares'' as follows.
 For $1 \leq i,j \leq 2k$ let 
 \[
    H_{i,j} \coloneqq \bigcup \{ C_{s,t} \mid (i-1)\cdot k +1 \leq s \leq i\cdot k, (j-1)\cdot k + 1 \leq j \leq j\cdot k \}.
 \]
Thus, each $H_{i,j}$ is a $k\times k$-Euler grid. If for some $i, j$ all $4$-cycles $C_{s, t}$ in $H_{i,j}$ are oriented in the same direction, then $H_{i,j}$ is a swirl and we are done.

Otherwise, in each $H_{i,j}$ there is a $4$-cycle oriented 
clockwise and a $4$-cycle oriented anti-clockwise.
We choose for all $1 \leq i, j \leq k$, with $i,j$ both odd, 
a $4$-cycle $O_{i,j} \coloneqq C_{x_{i,j}, y_{i,j}}$ in 
anti-clockwise orientation. Let $l_{i,j} := l_{x_{i,j}, y_{i,j}}$ 
and likewise for $r_{i,j}, b_{i,j}$, and $t_{i,j}$.

Now it is easily seen that 
\begin{itemize}
\item for all $1 \leq j \leq k$ and all $1 \leq i < k$ there are two edge-disjoint paths $R_{i,j}$ from $r_{i,j}$ to $l_{i+1,j}$ and $L_{i,j}$ from $l_{i+1,j}$ to $r_{i,j}$ and
\item  for all $1 \leq i \leq k$ and all $1 \leq j < k$ there are two edge-disjoint paths $D_{i,j}$ from $b_{i,j}$ to $t_{i,j+1}$ and $U_{i,j}$ from $t_{i+1,j}$ to $b_{i,j}$
\end{itemize}
such that all these paths are pairwise edge disjoint. 

Now, for all $1 \leq j  \leq k$ and all $1 \leq i < k$ let 
$T^+_{i,j}$ be the path from $t_{i+1,j}$ to $r_{i,j}$ that starts 
with the edge $\big(t_{i+1,j}, l_{i+1,j}\big)$ and then follows 
$L_{i,j}$. Similarly we obtain from $R_{i,j}$ a path from 
$r_{i,j}$ to $b_{i+1,j}$.
We now split off along these paths. As a result we get, for each $1 \leq i,j \leq k$ a sequence of new $4$-cycles $O'_{i,j}$ with vertices $t'_{i,j}, l'_{i,j}, b'_{i,j}$, and $r'_{i,j}$ such that $t'_{i,j} = t_{i,j}$, $r'_{i,j} = r'_{i,j}$, $b'_{i,j} = b_{i,j}$, and, for $1 < i \leq k$, $l'_{i,j} = r'_{i-1, j}$. 

Thus, in each row $j$ we now have a sequence of consecutive $4$-cycles all oriented anti-clockwise as required. In the next step we connect the cycles to the cycles in the row below using the same idea. 

For $1 \leq j < k$ and $1 \leq i \leq k$ we obtain 
a path $D'_{i,j}$ from $l'_{i,j}$ to $t'_{i, j+1}$ from $D_{i,j}$ and the edge $(l'_{i,j}, b'_{i,j})$. Similarly we obtain a path $U'_{i,j}$ from $t'_{i, j+1}$ to $r'_{i,j}$ from $U_{i,j}$ and the edge $(b'_{i,j},  r'_{i,j})$.
Splitting off along these paths yields the required swirl. 
\end{proof}

While, as we will prove below, immersed swirls are to Eulerian digraphs (of bounded degree) and~carving width what grid-minors are to undirected graphs and treewidth, immersed \emph{Routers} take the role of clique-minors in the directed Eulerian immersion setting. 

\begin{definition}[Router]\label{def:router}
    Let~$G$ be an Eulerian digraph on~$\binom{k}{2}$ vertices given by~$V(G)=\{v_{i,j} \mid 1 \leq i \neq j \leq k\}$. Let~$\CCC \coloneqq \{C_1,\ldots,C_k\}$ be a circle cover of~$G$ such that~$V(C_i) \cap V(C_j) = \{v_{i,j}\}$ for all $1 \leq i < j \leq k$. Then we call~$G$ \emph{a~$k$-router} and we denote it by~$\RRR_k$.
\end{definition}

We take the time to prove a nice result strengthening the intuition that Euler-Grids (more precisely swirls) are the structural equivalents to grids with regards to Eulerian digraphs and immersion. 
In \cite{GMV}, Robertson and Seymour proved the following theorem. 

\begin{theorem}
    There exists a function~$f_{\WWW}:\N \to \N$ such that the following holds. Let~$G$ be an undirected graph. Then~$\tw{G} \geq f_{\WWW}(k)$ implies that~$G$ admits a~$k\times k$-wall as a subgraph.
    \label{thm:undirected_wall}
\end{theorem}

In the spirit of \cref{thm:undirected_wall}, Cavallaro, Kawarabayashi, and Kreutzer have shown that Eulerian digraphs of degree four and high treewidth admit large induced swirls \cite{EDP_Euler}\footnote{Note that the definition of (induced) swirls in \cite{EDP_Euler} differs from our definition but as mentioned previously they are qualitatively the same.}. An equivalent result has also been obtained by Johnson in \cite[Theorem 2.5]{Johnson2002}, where swirls are called medial grids. 

\begin{theorem}[\cite{Johnson2002,EDP_Euler}]\label{thm:4_reg_swirl}
    There exists a function~$f_{\SSS}:\N \to \N$ such that the following holds. Let~$G$ be an Eulerian directed graph of maximum degree~$4$. Then~$\tw{G} \geq f_{\SSS}(k)$ implies that~$G$ strongly immerses a~$k$-swirl.
\end{theorem}

The previous theorem will be an important tool in the sequel. We show next that, using \cref{thm:undirected_wall} and \cref{thm:4_reg_swirl}, its statement can be generalised to Eulerian digraphs of unbounded degree, but at the expense that we are no longer guaranteed to get a strong immersion.

\begin{lemma}
    Let~$G$ be an Eulerian digraph and~$k \geq 4$. Then there exists~$f: \N \to \N$ such that~$\tw{G} \geq f(k)$ implies that $G$ immerses $\SSS_{k,k}$ (not necessarily strong). In particular some Euler-Grid immerses in~$G$.
    \label{thm:high_tw_implies_high_ebw_grid_version}
\end{lemma}
\begin{proof}
    Let~$h(k) \coloneqq f_\SSS(k)$ and choose~$f(k) \coloneqq f_\WWW(h(k))$ where~$f_\SSS$ is as in \cref{thm:4_reg_swirl} and~$f_\WWW$ is as in \cref{thm:undirected_wall}.
    
    Since~$\tw{G} \geq f_\WWW(h(k))$, $G$ admits an~$h(k)\times h(k)$-wall~$\WWW$ as a subgraph. In particular the subgraph is of maximum degree~$3$. Now let~$\CCC$ be any circle cover of~$G$ and split~$G$ off along~$\CCC$ in direction~$(e_1,e_2) \subset C$ for any~$C \in \CCC$ whenever~$\{e_1,e_2\} \cap E(\WWW) = \emptyset$; refer to the resulting graph as~$\epsilon\WWW$. Clearly~$\Delta(\epsilon\WWW) \leq 6$ for any vertex of~$\WWW$ is part of at most three circles in~$\CCC$ that use edges of~$\WWW$. By construction~$\WWW \subset \epsilon\WWW$ and thus~$\tw{\epsilon \WWW} \geq h(k)$. We continue to deform~$\epsilon\WWW$ into a graph of maximum degree~$4$. To this extent let~$\FFF_3 = (V,E)$ be a~$3\times 3$-fence defined via
    \begin{align*}
        V \coloneqq& \{v_i^j,s_i,t_i \mid 1 \leq i,j \leq 3\},\\
        E \coloneqq& \{(v_i^1,v_i^2),(v_i^2,v_i^3),(v_i^3,v_i^1) \mid 1 \leq i \leq 3\}  \cup\\& \{(v_1^j,v_2^j),(v_2^j,v_3^j) \mid 1 \leq j \leq 3\} \cup \\
         & \{(s_j,v_1^j) \mid 1 \leq j \leq 3\} \cup \{(v_3^j,t_j) \mid 1 \leq j \leq 3\}.
    \end{align*}
    
    Then replacing every degree~$6$ vertex of~$\epsilon\WWW$ by a copy of~$\FFF_3$ and identify the in-neighbors of~$v$ with~$s_1,s_2,s_3$ and the out-neighbors of~$v$ with~$t_1,t_2,t_3$; call the resulting graph~$\SS$.
    Clearly~$\tw{\SS} \geq h(k)$ for we can simply contract the gadgets. 

    \begin{claim}
        $\SSS$ (strongly) immerses a~$k$-swirl.
    \end{claim}
    \begin{claimproof}
        Since~$\Delta(\SSS) = 4$ this follows from \cref{thm:4_reg_swirl} using~$\tw{\SSS} \geq h(k)$ and the definition of~$h$. 
    \end{claimproof}

    \begin{claim}
        $\epsilon\WWW$ (strongly) immerses a~$k$-swirl.
    \end{claim}
    \begin{claimproof}
        This is imminent from the fact that~$k\geq 4$ implies that the swirl cannot be contained in any~$\FFF_3$ gadget and thus whenever the~$k$-swirl uses an in-edge through $s_i$ of an~$\FFF_3$ gadget it must use an out-edge through some $t_j$ too (using Eulerianness) and thus we can replace this path in~$\epsilon\WWW$ by using the degree~$6$ vertex and routing accordingly. Note here that the vertices of the swirl are all of degree exactly four and since at most three paths can enter and leave an~$\FFF_3$ gadget, at most one vertex of the gadget is used in the immersion model of the swirl for otherwise we would need 4 paths to enter and leave the gadget.
    \end{claimproof}
    This concludes the proof (note that we may loose strong immersion by going back).
\end{proof}

Note that \cref{thm:high_tw_implies_high_ebw_grid_version} strengthens the relation given by \cref{thm:qualitative_equivalence_of_tw_and_ebw} with a structural dual. That is, we know that high undirected tree-width does not only imply high carving width, but it moreover implies the existence of an Euler-Grid witnessing high carving width. Unfortunately, as we have seen above, high carving width in general does not imply the existence of an Euler-Grid; but \cref{thm:high_tw_implies_high_ebw_grid_version} implies that it does if we assume bounded degree.

\begin{theorem}%
    Let~$G$ be an Eulerian digraph of degree bounded by~$\Delta \in \N$. There exists a computable function~$f_\Delta:\N \to \N$ such that~$\ebw{G} \geq f_\Delta(k)$ implies that $G$ immerses~$\EEE_{k,k}$.
    \label{thm:grid_theorem_for_etb}
\end{theorem}
\begin{proof}
    This is immediate from \cref{thm:qualitative_equivalence_of_tw_and_ebw} and \cref{thm:high_tw_implies_high_ebw_grid_version}.
\end{proof}

As our graphs will mostly only immerse swirls, we define the following.

\begin{definition}[Immersed Swirls]\label{def:immersed_swirl}
    Let~$G$ be Eulerian and~$H \subseteq G$ a subgraph such that there exists a~$(k,n)$-swirl~$\SSS_{k,n}$ and a strong immersion~$\gamma:\SSS_{k,n} \hookrightarrow G$ with~$\gamma(\SSS_{k,n}) = H$, then we call~$H$ an \emph{immersed~$(k,n)$-swirl}.
\end{definition} %

\subsection{Rooted Graphs and Knitworks}
\label{subsec:knitworks}
As so often when dealing with well-quasi-ordering graphs by some containment relation \cite{GMIV,Gee02}, it will be very helpful to switch to a \emph{rooted} setting simplifying inductive reasoning for surface-embeddings by exploiting structural properties of the graph class (and embeddings) at hand.
\subsubsection{Rooted Eulerian Digraphs} 

\begin{definition}\label{def:rooted_graph}
Let $k \in 2\N$. Let~$G$ be an Eulerian digraph and let~$E=\{e_1,\ldots,e_k\} \subseteq E(G)$ be a set of edges. Let $\pi_G(E)=(e_1,\ldots,e_k)$ be an ordering of $E$. Then we call~$\bar{G} \coloneqq (G,\pi_G(E))$ a \emph{rooted Eulerian digraph}. We call each~$e\in E$ a \emph{root edge} and~$k$ the \emph{index} of the rooted digraph.

Let $X \subset V(G)$ induce a $k$-cut $\rho(X) = \{e_1,\ldots,e_k\}$ in $G$ and let~$\pi_G(X)\coloneqq (e_1,\ldots,e_k)$ be an ordering of~$\rho(X)$. Then we write~$\bar{G} \coloneqq (G,\pi_G(X))$ instead of $\bar{G} \coloneqq (G,\pi_G(\rho(X)))$ for the rooted Eulerian digraph highlighting the respective induced cut the root edges are coming from.

Finally, given a rooted digraph~$\bar{G}=(G,\pi_G(E))$ (or $\bar G = (G,\pi_G(X))$ we call~$G$ its \emph{underlying digraph}.
\end{definition}
\begin{remark}
Recall that for Eulerian digraphs, induced cuts always have even order and thus the above is well-defined.
    
If the graphs we are talking about are clear from context we write~$\pi(E)$ and $\pi(X)$ instead of~$\pi_G(E),\pi_G(X)$, and vice-versa we may add the subscript to remove ambiguity.

    Furthermore, we will write $\pi(e) < \pi(e')$ if $e$ precedes $e'$ given some order $\pi$. 
\end{remark}

We will for most of the exposition work with graphs rooted in induced cuts $X$ as induced cuts come with a natural partition into two sides $X$ and $\bar{X}$. However, in \Cref{sec:cylinder,sec:disc,sec:la-grande-inductione} we will need to switch to the more general setting which is why we introduce it here already and try to state most of the results just as general as we need them.

\smallskip

We fix the following. Given sets~$E,E'$, a map~$\gamma:E \to E'$, and an ordering~$\pi=(e_1,\ldots,e_k)$ of elements~$e_i \in E$ for~$k \in \N$, we define~$\gamma(\pi) = (\gamma(e_1),\ldots,\gamma(e_k))$ and thus~$\gamma$ naturally lifts to a map~$\gamma:E^k \to {E'}^k$ by applying~$\gamma$ componentwise. Similarly for a set~$S$ we define~$\gamma(S) \coloneqq \{\gamma(s) \mid s \in S\}$. Given a second ordering $\pi'=(f_1,\ldots,f_\ell)$, with $\{ f_1, \dots, f_{\ell} \} \cap \{ e_1, \dots, e_k \} = \emptyset$,  we write $\pi' \circ \pi \coloneqq (f_1,\ldots,f_\ell,e_1,\ldots,e_k)$ for the \emph{concatenation of orderings} in accordance with our definition of concatenation of paths.%

We lift the definition of isomorphism to rooted graphs in a natural way by taking an isomorphism between the underlying digraphs and additionally requiring it to map roots to roots respecting their order.
\begin{definition}
    Let~$\bar{G}=(G,\pi(E)), \bar{G}'=(G',\pi(E'))$ be rooted Eulerian digraphs. We say that~$\bar{G}$ and~$\bar{G'}$ are \emph{isomorphic} written~$\bar{G} \cong \bar{G}'$ if and only if they share the same index~$k\in 2\N$ and there exists a bijection~$\xi: G \to G'$ such that~$\restr{\xi}{V(G)}: V(G) \to V(G')$ and~$\restr{\xi}{E(G)}: E(G) \to E(G')$ are bijections satisfying the following
    \begin{enumerate}
        \item for~$e=(u,v) \in E(G)$ it holds~$\xi(e) = e' \in E(G')$ if~$e'=(\xi(u),\xi(v))$, and
        \item for~$\pi(E)=(e_1,\ldots,e_k)$ and~$\pi(E')=(e_1',\ldots,e_k')$ it holds~$\xi(e_i) = e_i'$, i.e.,~$\xi(\pi(E)) = \pi(E')$.
    \end{enumerate}
    That is~$\xi(\bar{G}) = \bar{G}'$, and we write~$\xi:(G,\pi(E)) \cong (G',\pi(E'))$.
\end{definition}%

Similarly to the above we lift the definition of immersions to rooted Eulerian digraphs as follows. 

\begin{definition}[Immersion of Rooted Graphs]\label{def:rooted_immersion}
    Let~$\bar{G} = (G,\pi(E))$ and~$\bar{G'}=(G',\pi(E'))$ be rooted Eulerian digraphs. We say that~$\bar{G}'$ \emph{(strongly) immerses}~$G$ if both have the same index~$k \in 2\N$ and additionally there is a map $\gamma: V(G)\cup E(G) \to G'$ satisfying the following:
    \begin{enumerate}
        \item $\gamma$ is a (strong) immersion, and
        \item given the natural bijection~$\eta: E \to E'$ satisfying~$\eta(\pi_G(E)) = \pi_{G'}(E')$, the path~$\gamma(e)$ for~$e \in E$ contains~$e' \in E'$ if and only if~$\eta(e) = e'$. We say that \emph{$\gamma$ respects the order of the roots}.
    \end{enumerate}
    We call~$\gamma$ a \emph{(strong) rooted immersion} or simply a (strong) immersion of~$\bar{G}$ in~$\bar{G'}$ and write~$\gamma: \bar{G} \hookrightarrow \bar{G}'$ for strong immersion and $\gamma: \bar{G} \hookrightarrow^* \bar{G}'$ for general immersion.
\end{definition}
\begin{remark}
    Note that when taking the root edges from induced cuts $X,X'$, then $E=\rho_G(X)$ and $E'=\rho_{G'}(X')$.
\end{remark}

Further we define what we mean by \emph{rooted cuts}; this makes the above discussion about ``keeping track of sides'' when rooting in cuts more apparent.
\begin{definition}[Rooted Cut]
    Let~$\bar{G}= (G,\pi(X))$ be a rooted Eulerian digraph. Let~$Y \subset V(G)$, then we say that $Y$ \emph{induces a rooted cut (in $\bar G$)} if~$X \subseteq Y$. 
\end{definition}
\begin{remark}
    In particular~$\rho(X) \cap E(G[\bar{Y}]) = \emptyset$ but~$\rho(X) \cap \rho(Y)$ may not be empty.
\end{remark}%

\subsubsection{Knitworks} The following definitions are inspired by the definition of \emph{patchworks} due to Robertson and Seymour \cite{GMIV,GMXX, GMXXIII} but are different and of their own interest as they are suited for immersions rather than minors which comes with new and distinct caveats. The intuition for ``knitworks'' comes from the idea that Eulerian digraphs may be viewed as collections of disjoint cycles that have been ``knitted'' together at vertices. We will need \emph{links} to keep track of ``how'' certain parts of the graph can be knitted to the graph.

\begin{definition}[Matchings and Links]\label{def:well-linked_links}
  Let~$E_1,E_2$ be two sets of distinct elements with~$\Abs{E_1} = \Abs{E_2}$. A \emph{matching of $(E_1,E_2)$} is a set~$M\subset E_1 \times E_2$ such that for every element~$e_1 \in E_1$ there is at most one element~$m=(e,e') \in M$ with~$e_1 = e$ and similarly for~$e_2 \in E_2$ there is at most one element~$m=(e,e') \in M$ with~$e_2 = e'$. We call a matching~$M$ \emph{perfect} if for every element~$e_1 \in E_1$ and every element~$e_2 \in E_2$ there exist~$e_1'\in E_1$ and~$e_2' \in E_2$ with~$(e_1,e_1'),(e_2,e_2') \in M$. We denote by~$\operatorname{Match}(E_1,E_2)$ the set of all matchings of~$(E_1,E_2)$.

   Let~$\mathfrak{M} \subseteq \operatorname{Match}(E_1,E_2)$ be a set of matchings, then we call~$\mathfrak{M}$ a \emph{link}. We call $\mathfrak{M}$ \emph{reliable} if for every $M \in \mathfrak{M}$ and every $M' \subseteq M$ it holds $M' \in \mathfrak{M}$. 
  We call~$\mathfrak{M}$ \emph{linkable} if for every~$(e_1,e_2) \in E_1 \times E_2$ there exists~$M \in \mathfrak{M}$ with~$(e_1,e_2) \in M$.

  We call~$\mathfrak{M}$ \emph{well-linked} if for all pairs~$E_1',E_2'$ with~$E_j' \subseteq E_j$ and~$\Abs{E_1'} = \Abs{E_2'}$  and for all partitions~$E_j'= X_j \cup Y_j$ with~$\Abs{X_1} = \Abs{X_2}$ and~$\Abs{Y_1} = \Abs{Y_2}$ for~$j=1,2$ there exists~$M \in \mathfrak{M}$ such that~$M = M_X \cup M_Y$ where~$M_X \in \operatorname{Match}(X_1,X_2)$ and~$M_Y \in \operatorname{Match}(Y_1,Y_2)$ are perfect matchings.

\end{definition}

We will use links and matchings to keep track of ``how'' a strong immersion is allowed to interact with parts of the graph that are to be knitted back in during the induction process. 

\begin{definition}[Matchings of~$\LLL$ and~$\gamma$]\label{def:matching_of_gamma}
    Let~$G$ be a digraph and~$\LLL$ a linkage in~$G$. Let~$v \in V(G)$. We define~$M_\LLL(v) \coloneqq \{(e,e') \mid e,e' \in \rho(v) \text{ and } (e,e') \subseteq L \text{ for some } L\in\LLL \}$.
    
    Let~$G'$ be another digraph and~$\gamma: E(G) \cup V(G) \to G'$ an immersion. Let~$v' \in V(G')$. We define~$M_\gamma(v') \in \operatorname{Match}(\rho^-(v'),\rho^+(v'))$ to be the matching defined via~$(e,e') \in M_\gamma(v')$ if and only if there is~$e^* \in E(G)$ such that~$(e,e') \subseteq \gamma(e^*)$ is a sub-path.
\end{definition}
\begin{remark}
    Clearly~$\gamma$ is strong if, and only if, $M_\gamma(v') = \emptyset$ for every~$v' \in \gamma(V(G))$.
\end{remark}

The following is the main data structure for \cref{sec:bounded-case} and will be of interest for our future work. Since for this paper we will only need to root Eulerian digraphs in \emph{induced} cuts whenever we work with ``knitworks'', we define \emph{$\Omega$-knitworks} accordingly.
\begin{definition}[$\Omega$-Knitworks]\label{def:knitwork}
    Let~$\Omega = (V(\Omega),\preceq)$ be a well-quasi-order. An \emph{$\Omega$-knitwork} is a tuple~$\GGG=(\bar{G},\mu,\m,\Phi)$ such that
    \begin{enumerate}
        \item $\bar{G}=(G,\pi(X))$ is a rooted Eulerian digraph for some $X \subset V(G)$, 
        \item $\mu$ is a function with~$\dom(\mu) \subseteq V(G)$ and~$\mu(v)$ is an ordering of~$\rho(v)$,
        \item $\m$ is a function with~$\dom(\m) \subseteq \dom(\mu)$ and~$\m(v) \subseteq \operatorname{Match}(\rho^-(v),\rho^+(v))$ is a reliable link, and
        \item $\Phi$ is a function with~$\dom(\Phi) \subseteq V(G)$ and~$\Phi(v) \in V(\Omega)$.
    \end{enumerate}
    If for all~$v \in \dom(\m)$ the links~$\m(v)$ are linkable or well-linked we call~$\GGG$ \emph{linkable} or  \emph{well-linked}, respectively.

    Given the rooted Eulerian digraph~$\bar{G}$ we call~$\GGG$ an~\emph{$\Omega$-knitwork for $\bar G$}.
\end{definition}
\begin{remark}
    Note that one may easily extend the definition of $\Omega$-knitworks to Eulerian digraphs rooted in arbitrary sets $E\subseteq E(G)$.
\end{remark}

The notions established for rooted Eulerian digraphs naturally lift to~$\Omega$-knitworks as exemplified by the following.
\begin{definition}
    Let~$\Omega$ be a well-quasi-order and let $\GGG=(\bar{G},\mu,\m,\Phi)$ be an $\Omega$-knitwork. Then~$Y \subset V(G)$ induces a \emph{rooted cut} in~$\GGG$ if and only if it induces a rooted cut in~$\bar{G}$.
\end{definition}

One may abstractly think of an~$\Omega$-knitwork~$(\bar{G},\mu,\m,\Phi)$ as a rooted Eulerian digraph together with a set of vertices~$\dom(\mu)$ which are placeholders marking cuts---namely~$\rho(v)$---where we will later ``knit'' another rooted Eulerian digraph, $(H,\pi(Y))$ say, into that cut by replacing~$v$, identifying~$\rho_G(v)$ with~$Y$ such that~$\mu_G(v)$ and~$\pi_H(Y)$ agree, and so that the ``immersion-type'' of~$(H,\pi(Y))$ agrees with the label~$\Phi(v)$. The map~$\m$ is used to keep track of the ``feasible linkages'' in~$H$ with both ends in $Y$, making sure that the immersion adheres to the restrictions posed by the Eulerian embeddings (two paths cannot ``cross''). With this in mind, the following definitions of \emph{isomorphism} and \emph{immersion of~$\Omega$-knitworks} may strengthen that intuition.

\begin{definition}[Isomorphic Knitworks]\label{def:isomorphism_knitwork}
     Let~$\Omega$ be a well-quasi-order and let~$\GGG=((G,\pi(X)),\mu,\m,\Phi)$ and~$\GGG'=((G',\pi(X')),\mu',\m',\Phi')$ be~$\Omega$-knitworks. Then we say that~$\GGG \cong \GGG'$ if and only if there is a map~$\xi:(G,\pi(X)) \cong (G',\pi(X'))$ satisfying~$\Phi'(\xi(x)) = \Phi(x)$, i.e., the vertices share the same~$\Omega$-label,~$\mu'(\xi(x)) = \xi(\mu(x))$ and~$\m'(\xi(x)) = \xi(\m(x))$. We write~$\xi:\GGG \cong \GGG'$.
\end{definition}

\begin{definition}[$\Omega$-knitwork immersion]\label{def:knitwork_immersion}
    Let~$\Omega=(V(\Omega),\preceq)$ be a well-quasi-order.
    
    Let $(\bar{G},\mu,\m,\Phi)$ and $(\bar{G}',\mu',\m',\Phi')$ be~$\Omega$-knitworks.
    We say that~$(\bar{G}',\mu',\m',\Phi')$ \emph{(strongly) immerses} $(\bar{G},\mu,\m,\Phi)$ if and only if there exists a map~$\gamma: V(G) \cup E(G) \to G'$ satisfying the following
    \begin{enumerate}
        \item[1.] $\gamma$ (strongly) immerses~$\bar{G}$ into~$\bar{G}'$; in particular~$\gamma$ respects the order of the roots,
        \item[2.] for each~$v \in V(G)$,~$\gamma(v) \in \dom(\mu')$ if and only if~$v \in \dom(\mu)$, and~$\gamma(v) \in \dom(\m')$ if and only if~$v \in \dom(\m)$,
        \item[3.] for each~$v \in \dom(\mu)$ the index of~$\mu(v)=(e_1,\ldots,e_k)$ and~$\mu'(\gamma(v))=(e_1',\ldots,e_k')$ agree and is $k \in 2\N$, say, and for~$1 \leq i \leq k$ the edge~$e_i'$ is contained in~$\gamma(e_i)$; we say that $\gamma$ \emph{respects~$\mu$ and~$\mu'$},
        \item[4.] for each~$v' \in \dom(\m') \setminus \gamma(\dom(\m))$,~$M_\gamma(v') \in \m'(v')$, 
        \item[5.] for each~$v \in \dom(\m)$,~$(e_i,e_j) \in \m(v)$ if and only if $(e_i',e_j') \in \m'(\gamma(v))$, where $(e_1,\ldots,e_k)=\mu(v)$ and~$(e_1',\ldots,e_k')=\mu'(\gamma(v))$, and
        \item[6.] for every~$v \in \dom(\Phi)$ we have~$\Phi(v) \preceq \Phi'(\gamma(v))$, in particular~$\gamma(\dom(\Phi)) \subseteq \dom(\Phi')$ .
    \end{enumerate}
    We call the map~$\gamma$ a \emph{(strong) $\Omega$-knitwork immersion} or simply \emph{(strong) immersion} if clear from context, and write~$\gamma:(\bar{G},\mu,\m,\Phi) \hookrightarrow (\bar{G}',\mu',\m',\Phi')$ for strong immersion and $\gamma:(\bar{G},\mu,\m,\Phi) \hookrightarrow^* (\bar{G}',\mu',\m',\Phi')$ for immersion. 
\end{definition}
\begin{remark}
    It is a natural idea to loosen the restriction of (3) and only require the index of~$\mu'(\gamma(v))$ to be at least as large as the index of~$\mu(v)$. Since this is not needed in this exposition, we settled for the above definition as it makes certain arguments easier. In future work however it seems, up until now, unavoidable to loosen (3).
    Condition (4) may intuitively be thought of as a check that, whenever $\gamma$ routes through a vertex $v' \in \dom(\m')$, then $\gamma$ is indeed \emph{realisable} in $v'$, where $v'$ is to be thought of as a placeholder vertex encoding the feasible linkages through the cut $\rho(v')$. Condition (5) may intuitively be thought of as a guarantee that the placeholder vertices and their respective cuts are mapped properly.
\end{remark}

The following is imminent.
\begin{lemma}\label{obs:knitwork_immersion_on_mu}
    Let~$\Omega$ be a well-quasi-order, let~$\GGG = (\bar{G},\mu,\m,\Phi)$ and~$\GGG'=(\bar{G}',\mu',\m',\Phi')$ be~$\Omega$-knitworks and let~$\gamma: \GGG\hookrightarrow^* \GGG'$ be an~$\Omega$-knitwork immersion. Let~$v\in \dom(\mu)$ and let~$(e_1,\ldots,e_k) = \mu(v)$. Further let~$(e_1',\ldots,e_k') = \mu(v')$ for~$v' = \gamma(v)$.
    Then~$e_i' \in \rho^+(v') \iff e_i \in \rho^+(v)$.
\end{lemma}
\begin{proof}
    By \cref{def:knitwork_immersion} of~$\Omega$-knitwork immersion~$\gamma$ respects~$\mu$ and $\mu'$ whence~$e_i' \in \gamma(e_i)$ for all~$1 \leq i \leq k$. Since~$\gamma$ is in particular an immersion of~$G$ into~$G'$, the edge~$e_i=(v,w) \in \rho^+(v)$ for some~$i \in \{1,\ldots,k\}$ is mapped to a path starting in~$v' = \gamma(v)$. If the path~$\gamma(e_i)$ does not start in~$e_i'$ then it must start in some other edge~$f \in N_{G'}(v')$, but clearly~$f$ cannot be a loop since~$e_i \in \rho_{G}(v)$ is not a loop by \cref{obs:cut_at_vertex}; thus~$f = e_j$ for some~$j \in \{1,\ldots,k\}$ with~$i \neq k$. But then~$e_j  \in \gamma(e_i') \cap \gamma(e_j')$; a contradiction to the paths being edge-disjoint using the fact that~$\gamma$ is an immersion of rooted Eulerian digraphs.
\end{proof}

\smallskip

One of our main goals is to prove that certain subclasses of~$\Omega$-knitworks (for example those of bounded carving width which we define and discuss in \cref{sec:bounded-case}) are well-quasi-ordered by~$\Omega$-knitwork immersion. To this extend we need to prove that this newly introduced immersion relation induces a quasi-order on~$\Omega$-knitworks, i.e., it is transitive and reflexive. 

\begin{definition}\label{def:quasi_order_induced_by_knitimm}
    Let~$\Omega=(V(\Omega),\ll)$ be a well-quasi-order. Define a binary relation~$\preceq$ on~$\Omega$-knitworks via~$(\bar{G},\mu,\m,\Phi) \preceq (\bar{G}',\mu',\m',\Phi')$ if and only if there is a strong immersion~$\gamma:  (\bar{G},\mu,\m,\Phi)\hookrightarrow (\bar{G}',\mu',\m',\Phi') $. We call~$\preceq$ the \emph{quasi-order induced by strong~$\Omega$-knitwork immersion}. Similarly define $\preceq^*$ for smore general immersion.

    Let~$\bar{G} = (G,\pi(X))$ and~$\bar{G}'=(G',\pi(X'))$. Then we define~$(\bar{G},\mu,\m,\Phi) \cong (\bar{G}',\mu',\m',\Phi')$ if there exists an isomorphism~$\alpha: \bar{G}' \to \bar{G}$  such that~$((\alpha(\bar{G}'),\mu'\circ\alpha,\m'\circ \alpha,\Phi'\circ\alpha) = (\bar{G},\mu,\m,\Phi)$.
\end{definition}
We next prove that the name ``quasi-order'' is indeed justified.

\begin{lemma}\label{lem:knitwork_imm_is_quasi_order}
    Let~$\Omega=(V(\Omega),\ll)$ be a well-quasi-order and~$\preceq,\preceq^*$ the quasi-orders induced by (strong)~$\Omega$-immersion. Then~$\preceq, \preceq^*$ (for both strong and standard immersion) are quasi-orders on~$\Omega$-knitworks.

\end{lemma}
\begin{proof}
    We provide a proof fro~$\preceq$ since the proof for general immersion is analogous.
    The relation is clearly reflexive via the identity map. We continue with transitivity; to this extent let~$(\bar{G}_i,\mu_i,\m_i,\Phi_i)$ be~$\Omega$-knitworks with~$\bar{G}_i = (G_i,\pi(X_i))$ for~$i \in \{1,2,3\}$ such that~$\gamma_1: (\bar{G}_1,\mu_1,\m_1,\Phi_1) \hookrightarrow (\bar{G}_2,\mu_2,\m_2,\Phi_2)$ and~$\gamma_2: (\bar{G}_2,\mu_2,\m_2,\Phi_2) \hookrightarrow (\bar{G}_3,\mu_3,\m_3,\Phi_3)$ by strong~$\Omega$-knitwork immersion. In particular~$\gamma$ is a strong rooted immersion by 1. of \cref{def:knitwork_immersion} and thus~$\Abs{\rho(X_i)} = k \in 2\N$ for all~$i \in \{1,2,3\}$.
    Let~$\gamma: V(G_1) \cup E(G_1) \to G_3$ be defined via~$\gamma = \gamma_2 \circ \gamma_1$. 
    \begin{claim}
        $\gamma$ immerses~$\bar{G}_1$ into~$\bar{G}_3$.
    \end{claim}
    \begin{claimproof}
        Since~$\gamma_1,\gamma_2$ are strong immersions on the underlying unrooted Eulerian digraphs, using the transitivity for strong immersion we have that~$\gamma: G_1 \hookrightarrow G_3$ is a strong immersion. It remains to show that~$\gamma$ respects the order of the roots. To this extend let~$\pi(X_i)=(e_1^i,\ldots,e_k^i)$; by 1. of \cref{def:knitwork_immersion} and \cref{def:rooted_immersion} we derive that~$e_i^2 \in \gamma_1(e_i^1)$ and~$e_i^3 \in \gamma_2(e_i^2)$ for~$1 \leq i \leq k$. By definition of~$\gamma$ together with \cref{obs:immersion_maps_path_to_path} this implies that~$e_i^3 \in \gamma(e_i^1)$. Clearly no~$e_j^3 \in \gamma(e_i^1)$ for~$j \neq i$ since~$E(\gamma(e_j^3)) \cap E(\gamma(e_i^3)) = \emptyset$ using the fact that~$\gamma$ is an immersion on the underlying Eulerian digraphs. 
    \end{claimproof}
 
    \begin{claim}
        For each~$v_1 \in V(G_1)$,~$\gamma(v_1) \in \dom(\mu_3)$ if and only if~$v_1 \in \dom(\mu_1)$, as well as $\gamma(v_1) \in \dom(\m_3)$ if and only if $v_1 \in \dom(\m_1)$.
    \end{claim}
    \begin{claimproof}
        For each~$v_1 \in V(G_1)$ we know by 2. of \cref{def:knitwork_immersion} that~$v_2 \coloneqq \gamma_1(v_1) \in \dom(\mu_2)$ if and only if~$v_1 \in \dom(\mu_1)$. Then again by the same reasoning~$\gamma_2(v_2) \in \dom(\mu_3)$ if and only if~$v_2 \in \dom(\mu_2)$ concluding the proof of the first part of the claim; the second part is analogous.
    \end{claimproof}

    \begin{claim}
        For each~$v_1 \in \dom(\mu_1)$ the index of~$\mu_1(v_1)=(f_1^1,\ldots,f_\ell^1)$ and~$\mu_3(\gamma(v_1))=(f_1^3,\ldots,f_{\ell'}^3)$ agree, i.e.,~$\ell = \ell'$ where $\ell\in 2\N$, and for every~$1 \leq i \leq \ell$ the edge~$f_i^3$ is contained in~$\gamma(f_i^1)$.
    \end{claim}
    \begin{claimproof}
        Using 3. of \cref{def:knitwork_immersion} for~$\gamma_1$ and~$\gamma_2$ we derive that~$\ell = \ell'$ by transitivity of equality, and by eulerianness $\ell \in 2\N$. Further, 3. of \cref{def:knitwork_immersion} implies that for~$1 \leq i \leq \ell$ the edge~$f_i^2$ is contained in~$\gamma_1(f_i^1)$ where~$\mu_2(\gamma_1(v_1)) = (f_1^2,\ldots,f_\ell^2)$. Now~$P_i^2\coloneqq \gamma_1(f_i^1)$ is a path in~$G_2$ containing the edge~$f_i^2$. By \cref{obs:immersion_maps_path_to_path}~$\gamma_2(P_i^2)$ is a path in~$G_3$ and since~$f_i^2 \in E(P_i^2)$ it follows from 3. of \cref{def:knitwork_immersion} that~$f_i^2 \in \gamma_2(P_i^2) = \gamma_2(\gamma_1(e_i^1)) = \gamma(e_i^1)$. This concludes the proof.
    \end{claimproof}

      \begin{claim}
        For each~$v_3 \in \dom(\m_3)\setminus \gamma(\dom(\m_1))$, $M_\gamma(v_3) \in \m_3(v_3)$.
    \end{claim}
    \begin{claimproof}
        Since~$v_3 \notin \gamma(\dom(\m_1))$ there are two cases to consider. Again let~$(f_1^3,\ldots,f_\ell^3) = \mu_3(v_3)$.

        Assume first that there is~$v_2 \in \dom(\m_2)$ with~$\gamma_2(v_2) = v_3$, which by the assumption of the claim implies~$v_2 \notin \gamma_1(\dom(\m_1))$ and thus~$v_2 \notin \gamma_1(V(G))$ by 2. of \cref{def:knitwork_immersion} applied to $\gamma_1$. By 4. of \cref{def:knitwork_immersion} we derive~$M_{\gamma_1}(v_2) \in \m_2(v_2)$. The claim follows by 5. of \cref{def:knitwork_immersion} applied to~$\gamma_2$.

        Next assume that there is no such~$v_2 \in \dom(\m_2)$, whence~$M_{\gamma_2}(v_3) \in \m_3(v_3)$ by 4. of \cref{def:knitwork_immersion} applied to~$\gamma_2$. Let~$E_2 \subseteq E(G_2)$ be the set of edges for which~$\LLL^2_3=\{\gamma_2(e) \mid e \in E_2\}$ witnesses~$M_{\LLL^2_3}(v_3) = M_{\gamma_2}(v_3)$.
        Let~$E_1 \subseteq E(G_1)$ be the minimal set of edges maximising~$E_2 \cap E(\gamma_1(E_1))$. Since~$\gamma:G_1 \hookrightarrow G_3$
       is an immersion it follows that~$\LLL^1_3 = \{\gamma(e) \mid e \in E_1\}$ is a linkage in~$G_3$ such that if~$(f_i^3,f_j^3) \subset L$ is a subpath of some~$L \in \LLL^1_3$ then~$(f_i^3,f_j^3)$ is a subpath of some~$L \in \LLL^2_3$ for any~$1 \leq i,j \leq \ell$. The claim now follows by definition of~$M_\gamma(v_3)$ and the fact that all the links in $\m_3$ are reliable by \cref{def:knitwork}.
    \end{claimproof}

    \begin{claim}
        For each~$v_1 \in \dom(\m_1)$ with $\mu_1(v_1)=(f_1^1,\ldots,f_\ell^1)$ and~$\mu_3(\gamma(v_1))=(f_1^3,\ldots,f_{\ell'}^3)$,~$(f_i^1,f_j^1) \in \m_1(v_1)$ if and only if~$(f_i^3,f_j^3) \in \m_3(\gamma(v_1))$ for all~$1\leq i,j \leq \ell$.
    \end{claim}
    \begin{claimproof}
        This follows at once by applying 5. of \cref{def:knitwork_immersion} once for~$\gamma_1$ and once for~$\gamma_2$, i.e.,~$(f_i^1,f_j^1) \in \m(v_1) \iff (f_i^2,f_j^2) \in \m(\gamma_1(v_1)) \iff (f_i^3,f_j^3) \in \m(\gamma_2(\gamma_1(v_1)))$.
    \end{claimproof}

    \begin{claim}
        For every~$v_1 \in V(G_1)$ we have~$\Phi(v_1) \ll \Phi_3(\gamma(v_1))$.
    \end{claim}
    \begin{claimproof}
        This follows at once from 6. of \cref{def:knitwork_immersion} for~$\gamma_1,\gamma_2$ and the fact that~$\ll$ is transitive, i.e.,~$\Phi(v_1) \ll \Phi(\gamma_1(v_1)) \ll \Phi(\gamma_2(\gamma_1(v_1))) = \Phi(\gamma(v))$.
    \end{claimproof}
    All in all we verified all of the conditions for \cref{def:knitwork_immersion} concluding the proof for the transitivity (and thus the proof that~$\Omega$-knitwork immersion induces a quasi-order on~$\Omega$-knitworks).
\end{proof}

Next, we deal with two crucial operations that we call \emph{stitching} and \emph{knitting}. To this extent recall that given~$Y \subset V(G)$ we defined~$\bar{Y} = V(G) \setminus Y$. Concisely, given a rooted cut~$Y$ we define the up- and down-stitch to be the graphs obtained by contracting~$\bar{Y}$ and~$Y$ into a single vertex~$y^*$ and~$y_*$ respectively. The following is the formal definition.

\begin{definition}[Stitching]\label{def:stitching_std}
      Let~$\bar{G}=(G,\pi(X))$ be a rooted Eulerian digraph where~$G=(V,E,\operatorname{inc})$. Let~$\pi(X) \coloneqq (f_1,\ldots,f_\ell)$ for~$\ell \in 2\N$. Let $k\in 2\N$ and~$Y \subset V(G)$ induce a rooted~$k$-cut in~$G$. Let~$ \pi(Y) = (e_1,\ldots,e_k) = \pi(\bar{Y}) $ be an ordering of~$\rho(Y)$ and~$\rho(\bar Y)$. Let~$y_\ast,y^\ast$ be two new elements that are not part of~$V(G) \cup E(G)$. We define~$G_Y \coloneqq (V_Y,E_Y,\operatorname{inc}_Y)$ and~$\pi_Y$ via~
    \begin{align*}
        &V_Y \coloneqq Y \cup \{y_\ast\}\\
        & E_Y \coloneqq E(G[Y]) \cup \{e_1,\ldots,e_k\},\\
        &(e,v) \in \operatorname{inc}_Y :\iff \begin{cases} 
  (e,v) \in \operatorname{inc} \text{ and } e \in E_Y, v\in Y,\\
   e \in \rho^+(Y) \text{ and } v=y_\ast\\
        \end{cases}
        &(v,e) \in \operatorname{inc}_Y :\iff \begin{cases} 
  (v,e) \in \operatorname{inc} \text{ and } e \in E_Y, v\in Y,\\
  e \in \rho^-(Y) \text{ and } v=y_\ast\\
        \end{cases}\\
        & \pi_Y = (f_1,\ldots,f_\ell)\\     
    \end{align*} 
    We define~$\stitch(\bar{G};\pi(Y)) \coloneqq (G_Y,\pi_Y)$ and say that~\emph{$G_Y$ is obtained from~$G$ by down-stitching~$Y$} and call~$y_\ast$ the \emph{down-stitch vertex resulting from~$Y$}. 

    Similarly we define~$G^Y \coloneqq (V^Y,E^Y,\operatorname{inc}^Y)$ and~$\pi^Y$ via~
    \begin{align*}
        &V^Y \coloneqq \bar{Y} \cup \{y^\ast\}\\
        & E^Y \coloneqq E(G[\bar{Y}]) \cup \{e_1,\ldots,e_k\}, \\
         &(e,v) \in \operatorname{inc}^Y :\iff \begin{cases} 
  (e,v) \in \operatorname{inc}, \text{ and } e \in E^Y, v\in \bar{Y},\\
   e \in \rho^+(\bar{Y}) \text{ and } v=y^\ast\\
        \end{cases}
        &(v,e) \in \operatorname{inc}^Y :\iff \begin{cases} 
  (v,e) \in \operatorname{inc} \text{ and }  e \in E^Y, v\in \bar{Y},\\
   e \in \rho^-(\bar{Y}) \text{ and } v=y^\ast \\
        \end{cases}\\
    & \pi^Y = (e_1,\ldots,e_k)
    \end{align*} 
    We define~$\stitch(\bar{G};\pi(\bar{Y})) \coloneqq (G^Y,\pi^Y)$ and say that~\emph{$G^Y$ is obtained from~$G$ by up-stitching~$Y$} and call~$y^\ast$ the \emph{up-stitch vertex resulting from~$Y$}. 
\end{definition}

The following relations are readily extracted from the definition.
\begin{observation}\label{obs:stitching_fundamentals}
     Let~$\bar{G}=(G,\pi(X))$ and let~$Y$ induce a rooted cut with a fixed ordering~$\pi(Y) =\pi(\bar{Y})$. Let~$(G_Y,\pi_Y) = \stitch(\bar{G};\pi(Y))$ and~$(G^Y,\pi^Y) = \stitch(\bar{G};\pi(\bar{Y}))$ be the down- and up-stitches of~$G$ at~$Y$ with down- and up-stitch vertices~$y_\ast,y^\ast$ respectively. Then
     \begin{enumerate}
         \item $(G_Y,\pi_{G_Y}(X))$ and~$(G^Y,\pi_{G^Y}(y^*))$ are well-defined rooted Eulerian digraphs where
         \item $Y \subset V(G_Y)$ and~$\bar{Y} \subset V(G^Y)$ with~$G_Y[Y] = G[Y]$ and~$G^Y[\bar{Y}] = G[\bar{Y}]$, and
         \item $\rho_{G_Y}(y_\star) = \rho_{G_Y}(Y) = \rho_G(Y)$, and~$\rho_{G_Y}(X) = \rho_G(X)$, whence $\pi_{G_Y}(X) \coloneqq \pi_Y =\pi_G(X) $ is well-defined , and
         \item $\rho_{G^Y}(\bar{Y}) = \rho_{G^Y}(y^\ast) = \rho_G(\bar{Y})$, whence~$\pi^Y \eqqcolon \pi_{G^Y}(y^*) =\pi_{G^Y}(\bar{Y}) = \pi_G(\bar{Y})$ is well-defined.
     \end{enumerate} 
\end{observation}
\begin{remark}
    Although technically the graphs obtained by stitching are new graphs, and the incidences of edges vary due to introducing $y_\ast,y^\ast$, the edges themselves as objects remain the same. This way we may indeed write~$\rho_G(\bar{Y}) = \rho_{G^Y}(y^*)$ instead of carrying a bijection between these sets which turns out more tedious and annoying than insightful.
\end{remark}

Using \cref{obs:stitching_fundamentals} we will unambiguously set~$(G_{Y},\pi_{G_Y}(X)) = \stitch((G,\pi_G(X)),\pi(Y))$ as well as~$(G^{Y},\pi_{G^Y}(y^*)) = \stitch((G,\pi_G(X)),\pi(\bar{Y}))$ to be the rooted digraphs resulting by taking up- and down-stitches. Since~$y^* = \bar{Y}$ in~$G^Y$ we may equally well write $(G^{Y},\pi_{G^Y}(\bar{Y})) = \stitch((G,\pi_G(X)),\pi(\bar{Y}))$.

We extend the definition of stitching to~$\Omega$-knitworks as follows.

\begin{definition}[Types of feasible linkages]\label{def:types_of_linkages_on_a_cut}
    Let~$G$ be a rooted Eulerian digraph and~$X \subset V(G)$ induce a~$k$-cut for some $k\in 2\N$. Let~$\{e_1,\ldots,e_k\} = \rho(X)$.  Let~$\LLL(\rho(X))$ be the set of all~$(\rho^-(X),\rho^+(X))$-linkages~$\LLL$ such that for very $P \in \LLL$ it holds $V(P) \subseteq X$. For~$\LLL \in \LLL(\rho(X))$ we define~$\mathfrak{M}(\rho(X)) \coloneqq \{ \tau(\LLL) \mid \LLL \in \LLL(\rho(X))\}$ to be the \emph{set of feasible~$\rho(X)$-types}.
\end{definition}
\begin{remark}
    Recall that $V(P) \subseteq X$ implies that all the paths in $\LLL$ use solely vertices in $X$ except for their endpoints which may be part of $\bar X$.

    Note that while~$\rho(X) = \rho(\bar{X})$,~$\mathfrak{M}(\rho(X))$ and~$\mathfrak{M}(\rho(\bar X))$ may differ. We choose this notation to emphasize the importance of the cut-edges. Note further that for $M \in \mathfrak{M}(\rho(X))$ if holds$ M \subseteq \operatorname{Match}(\rho^-(X),\rho^+(X))$ and for every $M' \subseteq M$ it holds $M' \in \mathfrak{M}(\rho(X))$ by definition; in particular $\mathfrak{M}(\rho(X))$ is a reliable link by definition.
\end{remark}

\begin{definition}\label{def:stitching_knitwork}
   Let~$\Omega$ be a well-quasi-order, $k \in 2\N$, and let~$\GGG = (\bar{G},\mu,\m,\Phi)$ be an~$\Omega$-knitwork. Let~$Y \subset V(G)$ induce a rooted~$k$-cut in~$\bar{G}=(G,\pi(X))$. Let~$\pi(Y) = (e_1,\ldots,e_k) = \pi(\bar{Y})$ be an ordering of~$\rho(Y)=\rho(\bar{Y})$. Then we define $\stitch(\GGG;\pi(Y)) \coloneqq ((G_Y,\pi(X),\mu_Y,\m_Y,\Phi_Y)$ as follows:
   \begin{align*}
       &(G_Y,\pi(X)) \coloneqq \stitch(\bar{G};\pi(Y)) \text{ with down-stitch vertex } y_\ast\\
       &\mu_Y(v) \coloneqq \begin{cases}
           \mu(v), &\ v\in Y\\
           (e_1,\ldots,e_k), &v = y_\ast\end{cases},\\
        &\m_Y(v)\coloneqq \begin{cases}
           \m(v), &\ v\in Y\\
           \mathfrak{M}(\rho(\bar{Y})), &v = y_\ast\end{cases}, \text{ and}\\
       & \Phi_Y(v) \coloneqq \Phi(v) \text{ for } v\in Y,
   \end{align*}
   and we do not specify~$\Phi_Y$ on~$y_\ast$.

    We further define $\stitch(\GGG;\pi(\bar{Y})) \coloneqq ((G^Y,\pi(\bar{Y}),\mu^Y,\Phi^Y)$ as follows:
   \begin{align*}
       &(G^Y,\pi(\bar{Y})) \coloneqq \stitch(\bar{G};\bar{Y}) \text{ with up-stitch vertex } y^\ast\\
      &\mu^Y(v) \coloneqq \begin{cases}
           \mu(v), &\ v\in \bar Y\\
           (e_1,\ldots,e_k), &v = y^\ast\end{cases},\\
      &\m^Y(v)\coloneqq \m(v), \text{ for }\ v\in \bar Y, \text{ and}\\
       & \Phi^Y(v) \coloneqq \Phi(v) \text{ for } v\in \bar{Y},
   \end{align*}
   where~$\pi(\bar{Y}) = (e_1,\ldots,e_k)$ and we do not specify~$\m^Y$ and $\Phi^Y$ on~$y^\ast$. 
\end{definition}
\begin{remark}
    Note that setting~$\mu^Y(y^\ast) =(e_1,\ldots,e_k)$ is technically not necessary since the order is already taken care off by the newly introduced roots but it makes intuitive sense which is why we keep it. 
\end{remark}

Using the above definition, the fact that $\mathfrak{M}(\rho(\bar Y))$ is a reliable link, and \cref{obs:stitching_fundamentals}, the following is straightforward.

\begin{observation}\label{obs:stitching_fundamentals_knitworks}
     Let~$\Omega$ be a well-quasi-order and let~$\GGG = (\bar{G},\mu,\m,\Phi)$ be an~$\Omega$-knitwork. Let~$Y \subset V(G)$ induce a rooted~$k$-cut in~$\bar{G}$. Let~$\pi(Y) = (e_1,\ldots,e_k) = \pi(\bar{Y})$ be an ordering of~$\rho(Y)=\rho(\bar{Y})$. Let $(\bar G_Y,\mu_Y,\m_Y,\Phi_Y) = \stitch(\GGG;\pi(Y))$ as well as~ $(\bar G^Y,\mu^Y,\m^Y,\Phi^Y) = \stitch(\GGG;\pi(\bar{Y}))$. Then
     \begin{enumerate}
     \item  $\stitch(\GGG,\pi(Y))$ and~$\stitch(\GGG,\pi(\bar{Y}))$  are well-defined~$\Omega$-knitworks,
         \item $\mu(x) = \mu_Y(x)$ for all~$x \in Y$ and~$\mu(x) = \mu^Y(x)$ for all~$x \in \bar{Y}$, and moreover~$\mu_Y(y_\ast) =\pi(Y) = \pi(\bar Y )=  \mu^Y(y^\ast)$ 
        \item $\m(x) = \m_Y(x)$ for all~$x \in Y$ and~$\m(x) = \m^Y(x)$ for all~$x \in \bar{Y}$, and
         \item $\Phi(x) = \Phi_Y(x)$ for all~$x \in Y$ and~$\Phi(x) = \Phi^Y(x)$ for all~$x \in \bar{Y}$.
     \end{enumerate}
\end{observation}
\smallskip

We further get the following useful corollary straight from the \cref{def:stitching_knitwork} of up-stitches.
\begin{observation}\label{obs:up-stitch_of_well-linked_is_well-linked}
Let~$\Omega$ be a well-quasi-order and~$\GGG$ a well-linked~$\Omega$-knitwork for a rooted Eulerian digraph~$(G,\pi(X))$. Let~$Y \subseteq V(G)$ be a rooted cut in~$\bar G$. Then~$\stitch(\GGG;\pi(\bar{Y}))$ with up-stitch vertex~$y^*$ is a well-linked~$\Omega$-knitwork.
\end{observation}
\begin{proof}
    This follows immediatley from the fact that~$y^* \notin \dom(\m^Y)$ and~$\m^Y = \m$ for all~$y \in Y$.
\end{proof}
\begin{remark}
    Note that \cref{obs:up-stitch_of_well-linked_is_well-linked} is \emph{not} generally true for down-stitches.
\end{remark}

Next we introduce \emph{knitting}---sort of the inverse operation to stitching---an operation that given two rooted Eulerian digraphs with distinguished ordered cuts produces a new rooted Eulerian digraph by ``knitting'' the graphs together at said cuts respecting their order. 

\begin{definition}[Knitting]\label{def:knitting_rooted_graphs}
   Let~$\bar{G}_i=(G_i,\pi(X_i))$ be rooted Eulerian digraphs for some~$X_i \subset V(G_i)$ and~$i \in \{1,2\}$. Let~$\pi_{G_1}(X_1) \coloneqq (f_1^1,\ldots,f_\ell^1)$ and~$\pi_{G_2}(X_2) = (e^2_1,\ldots,e^2_k) = \pi_{G_2}(\bar X_2)$ for~$\ell,k \in 2\N$. Let~$Y \subset V(G_1)$ induce a rooted~$k$-cut in~$G_1$. Let~$\pi_{G_1}(Y) \coloneqq (e_1^1,\ldots,e^1_k)$ be an ordering of~$\rho(Y)$ such that~$e_i^1 \in \rho^+(Y) \iff e_i^2 \in \rho^-(\bar X_2)$ for every~$i\in\{1,\ldots,k\}$.

   We define~$\knit({G_2},G_1;\pi_{G_1}(Y)) \coloneqq (V',E',\operatorname{inc}')$ and~$\knit(\bar{G_2},\bar{G_1};\pi_{G_1}(Y)) \coloneqq ((V',E',\operatorname{inc}'),\pi)$ where
   \begin{align*}
       & V' \coloneqq Y \cup \bar X_2, \\
        & E' \coloneqq E(G_1[Y]) \cup E(G_2[\bar X_2]) \cup \{e_1,\ldots,e_k\},  \\
        &(e,v) \in \operatorname{inc}' :\iff 
            \begin{cases} 
                &(e,v) \in \operatorname{inc}_{G_1} \text{ and } e \in E(G_1[Y]), v\in Y,\\
                &(e,v) \in \operatorname{inc}_{G_2} \text{ and } e \in E(G_2[X_2]), v\in \bar X_2,\\
                &e= e_i,\ v\in Y \text{ and } (e_i^1,v) \in \operatorname{inc}_{G_1},\\
                &e= e_i,\ v\in \bar X_2 \text{ and } (e_i^2,v) \in \operatorname{inc}_{G_2}
                \end{cases}\\
         &(v,e) \in \operatorname{inc}' :\iff 
            \begin{cases} 
                &(v,e) \in \operatorname{inc}_{G_1} \text{ and } e \in E(G_1[Y]), v\in Y,\\
                &(v,e) \in \operatorname{inc}_{G_2} \text{ and } e \in E(G_2[{\bar X_2}]), v\in {\bar X_2},\\
                &e= e_i,\ v\in Y \text{ and } (v,e_i^1) \in \operatorname{inc}_{G_1},\\
                &e= e_i,\ v\in \bar X_2 \text{ and } (v, e_i^2) \in \operatorname{inc}_{G_2}
             \end{cases}\\
        & \pi \coloneqq (f_1^*,\ldots,f_\ell^*), \text{ where } \quad f_i^*\coloneqq 
                \begin{cases} 
                      f_i, & f_i \in E',\\
                      e_j, & f_i= e_j^1\\
                \end{cases}.
   \end{align*}

Finally rename~$f_i^*$ to~$f_i$ in~$G'$.
We say that~$(G',\pi) \coloneqq ((V',E',\operatorname{inc}'),\pi)$ is obtained by \emph{knitting~$\bar G_2$ to~$\bar G_1$ at~$\pi_{G_1}(Y)$}.
\end{definition}

By construction~$X_1 \subseteq V(G')$. Again the following observations, although a bit tedious, are straightforward to verify from the definition and are left to the reader. 

\begin{observation}\label{obs:knitting_fundamentals}
    Let~$(G',\pi) \coloneqq \knit(\bar{G_2},\bar{G_1};\pi_{G_1}(Y))$, as in \cref{def:knitting_rooted_graphs}. Then
    \begin{enumerate}
        \item $(G',\pi)$ is a well-defined rooted Eulerian digraph where~$\pi=\pi_{G'}(X_1) = \pi_G(X_1)$ is an ordering of~$\rho_{G'}(X_1)$,
         \item the sets~$\bar X_2,Y \subset V(G')$ induce the same cut~$\rho_{G'}(\bar X_2) = \rho_{G'}(Y)$ of size~$k\in 2\N$, where~$\bar X_2 \cup Y$ is a partition of~$V(G')$,
        \item $G'[Y] = G_1[Y]$ and~$G'[\bar X_2] = G_2[\bar X_2]$, and
        \item $\stitch(\bar{G_1};\pi_{G_1}(X_1)) \cong \stitch(\bar{G'};\pi_{G'}(X_1))$.
       
    \end{enumerate}
\end{observation}

Using \cref{obs:knitting_fundamentals} we will unambiguously define~$\big(\knit(G_1,G_2;\pi_{G_1}(Y)),\pi_{G'}(X)\big) = \knit(\bar{G_2},\bar{G_1};\pi_{G_1}(Y))$ to be the rooted Eulerian digraph obtained by knitting~$\bar G_2$ to~$\bar G_1$.

The definition of knitting readily extends to~$\Omega$-knitworks.

\begin{definition}\label{def:knitting_knitworks}
     Let~$\Omega$ be a well-quasi-order and let~$\GGG_i = (\bar{G}_i,\mu_i,\m_i,\Phi_i)$ be~$\Omega$-knitworks with rooted Eulerian digraphs~$\bar{G_i}=(G_i,\pi(X_i))$ for~$i \in \{1,2\}$. Let~$\pi_{G_2}(X_2) = (e^2_1,\ldots,e^2_k) = \pi_{G_2}(\bar X_2)$ for some~$k \in 2\N$. Let~$Y \subset V(G_1)$ induce a rooted~$k$-cut in~$G_1$. Let~$\pi_{G_1}(Y) \coloneqq (e_1^1,\ldots,e^1_k)$ be an ordering of~$\rho(Y)$ such that~$e_i^1 \in \rho^+(Y) \iff e_i^2 \in \rho^-(\bar X_2)$ for every~$i\in\{1,\ldots,k\}$. Let~$(G',\pi_{G'}(X_1)) \coloneqq \knit(G_2,G_1;\pi_{G_1}(Y))$. We define 
     \begin{align*}
         \mu' &\coloneqq \begin{cases}
             \mu_1(x) \text{ renaming } e_i^1 \text{ to } e_i, &\text{ if } x \in Y\\
             \mu_2(x) \text{ renaming } e_i^2 \text{ to } e_i, &\text{ if } x \in \bar X_2\\
         \end{cases}, \\
           \m' &\coloneqq \begin{cases}
             \m_1(x) \text{ renaming } e_i^1 \text{ to } e_i, &\text{ if } x \in Y\\
             \m_2(x) \text{ renaming } e_i^2 \text{ to } e_i, &\text{ if } x \in \bar X_2\\
         \end{cases}, \text{ and }\\
         \Phi'(x) &\coloneqq \begin{cases}
             \Phi_1(x), &\text{ if } x \in Y\\
             \Phi_2(x), &\text{ if } x\in \bar X_2\\
         \end{cases}.
     \end{align*}

  Finally we define~$\GGG'\coloneqq \knit(\GGG_2,\GGG_1;\pi_{G_1}(Y)) \coloneqq \big((G',\pi_{G'}(X_1)), \mu', \m', \Phi'\big)$ and say that~$\GGG'$ is obtained by \emph{knitting~$\GGG_2$ to~$\GGG_1$ at~$\pi_{G_1}(Y)$}.
\end{definition}

Again, using the above definition and \cref{obs:knitting_fundamentals} the following is straightforward.

\begin{observation}\label{obs:knitting_fundamentals_knitworks}
      Let~$\Omega$ be a well-quasi-order and let~$\GGG_i = (\bar{G}_i,\mu_i,\m_i,\Phi_i)$ be~$\Omega$-knitworks with rooted Eulerian digraphs~$\bar{G_i}=(G_i,\pi(X_i))$ for~$i \in \{1,2\}$. Let~$k \in 2\N$ be the index of~$\bar{G_2}$ and let~$Y \subset V(G_1)$ induce a rooted~$k$-cut with an ordering~$\pi_{G_1}(Y)$ in~$\bar{G_1}$. Let~$\GGG' \coloneqq \knit(\GGG_2,\GGG_1;\pi_{G_1}(Y)) =\big((G',\pi_{G'}(X_1)), \mu',\m', \Phi'\big)$. Then
     \begin{enumerate}
        \item  ~$\GGG'$ is a well-defined~$\Omega$-knitwork,
         \item $\mu'(x) = \mu_1(x)$ for all~$x \in Y$ and~$\mu'(x) = \mu_2(x)$ for all~$x \in \bar X_2$,
         \item $\m'(x) = \m_1(x)$ for all~$x \in Y$ and~$\m'(x) = \m_2(x)$ for all~$x \in \bar X_2$, and
         
         \item $\Phi'(x) = \Phi_1(x)$ for all~$x \in Y$ and~$\Phi'(x) = \Phi_2(x)$ for all~$x \in \bar X_2$.
     \end{enumerate}
\end{observation}

The definitions of stitching and knitting allow us to unambiguously decompose~$\Omega$-knitworks at rooted cuts and knit them back together in a unique and reversible way (up to isomorphisms).

\begin{lemma}[Stitch-and-Knit]\label{lem:stitch-and-knit}
    Let~$\Omega$ be a well-quasi-order. Let~$\GGG=((G,\pi_G(X)),\mu,\m,\Phi)$ be an~$\Omega$-knitwork. Let~$Y \subset V(G)$ induce a rooted~$k$-cut in~$(G,\pi_G(X))$ for some~$k\in 2\N$ and let~$\pi_G(Y)=(e_1,\ldots,e_k)=\pi_G(\bar Y)$ be an ordering of~$\rho(Y)$. Let
    \begin{align*}
        &\text{$\GGG_Y=((G_Y,\pi_{G_Y}(X)),\mu_Y,\m_Y,\Phi_Y)\coloneqq \stitch(\GGG;\pi_G(Y))$, and}\\
        &\text{$\GGG^Y=((G^Y,\pi_{G^Y}(\bar{Y})),\mu^Y,\m^Y,\Phi^Y) \coloneqq \stitch(\GGG;\pi_G(\bar{Y}))$.}
    \end{align*}
    Then there exists an ordering~$\pi_{G_Y}(Y) = \pi_{G^Y}(\bar{Y})$ of $\rho_G(Y)$ such that~$\knit(\GGG^Y,\GGG_Y; \pi_{G_Y}(Y)) = \GGG$.
\end{lemma}
\begin{proof}
    Let~$\bar G_Y = (G_Y,\pi_{G_Y}(X))$ and~$\bar G^Y = (G^Y,\pi_{G^Y}(\bar{Y}))$ as usual with respective down and up-stitch vertices~$y_*,y^*$. By 3. of \cref{obs:stitching_fundamentals} we have~$\rho_{G_Y}(Y) = \rho_G(Y) = \{e_1,\ldots,e_k\}$, whence we may set~$\pi_{G_Y}(Y) = \pi_G(Y)$. By 4. of the same \cref{obs:stitching_fundamentals} note that~$\pi_{G}(\bar Y) = \pi_{G^Y}(y^*) =\pi_{G^Y}(\bar{Y})$ and by assumption $\pi_G(\bar Y) = \pi_G(Y)$. In particular we have that~$e_i \in \rho_{G_Y}^+(Y) \iff e_i \in \rho_{G^Y}^-(\bar Y)$ for~$i\in\{1,\ldots,k\}$, whence $\knit(\GGG^Y,\GGG_Y; \pi_{G_Y}(Y))$ is well defined using \cref{def:knitting_rooted_graphs}. Note here that knitting looses the information of~$\mu_y,\m_Y$ on~$y_\ast$, and of ~$\mu^Y,\m^Y$ on~$y^\ast$ respectively, for they are not part of the resulting graph. 
    
    Let~$ \knit(\GGG^Y,\GGG_Y; \pi_G({Y})) = \big(\knit(\bar G^Y,\bar G_Y; \pi_G({Y})), \mu',\m',\Phi'\big)$ with respective~$\mu',\m',\Phi'$ as in \cref{def:knitting_knitworks}.
    \begin{claim}
        $ \bar{G} = \knit(\bar G^Y,\bar G_Y; \pi_G({Y}))$.
    \end{claim}
    \begin{claimproof}
        Let~$(H,\pi_H(Y)) = \knit(\bar G^Y,\bar G_Y; \pi_G(Y))$. Then by 1. and 3. of \cref{obs:knitting_fundamentals} $\pi_H(Y) = \pi_G(Y)$ and~$H[Y] = G_Y[Y]$ as well as~$H[\bar{Y}] = G^Y[\bar{Y}]$.
        
        By 2. of \cref{obs:stitching_fundamentals}~$Y$ induces a rooted~$k$-cut in~$\bar G_Y$ and~$\bar{Y}$ induces a rooted~$k$-cut in~$\bar G^Y$ such that~$G_Y[Y] = G[Y]$ and~$G^Y[\bar{Y}] = G[\bar Y ]$. In particular then~$H[Y] = G[Y]$ and $H[\bar{Y}] = G[\bar{Y}]$.

        By 3. of \cref{obs:stitching_fundamentals} together with our choice of~$\pi_{G_Y}(Y)$ we derive $\pi_{G^Y}(\bar{Y}) = (e_1,\ldots,e_k) = \pi_{G_Y}(Y) = \pi_G(Y)$. Thus by \cref{def:knitting_rooted_graphs} we get~$\rho_H(Y) = \rho_G(Y)$ respecting the same incidences; the claim follows.

    \end{claimproof}

    Recall the \cref{def:knitting_knitworks} of knitting $\Omega$-knitworks. Combining 2. 3. and 4. of \cref{obs:stitching_fundamentals_knitworks} with 2. 3. and 4. of \cref{obs:knitting_fundamentals_knitworks} we immediately get~$\mu'(x) = \mu(x)$,~$\m'(x) = \m(x)$ and~$\Phi'(x) = \Phi(x)$ for all~$x \in V(G)$. Hence the conditions of \cref{def:isomorphism_knitwork} are satisfied, concluding the proof.
\end{proof}
\begin{remark}
  Note that the lemma holds with equality instead of isomorphism due to the intrinsic identifications in all the definitions (renaming edges).  
\end{remark}

Similar to the Stitch-and-Knit \cref{lem:stitch-and-knit}, stitching and knitting can be used to reduce the question of whether a graph~$H$ strongly immerses into~$G$ by decomposing the graphs at cuts via stitches and asking whether the resulting stitches can be immersed into each other; if so, we can knit the immersions back together to yield an immersion of~$H$ in~$G$.
The following is the main result of this subsection.

\begin{theorem}\label{thm:knitting_knitwork_immersion}
    Let~$\Omega=(V(\Omega),\preceq)$ be a well-quasi-order and let~$\GGG=((G,\pi(X)),\mu,\m,\Phi)$ and~$\HHH=((H,\pi(A)),\nu,\n,\Psi)$ be~$\Omega$-knitworks of common index~$\ell \in 2\N$. Let~$Y \subset V(G)$ and~$B \subset V(H)$ induce rooted~$k$-cuts in~$\bar{G}$ and~$\bar{H}$ respectively for some~$k \in 2\N$, and let~$\pi(Y)=(e_1^Y,\ldots,e_k^Y)=\pi(\bar Y)$ and~$\pi(B)=(e_1^B,\ldots,e_k^B) =\pi(\bar{B})$ be orderings of~$\rho(Y),\rho(B)$ respectively. Let
    \begin{align*}
        \big((G_Y,\pi(X)),\mu_Y,\m_Y,\Phi_Y\big) &=  \stitch(\GGG;\pi(Y)) \text{ with down-stitch vertex } y_\ast,\\
       \big ((G^Y,\pi(\bar Y)),\mu^Y,\m^Y,\Phi^Y\big) &=\stitch(\GGG;\pi(\bar{Y})) \text{ with up-stitch vertex } y^\ast,\\
        \big((H_B,\pi(A)),\nu_B,\n_B,\Psi_B\big) &=  \stitch(\HHH;\pi(B)) \text{ with down-stitch vertex } b_\ast,\\
        \big((H^B,\pi(\bar{B})),\nu^B,\n^B,\Psi^B\big) &=\stitch(\HHH;\pi(\bar{B})) \text{ with up-stitch vertex } b^\ast, 
    \end{align*}

    respectively. Further let
    \begin{align*}
        \gamma_d:&  \big((H_B,\pi(A)),\nu_B,\n_B,\Psi_B\big) \hookrightarrow \big((G_Y,\pi(X)),\mu_Y,\m_Y,\Phi_Y\big), \text{ with } \gamma_d(b_\ast) = y_\ast,\text{ and}\\
        \gamma_u:&  \big((H^B,\pi(\bar{B})),\nu^B,\n^B,\Psi^B\big) \hookrightarrow \big((G^Y,\pi(\bar Y)),\mu^Y,\m^Y,\Phi^Y\big) , \text{ with } \gamma_u(b^\ast) = y^\ast,
    \end{align*}be strong~$\Omega$-knitwork immersions. 
    
    Then there exists a strong~$\Omega$-knitwork immersion~$\gamma: \knit(\HHH^B,\HHH_B;\pi({B})) \hookrightarrow \knit(\GGG^Y,\GGG_Y;\pi({Y}))$ and in particular~$\gamma:\HHH \hookrightarrow \GGG$ such that~$\restr{\gamma}{H[B]} = \restr{\gamma_u}{H_B[B]}$ and~$\restr{\gamma}{H[\bar{B}]} = \restr{\gamma_d}{H^B[\bar{B}]}$.

    The same holds true for general immersion $\hookrightarrow^*$.
\end{theorem}

\begin{proof}
We prove the theorem for strong immersion; the other case is analogous. By the assumptions of the theorem and \cref{def:knitwork_immersion} we derive the following.

    \begin{claim}\label{claim:stitch_knit_corr_paths_under_partial_immersion_uno}
        Let~$e_i^Y \in \rho_G^+(Y)$ for some $i \in \{1,\ldots,k\}$. Then~$\gamma_d(e_i^B) \subset G_Y$ is a path ending in~$e_i^Y \in \rho_{G_Y}^+(Y)$ that is otherwise disjoint from~$\rho_{G_Y}(Y)$ and~$\gamma_u(e_i^B) \subset G^Y$ is a path starting in~$e_i^Y \in \rho_{G^Y}^-(\bar{Y})$ that is otherwise disjoint from~$\rho_{G^Y}(\bar{Y})$.
    \end{claim}
    \begin{claimproof}
       We start with a proof for~$\gamma_d$. To this extent note that by 3. of \cref{obs:stitching_fundamentals}~$\rho_{H_B}(b_\ast) = \{e_1^B,\ldots,e_k^B\}$ and~$\rho_{G_Y}(y_\ast) = \{e_1^Y,\ldots,e_k^Y\}$ and by 3. of \cref{obs:stitching_fundamentals_knitworks} we know that~$\nu_B(b_\ast) = (e_1^B,\ldots,e_k^B)$ as well as~$\mu_Y(y_\ast) = (e_1^Y,\ldots,e_k^Y)$. By the assumption of the theorem we know that~$\gamma_d(b_\ast) = y_\ast$ and since~$\gamma_d$ is an~$\Omega$-knitwork immersion we derive from 3. of \cref{def:knitwork_immersion} that for every~$1 \leq i \leq k$ the path~$\gamma_d(e_i^B)$ contains the edge~$e_i^Y$. Since all these paths are edge-disjoint the path~$\gamma_d(e_i^B)$ contains no other edge of~$\rho_{G_Y}(y_\ast)$. Further, by \cref{obs:knitwork_immersion_on_mu}, using~$\rho_{G_Y}^-(Y) = \rho_{G_Y}^-(y_\ast)$, the path must end in~$e_i^Y$ concluding the proof for this case.

       For~$\gamma_u$ the proof is analogous using~$\gamma_u(b^\ast) = y^\ast$ and the fact that the~$\Omega$-knitwork immersion is in particular a rooted immersion whence~$\gamma_u$ respects the order of roots $\pi_{H^B}(\bar{B})$ and $\pi_{G^Y}(\bar{Y})$---note here that by \cref{def:stitching_knitwork}~$\nu^B(b^\ast) = \pi_{H^B}(\bar{B})$ and~$\mu^Y(y^\ast) = \pi_{G^Y}(\bar{Y})$---and the proof follows as for~$\gamma_d$ with the difference that the path starts in the respective edge~$e_i^Y \in \rho^-(\bar{Y})$ by symmetry.  
    \end{claimproof}

    Analogously we get the following.
     \begin{claim}\label{claim:stitch_knit_corr_paths_under_partial_immersion_dos}
        Let~$e_i^Y \in \rho_G^-(Y)$ for some $i \in \{1,\ldots,k\}$. Then~$\gamma_d(e_i^B) \subset G_Y$ is a path starting in~$e_i^Y \in \rho_{G_Y}^-(Y)$  that is otherwise disjoint from~$\rho_{G_Y}(Y)$ and~$\gamma_u(e_i^B) \subset G^Y$ is a path ending in~$e_i^Y \in \rho_{G^Y}^-(\bar{Y})$ that is otherwise disjoint from~$\rho_{G^Y}(\bar{Y})$.
    \end{claim}

    Combining Claims \ref{claim:stitch_knit_corr_paths_under_partial_immersion_uno} and \ref{claim:stitch_knit_corr_paths_under_partial_immersion_dos} we derive the following.
    \begin{claim}\label{claim:stitch_knit_effect_on_cut}
        For~$1 \leq i \leq k$,~$e_i^Y$ is an end to both paths~$\gamma_u(e_i^B) \subset G^Y$ and~$\gamma_d(e_i^B)\subset G_Y$, and~$e_i^Y \in \gamma(e)$ if and only if~$e = e_i^B$.
    \end{claim}

    Finally we use \cref{claim:stitch_knit_corr_paths_under_partial_immersion_uno} and \ref{claim:stitch_knit_corr_paths_under_partial_immersion_dos} to ``knit'' both strong immersions~$\gamma_u,\gamma_d$ when knitting the stitched graphs back together. By \cref{lem:stitch-and-knit} it suffices to give a strong immersion~$\gamma: H \hookrightarrow G$; we define it as follows. Recall \cref{def:paths} and \cref{def:concatenation_of_paths}, in particular recall that the same sequence of edges~$(f_1,\ldots,f_\ell)$ with~$f_i \in E(G') \cap E(G'')$ for two graphs~$G',G''$, say, may represent unique paths in~$G',G''$ respectively, i.e., paths using the same edges but with possibly distinct vertices (in our case the endpoints may differ). 

    \begin{align*}
        &\gamma\colon V(H)\cup E(H) \to G,\\
        &\restr{\gamma}{H[B]} \coloneqq \restr{\gamma_d}{H_B[B]},\\
        &\restr{\gamma}{H[\bar{B}]} \coloneqq \restr{\gamma_u}{H^B[\bar{B}]}, \text{ and }\\
        &\gamma(e_i^B) = \begin{cases}
            \gamma_u(e_i^B) \circ \gamma_d(e_i^B),& e_i^B \in \rho_H^+(B)\\
            \gamma_d(e_i^B) \circ \gamma_u(e_i^B),& e_i^B \in \rho_H^-(B).
        \end{cases}
    \end{align*}
    By construction the above is well-defined. To see it this, it suffices to show that~$ \gamma_u(e_i^B) \circ \gamma_d(e_i^B)$ and~$\gamma_d(e_i^B) \circ \gamma_u(e_i^B)$ are paths in~$G$. Now clearly their summands are internally edge-disjoint respectively, and for~$\gamma_u(e_i^B)=(f_1,\ldots,f_\ell)$ with~$f_1,\ldots,f_\ell \in E(G_Y) \cap E(G)$ it holds that~$(f_1,\ldots,f_\ell)_G$ is a path in~$G$, in particular~$\gamma_u(e_i^B)$ is. Analogously, all the summands are paths in~$G$ and the claim follows by \cref{obs:concat_paths}.
    
    It is now straightforward to verify that~$\gamma$ is a strong rooted immersion.
\begin{claim}
        $\gamma:(H,\pi_H(A)) \hookrightarrow (G,\pi_G(X))$ is a strong rooted immersion. 
    \end{claim}
    \begin{claimproof}
        The fact that~$\gamma: V(H) \to V(G)$ is injective is clear by construction together with the fact that~$\gamma_u,\gamma_d$ are strong immersions for two disjoint sets~$B,\bar{B}$ where~$V(H) = B \cup \bar{B}$. Further, since $E(H[B]) \cap E(H[\bar{B}]) = \emptyset$, it follows by construction, the \cref{def:immersion} of strong immersion, and \cref{claim:stitch_knit_effect_on_cut}, that every~$e \in E(H) \setminus \rho(B)$ is mapped to a unique path~$\gamma(e)$ in~$G$. In addition, the respective paths are edge-disjoint where no such path contains any of~$\gamma(V(G))$ as an internal vertex, and no path contains an edge of~$\rho(Y)$. In particular~$\gamma(e)$ is either completely contained in~$G[Y]$ or~$G[\bar{Y}]$. Finally for~$e_i^B \in \rho(B)$ we have seen that~$\gamma(e_i^B)$ is a path in~$G$ containing~$e_i^Y$, for every~$1 \leq i \leq k$; note that~$E(\gamma_u(e_i^B)\cap E(\gamma_d(e_i^B)) =\{e_i^Y\}$ and~$V(\gamma_u(e_i^B)) \cap V(\gamma_d(e_i^B)) = \emptyset $, concluding the proof that $\gamma$ is a strong immersion.

        \smallskip
        
        To see that it is a \emph{rooted} strong immersion, note that~$\gamma_d$ is a rooted immersion (see \cref{def:rooted_immersion}) which implies that~$\gamma_d$ respects the order of the roots; the claim follows by combining 3. of \cref{obs:stitching_fundamentals} and 1. of \cref{obs:knitting_fundamentals}.
    \end{claimproof}

Finally we have the following.
    \begin{claim}
        $\gamma:\big((H,\pi_H(A)),\nu,\n,\Psi\big) \hookrightarrow \big((G,\pi_G(X)),\mu,\m,\Phi\big)$ is a strong~$\Omega$-knitwork immersion. 
    \end{claim}
    \begin{claimproof}
        By the previous claim we are left to verify the conditions of \cref{def:knitwork_immersion} on~$\nu,\mu$ as well as~$\n,\m$ and~$\Psi,\Phi$. But these are imminent from the definition:~$\Psi(v) \preceq \Phi(\gamma(v))$ using the fact that for~$v \in B$ we have~$\gamma(v) \in Y$ and~$\Psi(v) = \Psi_{B}(v)$ as well as~$\Phi(\gamma(v)) = \Phi_{Y}(\gamma_d(v))$ whence the claim follows from the fact that~$\Phi_{B}(v) \preceq \Phi_{Y}(\gamma_d(v))$ using that~$\gamma_d$ is an~$\Omega$-knitwork immersion, and analogously for~$v \in \bar{Y}$ implying~$\gamma(v) \in \bar{B}$. 

        The claims for~$\nu,\mu$ and~$\n,\m$ follow similarly from \cref{lem:stitch-and-knit} using the fact that~$\gamma_u,\gamma_d$ are~$\Omega$-knitwork immersions together with \cref{claim:stitch_knit_effect_on_cut} and 2. and 3. of \cref{obs:knitting_fundamentals_knitworks} (for~$\mu$ and~$\m$ respectively) and the fact that~$(B,\bar{B})$ and~$(Y,\bar{Y})$ partition the vertex sets~$V(H)$ and~$V(G)$.  
    \end{claimproof}

This concludes the proof.

\end{proof}

\cref{thm:knitting_knitwork_immersion} allows us to make the following simplification regarding~$\Omega$-knitworks.

\begin{corollary}\label{cor:immersion_of_stitches_yields_immersion}
    Let~$\GGG= \big((G,\pi(X)),\mu,\m,\Phi\big)$ and~$\HHH= \big((H,\pi(A)),\nu,\n,\Psi\big)$ be~$\Omega$-knitworks and~$\GGG_X \coloneqq \stitch(\GGG;\pi(X))$,$\GGG^X \coloneqq \stitch(\GGG;\pi(\bar{X}))$, as well as~$\HHH_A \coloneqq \stitch(\HHH;\pi(A))$ and~$\HHH^A \coloneqq \stitch(\HHH;\pi(\bar{A}))$. Then~$$\GGG_1 \hookrightarrow \GGG_2 \text{ if and only if } \GGG_X \hookrightarrow \HHH_A \text{ and } \GGG^X \hookrightarrow \HHH^\AAA,$$ 
   and the same holds true for general immersion $\hookrightarrow^*$.
\end{corollary}

By \cref{obs:stitching_fundamentals} $\rho_{G_X}(X) = \rho_{G_X}(x_\ast)$ and~$\rho_{G^X}(X) = \rho_{G^X}(x^\ast)$ for respective down- and up-stitch vertices~$x_\ast,x^\ast$ as well as~$\rho_{H_A}(A) = \rho_{H_A}(a_\ast)$ and~$\rho_{H^A}(A) = \rho_{H^A}(a^\ast)$ for respective down- and up-stitch vertices~$a_\ast,a^\ast$. Using the above corollary we deduce that it suffices to look at rooted Eulerian digraphs whose roots ``come from a single cut-vertex'' when proving well-quasi-ordering for~$\Omega$-knitworks under (strong) immersion. 

\begin{definition}[Closed under taking stitches]
Let~$\Omega$ be a well-quasi-ordering and let~$\mathbf{G}$ be some class of~$\Omega$-knitworks. We say that~$\mathbf{G}$ is \emph{closed under taking stitches} if for every~$\GGG=((G,\pi(X)),\mu,\m,\Phi)$ the graphs~$\stitch(\GGG;\pi(X))$ and~$\stitch(\GGG;\pi(\bar{X}))$ are contained in~$\mathbf{G}$.
\end{definition}

We have the following useful corollary.

\begin{corollary}\label{cor:knitwork_immersion_restriction_to_rooted_in_cutvertex}
    Let~$\Omega$ be a well-quasi-ordering and let~$\mathbf{G}$ be some class of~$\Omega$-knitworks. 
    
    Let~$\stitch(\mathbf{G})$ be a minimal class of~$\Omega$-knitworks such that for~$\GGG = ((G,\pi(X)),\mu,\m,\Phi)) \in \mathbf{G}$ it holds $\stitch(\GGG,\pi(X)), \stitch(\GGG,\pi(\bar{X})) \in \stitch(\mathbf{G})$. Then~$\mathbf{G}$ is well-quasi-ordered with respect to~$\Omega$-knitwork immersion if and only if~$\stitch(\mathbf{G})$ is.
\end{corollary}
\begin{proof}
    This follows at once from \cref{cor:immersion_of_stitches_yields_immersion} and Higman's \cref{thm:higman} applied to the tuple of respective down- and up-stitches.
\end{proof}

In particular if~$\mathbf{G}$ is a class closed under taking stitches, it suffices to look at the class of respective stitches, and if the number of vertices in $X$ is bounded for every $(G,\pi(X)) \in \mathbf{G}$, the question essentially boils down to whether or not the up-stitches are well-quasi-ordered as we will see in \cref{sec:bounded-case}.

\subsection{Well-quasi-ordering Graphs encoded in Trees}\label{subsec:wqo_framework_trees}

In their seminal work on graph minors, Robertson and Seymour have proved Wagner's Conjecture \cite{GMXX}, which states that undirected graphs are well-quasi-ordered by the minor relation. A first important step towards a full proof was to prove that graphs admitting bounded treewidth are well-quasi-ordered by the minor relation \cite{GMIV}. The idea behind that proof was to leverage techniques due to Kruskal \cite{Kru60} and Nash \cite{nash63} developed to prove that trees are well-quasi-ordered by topological containment, to the setting of general undirected graphs admitting some ``nice'' encoding into trees. That is, prove that trees admitting some restricted labelling are well-quasi-ordered by some form of ``labelled topological containment'' from which one can retrieve a well-quasi-order for the underlying encoded graphs. This resulted in their so-called \emph{Tree Lemma} \cite[Theorem 2.1 and 2.2]{GMIV}; we will use a version of this lemma in the main proof of \cref{sec:bounded-case}. While tree decompositions give an encoding of graphs into trees, a generic ``minimal'' tree decomposition of a graph was not enough to leverage the proof due to Kruskal and Nash; the tree decompositions need to be \emph{linked}. A tree decomposition~$(T,\beta)$ is called \emph{linked} if for every path~$P \subseteq T$ between two vertices~$v_T,w_T \in V(T)$ with~$\Abs{\beta(v_T)} = \Abs{\beta(w_T)} = k \in \N $ there either exist~$k$ edge-disjoint paths between~$\beta(v_T)$ and~$\beta(w_T)$ in~$G$ or there is an edge~$e=\{u,u'\} \in E(P)$ with~$\Abs{\beta(u) \cap \beta(u')} < k$. Thomas \cite{Tho90} proved that graphs of treewidth~$k$ admit a linked tree decomposition witnessing said width; a result later used by Robertson and Seymour in the published version of their result. 

\smallskip

A few years later Geelen, Gerards, and Whittle \cite{Gee02} dissected the proof due to Robertson and Seymour and generalized it to the setting of \emph{matroids} admitting bounded ``branch width''. They adapted the aforementioned Tree Lemma by giving a different version of it tied to cubic trees and suited for graphs of small \emph{branch width} (we provide the needed definition below). They generalize the definition of branch width of a graph as introduced by Robertson and Seymour to a whole class of parameters, where the respective branch width depends on an integer-valued symmetric submodular function; essentially depending on a choice of a separator function. In the same paper they prove that after a choice of such a separator function coming with a respective branch decomposition and branch width of a graph---in particular the carving width falls into that class---one can find a \emph{linked branch decomposition} witnessing the branch width of the graph \cite[Theorem 2.1]{Gee02}.
We want to make use of this result, which is another reason why we chose to work with carving width (besides the fact that it ``feels'' more natural when dealing with immersions): we do not need to provide another proof witnessing the existence of a witness adhering to some``linkedness property'' of some tree-like decomposition given some separator function. Hence, we will shortly introduce the results presented in \cite{Gee02} that are of direct use to us.

\smallskip

We start with some definitions that we transcribe from \cite{Gee02} to our setting.

\begin{definition}[Linked~$k$-labelled rooted Trees]
\label{def_linked_trees}
    Let~$T$ be a tree and~$k \in \N$. Let~$r \in V(T)$ be some vertex and direct all the edges of~$T$ away from~$r$, such that~$r$ has in-degree~$0$ and every other vertex has in-degree exactly~$1$. Then we call~$(T,r)$ a \emph{tree rooted at~$r$} or simply \emph{a rooted tree}.
    
    Let~$\omega: E(T) \to \{1,\ldots,k\}$ be some lanbelling function, then we call~$(T,\omega)$ a~\emph{$k$-labelled tree with label~$\omega$}, and similarly we call~$(T,r,\omega)$ a~\emph{$k$-labelled rooted tree}.

    Let~$e_1,e_2$ be two edges in~$E(T)$ with~$\omega(e_1) = \omega(e_2) = \rho$ such that there is a directed root-to-leaf path~$P$ in~$(T,r)$ visiting~$e_1$ prior to~$e_2$, and let~$P^1_2 \subseteq P$ be the sub-path with first edge~$e_1$ and last edge~$e_2$. If~$\omega(e) \geq \rho$ for every~$e \in E(P^1_2)$, then we say that~$(e_1,e_2)$ is \emph{$\omega$-linked in $(T,r,\omega)$} or~\emph{$e_2$ is~$\omega$-linked to~$e_1$}. We say that~$(T,r,\omega)$ is \emph{linked} if every pair of edges~$e,e' \in E(T)$ with~$\omega(e) = \omega(e')$ is linked.

\end{definition}
\begin{remark}
Clearly~$\omega$-linkedness is a transitive relation, i.e., if~$(f,f')$ and~$(e,f)$ are~$\omega$-linked, then~$(e,f')$ is~$\omega$-linked.
\end{remark}

The above definition is easily extended to rooted and labelled forests, where a rooted forest is a (possibly infinite) set of rooted trees with a common label~$\omega$ obtained by combining the labels of each tree to a single one. (Here we assume all the trees to have pairwise disjoint vertex and edge sets).

Let~$\FFF$ be a rooted forest and let~$S \subseteq E(\FFF)$ be some set of edges. Then we define~$u_\FFF(S)$ to be the set of those edges in~$\FFF$ whose tail is a head of an edge in~$S$; in a sense the ``child-edges'' with~$u$ referring to ``under''.

The following is the aforementioned version of the Tree Lemma \cite[Lemma 2.2]{GMIV} due to Robertson and Seymour, in the setting presented by Geelen, Gerards, and Whittle.

\begin{lemma}[Lemma on Trees; Theorem 3.1 in \cite{Gee02}]
\label{tree-lemma}
    Let~$\mathcal{F}$ be a rooted forest of~$k$-labeled trees with common label~$\omega$ for some~$k \in \N$. Let~$\preceq$ denote a quasi-order on the edges of~$\mathcal{F}$ with no infinite strictly descending sequence and such that~$e \preceq f$ whenever~$(f,e)$ is~$\omega$-linked. If the edges of~$\FFF$ are not well-quasi-ordered by~$\preceq$ then there exists an infinite antichain~$\AAA$ of edges of~$\FFF$ such that~$(u_\FFF(\AAA),\preceq)$ is a well-quasi-order.
\end{lemma}
\begin{remark}
    Think of the quasi-order on the edges as~$e \preceq e'$ if and only if the rooted tree below~$e$ can be ``(strongly) immersed'' in the rooted tree below~$e'$ in some sense that needs to be precise. 
\end{remark}

Geelen, Gerards, and Whittle extracted a very useful corollary to \cref{tree-lemma} suitable for cubic trees; we introduce the needed notation before stating said result. A \emph{leaf edge} of a forest~$\FFF$ is an edge~$e \in E(\FFF)$ that is adjacent to some~$v \in \leaves{\FFF}$, \emph{root-edges} and \emph{non-leaf edges} are defined in the obvious way.

\begin{definition}[Binary Forest]
    A \emph{($k$-labelled) binary forest}~$(\FFF,\mathrm{left},\mathrm{right})$ is a ($k$-labelled) rooted forest~$\FFF$ where the roots have out-degree~$1$ together with maps~$\mathrm{left},\mathrm{right}$ which are defined on the non-leaf edges such that the head of each non-leaf edge~$e \in E(\FFF)$ has exactly two distinct out-going edges~$\lenks{e},\riets{e}$.
\end{definition}

\begin{lemma}[Lemma on Cubic Trees, Theorem 3.2 in \cite{Gee02}]
\label{cubic-tree-lemma}

    Let~$k \in \N$ and let~$(\FFF,\mathrm{left},\mathrm{right})$ be an infinite~$k$-labelled binary forest with label~$\omega$. Let~$\preceq$ denote a quasi-order on~$E(\FFF)$ with no infinite strictly decreasing sequences, such that~$e\preceq f$ whenever~$(e,f)$ is~$\omega$-linked. If the leaf edges of~$\FFF$ are well-quasi-ordered by~$\preceq$ but the root edges of~$\FFF$ are not, then~$\FFF$ contains an infinite sequence~$(e_0,e_1,\ldots,)$ of non-leaf edges such that:
    \begin{itemize}
        \item[(i)] $(e_0,e_1,\ldots)$ is an antichain with respect to~$\preceq$,
        \item[(ii)] $\lenks{e_0} \preceq \ldots \preceq \lenks{e_{i-1}} \preceq \lenks{e_i} \preceq \ldots$
        \item[(iii)] $\riets{e_0} \preceq \ldots \preceq \riets{e_{i-1}} \preceq \riets{e_i} \preceq \ldots$.
    \end{itemize}
\end{lemma}

To make use of \cref{cubic-tree-lemma}, and inspired by the above \cref{def_linked_trees}, we define \emph{linked}~carvings as follows.

\begin{definition}[Linked carving]
\label{def:linked_carving}
    Let~$G$ be an Eulerian digraph and let~$(T,\ell)$ be a carving of~$G$ of width~$k \in \N$ and let~$\epsilon\colon E(T) \to 2^{E(G)}$ be the respective edge-labelling. Let~$\omega\colon E(T) \to \{1,\ldots,k\}$ be defined via~$\omega(e) = \Abs{\epsilon(e)}$. Let~$e_1,e_2 \in E(T)$ and let~$T_1,T_2$ be subtrees of~$T$ such that~$T_i$ is the component of~$T-e_i$ not containing~$e_{(i \mod 2) + 1}$ for~$i=1,2$. Let~$P$ be a shortest path in~$T$ containing both~$e_1$ and $e_2$. We call~$\{e_1,e_2\}$ \emph{linked} if the minimum width of an edge in~$P$ equals~$\delta(V(T_1),V(T_2))$. We call~$(T,\ell)$ \emph{linked}, if every pair of edges is linked.
\end{definition}

The following is the aforementioned result due to Geelen, Gerards, and Whittle, where the definition of ``branch width'' depends on the symmetric submodular function.

\begin{theorem}[Theorem 2.1 in \cite{Gee02}]
    An integer-valued symmetric submodular function with branchwidth $n$ has a linked branch decomposition of width~$n$.
\end{theorem}

We immediately derive the following.
\begin{corollary}
\label{cor:linked_ebw}
    Let~$G$ be an Eulerian digraph with~carving width~$n$. Then~$G$ admits a linked carving of width~$n$.
\end{corollary}

With \cref{cubic-tree-lemma} and \cref{cor:linked_ebw} we have gathered all the needed tools from the literature to prove the main theorem of the following section.

\section{The bounded width case} \label{sec:bounded-case}

Recall the notions introduced in \cref{subsec:knitworks}. We extend the \cref{def:carving} of carvings of Eulerian digraphs to carvings of rooted Eulerian digraphs as follows.

\begin{definition} 
    Let~$(G,\pi(X))$ be a rooted digraph and~$(T,\ell)$ a carving of~$G$ such that there exists~$e\in E(T)$ with~$\epsilon_{(T,\ell)}(e) = \rho(X)$. Then we define~$\rt(T) \coloneqq e$ to be a choice of such an edge and say that~$(T,\ell; e)$ is a \emph{rooted carving of~$(G,\pi(X))$}. The width is defined via~$\w(T,\ell;e) \coloneqq \operatorname{min}(\w(T,\ell), \Abs{X})$ and the carving width~$\ebw{G,\pi(X)}$ of the rooted Eulerian digraph is defined to be the minimum width over all rooted carvings of~$(G,\pi(X))$. 
\end{definition}
\begin{remark}
    A rooted carving always exists since~$\rho(X)$ is an induced cut.
\end{remark}

The carving width cannot change drastically when switching to rooted carvings.

\begin{lemma}\label{lem:cw_of_rooted_vs_unrooted}
    Let~$(G,\pi(X))$ be a rooted Eulerian digraph of index~$k \in 2\N$ with~$\Abs{X} \leq k$. Then~$k \leq \ebw{G,\pi(X)} \leq \ebw{G} + k$.
\end{lemma}
\begin{proof}
    The lower bound is clear. For the upper bound we construct a rooted carving witnessing said width as follows: Let~$(T_0,\ell_0)$ and~$(T_1,\ell_1)$ be carvings of~$G[X]$ and~$G[\bar{X}]$ witnessing the respective carving-width. Subdivide an edge of~$T_i$ adding a new vertex~$t_i$ of degree two for~$i=1,2$. Then define~$T$ to be the tree obtained from the forest~$\{T_0,T_1\}$ by adding the edge~$e=(t_0,t_1)$ between the two trees; note that~$\leaves{T} = \leaves{T_0} \cup \leaves{T_1}$. Define~$\ell \colon \leaves{T} \to V(G)$ via~$\restr{\ell}{\leaves{T_i}} = \ell_i$ for $i=1,2$. Then~$(T,\ell,e)$ is a rooted carving of~$(G,\pi(X))$ satisfying the claimed width since~$\w{(T_i,\ell_i)} \leq \ebw{G}$.
\end{proof}

The main result of this section reads as follows.

\begin{theorem}\label{thm:bounded_carving_with_non_labelled}
    Let~$k\in \N$, then the class of rooted Eulerian digraphs admitting carving-width at most~$k$ are well-quasi-ordered by strong immersion.
\end{theorem}

However, we will not prove \cref{thm:bounded_carving_with_non_labelled} in the stated form, since we will need a much stronger version for later sections and future work we are pursuing: We want to prove the result for well-linked~$\Omega$-knitworks. To this extent, recall \cref{def:knitwork}.

\begin{definition}
    Let~$\Omega$ be a well-quasi-order and let~$\GGG=(\bar{G},\mu,\m,\Phi)$ be an~$\Omega$-knitwork. We define the \emph{carving width} of~$\GGG$ to be the carving width of~$\bar{G}$ and write~$\ebw{\GGG} = \ebw{\bar{G}}$. 
\end{definition}

\begin{theorem}[Well-Quasi-Order for bounded Carving Width]\label{thm:wqo_bounded_carvingwidth_knitworks}
    Let~$k \in \N$ and let~$\Omega$ be a well-quasi-order. Let~$(\GGG_i)_{i\in \N}$ be a sequence of well-linked~$\Omega$-knitworks such that~$\ebw{\GGG_i} \leq k$ for all~$i \in \N$. Then there exist~$j>i \geq 1$ such that~$\GGG_i \hookrightarrow \GGG_j$ by strong $\Omega$-knitwork immersion.
\end{theorem}

 We prove the following theorem which easily implies \cref{thm:wqo_bounded_carvingwidth_knitworks} as we will see later.
 
\begin{theorem}\label{thm:wqo_bounded_carvingwidth_knitworks_red}
     Let~$k \in \N$ and let~$\Omega=(V(\Omega),\ll)$ be a well-quasi-order. Let~$(\GGG_i)_{i\in \N}$ be a sequence of well-linked~$\Omega$-knitworks~$\GGG_i=((G_i,\pi(x_i)),\mu_i,\m_i,\Phi_i)$ for some~$x_i \in V(G_i)$ with~$x_i \notin \dom(\Phi_i) \cup \dom(\m_i)$, $\mu_i(x_i) = \pi(x_i)$, and such that~$\ebw{\GGG_i} \leq k$ for all~$i \in \N$. Then there exist~$j>i \geq 1$ such that~$\GGG_i \hookrightarrow \GGG_j$ by strong $\Omega$-knitwork immersion.
\end{theorem}

\begin{proof}
Towards a contradiction, assume that the theorem is false. Let~$(\GGG_i)_{i \in \N}$ be a bad sequence of well-linked~$\Omega$-knitworks with~$\ebw{\GGG_i} \leq k $ for all~$i \in \N$ with respect to strong~$\Omega$-knitwork immersion. Let~$\GGG_i=(\bar{G_i},\mu_i,\m_i,\Phi_i)$ where~$\bar G_i = (G_i,\pi_{G_i}(x_i))$ for some~$x_i \in V(G_i)$ with~$x_i \notin \dom(\Phi_i)\cup \dom(\m_i)$ for every~$i \in \N$. By \cref{lem:cw_of_rooted_vs_unrooted}~$\delta(x_i) \leq k$ for every~$i \in \N$ and by switching to a suitable infinite subsequence indexed by $I\subseteq \N$ we may assume that~$\delta(x_i) = \delta(x_j)$ for every~$i,j \in I$. Additionally assuming minimality on $k$ we find $\delta(x_i) = k$ for every $i \in I$; identify $I\cong \N$ for simplicity.

    \smallskip
    
Since the~$\GGG_i$ are rooted in a single vertex~$x_i$ we have~$\ebw{\GGG_i} = \ebw{G_i}$ for every~$i \in \N$ (see the proof of \cref{lem:cw_of_rooted_vs_unrooted}). By \cref{cor:linked_ebw} there exists a linked carving~$(T_i,\ell_i)$ for~$G_i$ and every~$i \in \N$ witnessing its carving width and in particular witnessing~$\ebw{\GGG_i}$. For every~$i \in \N$, let~$e_i \in E(T_i)$ be the unique edge adjacent to~$r_i \coloneqq \ell_i(x_i)$, then~$(T_i,\ell_i;e_i)$ is a rooted carving of~$\bar G_i$ and we root the cubic trees~$T_i$ at~$r_i$ (thus the edges of the trees~$T_i$ are directed from here on). Let~$\epsilon_i \coloneqq \epsilon_{(T_i,\ell_i)}$ and define~$\omega_i(e) \coloneqq \Abs{\epsilon_i(e)}$ for every~$e \in E(T_i)$ and every~$i \in \N$. Then, for every $i\in \N$,~$\omega_i\colon E(T_i) \to \{1,\ldots,k\}$ is a labelling function and~$(T_i,r_i,\omega_i)$ is a~$k$-labelled rooted tree with label~$\omega_i$. Fix $i\in \N$ and let~$e=(u,v) \in E(T_i)$ be some directed edge. Let~$T_i^u$,~$T_i^v$ be the two unique components of~$T_i-e$, where~$T_i^u$ contains~$u$ and~$T_i^v$ contains~$v$. We define~$X_e \coloneqq \ell_i^{-1}(\leaves{T_i^u})$ and~$X^e \coloneqq \ell_i^{-1}(\leaves{T_i^v})$. Clearly~$X^e = \bar{X_e}$ by definition; in particular we have the following.

\begin{claim}\label{claim:thm_wqo_bounded_carvingwidth_knitworks_red_uno}
    $X_e$ induces a rooted cut in~$\bar G_i$ and~$\rho(X_e) = \epsilon_i(e) = \rho(X^e)$.
\end{claim}
\begin{claimproof}
    Since~$(T_i,r_i)$ is a rooted tree underlying $(T_i,\ell_i;e_i)$, where~$r_i =\ell(x_i)$, it follows that~$x_i \in X_e$. Thus~$X_e$ induces a rooted cut; the rest follows trivially by definition.
\end{claimproof}

 Assuming that~$(V(T_i)\cup E(T_i)) \cap (V(T_j)\cup E(T_j)) = \emptyset$ for distinct~$ i,j \in \N$---otherwise rename the trees accordingly---we let~$(\FFF,\mathrm{left},\mathrm{right})$ be the binary forest obtained by the union of all the rooted trees~$(T_i,r_i)$ after a choice of~$\mathrm{left}$ and~$\mathrm{right}$ maps for each rooted tree which is feasible for they are cubic. Let~$\omega$ be the respective labelling function for~$\FFF$ obtained via~$\restr{\omega}{E(T_i)} = \omega_i$ for every $i\in \N$. We continue by defining a quasi-order~$\preceq$ on the edges of the forest that respect~$\omega$-linkedness. To this extent we need the following technical result building the bridge between the \cref{def:linked_carving} of linked carvings and $\omega$-linkedness. Intuitively speaking the following claim states that, given a pair $(e,f)$ of $\omega$-linked edges in $T_i$, the cuts (of equal order $t\in 2\N$) that they induce in $G_i$ can be connected by a $t$-linkage $\LLL_i$ such that $\LLL_i$ respects the restrictions imposed by $\m_i$, i.e., it is a ``feasible'' linkage for every $v \in \dom(\m_i)$.  

 \begin{claim}\label{claim:thm_wqo_bounded_carvingwidth_knitworks_red_feasible_linkages}
      Let~$i \in \N$ be fixed and~$e,f \in E(T_i)$ be two edges such that~$(e,f)$ is~$\omega$-linked and let~$t \coloneqq \omega(e) = \omega(f)$. Then there exists a linear~$\{\rho(X^f),\rho(X_e)\}$-linkage~$\LLL_i$ of order~$t$ such that~$M_{\LLL_i}(v) \in \m_i(v)$ for every~$v \in \dom(\m) \cap X_f \cap X^e$. Further~$\bigcup_{P \in \LLL_i}V(P) \cap (X^f \cup X_e) = \emptyset$ and every path~$P \in \LLL_i$ has one endpoint in~$X^f$ and one endpoint in~$X_e$.
 \end{claim}
 \begin{claimproof}
    Since the carving~$(T_i,\ell_i)$ is linked, using the \cref{def:linked_carving}, \cref{lem:Menger_for_Euler} together with the \cref{def:linked_carving} together with \cref{obs:linkage_gives_linear_linkage} imply the existence of a linear~$\{\rho(X^f),\rho(X_e)\}$-linkage~$\LLL_i$ of order~$t$ in $G_i$ such that for every~$P \in \LLL_i$ it holds~$V(P) \subseteq X_f \cap X^e$ by restricting the linkage accordingly (recall that~$V(P)$ are the internal vertices of~$P$ by \cref{def:paths}). 
    
    Assume towards a contradiction that~$\LLL_i$ is chosen to satisfy the above and so that the set of vertices~$F=\{v \in \dom(\m) \cap X_f \cap X^e \mid M_{\LLL_i}(v) \not \in \m_i(v)\}$ has minimal cardinality; if~$\Abs{F} = 0$ we are done. Let~$v \in F$, then since~$\LLL_i$ is a linkage, every~$(f,f') \in M_{\LLL_i}(v)$ is a sub-path of exactly one path in~$\LLL_i$, and since it is linear every $P \in \LLL_i$ admits at most one tuple $(f,f') \in M_{\LLL_i}(v)$ as a subpath. Let  $\LLL_i=\{P_1,\ldots,P_t\}$ and without loss of generality let~$\{(f_j,f_j') \mid 1 \leq j \leq \tau\} = M_{\LLL_i}(v)$ for some~$\tau \leq t$ and assume~$(f_j,f_j') \subset P_j$ for every~$1 \leq j \leq \tau$. Let~$E^-\coloneqq \{f_j \mid 1 \leq j \leq \tau\}$ and~$E^+ \coloneqq \{f_j' \mid 1\leq j \leq \tau\}$, then
     \begin{equation}\label{eq:well-linked_uno}
         E^-\subseteq \rho^-(v),\ E^+ \subseteq \rho^+(v), \text{ and } \Abs{E^-} = \Abs{E^+}.
     \end{equation} 
     Let~$P_j^1,P_j^2 \subset P_j$ be the unique edge-disjoint sub-paths satisfying~$E(P_j) = E(P_j^1) \cup E(P_j^2)$ and~$P_j^1$ ends in~$f_j$ and~$P_j^2$ starts in~$f_j'$; in particular~$V(P_j^1) \cap V(P_j^2) = \{v\}$ and $P_j = P_j^1 \circ P_j^2$ for every~$1 \leq j \leq \tau$. Let now~$\LLL_i^f \coloneqq \{P_j^q \mid P_j^q \text{ has an end in } \rho(X^f), \text{ for } 1 \leq j \leq \tau,\text{ and } q\in\{1,2\}\}$ and define~$\LLL_i^e$ analogously using $\rho(X_e)$. Clearly~$\LLL_i^f \cup \LLL_i^e = \{P_j^q \mid 1 \leq j \leq \tau \text{ and } q \in \{1,2\}\}$ as well as  $\LLL_i^f \cap \LLL_i^e = \emptyset$ by construction . 

     Define~$E_f^- \subset E^-$ via~$e \in E_f^-$ if and only if~$e \in E^-$ and~$e$ is an end of a path in~$\LLL_i^f$, and~$E_e^-\subset E^-$ via~$e \in E_e^-$ if and only if~$e \in E^-$ and~$e$ is an end of a path in~$\LLL_i^e$. Define~$E_f^+$ and~$E_e^+$ analogously using~$E^+$. Note that the linear linkage~$\LLL_i$ witnesses
     \begin{equation}\label{eq:well-linked_dos}
         \Abs{E_f^- } = \Abs{ E_e^+} \text{ and } \Abs{E_f^+ } = \Abs{ E_e^-}.
     \end{equation}

     Finally, by \cref{eq:well-linked_uno} and \cref{eq:well-linked_dos} together with the \cref{def:well-linked_links} of well-linked links, there exists~$M\in \m_i(v)$ such that~$M= M_1 \cup M_2$ where~$M_1 \in \operatorname{Match}(E_f^-,E_e^+)$ and $M_2 \in\operatorname{Match}(E_e^-,E_f^+)$ are perfect matchings. Let~$(f_p,f_q') \in M_1$ for respective~$1 \leq p,q \leq \tau$, then~$P_p^1 \circ P_q^2$ is a linear path with first edge in~$\rho(X^f)$ and last edge in~$\rho(X_e)$. Similarly for~$(f_p,f_q') \in M_2$ for respective~$1 \leq p,q \leq \tau$, then~$P_p^1 \circ P_q^2$ is a linear path with first edge in~$\rho(X_e)$ and last edge in~$\rho(X^f)$. It is easily verified that the collection of paths~$\PPP_i\coloneqq \{P_p^1 \circ P_q^2 \mid (f_p,f_q') \in M_1\cup M_2 \text{ for } 1\leq p,q \leq \tau\}$ is a linear~$\{\rho(X^f),\rho(X_e)\}$-linkage, and finally~$\LLL_i'=\PPP_i \cup\{P_p \mid \tau+1 \leq p \leq t\}$ is a linear~$\{\rho(X^f),\rho(X_e)\}$-linkage of order $t$ with~$M_{\LLL_i'}(v) \in \m_i(v)$. Since~$M_{\LLL_i'}(w) = M_{\LLL_i}(w)$ for all~$w \neq v$ with~$w \in \dom(\m) \cap X_f \cap X^e$by construction,~$\LLL_i'$ refutes the choice of~$\LLL_i$ minimizing $\Abs{F}$; a contradiction.
 \end{claimproof}

Recall that~$X^f = \bar{X_f}$.
  \begin{claim}\label{claim:thm_wqo_bounded_carvingwidth_knitworks_red_linked}
        Let~$i \in \N$ be fixed and~$e,f \in E(T_i)$ be two edges such that~$(e,f)$ is~$\omega$-linked. Let~$\pi_{G_i}(X_e)=\pi_{G_i}(X^e)$ be an ordering of~$\rho_{G_i}(X_e)$. Then there exists an ordering~$\pi_{G_i}(X_f)=\pi_{G_i}(X^f)$ of~$\rho_{G_i}(X_f)$ such that~$\stitch(\GGG_i,\pi_{G_i}(X^f)) \hookrightarrow \stitch(\GGG_i,\pi_{G_i}(X^{e}))$ by strong $\Omega$-knitwork immersion.
    \end{claim}
    \begin{claimproof}
      For notational convenience we write~$\GGG= \GGG_i$. Let~$P$ be the shortest directed path in~$T_i$ visiting the edges~$e$ and~$f$ in that order. Then, since~$(e,f)$ is~$\omega$-linked, we know that~$\omega(e) = t = \omega(f)$ and~$\omega(\eta) \geq t$ for all~$\eta \in E(P)$ and some~$t \in 2\N$. In particular, since the carving is linked, we know that~$\Abs{\epsilon(e)} = \Abs{\epsilon(f)}$ and there is no cut~$(X,\bar{X})$ in~$G$ separating~$X^f$ from~$X_e$ of lower order by \cref{def:linked_carving} of linkedness. By \cref{claim:thm_wqo_bounded_carvingwidth_knitworks_red_feasible_linkages} there exists a linear~$\{\rho(X^f),\rho(X_e)\}$-linkage~$\LLL_i$ of order~$t$ such that~$M_{\LLL_i}(v) \in \m_i(v)$ for every~$v \in \dom(\m) \cap X_f \cap X^e$ and such that~$\bigcup_{i=1}^t V(P_i)  \cap (X^f \cup X_e) = \emptyset$.

      \smallskip
      
      Let~$\pi_{G}(X_e) = (e_1,\ldots,e_t)$ be the ordering of the claim and without loss of generality assume that~$e_j$ is an end of~$P_j$ for every~$1 \leq j \leq t$. 
        Let~$\rho(X^f) = \{f_1,\ldots,f_t\}= \rho(X_f) $ and without loss of generality assume that~$f_j$ is the other end of~$P_j$---note that~$f_j = e_j$ is possible---for every~$1\leq j \leq t$, else rename them accordingly. We set~$\pi_{G}(X_f) \coloneqq (f_1,\ldots,f_t)$. Finally let~$\GGG^{X_f}\coloneqq (\bar G^{X_f}, \mu^{X_f},\m^{X_f}.\Phi^{X_f}) =\stitch(\GGG,\pi_{G}(X^f))$ and~$\GGG^{X_e} \coloneqq (\bar G^{X_e}, \mu^{X_e},\m^{X_e}.\Phi^{X_e}) =\stitch(\GGG,\pi_{G}(X^e))$ be the respective up-stitches with up-stitch vertices~$f^*$ and~$e^*$\footnote{Recall that $X_f,X_e$ are the rooted cuts whence we write $G^{X_f}, G^{X_e}$ for the up-stitches whereas $V(G^{X_f}) = X^f \cup \{f^*\}$.}. Note that by \cref{claim:thm_wqo_bounded_carvingwidth_knitworks_red_uno}~$X_e,X_f$ are rooted cuts with $X^f = \bar X_f$ as well as $X^e = \bar X_e$ whence the above defined up-stitches are well-defined. Note further that~$\mu^{X_f}(f^*) = (f_1,\ldots,f_t) = \pi_{G^{X_f}}(f^*)$ and~$\mu^{X_e}(e^*) = (e_1,\ldots,e_t) = \pi_{G^{X_e}}(e^*)$ by \cref{def:stitching_knitwork} of up-stitches.
      
      We claim that~$\GGG^{X_f} \hookrightarrow \GGG^{X_e}$ by strong $\Omega$-knitwork immersion. To see this, define~$\gamma\colon \bar G^{X_f} \to \bar G^{X_e}$ to be the identity on~$X^f \subseteq X^e$  and set~$\gamma(f^*) \coloneqq e^*$. Further let~$\gamma$ be the identity on~$E(G[X^f]) \subseteq E(G[X^e])$ and set~$\gamma(f_j) \coloneqq P_j$ for every~$1 \leq j \leq t$. We claim that~$\gamma$ is a strong~$\Omega$-knitwork immersion as desired. We verify (1)-(6) of \cref{def:knitwork_immersion}.
      \begin{itemize}
      \item[(1)] By construction~$\gamma\colon G^{X_f} \hookrightarrow G^{X_e}$ is a strong immersion where further~$\gamma\big((f_1,\ldots,f_t)\big) = (e_1,\ldots,e_t)$ since~$e_j \in P_j = \gamma(f_j)$, implying that~$\gamma\colon \bar G^{X_f} \hookrightarrow \bar G^{X_e}$ is a strong rooted immersion; thus (1) is satisfied. 
      \item[(2)] For $\m^{X_f}$ and $\m^{X_e}$ this is trivially satisfied since $\restr{\gamma}{X^f}$ is the identity and they are both not defined on $f^*$ and $e^*$ respectively. For $\mu^{X_f}$ and $\mu^{X_e}$ satisfied since~$\dom(\mu^{X_f})\setminus\{f^*\} \subset \dom(\mu^{X_e})$ where $\restr{\gamma}{\dom(\mu^{X_f})\setminus\{f^*\}}$ is the identity and~$\gamma(f^*) = e^* \in \dom(\mu^{X_e})$.
      \item[(3)] The claim is trivially true for all~$v \in X^f$ with~$\rho(v) \subset E(G[X^f])$ again since~$\restr{\gamma}{E(G[X^f])}$ is the identity and thus~$\mu^{X_f}(v) = \mu^{X_e}(v)$ and in particular (3) holds. The claim holds for~$f^*$ since~$\mu^{X_f}(f^*) = (f_1,\ldots,f_t)$ and~$\mu^{X_e}(\gamma(f^*)) = \mu^{X_e}(e^*) = (e_1,\ldots,e_t)$ with~$e_j \in P_j = \gamma(f_j)$ for all~$1 \leq j \leq t$. Thus let~$v \in X^f$ be a vertex with~$\rho_{G^{X_f}}(v) \cap \{f_1,\ldots,f_t\} \neq \emptyset$. By \cref{def:stitching_std,def:stitching_knitwork}~$\rho_{G^{X_f}}(v) = \rho_{G^{X_e}}(v) $ and~$\mu_{G^{X_f}}(v) = (g_1,\ldots,g_p) = \mu_{G^{X_e}}(v)$ for some~$p \leq k$. Again (3) is trivially satisfied for all~$1 \leq j \leq p$ with ~$g_j \notin \{f_1,\ldots,f_t\}$ since~$\gamma(g_j) = g_j$, and for~$g_j = f_\ell$ for some~$1 \leq \ell \leq t$ it holds~$\gamma(g_j)=P_\ell \subset G^{X_e}$ which is a path starting in~$f_\ell = g_j$; thus (3) holds.
      \item[(4)] Note that~$f^* \notin \dom(\m^{X_f})$ and~$e^* \notin \dom(\m^{X_e})$ by \cref{def:stitching_knitwork} for they are up-stitch vertices, thus (4) does not apply to~$e^*$. Finally for~$v' \in \dom(\m^{X_e}) \setminus \gamma(\dom(\m^{X_f}))$ with~$v' \notin X_f \cap X^e$ the claim follows once again since~$\gamma$ is the identity (in particular if~$v' \in \dom(\m^{X_e})$ then~$v' = \gamma(v')$ whence~$v' \in \gamma(\dom(\m^{X_f}))$ and (4) does not apply).
      
      If~$v' \in X_f \cap X^e$, then by our choice of~$\LLL_i$ and construction of~$\gamma$,~$M_\gamma(v') = M_{\LLL_i}(v') \in \m^{X_e}(v')$ whence (4) is verified.

      \item[(5)] Since by definition~$\m^{X_e}$ and~$\m^{X_f}$ are not defined on~$e^*$ and~$f^*$ respectively, the claim follows again by~$\restr{\gamma}{X^f}$ being the identity. 
      
      \item[(6)] Since~$f^* \notin \dom(\Phi^{X_f})$ and $e^* \notin \dom(\Phi^{X_e})$ by \cref{def:stitching_knitwork}, and~$\restr{\gamma}{X^f}$ is the identity, (6) is satisfied.
      \end{itemize}
    Thus the claim follows.
    \end{claimproof}

We derive the following.

\begin{claim}\label{claim:thm_wqo_bounded_carvingwidth_knitworks_red_pi_i}
    Let~$i \in \N$ be fixed. Then there is a map~$\pi_i$ which for each edge~$e \in E(T_i)$ fixes an ordering~$\pi_i(X_e) = \pi_i(X^e)$ of the cut~$\rho_{G_i}(X_e)$, such that if~$(e,f)$ is~$\omega_i$-linked for edges~$e,f \in E(T_i)$, then $\stitch(\GGG_i;\pi_i(X^f)) \hookrightarrow \stitch(\GGG_i;\pi_i(X^e))$ by strong $\Omega$-knitwork immersion. Further~$\pi_i(X_{e_i}) = \pi_{G_i}(x_i) = \pi_i(X^{e_i})$.
\end{claim}
\begin{claimproof}
    The proof is via induction by pushing orders starting from the root edge~$e_i$ as follows. Let~$F_i(e_i) \subseteq E(T_i)$ be the maximal set of edges containing all~$f \in E(T_i)$ with~$(e_i,f)$ being~$\omega_i$-linked such that there is no other edge~$f' \in E(T_i)$ on the unique directed~$(e_i,f)$-path in~$T_i$ with~$f' \in F_i(e)$. For each such edge $f \in F_i$ define~$\pi_i(X_f)$ via \cref{claim:thm_wqo_bounded_carvingwidth_knitworks_red_linked}. Let~$T_i^1,T_i^2$ be the sub-trees rooted in~$\lenks{e_i}$ and $\riets{e_i}$ respectively. If~$\lenks{e_i}\in F_i(e)$ then~$\pi_i(X_{\lenks{e_i}})$ was already defined and we may continue inductively as above by defining $F_i(\lenks{e_i})$ analogously. Otherwise, choose an arbitrary ordering~$\pi_i(X_{\lenks{e_i}})$ and again continue as above by defining $F_i(\lenks{e_i})$. Repeat inductively until~$\pi_i$ has been defined on all of $E(T_i)$.

    This procedure clearly terminates and defines~$\pi_i$ for each edge of~$E(T_i)$; recall that if~$(f,f')$ and~$(e,f)$ are~$\omega_i$-linked then so is~$(e,f')$, i.e., the relation is transitive. Hence, we do not ``doubly define"~$\pi_i$.

    Finally, the proof of the claim follows from the transitivity of strong~$\Omega$-knitwork immersion. To see this let~$(e,f')$ be~$\omega_i$-linked and assume that~$\stitch(\GGG_i,\pi_i(X^{f'}))$ does not strongly immerse in $\stitch(\GGG_i,\pi_i(X^e))$. If there is no~$f' \in E(T_i)$ distinct from~$e$ and~$f$ such that~$(e,f)$ and~$(f,f')$ are~$\omega_i$-linked, then the definition of~$\pi_i$ and \cref{claim:thm_wqo_bounded_carvingwidth_knitworks_red_linked} yield a contradiction. Thus let~$e,f'$ be chosen so that the maximal set~$F \subseteq E(T_i)$ with~$f \in F$ if and only if~$(e,f),(f,f')$ are~$\omega_i$-linked has minimal cardinality; by the above~$\Abs{F} \geq 1$. Let~$f \in F$, then~$(e,f),(f,f')$ are~$\omega_i$-linked and by the minimality choice of~$e,f'$ it follows that~$\stitch(\GGG_i,\pi_i(X^{f'})) \hookrightarrow \stitch(\GGG_i,\pi_i(X^f))$ and~$\stitch(\GGG_i,\pi_i(X^f)) \hookrightarrow \stitch(\GGG_i,\pi_i(X^e))$. By \cref{lem:knitwork_imm_is_quasi_order} strong~$\Omega$-knitwork immersion is transitive, and thus $\stitch(\GGG_i,\pi_i(X^{f'})) \hookrightarrow \stitch(\GGG_i,\pi_i(X^e))$; a contradiction to the choice of~$e,f'$. The claim follows.
\end{claimproof}

By collecting all of the maps~$\pi_i$ of \cref{claim:thm_wqo_bounded_carvingwidth_knitworks_red_pi_i} we may define~$\pi$ for~$E(\FFF)$ via~$\restr{\pi}{E(T_i)} = \pi_i$.

Finally, in light of the \cref{cubic-tree-lemma} on Cubic Trees, we define a quasi-order~$\preceq$ on~$E(\FFF)$ as follows.

 \begin{itemize}
        \item[$\preceq:$] Let~$e=(u,v),e'=(u',v') \in E(\FFF)$ with~$e \in E(T_i)$ and~$e' \in E(T_j)$ for some~$i,j \in \N$. Then we define~$e\preceq e'$ if and only if~$\stitch(\GGG_i,\pi(X^e)) \hookrightarrow \stitch(\GGG_j,\pi(X^{e'}))$ by strong~$\Omega$-knitwork immersion.
\end{itemize}

We next verify that the above-defined objects satisfy the conditions of the \cref{cubic-tree-lemma} on Cubic Trees. By definition of~$\preceq$ and the choice of~$\pi$ respecting \cref{claim:thm_wqo_bounded_carvingwidth_knitworks_red_pi_i} we derive that~$e \preceq f$ whenever~$(e,f)$ is~$\omega$-linked (in particular $e,f$ are edges of the same tree in $\FFF$).

We start with the root edges.
    \begin{claim}
        The root edges are not well-quasi-ordered by $\preceq$.
    \end{claim}
    \begin{claimproof}
        If the root edges were well-quasi-ordered, then this would imply the existence of~$1\leq i<j$ such that $\stitch(\GGG_i, \pi(X^{e_i})) \hookrightarrow \stitch(\GGG_j, \pi(X^{e_j}))$ by strong $\Omega$-knitwork immersion, where~$\pi(X^{e_i}) = \pi_{G_i}(x_i)$ and~$\pi(X^{e_j}) = \pi_{G_j}(x_j)$ by definition. By the \cref{def:stitching_knitwork} of up-stitches and the fact that~$x_i \notin \dom(\Phi_i) \cup \dom(\m_i)$ as well as~$\mu_i(x_i) = \pi_{G_i}(x_i)$, we derive that~$\stitch(\GGG_i, \pi(X^{e_i})) \cong \GGG_i$ and~$\stitch(\GGG_j, \pi(X^{e_j})) \cong \GGG_j$. This contradicts the assumption that the initial sequence $(\GGG_i)_{i\in \N}$ was bad.
    \end{claimproof}

 Next, we deal with the leaf edges; note that the degree of vertices may still be arbitrarily large since we allow loops, since loops will not be ``seen'' by the cuts carved by~$\epsilon$.
    \begin{claim}\label{claim:thm_wqo_bounded_carvingwidth_knitworks_red_leafs}
        The leaf edges are well-quasi-ordered by $\preceq$.
    \end{claim}
    \begin{claimproof}
        For any $i \in \N$ and any leaf edge~$e=(t,t')\in E(T_i)$ note that the graph~$G_i^e\coloneqq G_i^{X_e}$ of~$\stitch(\GGG_i,\pi(X^e))$ consists of exactly two vertices: namely the \emph{central vertex}~$v \coloneqq \ell_i^{-1}(t')$ and the newly introduced up-stitch vertex~$x^*$ say. It is easily seen that~$E(G_i^e) = \rho(x^*) \cup \loops(v)$ where~$\delta(x^*) \leq k$ by assumption on the carving width. Define~$\xi(v) \coloneq \Abs{\loops(v)} \in \N$. 

        Finally, let~$(e(i))_{i\in \N}$ be a sequence of leaf edges and without loss of generality---by switching to an infinite subsequence indexed by $I\subseteq \N$---assume that~$e(i) \in E(T_i)$, for every tree has only a finite number of leaf edges. For every~$i \in I$ let~$\big((G_i^{e(i)},\pi(x_i^*), \nu_i, \n_i, \Psi_i\big) = \stitch(\GGG_i,\pi(X^{e(i)}))$ and denote the respective central vertex of~$G_i^{e(i)}$ by~$v(i)$. By \cref{obs:wqo_of_tuples}~$\Omega' \coloneqq \Omega \times \N$ is a well-quasi-order; let~$\GGG_i^{e(i)} \coloneqq \big((G_i^{e(i)},\pi(x_i^*), \nu_i, \n_i, \Psi_i'\big)$ be an~$\Omega'$-knitwork with~$\Psi_i'(v(i)) \coloneqq \big(\Psi_i(v(i)),\xi(v_i)\big)$. 

        We start filtering the sequence~$(\GGG_i^{e(i)})_{i \in I}$ to derive the claim via straightforward applications of the pigeonhole principle. Since~$\Omega'$ is a well-quasi-order, using \cref{obs:wqo_yields_infinite_chain} we take an infinite subsequence indexed by~$I_1 \subseteq I$ such that~$\Psi_i(v(i)) \ll \Psi_j(v(j))$ and~$\xi(v(i)) \leq \xi(v(j))$ for every~$i\leq j$ and~$i,j \in I_1$. Further, since by \cref{def:stitching_knitwork}~$x_i^* \notin \dom(\Psi_i)$ and since~$\delta(x_i) \leq k$, we may switch to an infinite subsequence indexed by~$I_2 \subseteq I_1$ such that~$\delta(x_i) = \delta(x_j) = t$ for every~$i,j \in I_2$ and some~$t \leq k$.

        In addition, there exists an infinite index set~$I_3 \subseteq I_2$ such that, without loss of generality,~$v(i) \in \dom(\mu_i) \cap \dom(\m_i)$ for every~$i \in I_3$ (the other case is analogous). Since~$\delta(v(i)) = t$ we may rename the edges of the graphs so that~$\rho(v(i)) = \rho(v(j))=\{g_1,\ldots,g_t\}$ for all~$i,j \in I_3$. Since there are only finitely many distinct linear orderings on~$\{g_1,\ldots,g_t\}$, again applying the pigeonhole principle implies that there exists an infinite index set~$I_4 \subseteq I_3$ with~$\mu_i(v(i)) = \mu_j(v(j))$ and~$\m_i(v(i)) = \m_j(v(j))$ for all~$i,j \in I_4$.
        
        But then this immediately implies the claim for one now easily verifies that~$\GGG_i^{e(i)} \hookrightarrow \GGG_j^{e(j)}$ by strong $\Omega'$-knitwork immersion for every~$i<j$ with~$i,j \in I_4$ by mapping~$v(i)$ to~$v(j)$ and~$x_i^*$ to~$x_j^*$. Note that the strong~$\Omega'$-knitwork immersion implicitly yields a strong~$\Omega$-knitwork immersion.
    \end{claimproof}

We prove the last missing ingredient for \cref{cubic-tree-lemma}.

\begin{claim}
    There is no infinite strictly decreasing sequence~$(f_i)_{i \in \N}$ for $\preceq$.
\end{claim}
\begin{claimproof}
    Assume the contrary; let~$\GGG_i = (G_i,\pi(x_i)), \mu_i, \m_i, \Phi_i)$ be such that $f_i \in E(G_i)$ for every $i \in \N$ and without loss of generality (by switching to a respective infinite subsequence) assume that for every pair of distinct $i,j \in \N$, $G_i \neq G_j$. Let $\HHH_i \coloneqq \stitch(\GGG_i,\pi(X^{f_i}))$ be the respective up-stitches with up-stitch vertex $h_i^*$ and write $\HHH_i = (H_i,\pi(h_i^*)), \nu_i, \n_i, \Psi_i)$ for convenience. Let~$m\coloneqq \Abs{H_1} = \Abs{E(H_1)} + \Abs{V(H_1)}$. Then, by assumption on $(f_i)_{i\in \N}$ being strictly decreasing with respect to $\preceq$, for every~$ j>1$ we derive~$\HHH_j \hookrightarrow \HHH_1$ by definition of $\preceq$ implying~$\Abs{H_j} \leq m$. Since there are only finitely many Eulerian digraphs~$G$ (up to isomorphism) satisfying~$\Abs{G} \leq m$, by the pigeonhole principle there is an infinite subsequence indexed by~$I \subseteq \N$ such that~$H_i \cong H_j$ for every~$i \in I$. Let the vertices of~$H_i$ be given by~$\{u_1^i,\ldots,u_t^i\}$ with the respective isomorphisms mapping~$u_\ell^i$ to~$u_\ell^j$ for every~$i,j \in I$ and~$1 \leq \ell \leq t$. Similarly, by the pigeonhole principle, there is~$I_2 \subseteq I$ such that~$\dom(\nu_i) = \dom(\nu_j)$,$\dom(\n_i) = \dom(\n_j)$ and~$\dom(\Psi_i) = \dom(\Psi_j)$ where further~$\mu_i(u_\ell^i) = \mu_j(u_\ell^j)$---if defined---and $\n_i(u_\ell^i) = \n_j(u^j_\ell)$---if defined---for every~$i,j \in I_2$ and every~$1 \leq \ell \leq t$. Thus, since the sequence is strictly decreasing with respect to~$\preceq$, there must exist~$1 \leq \ell \leq t$ such that~$\Psi_j(u_\ell^j) \ll \Psi_i(u_\ell^i)$ for every~$j > i$ with~$i,j \in I_2$. Then this witnesses an infinite strictly decreasing sequence of elements in~$\Omega$; a contradiction to~$\Omega$ being a well-quasi-order. The claim follows.
\end{claimproof}

Combining the previous claims, the Cubic Tree \cref{cubic-tree-lemma} implies the existence of an infinite sequence~$(f_i)_{i\in\N}$ of non-leaf edges such that the sequence forms an antichain with respect to~$\preceq$ whereas~$(\lenks{f_i})_{i\in \N}$ and~$(\riets{f_i})_{i\in \N}$ form chains with respect to~$\preceq$.  By switching to an infinite subsequence similar to above, we may assume that~$f_i \in E(T_i)$ for every~$i \in \N$. We are left to refute the existence of said antichain by proving that we find a pair~$f_i,f_j$ such that we can ``knit'' the~$\Omega$-knitwork immersions given on their left and right children to derive that~$f_i \preceq f_j$ as a contradiction.
\smallskip

We start by filtering our sequence using the pigeonhole principle to get nicer properties in several steps as follows; let~$I = \N$ and let~$\GGG_i^f \coloneqq \stitch(\GGG_i, \pi(X^{f_i})$ with up-stitch vertex~$f_i^*$ for every~$i \in I$, then $(\GGG_i^f)_{i\in I}$ forms the respective antichain.

\begin{enumerate}
    \item[$I_1$] Recall that~$\omega(f_i) \leq k$ for every $i \in I$, so by taking an infinite subsequence indexed by~$I_1 \subseteq I$ we can assume that~$\omega(f_i) = \omega(f_j)$ for all~$i,j \in \N$ and similarly~$\omega(\lenks{f_i}) = \omega(\lenks{f_j})$ as well as~$\omega(\riets{f_i}) = \omega(\riets{f_j})$ for every $i,j \in I_1$.

    \item[$I_2$] By taking an infinite subsequence indexed by~$I_2 \subseteq I_1$ we can assume that, after relabelling the edges accordingly,~$\pi(X^{f_i}) = \pi(X^{f_j})$ and~$\mu_i(f_i^*) = \mu_j(f_j^*)$ again using the pigeonhole principle.
\end{enumerate}
By \cref{def:carving} of carving we have~$\epsilon(\lenks{f_i}) \cup \epsilon(\riets{f_i}) = \rho(f_i^*) \cup (\epsilon(\lenks{f_i}) \cap \epsilon(\riets{f_i})) $  where~$\rho(f_i^*) \cap (\epsilon(\lenks{f_i}) \cap \epsilon(\riets{f_i})) = \emptyset$. Since~$I_2 \subseteq I_1$ we may refine the sequence as follows.
\begin{enumerate}
    \item[$I_3$] By switching to an infinite subsequence indexed by~$I_3 \subseteq I_2$ we can assume that, after relabelling the edges of~$\epsilon(\lenks{f_i}) \cap \epsilon(\riets{f_i}))$ and $\epsilon(\lenks{f_j}) \cap \epsilon(\riets{f_j}))$ accordingly,~$\pi(X^{\lenks{f_i}}) = \pi(X^{\lenks{f_j}})$ and~$\pi(X^{\riets{f_i}}) = \pi(X^{\riets{f_j}})$ maintaining~$\pi(X^{f_i}) = \pi(X^{f_j})$.
\end{enumerate} 
Finally let~$1 \leq i<j$ with~$i,j \in I_3$ and let~$\gamma_l\colon \stitch(\GGG_i,\pi(X^{\lenks{f_i}})) \hookrightarrow \stitch(\GGG_j,\pi(X^{\lenks{f_j}}))$ and~$\gamma_r: \stitch(\GGG_i,\pi(X^{\riets{f_i}})) \hookrightarrow \stitch(\GGG_j,\pi(X^{\riets{f_j}}))$ be the respective strong~$\Omega$-knitwork immersions. It is now straightforward to ``knit'' the strong immersions to obtain a strong~$\Omega$-knitwork immersion~$\gamma\colon \stitch(\GGG_i,\pi(X^{f_i})) \hookrightarrow \stitch(\GGG_j,\pi(X^{f_j}))$ similarly to the proof of \cref{thm:knitting_knitwork_immersion}. Note here that for~$\GGG_j^{X_{f_j}} = \stitch(\GGG_j,\pi(X^{f_j}))$ we have~$V(G_j^{X_{f_j}}) = X^{\lenks{f_j}} \cup X^{\riets{f_j}} \cup \{f_j^*\}$. Now for any edge~$ e \in \epsilon(\lenks{f_i}) \cap \epsilon(\riets{f_i})$ with endpoints~$u = \tail(e) \in X^{\lenks{f_i}}$ and~$v = \head(e) \in X^{\riets{f_i}}$ say, we can concatenate the paths~$\gamma_l(e) \circ \gamma_r(e)$ to define~$\gamma(e)$ in~$\GGG_j^{X_{f_j}})$, for they are both paths in~$G_j^{X_{f_j}}$ by construction, which share exactly one common edge in~$\epsilon(\lenks{f_j}) \cap \epsilon(\riets{f_j})$ and respect~$\m_j^{X_{f_j}}$ by construction (since~$M_\gamma(v) = M_{\gamma_l}(v) \cup M_{\gamma_r}(v)$ for all~$v \in V(G_j^{X_{f_j}}) \setminus \{f_j^*\}$). For the edges of $\rho(f_i^*)$ it is clear and for all other edges there is nothing to show.%
This is a contradiction to $(f_i)_{i\ in \N}$ being an antichain, concludes the proof.
\end{proof}

We are ready to derive \cref{thm:wqo_bounded_carvingwidth_knitworks} from \cref{thm:wqo_bounded_carvingwidth_knitworks_red}.

\begin{proof}[of \cref{thm:wqo_bounded_carvingwidth_knitworks}]

Towards a contradiction, assume the theorem to be false and let~$(\GGG_i)_{i \in \N}$ be a counterexample to the claim with~$\GGG_i = \big((G_i,\pi_{G_i}(X_i)),\mu_i,\m_i,\Phi_i)$. By assumption~$\ebw{\GGG_i} \leq k$ whence~$\pi_{G_i}(X_i) \leq k$ for every $i \in \N$.

By moving to a suitable subsequence we may assume that~$\Abs{X_i} =t \leq k$ for every~$i \in \N$. Let~$\GGG^{X_i},\GGG_{X_i}$ be the respective up- and down-stitches of~$\GGG_i$ with respective up- and down-stitches~$x^i$ and~$x_i$. By \cref{obs:up-stitch_of_well-linked_is_well-linked} the up-stitches~$\GGG^{X_i}$ are again well-linked and $x^i \notin \dom(\Phi_i^{X_i}) \cup \dom(\m_i^{X_i})$. Thus, by \cref{thm:wqo_bounded_carvingwidth_knitworks_red} together with \cref{obs:wqo_yields_infinite_chain} there exists~$I \subseteq \N$ such that~$(\GGG_{X_i})_{i \in \N}$ is a chain with respect to strong~$\Omega$-knitwork immersion. 

\smallskip

Let~$\GGG_{X_i}=((G_{X_i},\pi_{G_{X_i}}(X_i)),\mu_{X_i},\m_{X_i},\Phi_{X_i})$ for every~$i \in I$, then~$\Abs{V(G_{X_i})} \leq t+1$ by \cref{def:stitching_std}. After removing all the loop-edges there exist only finitely many non-isomorphic Eulerian digraphs of carving width~$k$ (note that the maximum degree after deleting loops is~$k$). Thus there is~$I' \subset I$ such that all the graphs $G_i,G_j$ with $i,j \in I'$ are isomorphic (up-to loops); fix~$V(G_i) = \{u_1,\ldots,u_t,u{t+1}\}$ for every~$i \in I'$ for simplicity by renaming the graphs accordingly. Since~$\Abs{\rho(u_i)} \leq k$ for every~$1 \leq i \leq t+1$, we may further assume that~$\mu_i(u_\ell) = \mu_j(u_\ell)$ and~$\m_i(u_\ell) = \m_j(u_\ell)$ for all~$i,j \in I'$ and all~$1\leq \ell \leq t$. 
Since~$(\Phi_i(u_\ell))_{i \in I'}$ is a well-quasi-order (with respect to $\ll$) we may iteratively apply \cref{obs:wqo_yields_infinite_chain}~$t+1$ times yielding an infinite~$I'' \subseteq I'$ such that~$(\Phi_i(u_\ell))_{i \in I''}$ is a chain for every~$1 \leq \ell \leq t$. 

Lastly, for each vertex~$u_\ell$ denote by~$\ell_i(u_\ell) \in \N$ the number of loops at~$u_\ell$. Since~$(\N,<)$ is a well-quasi-order, the same argument as above applies and there is~$J\subseteq I''$ such that~$(\ell_i(u_\ell))_{i\in J}$ is a chain with respect to~$<$ in~$\N$ (see the proof of claim \ref{claim:thm_wqo_bounded_carvingwidth_knitworks_red_leafs} for more details). 

By construction~$(\GGG_{X_i})_{i \in J}$ is a chain with respect to strong~$\Omega$-knitwork immersion. Since~$(\GGG^{X_i})_{i\in J}$ is also a chain by strong~$\Omega$-knitwork immersion we deduce that~$\GGG_i \hookrightarrow \GGG_j$ for every~$i<j$ with~$i,j \in J$ using \cref{cor:immersion_of_stitches_yields_immersion}. This concludes the proof.
\end{proof}

Note that the \emph{well-linkedness} of the~$\Omega$-knitworks is crucial for a proof of \cref{thm:wqo_bounded_carvingwidth_knitworks}. Without the assumption, ``linkedness'' in \cref{cor:linked_ebw} can not directly be used to ``lift''~$\Omega$-knitwork immersions from children to parents through gluing at cuts. This is due to the restrictions imposed by the feasible linkages encoded by $\m$. To see this let~$Y$ be a rooted cut in some~$\Omega$-knitwork~$\GGG=((G,\pi(X)), \mu ,\m,\Phi)$ and let~$\LLL$ be a~$\{\rho(X),\rho(Y)\}$ in~$G$. One would hope to conclude~$\stitch(\GGG,\pi(\bar Y)) \hookrightarrow \GGG$ for a suited order~$\pi(\bar Y)$. But it may happen that two edge-disjoint paths~$L_1,L_2 \in \LLL$ cross in a common vertex~$v \in \dom(\m)$ whence we may not use these paths to define the natural immersion if~$\{\tau(L_1),\tau(L_2)\} \notin \m(v)$ (this would violate (4) of \cref{def:knitwork_immersion}). Moreover, given arbitrary~$\Omega$-knitworks it could happen that there exists no feasible linkage that can be used to lift it. This is in stark contrast with the undirected case using standard minors as studied by Robertson and Seymour in \cite{GMIV}. Further, it is also highly different to the undirected case for immersion, since in the undirected setting edge-disjoint paths may cross freely in vertices which would rather correspond to~$\m(v) = \operatorname{Match}(\rho^-(v),\rho^+(v))$, i.e., there are no further restrictions on the linkages, which is a much easier case. The following is an easy ``counterexample'' to the version of \cref{thm:wqo_bounded_carvingwidth_knitworks_red} where we do \emph{not} impose well-linkedness.

\begin{example}
    Let $V_i = \{v_0,\ldots,v_i\}$ for every $i \in \N$. Then let $G_i$ be the graph on $V_i$ where we add back and forth edges between subsequent vertices $v_j,v_{j+1}$ for $1 \leq j < i$. Finally let $\Omega=(\{0,1\},\emptyset)$ be the well-quasi-order on two labels that are not comparable. Let $\Phi_i(v_0) = 0$ and $\Phi_i(v_i) = 1$. Let $\mu_i$ be arbitrarily defined for every vertex of $V(G_i)$ and set $\mu_i(v_j)$ to be the two-paths that start and end in $v_{j-1}$ and start and end in $v_{j+1}$ respectively; in particular we cannot pass through vertices with immersions, the links are not well-linked. One easily verifies that $(\GGG_i)_{i\geq 3}$ is an antichain. While this example may be slightly trivial as we basically encode our graphs to be disconnected (or behave like alternating paths), one easily extends this idea to more interesting examples (even on higher degree).
\end{example}

\section{The Router Case}

\label{sec:router}

To highlight the importance of routers---especially for our future work---we prove that a bad sequence defying \cref{conj:wqo_4reg} cannot contain arbitrarily large routers as (strong) immersion.

\begin{theorem}
    Let $G$ be an Eulerian digraph of maximum degree $4$, and let $n \coloneqq \Abs{V(G)}$. Then $G \hookrightarrow \RRR_{6n}$ by strong immersion. 
\end{theorem}
\begin{proof}
Let $V(G) \coloneqq \{ v_1, \dots, v_n \}$ and let $\CCC = \{C^i_j \mid 1 \leq i \leq n \text{ and }1 \leq j \leq 6\}$ be a circle cover of~$\RRR_{6n}$ as given by \cref{def:router}. 

    For all $1 \leq i \leq n$ we proceed as follows. 
        Let $v_{s_1}, \dots, v_{s_4}$ be the neighbours of $v_i$ such that $(v_i, v_{s_1})$, $(v_{s_2}, v_i)$, $(v_i, v_{s_3})$, $(v_{s_4}, v_i) \in E(G)$.

    See \cref{fig:turtle} for an illustration of the following construction.     \begin{figure}
       \centering
         \includegraphics[height=6cm]{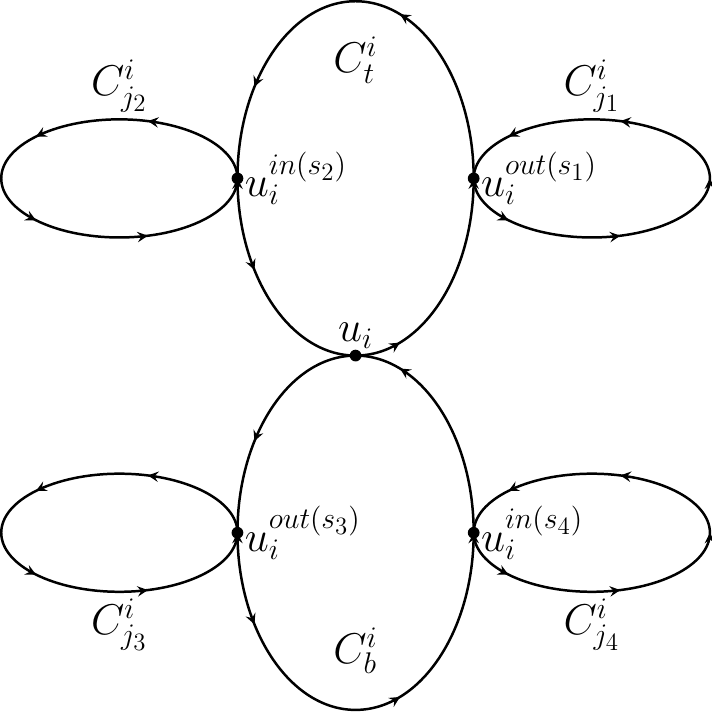}
       \caption{The gadget for immersing Eulerian digraphs into a router.}
        \label{fig:turtle}
    \end{figure}

    Let  $C^i_tt \coloneqq C^i_5$ and $C^i_b \coloneqq C^i_6$. Let $u_i$ be the unique vertex in $C_i^t \cap C_i^b$. Let $j_1 \in \{1, \dots, 4\}$  be such that the first vertex on $C_i^t$ after $u_i$ that $C_i^t$ has in common with $\{C^i_1, \dots, C^i_4\}$ is on $C^i_{j_1}$ and the second such vertex after $u_i$ is on $C^i_{j_2}$ where $j_2 \in  \{ 1, \dots, 4\}\setminus \{j_1\}$.
    Similarly, let $j_3 \in \{1, \dots, 4\} \setminus \{ j_1, j_2\}$ 
    be such that the first vertex on $C_i^b$ after $u_i$ that 
    $C_i^b$ has in common with $\{C^i_1, \dots, C^i_4\} \setminus 
    \{ C^i_{j_1}, C^i_{j_2}\}$ is on $C^i_{j_3}$. 
    Finally, let $j_4$ be the remaining index in $\{ 1, \dots, 4 \} \setminus \{ j_1, j_2, j_3\}$.
    We define $C^i_{out(s_1)} := C^i_{j_1}$, $C^i_{in(s_2)} := C^i_{j_2}$, 
    $C^i_{out(s_3)} := C^i_{s_3}$, and $C^i_{in(s_4)} := C^i_{j_4}$.
    Let $u_i^{out(s_1)}$ be the unique vertex in $C^i_{j_1} \cap C^i_t$ and let $u_i^{in(s_2)}, u_i^{out(s_3)}, u_i^{in(s_4)}$ be defined analogously. 

    By construction, $\bigcup \{ C^i_j \sth 1 \leq j \leq 6\}$ contains pairwise edge-disjoint paths 
    $P_i^{out(s_1)}, P_i^{in(s_2)}, P_i^{out(s_3)}$, and $P_i^{in(s_4)}$ such that $P_i^{out(s_1)}$ links $u_i$ to $u_i^{out(s_1)}$, $P_i^{in(s_2)}$ links 
    $u_i^{in(s_2)}$ to $u_i$, $P_i^{out(s_3)}$ links $u_i$ to $u_i^{in, s_3}$, and $P_i^{in(s_4)}$ links $u_i^{in(s_4)}$ to $u_i$.

    We are now ready to define the strong immersion $\gamma$ of $G$ into $\RRR_{6n}$. For every~$1\leq i \leq n$ we define $\gamma(v_i)\coloneqq u_i$. Now consider an edge $e_{i,j} = (v_i, v_j) \in E(G)$. 
    By definition of a router, there is a path $P_{i,j}$ from $u_i^{out(j)}$ to $u_j^{in(i)}$ in $C^i_{out(j)} \cup C^j_{in(i)}$. Let $E_{(i,j)} \coloneqq P_{i,j} \cup P_i^{out(j)} \cup P_j^{in(i)}$. Then, for $(v_i,v_j) \neq (v_{i'}, v_{j'}) \in E(G)$ the paths $E_{(i,j)}$ and $E_{(i',j')}$ are edge disjoint. 
    We define $\gamma(e_{i,j}) \coloneqq E_{(i,j)}$, for all $(v_i,v_j) \in E(G)$. Then $\gamma$ is a strong immersion of $G$ into $\RRR_{6n}$. %
\end{proof}

Obviously no Eulerian digraph containing a vertex of degree more than $4$ can (strongly) immerse into a router of any size. 

\begin{remark}
    A router of order $k$ with capacity $c$ and multiplicity $s$ is defined as a router with the difference that there are $s$ vertices which are the intersection of $c$ cycles and such that the cycles intersecting at such vertices are distinct. 
    Then we would be able to strongly immerse any Eulerian digraph with maximum degree $d$ and $n$ vertices into a router with capacity $d$ and multiplicity $n$ of order $2dn$. The extension of the proof is easy and not needed until later in future work we are pursuing. Thus we omit the details here.
\end{remark}

\section{Embeddings of rooted Eulerian digraphs}
Following the definition of rooted Eulerian digraphs in \cref{subsec:knitworks}, we devote this section to present the relevant definitions and tools to lift the definitions of Eulerian embeddings appropriately. The definitions and results established in this section will be extensively used throughout the rest of the paper
\smallskip

The main theorem of the remaining paper reads as follows.

\begin{theorem}\label{thm:wqo_on_surfaces}
  Let~$\Sigma$ be a fixed surface without boundary and let~$\mathbf{G}(\Sigma)$ be the class of all Eulerian digraphs of maximum degree four that are Euler-embeddable in~$\Sigma$. Then~$\mathbf{G}(\Sigma)$ is well-quasi-ordered by strong immersion.
\end{theorem}

The proof technique is via induction on the genus of~$\Sigma$ as well as by induction on the ``generality'' of the structure of the graphs in $\mathbf{G}(\Sigma)$. As usual, we will need to make a stronger induction hypothesis and thus end up proving a stronger claim. We start with the definitions needed to state (and prove) \cref{thm:wqo_on_surfaces} in its final form. In particular, we will need to loosen the condition on the graphs being of maximum degree four for inductive purposes.

Given a set $W$, a \emph{partition} $(W_1,\ldots,W_t)$ of $W$ consists of (possibly empty) sets $W_1,\ldots,W_t \subseteq W$ for some $t \geq 1$ such that $W_i \cap W_j=\emptyset$ for all distinct $1 \leq i,j \leq t$ and such that $W = \bigcup_{i=1}^t W_i$.

\begin{definition}[Bead-rooted Graphs] \label{def:bead-rooted_graph}
Let~$G$ be an Eulerian digraph and let~$W \subseteq V(G)$ be a set containing all vertices in~$G$ of degree at least~$6$. Then we call~$W$ a \emph{set of beads of~$G$}.\Index{beads} Let $\WWW \coloneqq (W_1,\ldots,W_t)$ be a partition of $W$ such that no two vertices in $W_i$ are adjacent for every $1\leq i \leq t$. Fix an ordering~$\pi(W_i)$\Symbol{PIW@$\pi(W_\omega)$} of~$\rho(W_i)$ for every $1 \leq i \leq t$. Let~$\pi(\WWW)\coloneqq(\pi(W_1),\ldots,\pi(W_t))$. We call~$(G,\pi(\WWW))$ \emph{bead-rooted}. \Index{bead-rooted}
\end{definition}
\begin{remark}
    Note that technically speaking a bead-rooted graph is not a rooted graph as given by \cref{def:rooted_graph}.
\end{remark}

We extend the \cref{def:rooted_immersion} of strong rooted immersion to bead-rooted Eulerian digraphs in the obvious way.

As mentioned above, we will prove our main theorem by induction on the genus of the surface $\Sigma$. In the reduction to surfaces of lower genus, however, we may cut the surface along curves producing surfaces of lower genus but additionally with cuffs, for example when cutting along curves that separate the surface. The curves we cut at will intersect the embedding only at edges and thus cutting along the curve will result in cutting the respective edges into two edges, each connected to one of its original endpoints. The edges will then be embedded on the cuffs and their respective original endpoint will be drawn in the interior of the surface. As a consequence, the graph represented by the resulting embedding would no longer correspond to an Eulerian digraph, but to a digraph with ``half-edges''. We therefore introduce ``virtual nodes'', the \emph{beads}, which provide the missing half of the edges that we cut\footnote{It would be totally fine and in accordance with \cite{EDP_Euler} to work with half-edges; we deliberately choose not to do so here, since, as of the date of writing, we still need beads and several results presented here for future work.}. We keep track of these beads by drawing them on the discs of $\hat\Sigma$ bounded by a cuff. This is only one choice to represent the above, but unfortunately not the only one we will need in the paper. The beads will help us to keep track of the cuts we made keeping the resulting graphs Eulerian and thus simplifying some proofs, but come with a different flaw that we will need to smooth out in later sections.
\smallskip

\begin{definition}[Embeddings fixing beads]\label{def:Euler_embedding_upto_beads}
  Let~$\Sigma$ be a surface and let $c(\Sigma) = \{\zeta_1,\ldots,\zeta_t\}$ be its set of cuffs for some~$k \geq 1$. Let~$\Delta_1,\ldots,\Delta_t$ be the unique closed discs bounded by the respective cuffs in~$\hat{\Sigma} \setminus \Sigma$ and denote their topological interior by~$\Delta_1^\circ,\ldots,\Delta_t^\circ$ respectively. Let~$G$ be an Eulerian digraph with beads $W$ and partition $\WWW = (W_1,\ldots,W_t)$ with an ordering $\pi(\WWW)$. Let~$(\hat \Gamma,\hat \nu)$ be an embedding of~$G$ in~$\hat{\Sigma}$ satisfying the following:

  \begin{itemize}
      \item[(i)] every~$v \notin W$ is Eulerian embedded in~$\Sigma$,
      \item[(ii)] it holds~$W_i = \hat\nu^{-1}(\Delta_i^\circ)$,
      \item[(iii)] for every~$1 \leq i \leq t$ and every~$e \in \rho(W_i)$, it holds~$\hat\nu(e) \in \zeta_i$. Further, for~$x \in \hat{\nu}^{-1}(\zeta_i)$ it holds~$x \in \rho(W_i)$. 
    \end{itemize}
      
  We define a map~$\omega\colon W\to \{1,\ldots,t\}$ via~$\omega(w) = i$ if and only if~$w \in W_i$ and we let~$W_\omega \coloneqq \WWW$ be the \emph{$\omega$-partition (of $W$)}. Finally let~$\Gamma \coloneqq \hat{\Gamma} \cap \Sigma$ and~$\nu \coloneqq \restr{\hat \nu}{\Gamma}$. Then we call~$(\Gamma,\nu,\omega)$\Symbol{GAMMA@$(\Gamma,\nu,\omega)$} an \emph{Eulerian embedding of~$(G,\pi(W_\omega))$ in~$\Sigma$}.\Index{Eulerian embedding of bead-rooted digraphs}%

\end{definition}
\begin{remark}
    Note that while the beads are not drawn on~$\Sigma$, every edge of the graph is drawn on~$\Sigma$, and the only elements drawn on cuffs are exactly the edges~$\bigcup_{i=1}^t\rho(W_i)$. Further, if~$\Sigma$ is a disc, we may omit~$\omega$ for it is independent of the embedding.
\end{remark}
The way \cref{def:Euler_embedding_upto_beads} is phrased, given a bead-rooted Eulerian digraph $(G,\pi(\WWW))$ together with an Eulerian embedding $(\Gamma,\nu,\omega)$ in $\Sigma$, the partition $\WWW$ is implicitly defined via $\omega$. 

\begin{observation}\label{obs:bead-rooted_edges_on_cuffs}
    Let $\Sigma$ be a surface with cuffs $c(\Sigma) = \{\zeta_1,\ldots,\zeta_t\}$ for some $t \geq 1$. Let $(G,\pi(\WWW))$ be a bead-rooted Eulerian digraph with Eulerian embedding $(\Gamma,\nu,\omega)$ in $\Sigma$. Let $W_\omega = (W_1,\ldots,W_t)$. Then $\rho(W_i) = \nu^{-1}(\zeta_i)$ and in particular $\rho(W) \subseteq \nu^{-1}(\bd(\Sigma)) \cap E(G)$.
\end{observation}

For notational convenience---highlighting the importance of beads---and since we will only work with Eulerian embeddings of bead-rooted graphs we will thus write bead-rooted Eulerian digraphs as $(G,\pi(W_\omega))$ since given an Eulerian embedding $(\Gamma,\nu,\omega)$, the partition $W_\omega$ is clear from context.

We fix the following.
\begin{definition}
    Let $(\Gamma,\nu,\omega)$ be an embedding of a bead-rooted graph $(G,\pi(W_\omega))$ in a surface $\Sigma$. Let $e \in \nu^{-1}(c(\Sigma))$ be a boundary edge. Then $e$ is called an \emph{in-edge} if $\nu(\head(e)) \in \Sigma$ and an \emph{out-edge} if $\nu(\tail(e)) \in \Sigma$. %
\end{definition}
Note that $e$ is an in-edge if, and only if, $e \in \rho^+(W_i)$ and an out-edge if and only if $e \in \rho^-(W_i)$ for some $1 \leq i \leq t$ given the partition $W_\omega=(W_1,\ldots,W_t)$. 

We will need a slight generalisation of stitching for bead-rooted Eulerian digraphs, where we allow to stitch at general cuts even if they only cut off part of the root-edges; since it is a true generalisation we will stick to the old the name. 

\begin{definition}[Stitching]\label{def:stitching_beadrooted}
      Let~$\bar{G}=(G,\pi(W_\omega))$ be a bead-rooted Eulerian digraph where~$G=(V,E,\operatorname{inc})$ and with partition $W_\omega=(W_1,\ldots,W_t)$ of its beads $W$ for some $t \in \N$. Let $I \cup J = \{1,\ldots, t\}$ be a partition of $\{1,\ldots, t\}$ and let $I = \{i_1,\ldots,i_p\}$ and $J = \{j_1,\ldots,j_q\}$. Let~$\pi(W_\omega) \coloneqq (\pi(W_1),\ldots,\pi(W_t)$). Let $k\in 2\N$ and~$Y \subset V(G)$ induce a~$k$-cut in~$G$ such that $W_i \subseteq Y$ and $W_j \cap Y = \emptyset$ for every $i \in I$ and $j\in J$. Let~$ \pi(Y) = (e_1,\ldots,e_k) = \pi(\bar{Y}) $ be an ordering of~$\rho(Y)$ and~$\rho(\bar Y)$. Let~$y_\ast,y^\ast$ be two new elements that are not part of~$V(G) \cup E(G)$. We define~$G_Y \coloneqq (V_Y,E_Y,\operatorname{inc}_Y)$ and~$\pi_Y$ via~
    \begin{align*}
        &V_Y \coloneqq Y \cup \{y_\ast\}\\
        & E_Y \coloneqq E(G[Y]) \cup \{e_1,\ldots,e_k\},\\
        &(e,v) \in \operatorname{inc}_Y :\iff \begin{cases} 
  (e,v) \in \operatorname{inc} \text{ and } e \in E_Y, v\in Y,\\
   e \in \rho^+(Y) \text{ and } v=y_\ast\\
        \end{cases}
        &(v,e) \in \operatorname{inc}_Y :\iff \begin{cases} 
  (v,e) \in \operatorname{inc} \text{ and } e \in E_Y, v\in Y,\\
  e \in \rho^-(Y) \text{ and } v=y_\ast\\
        \end{cases}\\
        & W_Y \coloneqq W_{i_1} \cup \ldots \cup W_{i_p} \cup \{y_\ast\},\\
        & \pi_Y = (\pi(W_{i_1}) , \ldots , \pi(W_{i_p}), \pi(Y))\\     
    \end{align*} 
    We define~$\stitch(\bar{G};\pi(Y)) \coloneqq (G_Y,\pi_Y)$ and say that~\emph{$G_Y$ is obtained from~$G$ by down-stitching~$Y$} with beads $W_Y$ and call~$y_\ast$ the \emph{down-stitch vertex resulting from~$Y$}. 

    Similarly we define~$G^Y \coloneqq (V^Y,E^Y,\operatorname{inc}^Y)$ and~$\pi^Y$ via~
    \begin{align*}
        &V^Y \coloneqq \bar{Y} \cup \{y^\ast\}\\
        & E^Y \coloneqq E(G[\bar{Y}]) \cup \{e_1,\ldots,e_k\}, \\
         &(e,v) \in \operatorname{inc}^Y :\iff \begin{cases} 
  (e,v) \in \operatorname{inc}, \text{ and } e \in E^Y, v\in \bar{Y},\\
   e \in \rho^+(\bar{Y}) \text{ and } v=y^\ast\\
        \end{cases}
        &(v,e) \in \operatorname{inc}^Y :\iff \begin{cases} 
  (v,e) \in \operatorname{inc} \text{ and }  e \in E^Y, v\in \bar{Y},\\
   e \in \rho^-(\bar{Y}) \text{ and } v=y^\ast \\
        \end{cases}\\
    & W^Y \coloneqq W_{j_1} \cup \ldots, \cup W_{j_q} \cup \{y^\ast\},\\
    & \pi^Y = (\pi(W_{j_1}) , \ldots , \pi(W_{j_q}), \pi(\bar Y))
    \end{align*} 
    We define~$\stitch(\bar{G};\pi(\bar{Y})) \coloneqq (G^Y,\pi^Y)$ and say that~\emph{$G^Y$ is obtained from~$G$ by up-stitching~$Y$} with beads $W^Y$ and call~$y^\ast$ the \emph{up-stitch vertex resulting from~$Y$}. 
\end{definition}

One easily verifies that for bead-rooted digraphs, the graph resulting from stitching is by definition again bead-rooted with analogous properties to the ones listed in \cref{obs:stitching_fundamentals}; we omit the details as we will not need them. 

\begin{definition}[$\mathbf{G}(\Sigma,k)$]
    Let~$\Sigma$ be a surface and~$k \in 2\N$. We define~$\mathbf{G}(\Sigma,k)$ to be the set of bead-rooted Eulerian digraphs~$(G,\pi(W_\omega))$ such that~$\Abs{\rho({W})},\Abs{W} \leq k$, and~$(G,\pi(W_\omega))$ Euler-embeds in~$\Sigma$. %
\end{definition}

With these definitions at hand, the following is the main theorem of the remainder of the paper.

\begin{theorem}\label{thm:wqo_on_surfaces_upto_beads}
    Let $\Sigma$ be a surface and $k \in 2\N$. The class~$\mathbf{G}(\Sigma,k)$ is well-quasi-ordered by strong immersion.
\end{theorem}

\begin{remark}
    Clearly~$\mathbf{G}(\Sigma,k) \subseteq \mathbf{G}(\Sigma,k+2)$ for every $k \in 2\N$.
\end{remark}

Notably \cref{thm:wqo_on_surfaces_upto_beads} implies \cref{thm:wqo_on_surfaces}.

\subsection{Properties of Eulerian Embeddings} \label{subsec:properties_of_eulerian_embeddings}
Prior to gathering relevant definitions and results regarding Eulerian embeddings of (bead-rooted) graphs, we discuss the notions of cut-cycles and~$O$-arcs which are defined for general embeddings. Recall the \cref{def:embedding} of embeddings.

\begin{definition}\label{def:internal_tracing_clean}
     Let~$G$ be a digraph and~$(U,\nu)$ an embedding of~$G$ in some surface~$\Sigma$. Let~$X,F,Y \subset \Sigma$ be topologically closed subsets. Then~$X$ is \emph{$U$-normal} if and only if~$X \cap U \subseteq \nu(V(G))$ and~$F$ is \emph{$U$-tracing} if and only if~$F \cap U\subseteq \nu(E(G))$. We call~$Y$ \emph{internal} if~$Y \cap \bd(\Sigma) = \emptyset.$ We call $Y$ \emph{clean} if $Y \cap \bd(\Sigma) \notin \nu(E(G) \cup V(G))$.
\end{definition}
\begin{remark}
    For simplicity, when given a curve $\gamma\colon[0,1] \to \Sigma$ we may call $\gamma$ $U$-tracing, internal or clean whenever $Y \coloneqq \gamma([0,1])$ is.
\end{remark}

\begin{definition}[Cut-cycles and $O$-arcs]\label{def:cut-cycle_Oarc}
    Let~$G$ be a digraph and let $(U,\nu)$ be an embedding of~$G$ in some surface~$\Sigma$. Let~$F \subset \Sigma$ be~$U$-tracing such that there exists a simple curve~$\gamma:[0,1] \to \Sigma$ with~$\gamma([0,1]) = F$. We call~$F$ a \emph{traced cut}\Index{traced cut} in~$U$ and if $\gamma$ can be chosen to be a simple closed curve then we call $F$ an $O$-trace. %
    
    If $F$ is an $O$-trace and in addition $\gamma$ can be chosen to be  null-homotopic in $\hat \Sigma$---whence~$F$ bounds a disc~$\Delta(F)$ in $\hat{\Sigma}$---where for $\gamma(0) = \gamma(1)$ we have $\gamma(0) \cap U = \emptyset$ and for every edge~$e \in \nu^{-1}(F)$ either its head or its tail but not both are embedded in~$\Delta(F)$, then we call it a \emph{cut-cycle (in~$U$)}.\Index{cut cycle} We define~\Symbol{RHO@$\rho(F)$}$\rho(F) \coloneqq \nu^{-1}(F)$ and~$\delta(F) = \Abs{\rho(F)}$ to be its \emph{order} or \emph{length}.\Symbol{DELTA@$\delta(F)$}

    Similarly, if~$F$ is~$U$-normal such that there exists a continuous map~$\gamma:[0,1] \to \Sigma$ with~$\gamma([0,1]) = F$ and with $\gamma$ being a simple closed curve, then we call~$F$ an \emph{$O$-arc}. We define~$\alpha(F)\coloneqq \nu^{-1}(F)$.\Symbol{ALPHA@$\alpha(F)$}
\end{definition}
\begin{remark}
    Note that by our \cref{def:embedding} of embeddings, if~$\gamma$ is a simple closed curve then there are no two distinct points~$t_1,t_2 \in F$ with~$\nu^{-1}(t_1) = \nu^{-1}(t_2)$. Informally speaking, $F$ does not pass twice through the same vertex or edge.
\end{remark}

Note that if~$F$ is a cut-cycle (or~$O$-arc) and~$\Sigma$ is the disc, then~$F$ bounds a unique closed disc~$\Delta(F) \subset \Sigma$ with~$F = \bd(\Delta(F))$. Further~$\delta(F)$ and~$\alpha(F)$ only depend on the underlying undirected graph.

\begin{definition}[$X(F)$ and $I(F)$]\label{def:X(F)_I(F)}
    Let~$F$ be a cut-cycle in some embedding $(U,\nu)$ of a digraph~$G$ in a disc~$\Delta$. 
    The \emph{cut $X(F)$ induced by~$F$} is defined as the set $X(F) \coloneqq \nu^{-1}(\Delta(F)) \cap V(G)$.\Symbol{XF@$X(F)$}

    Similarly, if~$F$ is an $O$-arc in~$\Delta$ we define~$I(F) \coloneqq \nu^{-1}(\Delta(F)) \cap E(G)$.\Symbol{XF@$X(F)$}

    We call two cut-cycles~$F,F'$ in~$U$ \emph{equivalent} if~$F \cap U = F' \cap U$ and~$X(F) = X(F')$. Similarly we call two~$O$-arcs~$F,F'$ \emph{equivalent} if~$F \cap U = F' \cap U$ and~$I(F) = I(F')$.
\end{definition}
Unfortunately, in general, traced cuts and $O$-traces do not correspond to induced cuts in digraphs, but for cut-cycles it is true by definition.

\begin{observation}\label{obs:cut-cycle_induce_cut}
   Let~$F$ be a cut-cycle in some embedding $(U,\nu)$ of a digraph~$G$ in a surface~$\Sigma$. Then~$\rho(F) = \rho(X(F))$.
\end{observation}

\begin{definition}
   Let~$\Sigma$ be some surface and let $G$ be a digraph with embedding~$\Pi \coloneqq (U,\nu)$ in~$\Sigma$. Let~$H \subseteq G$ with incidence relation $\operatorname{inc}_H$. We define~$\restr{\Pi}{H} \coloneqq (U(H),\restr{\nu}{H})$ as the  (unique) \emph{embedding of~$H$ induced by~$(U,\nu)$}, i.e.,~$U(H) \subseteq U$ is given by~$\nu(V(H) \cup E(H))$ together with the components of~$U\setminus \nu(V(G) \cup E(G))$ that have their ends in~$\nu(V(H) \cup E(H))$, i.e., $\nu((v,e))$ and $\nu((e,u))$ for $(v,e),(e,u) \in \operatorname{inc}_H$. 

   Similarly let~$\Delta \subseteq \Sigma$ be a disc and such that~$\bd(\Delta) \neq \emptyset$ and~$\bd(\Delta)$ is~$U$-normal. Then we call~$\Delta$ \emph{$U$-normal} and define~$G[\Delta; U] \coloneqq (V,E,\operatorname{inc})$ where~$V \coloneqq \nu^{-1}\big(\Delta \cap \nu(V(G))\big)$,~$E\coloneqq \nu^{-1}\big(\Delta \cap \nu(E(G))\big)$ and~$\operatorname{inc}$ is defined in the obvious way depending on whether the respective component of~$U\setminus \nu(V(G) \cup E(G))$ is contained in~$\Delta$.

   If $\bd(\Delta)$ is $U$-tracing we call $\Delta$ $U$-tracing. 
\end{definition}
\begin{remark}
    Using the shorthand notation as discussed in \cref{def:embedding} and the following remark, we have $U(H) = \nu(H)$ but to emphasize the distinction between $\nu$ as a map and the set $U(H)$ we will write $U(H)$ whenever we want to talk about the induced embedding $\restr{\Pi}{H}$. Given the above we may talk about the induced embeddings of cycles and paths in~$G$ or talk about paths and cycles in~$U$, meaning the graphs they induce in~$G$.
\end{remark}
One easily verifies that, if~$\Delta$ is~$U$-normal, then $G[\Delta;U]$ is a subdigraph of~$G$, i.e., every edge~$e \in E$ has exactly one head and one tail. In particular we have the following.
\begin{observation}\label{obs:graphs_inside_Oarcs}
    Let~$(U,\nu)$ be an embedding of some digraph~$G$ in a surface~$\Sigma$ and let~$F$ be an~$O$-arc in~$\Sigma$ bounding a disc~$\Delta(F) \subset \Sigma$. Then~$G[\Delta(F);U] \subseteq G$ is a subgraph of~$G$ with respective induced drawing.
\end{observation}

While~$U$-tracing subsets do not necessarily induce subgraphs, there is a natural way to define restrictions of paths to them.
In light of the above we define the following.
\begin{definition}\label{def:induced_paths_in_tracing_sets}
    Let~$(U,\nu)$ be an embedding of a digraph~$G$ in a surface~$\Sigma$ and let~$X \subset \Sigma$ be~$U$-tracing. Let~$\LLL$ be an edge-disjoint linkage in~$G$ and for each path~$P \in \LLL$ let~$\PPP$ be a collection of maximal subpaths~$P'$ of~$P$ satisfying~$\nu(E(P')) \in X$. We define~$\restr{\LLL}{X}$ to be the union over all sets~$\PPP$ for~$P \in \LLL$.
\end{definition}

The following follows immediately from the \cref{def:induced_paths_in_tracing_sets}.
\begin{observation}
     Let~$(U,\nu)$ be an embedding of a digraph~$G$ in a surface~$\Sigma$ and let~$X \subset \Sigma$ be~$U$-tracing. Let~$\LLL$ be an edge-disjoint linkage in~$G$ then~$\restr{\LLL}{X}$ is an edge-disjoint linkage in~$G$ such that~$\nu(V(P)) \subset X$ for every~$P \in \restr{\LLL}{X}$.
\end{observation}
\begin{remark}
    Note that $V(P) = \emptyset$ is possible.
\end{remark}

In particular we derive the following useful observation for cylinders. Recall that we write ${\Delta}^\circ$ for the topological interior of a set $\Delta$.

\begin{observation}\label{obs:linkages_reduced_to_cylinder}
     Let~$(U,\nu)$ be an embedding of a digraph~$G$ in a surface~$\Sigma$ and let~$F,F'$ be cut-cycles in~$(U,\nu)$ bounding discs~$\Delta(F'), \Delta(F) \subset \Sigma$ such that~$\Delta(F') \subset \Delta(F)$ is a strict subset (note that in general~$\bd(\Delta(F)) \cap \bd(\Delta(F')) \neq \emptyset$). Let~$\LLL$ be an edge-disjoint linkage in~$G$ such that for every~$P \in \LLL$ its endpoints are contained in~$\Sigma \setminus \Delta(F)$. Let~$\Sigma' \coloneqq {\Delta(F')\setminus \Delta(F)^\circ}$. Then~$\restr{\LLL}{X}$ is an edge-disjoint~$\{\rho(F),\rho(F')\}$-linkage.
\end{observation}

We extend the definition of in- and out-edges to cut-cycles as follows.

\begin{definition}
    Let~$(U,\nu)$ be an embedding of some digraph~$G$ in a surface~$\Sigma$ and let~$F$ be a cut-cycle in~$(U,\nu)$ bounding a disc~$\Delta(F) \subset \hat\Sigma$. We call $e \in \rho(F)$  an \emph{in-edge (of $F$)} if $\head(e) \in X(F)$ and an \emph{out-edge (of $F$)} if $\tail(e) \in X(F)$. 

    We call $F$ \emph{alternating} if there is a simple closed curve $\gamma:[0,1] \to F$ and an ordering $(e_1,\ldots,e_t)$ of $\rho(F)$ such that $e_i$ is an in-edge if and only if $e_{i+1}$ is an out-edge for all $1 \leq i < t$ and for every $1 \leq i < j \leq t$ there exist $t_i,t_j \in ]0,1[$ with $t_i < t_j$ such that $\gamma(t_i) = e_i$ and $\gamma(t_j) = e_j$. We write~$\pi_F \coloneqq (e_1,\ldots,e_k)$ and call~$\pi_F$ an \emph{alternating order of~$\rho(F)$}.
\end{definition}
\begin{remark}
    Note that since $G$ is Eulerian, $\rho(F)$ is partitioned into in- and out-edges, and the definition of alternating does not depend on a choice of $\gamma$ (as long as $\gamma(0) = \gamma(1)$ do not lie on $U$).
\end{remark}

We continue with standard definitions for embeddings.

\begin{definition}[Faces and~$2$-cell Embeddings]\label{def:faces_and_2cell}
    Let~$G$ be a digraph and let~$(U,\nu)$ be an embedding of~$G$ in some surface~$\Sigma$ (possibly with boundary). Let~$F(\Sigma,U)$\Symbol{FSIGMA@$F(\Sigma,U)$} be the set of connected components of~$\Sigma \setminus (U(G) \cup \bd(\Sigma))$, then we refer to~$F(\Sigma,U)$ as the \emph{faces} of the embedding and to each~$f \in F(\Sigma,U)$ as a \emph{face of the embedding}. 

    Let~$f \in F(\Sigma,U)$ with~$\bar{f} \cap \bd(\Sigma) = \emptyset$, then we call~$f$ an \emph{internal face} and otherwise a \emph{boundary face}. If every face is homeomorphic to an open disc we call~$(U,\nu)$ a~\emph{$2$-cell embedding}.
\end{definition}

Given embeddings in the disc~$\Delta$ away from the boundary, there is a unique face whose closure contains the boundary.

\begin{definition}[Outerface]
    Let~$G$ be a digraph and~$(U,\nu)$ an embedding of~$G$ in some disc~$\Delta$ such that~$U \cap \bd(\Delta) = \emptyset$. Let~$f \in F(\Delta,U)$ be such that~$\bd(\Sigma) \subset \bar{f}$. Then we refer to~$f$ as \emph{the outerface of~$U$}. 
\end{definition}

The following is a well-known toplogical fact.

\begin{observation}\label{obs:disc_embedd_is_2cell}
    Let~$G$ be a connected digraph and let~$(U,\nu)$ be an embedding of~$G$ in a sphere. Then~$(U,\nu)$ is~$2$-cell.
\end{observation}

Since graphs embeddable in a disc~$\Delta$ can be embedded in a sphere~$\hat \Delta$ we will refer to drawings in discs as~\emph{$2$-cell (up to their boundary face)}. Thus we will tacitly assume our embeddings to be~$2$-cell (up to boundary faces) whenever drawn in a disc, sphere or plane.

\begin{definition}[Outline of a Face]
    Let~$G$ be a digraph and~$(U,\nu)$ an embedding of~$G$ in some surface~$\Sigma$. Let~$f \in F(\Sigma,U)$ and let~$G'\subseteq G$ be minimal with respect to~$\Abs{V(G') \cup E(G')}$ such that~$\bar{f} \cap U(G) = U(G')$. Then we call~$G'$ the \Index{outline}\emph{outline of~$f$} and we write~$G' = \ol(f)$.\mSymbol[OUTL]{$\ol(F)$}
\end{definition}
\begin{remark}
    One easily verifies for an internal face~$f\in F(\Sigma,U)$ that~$U(\ol(f)) = \bar f \setminus f$, i.e., the outline of~$f$ \emph{traces} the boundary of~$f$.
\end{remark}
It is easy to see that~$G'$ exists regardless of whether or not~$U$ is~$2$-cell. Note however that~$G'$ may not be a directed circle a priori, and if the drawing is not~$2$-cell then~$\ol(f)$ may ``counter-intuitively'' contain vertices of degree~$\geq 3$ and thus not even the underlying undirected graph of~$G'$ is a circle. This changes when switching to Eulerian embedded graphs as discussed in the next subsection. 

Finally we will make use of \emph{homotopic paths} in the following sections.
\begin{definition}\label{def:homotopic}
    Let $\Sigma$ be a surface and let $(U, \nu)$ be an embedding of a (di)graph $G$ in $\Sigma$.
    \begin{enumerate}
    \item Let $P_1, P_2$ be internally vertex-disjoint paths between the same endpoints $u \neq v$. Then $P_1$ and $P_2$ are \emph{homotopic} if $\nu(P_1 \cup P_2)$ bounds a disc in $\Sigma$.
    \item Let $\zeta \neq \zeta'$ be distinct cuffs of $\Sigma$ and let $u_1, u_2$ and $v_1, v_2$ be distinct points on $\zeta, \zeta'$, respectively. Let $P_1$ be a path with endpoints $u_1, v_1$ and let $P_2$ be a path, vertex-disjoint from $P_1$ with endpoints $u_2, v_2$. We call $P_1$ and $P_2$ \emph{homotopic} if $\nu(P_1 \cup P_2)$ together with a segment of $\zeta$ between $u_1, u_2$ and a segment of $\zeta'$ between $v_1, v_2$ bounds a disc in $\Sigma$. 
    \end{enumerate}
\end{definition}
In contrast to the above, edge-disjoint paths may share common vertices whence they may not bound discs together. There is however a natural extension of homotopy to edge-disjoint paths as follows.
\begin{definition}\label{def:homotopic_edge-disjoint_paths}
    Let~$\Sigma$ be a surface and~$(U,\nu)$ an embedding of a (di)graph in~$\Sigma$. Let~$\Delta_1,\Delta_2 \subseteq \Sigma$ be~$U$-tracing discs such that~$\Delta_2 \subset \Delta_1$ is a strict subset. Fix an agreeing orientation for~$\bd(\Delta_1)$ and $\bd(\Delta_2)$. Let~$\Sigma' \coloneqq {\Delta_1 \setminus \Delta_2^\circ}$. Let~$P_1,P_2$ be edge-disjoint paths in~$G$ such that~$\nu(E(P_1) \cup E(P_2)) \in \Sigma'$ and~$P_1,P_2$ have each one end in~$\nu^{-1}(\bd(\Delta_1))$ and~$\nu^{-1}(\bd(\Delta_2))$ respectively (where $\nu^{-1}(\bd(\Delta)) \cap \nu^{-1}(\bd(\Delta')) \neq \emptyset$ is possible).  We call~$P_1$ and~$P_2$ \emph{homotopic (in the cylinder~$\Sigma'$)} if and only if there exists a simple curve~$F$ with one endpoint in~$\bd(\Delta_2)$ and one endpoint in~$\bd(\Delta_1)$ such that~$F \cap \nu(P_i) = \emptyset$ for $i=1,2$ and such that~$\nu(E(P_i))$ is contained in one component of~$\Sigma'\setminus F \cup \nu(P_j)$ for $\{i,j\} = \{1,2\}$.
\end{definition}
\begin{remark}
    Informally two edge-disjoint paths with their ends on two cuffs of a cylinder are homotopic if and only they do not ``cross'' but they may share vertices and can be continuously deformed into each as with usual homotopy.
\end{remark}

\subsubsection{Eulerian Embeddings.} 

The following is a well-known fact regarding Eulerian embeddings following easily from  the restrictions of the embedding.

\begin{observation}\label{obs:faces_in_2-cell_Euler_embeddings_bounded_by_cycle}
    Let~$G$ be an Eulerian digraph and~$(U,\nu)$ an Eulerian $2$-cell embedding of~$G$  in some surface~$\Sigma$. Let~$f \in F(\Sigma,U)$ be internal. Then there exists a cycle~$C \subseteq G$ such that~$C = \ol(f)$. 
    If~$f$ is not internal, then there is a path~$P \subset G$ such that~$P = \ol(f)$.
\end{observation}
\begin{remark}
    Note that there may be a vertex of degree four in~$C$, for example if~$G=(\{v\},\{e_1,e_2\})$, both of its edges being loops, then there is a face that is bounded by the cycle~$(e_1,e_2,e_1)$.
\end{remark}

Thus, if~$\Pi$ is an Eulerian embedding, the outline of an internal face is again Eulerian. However, as noted above, it may not be traced by a circle. By getting rid of \emph{cut-vertices} we can assure outlines to be circles.

\begin{definition}
    Let~$G$ be an Eulerian digraph. A vertex~$v \in V(G)$ is a \emph{cut-vertex} if there exist edge-disjoint~$G_1,G_2 \subset G$ satisfying~$E(G_i) \neq \emptyset$ for~$i=1,2$, such that~$G_1 \cup G_2 = G$ and~$V(G_1) \cap V(G_2) = \{v\}$.
\end{definition}

The following is straightforward.
\begin{observation}\label{obs:faces_in_2-cell_Euler_embeddings_bounded_by_circle}
     Let~$G$ be an Eulerian digraph without cut-vertices and let $(U,\nu)$ be an Eulerian~$2$-cell embedding of~$G$ in some surface~$\Sigma$. Let~$f \in F(\Sigma,U)$. If $f$ is internal then there exists a circle~$C \subseteq G$ such that~$C = \ol(f)$. If $f$ is not internal then there exists a linear path $P \subset G$ such that $P = \ol(f)$.
\end{observation}

We define discs in embeddings bounded by circles.

\begin{definition}
    Let~$G$ be an Eulerian digraph, let $C \subseteq G$ be a circle, and let $(U,\nu)$ be an Eulerian embedding of~$G$ in some disc~$\Delta$. We define~$\Delta(C)$ to be the unique disc bounded by~$\nu(C)$ in~$\Delta$.
\end{definition}

Again, the following is easily derivable from the previous observations together with the fact that a drawing~$(U,\nu)$ in a disc~$\Delta$ can be transferred to a drawing in a sphere~$\hat{\Delta}$.

\begin{observation}\label{obs:outline_of_plane_drawing_is_cycle}
    Let~$G$ be an Eulerian digraph and~$(U,\nu)$ an Eulerian embedding of~$G$ in a disc~$\Delta$ with~$U \cap \bd(\Delta) = \emptyset$. Let~$f \in F(\Sigma,U)$, then there exists a cycle~$C \subseteq G$ such that~$C = \ol(f)$. Further, if~$G$ admits no cut-vertex,~$C$ is a circle and it bounds a disc~$\Delta(C) \subset \Delta$, where~$\Delta(C) = f$ whenever~$f$ is not the outerface and~$\Delta(C) \cap \bar{f} = \nu(C)$ for the outerface.
\end{observation}

We define induced embeddings for immersions.
\begin{definition}
     Let~$\Sigma$ be a surface. Let~$G,H$ be Eulerian digraphs and~$\Pi=(U,\nu)$ an Eulerian embedding of $G$ in~$\Sigma$. Let $\gamma: H \hookrightarrow G$ be a (strong) immersion. We define $\restr{\Pi}{\gamma}=(U_\gamma,\nu_\gamma)$ via $U_\gamma = U(\gamma(H))$ and $\nu_\gamma = \nu \circ \gamma$. 
\end{definition}
By \cref{obs:eulerian_emb_closed_under_splitting_and_immersion} we derive that $\restr{\Pi}{\gamma}$ is an Eulerian embedding again. Similarly we have the following.

\begin{lemma}\label{obs:graph_induced_by_disc_in_2cellEulerembedd_is_eulerian}
    Let~$\Sigma$ be a surface,~$G$ an Eulerian digraph and~$\Pi =(U,\nu)$ an Eulerian~$2$-cell embedding of~$G$ in~$\Sigma$. Let $H$ be an Eulerian digraph and~$\gamma: H \hookrightarrow G$  a strong immersion with an induced~$2$-cell embedding~$\restr{\Pi}{\gamma}=(U_\gamma,\nu_\gamma)$. Let~$f \in F(\Sigma,U_\gamma)$ and~$\Delta = \bar{f}$. Then~$G[\Delta;U] \subseteq G$ is Eulerian.
\end{lemma}
\begin{proof}
    Note that~$\Delta$ is a disc since~$U_\gamma$ is~$2$-cell. Let~$v \in V(G[\Delta;U])$ be a non-eulerian vertex. Then there exists an edge~$e \in E(G)$ such that, without loss of generality, $v = \tail(e)$ and~$\nu(e) \in \Sigma \setminus \Delta$; in particular $\deg(v) > 2$ since $\nu(v) \in \bd(f)$ is adjacent to the edge $e$ not drawn on $\bd(f)$. Thus $\deg(v) = 4$ since $G$ is Eulerian embedded. Let~$G' \coloneqq G-\ol_{U'}(f)$, then by \cref{obs:outline_of_plane_drawing_is_cycle}~$G'$ is Eulerian; note that~$\ol_{U_\gamma}(f) \subseteq G[\Delta;U]$ by definition and $e \notin E(\ol_{U_\gamma}(f))$. Since~$G'$ is Eulerian, let~$C_e \subset G'$ be a cycle starting with~$e$, then clearly~$C_e$ contains two edges~$e,e' \in \rho(v)$, and since~$v$ must be of degree four in~$G$,~$e,e'$ are the remaining edges of~$v$ in~$G'$. 
    
    One easily verifies that~$U(C_e) \subseteq \Sigma \setminus \Delta$, for otherwise there is a path~$P \subset C_e$ with one end in~$\Sigma \setminus \Delta$ and one end in~$\Delta$, and since~$\Delta$ is a disc and the embedding is in particular planar,~$V(P) \cap V(\ol(f)) \neq \emptyset$; let~$w \in V(P) \cap V(\ol(f))$. Then~$\nu(w)$ must be strongly planar, a contradiction to the embedding being Eulerian as given by \cref{obs:eulerian_emb_closed_under_splitting_and_immersion}.
\end{proof}

By \cref{obs:graph_induced_by_disc_in_2cellEulerembedd_is_eulerian} we may define $G[\Delta(C);U] \coloneqq G[\Delta;U]$ where $\Delta$ is the closed disc bounded by $C$ extending \cref{obs:graphs_inside_Oarcs} to circles.

\begin{observation}\label{obs:graphs_inside_circles_are_Eulerian}
Let~$\Sigma$ a surface,~$G$ an Eulerian digraph and~$(U,\nu)$ an Eulerian~$2$-cell embedding of~$G$ in~$\Sigma$. Let~$C\subseteq G$ be a circle bounding a disc~$\Delta(C) \subset \Sigma$. Then~$G[\Delta(C);U] \subseteq G$ is Eulerian. In particular either $C$ bounds a face, or $G[\Delta(C);U] = G$, or $G-C$ is disconnected.
\end{observation}

It turns out that under certain restrictions cut-cycles in Eulerian embeddings behave similar to Eulerian vertices.

\begin{lemma}\label{lem:cut_cycles_are_alternating}
    Let~$\Sigma$ be a surface,~$G$ an Eulerian digraph and~$(U,\nu)$ an Eulerian~$2$-cell embedding of~$G$ in~$\Sigma$. Let~$F$ be a cut-cycle in~$(U,\nu)$ bounding a disc~$\Delta(F) \subset \hat\Sigma$. Then $F$ is alternating.
    
\end{lemma}
\begin{proof}
    Recall that by \cref{obs:cut-cycle_induce_cut} we have~$\rho(F) = \rho(X(F))$.
    Let~$d \coloneqq \delta(F)$ and let~$\gamma\colon[0,1] \to F$ be a simple closed curve such that there exists an ordering~$\pi_F=(e_0,\ldots,e_{d-1})$ of~$\rho(F)$ such that for~$0 \leq  i < j < d$ there exist~$t_i,t_j \in ]0,1[$ with~$t_i<t_j$ satisfying~$\nu^{-1}(\gamma(t_i)) = e_i$ and~$\nu^{-1}(\gamma(t_j)) = e_j$. We need to show that~$\pi_F$ is alternating.
    
    Since~$G$ is Eulerian~$2$-cell embedded, for every internal face~$f \in F(\Sigma,U)$,~$C_f \coloneqq \ol(f)$ is a cycle by \cref{obs:faces_in_2-cell_Euler_embeddings_bounded_by_cycle}. For boundary faces~$f$ we have~$P_f\coloneqq \ol(f)$ is a linear path~$P$, but the face~$f$ is still homeomorphic to an open disc, in particular there is a simple closed curve~$\gamma_f:[0,1] \to \bd(f)$ such that~$\gamma_f([0,1]) = \bd(f)$. Without loss of generality we assume~$C_f$ to be a circle, otherwise by Eulerianness we may decompose~$C_f$ into circles by separating them at cut-vertices using any choice of circle cover and then argue for each of the circles individually. Since~$F$ is a cut-cycle it must intersect~$\bd(f)$ an even number of times (possibly~$0$), and in particular it must intersect~$C_f$ (and/or $P_f$) an even number of times.
    
    Thus, if~$F \cap \nu(E(C_f)) \neq \emptyset$ there exist two edges~$e,e' \in E(C_f)$ and some~$0 \leq i \leq d-1$ such that~$e=e_i$ and~$e'= e_{(i+1)\mod d}$. We claim that~$e \in \rho^-(X(f)) \iff e' \in \rho^+(X(f))$. 
    This follows from a simple geometric observation: assume that~$\gamma_f:[0,1] \to \bd(f)$ respects the orientation of the cycle $C_f$ (or the path~$P_f$). Define~$\operatorname{ins}(\gamma_f)$ to be~$f$ and~$\operatorname{out}(\gamma_f)$ to be~$\Sigma \setminus \bar f$. Then~$\gamma([0,1]) \setminus \gamma_f([0,1])$ consists of components---call them \emph{segments}---that lie either completely in~$\operatorname{ins}(\gamma_f)$ or~$\operatorname{out}(\gamma_f)$. Further, following~$\gamma$ these components must alternate and since~$\gamma_f([0,1]) \cap \gamma([0,1])$ is even, there exist distinct consecutive segments~$I_1,I_2,I_3$ of~$ \gamma([0,1]) \setminus \gamma_f([0,1])$ such that~$v = \bd(I_1) \cap \bd(I_2)$ and~$w = \bd(I_2) \cap \bd(I_3)$ where~$I_1,I_3 \subset \operatorname{ins}(\gamma_f)$ and~$I_2 \subset \operatorname{out}(\gamma_f)$. Then~$e = \nu^{-1}(v)$ and~$e' = \nu^{-1}(w)$ do the trick. %
\end{proof}

\subsubsection{Eulerian Embeddings of bead-rooted Graphs}
All of the above definitions and results naturally transfer to Eulerian embeddings~$(\Gamma,\nu,\omega)$ of bead-rooted Eulerian digraphs~$(G,\pi(W_\omega))$ in some surface~$\Sigma$, where the faces~$F(\Sigma,\Gamma)$ are again given by the components of~$\Sigma \setminus \Gamma$ and faces are called \emph{internal} if they are disjoint from~$\bd{(\Sigma)}$ and the embedding is called \emph{$2$-cell} if all the faces are homeomorphic to open discs. Analogous to \cref{obs:faces_in_2-cell_Euler_embeddings_bounded_by_circle}, internal faces of Eulerian~$2$-cell embeddings without cut-vertices fixing beads are guaranteed to be traced by circles. As opposed to the general case, not every circle in~$G$ is embedded in~$\Sigma$ for they might use vertices in~$W$. To this extent, we define the following.

\begin{definition}
    Let~$(\Gamma,\nu,\omega)$ be an embedding of a bead-rooted Eulerian digraph~$(G,\pi(W_\omega))$ in some surface~$\Sigma$. Let~$C \subseteq G$ be a circle such that~$\nu(C) \subseteq \Gamma$. Then we call~$C$ an \emph{internal circle of~$\Gamma$}. We may write~$C \subset \Gamma$ to mean an internal circle of~$\Gamma$ for simplicity. %
\end{definition}

In particular we get the following analogue of \cref{obs:graphs_inside_circles_are_Eulerian} for bead-rooted Eulerian digraphs.

\begin{observation}\label{obs:graphs_inside_circles_are_Eulerian_beadrooted}
Let~$\Sigma$ a surface,~$(G,\pi(W_\omega))$ a bead-rooted Eulerian digraph and~$\Pi=(\Gamma,\nu,\omega)$ an Eulerian~$2$-cell embedding of~$(G,\pi(W_\omega))$ in~$\Sigma$. Let~$C\subseteq \Gamma$ be an internal circle bounding a disc~$\Delta(C) \subset \Sigma$. Then~$G[\Delta(C);U] \subseteq G$ is Eulerian. In particular either $C$ bounds a face, or $G[\Delta(C);U] = G$ (hence $W = \emptyset$), or $G-C$ is disconnected. %
\end{observation}

\smallskip

Next we discuss an extension of \cref{obs:graphs_inside_circles_are_Eulerian} using cut-cycles and bead-rooted Eulerian digraphs. 

\begin{observation}\label{obs:cut-cycle_induces_bead-rooted_subgraph}
    Let~$(G,\pi(W_\omega))$ be a bead-rooted Eulerian digraph and~$(\Gamma,\nu,\omega)$ an Eulerian embedding of~$(G,\pi(W_\omega))$ in a surface~$\Sigma$. Let~$F \subset \Sigma$ be a cut-cycle bounding a disc~$\Delta(F) \subseteq \Sigma$.
    Then~$\overline{X(F)}$ induces a cut in~$G$ such that $W_i \subseteq \overline{X(F)}$ for every $1 \leq i \leq t$ given the $\omega$-partition $W_\omega=(W_1,\ldots,W_t)$.%
\end{observation}
\begin{proof}
    We already established that~$\overline{X(F)}$ induces a cut in~$G$ (see \cref{obs:cut-cycle_induce_cut}), and clearly~$X(F) \cap W = \emptyset$ by definition.
\end{proof}

Note that for bead-rooted Eulerian digraphs the cuffs are cut-cycles--- they clearly bound a disc in $\hat{\Sigma}$. In particular, as an analogue to \cref{lem:cut_cycles_are_alternating} we have the following easy geometric observation following from the embedding restrictions; recall that for Eulerian embeddings only edges may be drawn on cuffs whence the cuffs are $\Gamma$-tracing by definition.

\begin{observation}\label{lem:cuffs_are_alternating}
    Let $(\Gamma,\nu,\omega)$ be an Eulerian $2$-cell embedding of a bead-rooted Eulerian digraph $(G,\pi(W_\omega))$ in a surface $\Sigma$. Let $c(\Sigma) = \{\zeta_1,\ldots,\zeta_t\}$ be the cuffs of $\Sigma$ for some $t \in \N$. Then for each $i \in \N$, $\zeta_i$ is an alternating cut-cycle.
\end{observation}%

Recall the \cref{def:stitching_beadrooted} of  stitching. The following is a consequence of \cref{obs:stitching_fundamentals}. 

\begin{observation}\label{obs:stiching_bead-rooted_graphs_disc}
    Let~$\bar{G}=(G,\pi(W_\omega))$ be a bead-rooted Eulerian digraph and let~$Y$ induce a rooted cut in $G$ with $W \subseteq Y$ with ordering~$\pi(Y)$. Let~$(G^Y,\pi(y^*)) \coloneqq \stitch(\bar{G};\pi(\bar{Y}))$ with up-stitch vertex~$y^*$. Then~$(G^Y,\pi(y^*))$ is a bead-rooted graph (with a single bead) such that~$G^Y[\bar{Y}] = G[\bar{Y}]$, and~$\rho_{G^Y}(y^*) = \rho_G(Y)$.
\end{observation}%

Combining \cref{obs:cut-cycle_induces_bead-rooted_subgraph} and \cref{obs:stiching_bead-rooted_graphs_disc}, we define an embedding for the bead-rooted graph obtained by stitching along a cut-cycle bounding a disc in $\Sigma$. While one easily defines it for general cut-cycles we will not need it and thus omit a definition.

\begin{definition}\label{def:induced_embedding_of_bead-rooted_graph_in_cut-cycle}
     Let~$(G,\pi(W_\omega))$ be a bead-rooted Eulerian digraph and~$(\Gamma,\nu,\omega)$ an Eulerian embedding of~$(G,\pi(W_\omega))$ in a surface~$\Sigma$. Let~$F \subset \Sigma$ be a cut-cycle bounding a disc~$\Delta(F) \subseteq \Sigma$. Fix an ordering~$\pi(X(F)) = \pi(\overline{X(F)})$.
     Let~$(G^F,\pi(w^*)) \coloneqq \stitch(\bar{G};\pi({X(F)}))$ with up-stitch vertex~$w^*$. We define the \emph{Eulerian embedding~$(\Gamma_F,\nu_F,\omega_F)$ of~$(G^F,\pi(w^*))$ in~$\Delta(F)$ induced by~$(\Gamma,\nu,\omega)$} via~$\Gamma_F \coloneqq \Gamma \cap \Delta(F)$ and~$\nu_F \coloneqq \restr{\nu}{\nu^{-1}(\Delta(F))}$ and finally set~$\omega_F(w^*) = F$. %
\end{definition}
Clearly~$(\Gamma_F,\nu_F,\omega_F)$ is an Eulerian embedding of~$(G^F,\pi(w^*))$: define~$\hat{\nu}(w^*)$ to be the center of the disc~$\hat{\Delta}(F) \setminus \Delta(F)$ and draw lines connecting~$\hat{\nu}(w^*)$ to the edges drawn on the cuff accordingly to get~$\hat{\Gamma}_F$. A different way to see this is: take~$\hat{\Gamma} \subset \hat{\Sigma}$ and contract~$\overline{X(F)}$ to the single vertex~$\hat{\nu}(w^*)$.

\subsubsection{A Menger Theorem on the cylinder} We continue with a version of Menger's Theorem for Eulerian embeddings (of bead-rooted) Eulerian digraphs in the disc and cylinder; note that this theorem crucially needs the Eulerian embedding as well as the Eulerianness of the digraph and does not lift to general digraphs. While the proof is straightforward, this theorem is one of the most crucial ingredients to the remainder of the paper, and the fact that it depends on the Eulerianness is why many of the results presented do not transfer to general digraphs.%

\begin{definition}\label{def:cycle_connectivity}
    Let~$(G,\pi(W_\omega))$ be a bead-rooted Eulerian digraph with beads~$W \subset V(G)$ and~$(\Gamma,\nu,\omega)$ an Eulerian embedding of~$(G,\pi(W_\omega))$ in a surface~$\Sigma$ with cuffs~$\zeta_1,\ldots,\zeta_t$ and~$\omega$-partition~$W_\omega=(W_1,\ldots,W_t)$ for some~$t \in \N$. Let~$1 \leq i,j \leq t$ be distinct and~$E^i \coloneqq \rho(W_i)$ and~$E^j \coloneqq \rho(W_j)$. We define~$\kappa(\zeta_i,\zeta_j;\Gamma)$ to be the maximum order of an~$\{E^i,E^j\}$-linkage in~$G$.

\end{definition}
\begin{remark}
    Recall that~$E^i = \nu^{-1}(\zeta_i)$ by \cref{obs:bead-rooted_edges_on_cuffs}.
\end{remark}

Recall the \cref{def:stitching_beadrooted} of stitching for bead-rooted digraphs.

\begin{lemma}\label{lem:boundary_linked_Menger_for_embeddings_in_cylinder}
    Let~$(G,\pi(W_\omega))$ be a bead-rooted Eulerian digraph. Let~$(\Gamma,\nu,\omega)$ be an Eulerian~$2$-cell embedding of~$G$ in a cylinder~$\Sigma$ with cuffs~$\zeta_1, \zeta_2$ fixing~$W$. 
    Then~$\kappa(\zeta_1,\zeta_2;\Gamma)$ is equal to the minimum of~$\delta(F)$ over all cut-cycles~$F \subset \Sigma$ that bound a disc~$\Delta(F) \subset \hat{\Sigma}$ such that~$\zeta_1 \subset \Delta(F)$ and~$\zeta_2 \subset \hat{\Sigma} \setminus \Delta(F)$.
\end{lemma}
\begin{proof}

Let~$\Sigma^\circ \coloneqq \Sigma \setminus (\zeta_1\cup \zeta_2)$, $W_\omega \coloneqq (W_1,W_2)$ and fix~$V \coloneqq V(G) \setminus W$; note that~$V = \nu^{-1}(\Sigma^\circ) \cap V(G)$. Let~$E^i \coloneqq \nu^{-1}(\zeta_i)$ for~$i=1,2$, by \cref{def:Euler_embedding_upto_beads} $E^i=\rho(W_i)$. The proof is by induction on~$\Abs{V}$.

\smallskip
If~$\Abs{V} = 0$, then~$\kappa(\zeta_1,\zeta_2;\Gamma) = \Abs{E^1 \cup E^2}$; note that loops at vertices in~$W$ are not drawn on cuffs.
\smallskip
Let~$\Abs{V} = n \geq 1$ and let~$F$ be a cut-cycle in~$(\Gamma,\nu)$ as in the hypothesis such that~$\delta(F)$ is minimal. We need to prove that there are~$\delta(F)$ edge-disjoint paths with ends in different cuffs (the other direction of the statement is trivial). Let~$v \in V$; if~$\deg(v) = 2$ the claim follows by dissolving~$v$ and applying induction, thus assume that~$\deg(v) = 4$. If~$v$ admits a loop, the claim follows by deleting the loops and subsequently dissolving~$v$ and applying induction, thus assume~$v$ admits no loop. Let~$\{e_1^i,e_2^i,e_1^o,e_2^o\} = \rho(v)$ where~$v = \head(e_j^i)$ for~$j=1,2$. Let~$(G^*,e^*) = \spl(G;(e_1^i,e_1^o))$ and~$(G^\star,e^\star) = \spl(G^*,(e_2^i,e_2^o))$; let~$(\Gamma^\star,\nu^\star,\omega)$ be the respective embedding of~$G^\star$ in~$\Sigma$ fixing $W$ such that~$\Gamma^\star$ agrees with~$\Gamma$ away from~$v,\rho(v)$ and define~$\nu^\star(e^*)$ and~$\nu^\star(e^\star)$ by drawing them close to~$\nu(v)$ in the obvious way (just trace along the old two-paths with the single new edges). 

Let~$F^\star$ be a cut-cycle in~$G^\star$ minimizing~$\delta_{\Gamma^\star}(F^\star)$. If~$\delta_{\Gamma^\star}(F^\star) = \delta_\Gamma(F)$, the claim follows by induction, for we get a~$\{E^1,E^2\}$-linkage~$\LLL^\star$ of order~$\delta(F)$ in~$G^\star$ and thus in~$G$ by \cref{obs:splitting_off_linkages_gives_linkages}. Now assume that $\delta_{\Gamma^\star}(F^\star) < \delta_\Gamma(F)$. Let~$f^\star \in F(\Sigma,\Gamma^\star) \setminus F(\Sigma,\Gamma)$ be the newly introduced face with~$e^*,e^\star \in \ol(f^\star)$. 

\begin{claim}\label{claim:boundary_linked_Menger_for_embeddings_in_discs_claim1}
$\delta_{\Gamma^\star}(F^\star) = \delta_{\Gamma}(F) - 2$ and in particular there exist cut-cycles~$F_1,F_2$ in~$\Gamma$ of order~$\delta_\Gamma(F)$ such that
\begin{itemize}
    \item ~$F_i$ are cut-cycles in~$\Gamma^\star$ with~$\delta_{\Gamma^\star}(F_i) = \delta_{\Gamma^\star}(F^\star)$ for~$i=1,2$,
    \item~$F_1 \cap \Gamma^\star = F_2 \cap \Gamma^\star$, and both visit~$\{e^*,e^\star\}$ consecutively up-to a cyclic rotation, and
    \item $\nu^{-1}(\Delta(F_2) \setminus \Delta(F_1)) \cap V(G) = \{v\}$ and~$\Delta(F_1) \subset \Delta(F_2)$.
\end{itemize} 
\end{claim}
\begin{claimproof}
    Since~$X(F)$ induces a cut in~$G$---since each~$w \in W$ has all its incident edges on the same cuff---it follows that~$\delta_{\Gamma^\star}(F^\star) \in 2\N$ using Eulerianness, hence~$\delta_{\Gamma^\star}(F^\star) \leq \delta -2$. Now if~$F^\star \cap f^\star = \emptyset$ then~$F^\star$ is also a cut-cycle in~$\Gamma$ that does not intersect~$\rho(v)$; in particular~$\delta_{\Gamma}(F^\star) = \delta_{\Gamma^\star}(F^\star) < \delta$; a contradiction.

    Thus~$F^\star \cap f^\star \neq \emptyset$ and clearly~$\Abs{F^\star \cap \bd(f^\star)} \leq 2$ for otherwise we could reroute the curve in~$f^\star$ accordingly resulting in a cut-cycle of lower order; in particular using the planarity of the drawing, since~$F^\star$ enters and leaves the face~$f^\star$, it visits the edges~$e^*,e^\star$ consecutively. Also~$\Abs{F^\star \cap \bd(f^\star)} >1$ since it is at least~$1$ and if it were exactly~$1$ this implies~$F^\star \subset \bar f^\star$ and therefore~$\delta(F^\star) = 1$ as a contradiction to it being even.

    Let~$E^+ \coloneqq \rho(v) \cap \nu^{-1}(\Delta(F^\star))$, and without loss of generality assume that~$\Abs{E^+} \geq 2$. By slightly re-routing~$F^\star$ inside~$f^\star$ we can guarantee that~$v \in \Delta(F^\star)$ (or equivalently~$v \notin \Delta(F^\star)$), producing a cut-cycle~$F_1$ (or~$F_2$ respectively) in~$\Gamma$ with~$\delta_G(F_i) \leq \delta_{\Gamma^\star}(F^\star) + 2$ for~$i=1,2$. The cut-cycles are easily seen to satisfy the remaining statements of the claim, concluding the proof.
\end{claimproof}

In particular, using \cref{claim:boundary_linked_Menger_for_embeddings_in_discs_claim1} we may assume that~$F_1,F_2$ are cut-cycles in~$\Gamma$ such that~$F_1 \cap \Gamma^\star = F_2 \cap \Gamma^\star$ with~$\delta_{\Gamma}(F_i) = \delta(F)$ and~$v \notin \Delta(F_1)$ but~$v \in \Delta(F_2)$. Note further that~$\Abs{F_i \cap \nu(\rho(v))} = 2$ for~$i=1,2$ and~$(\nu^{-1}(F_1) \cap \nu^{-1}(F_2))\cap \rho(v) = \emptyset$, i.e. the curves cut different edges of~$\rho(v)$. 

\begin{claim}\label{claim:boundary_linked_Menger_for_embeddings_in_discs_claim2}
   Let~$j \in \{1,2\}$ and~$e^i,e^o \in \nu^{-1}(F_j) \cap \rho(v)$. Then~$\head(e^i) = v$ and~$\tail(e^o) = v$.
\end{claim}
\begin{claimproof}
    By assumption~$v$ is Eulerian embedded and the embedding is~$2$-cell. Thus, any cut-cycle that visits two edges of~$\rho(v)$ consecutively must respect the cyclic order of the edges at~$v$ in the embedding since at most two edges of~$\rho(v)$ are part of the same face; see also \cref{obs:faces_in_2-cell_Euler_embeddings_bounded_by_cycle}.
\end{claimproof}

    For the following we focus on~$F_1$, whence~$v \notin \Delta(F_1)$ (the other case is analogous using down-stitches and centering around~$\zeta_2$ as we briefly discuss below). Let~$X(F_1) \coloneqq \big(\nu^{-1}(\Delta(F_1)) \cap V(G) \big ) \cup W_1$---note that since $\Sigma$ is not a disc, $X(F_1)$ is not defined a priori---then one easily verifies that $X(F_1)$ induces a cut in $G$ where $W_2 \cap X(F_1) = \emptyset$. Recall that $W_\omega=(W_1,W_2)$ is a partition, thus the up- and down-stitches at $X(F_1)$ are well-defined using \cref{def:stitching_beadrooted}. Let~$(G^\rho, \pi(W^\rho))$ be the respective bead-rooted Eulerian digraph obtained from~$(G,\pi(W_\omega))$ by up-stitching~$X(F_1)$ with up-stitch vertex~$x^*$ where~$W^\rho \coloneqq W_1 \cup \{x^*\}$; note that~$\rho(X(F_1)) = \rho(F_1)$. We define the cylinder~$\Sigma_1 \coloneqq \Delta(F_1) \cap \Sigma$ with cuffs~$\zeta_1$ and~$\zeta$, the latter of which is ``traced'' by~$F_1$. Next define~$(\Gamma^\rho,\nu^\rho,\omega^\rho)$ to be an Eulerian~$2$-cell embedding of~$G^\rho$ in~$\Sigma_1$ where~$\omega^\rho(w) = \zeta_2$ if and only if~$w=x^*$ and~$(\Gamma^\rho,\nu^\rho)$ is equal to~$(\Gamma,\nu)$ up to~$\rho(x^*) = \rho(F_1)$ which is drawn on the new cuff~$\zeta$ in the obvious way: contract~$\zeta_2$ to~$\zeta$ via the obvious homotopy. Note that the edges drawn on~$\zeta_1$ are exactly~$E^1$ while the edges drawn on~$\zeta$ are exactly~$\rho(F_1)$.

\smallskip

By assumption~$v \notin V(G^\rho)$, therefore our new embedding has at least one internal vertex less; in particular the lemma holds by induction and we conclude that there is a cut-cycle~$F_1' \subset \Sigma_1$ such that~$\Delta(F_1') \subset \Sigma_1$ is a disc with~$\zeta_1 \subset \Delta(F_1')$ such that~$\kappa(\zeta_1,\zeta;\Gamma^\rho) = \delta(F_1')$. By minimality of~$F_1$ we know further that~$\delta_{\Gamma}(F_1) \leq \delta_{\Gamma^\rho}(F_1')$ and in particular there is a~$\{E^1,\rho(F_1)\}$-linkage~$\LLL_1$ in~$G^\rho$ of order~$\delta_{\Gamma^\rho}(F_1')$. Since~$\Abs{\rho(F_1)} = \delta(F_1)$ by definition we conclude that~$\LLL_1$ has order exactly~$\delta(F_1)$.

\smallskip

Similarly, by setting~$F = F_2$ with~$v \in \Delta(F_2)$ and restricting our attention to the cylinder~$\Sigma_2 \coloneqq \Sigma \setminus \Delta(F_2)$ with cuffs~$\zeta_2$ and~$\zeta$ where the latter is traced by~$F_2$ we get a~$\{\rho(F_2),E^2\}$-linkage~$\LLL_2$---note that~$\rho(F_2)$ and~$\rho(F_1)$ agree up-to the edges adjacent to~$v$ by \cref{claim:boundary_linked_Menger_for_embeddings_in_discs_claim1}---and we finally concatenate both linkages by adding back~$v$ and using the fact that~$\rho(F_1)$ and~$\rho(F_2)$ both contain exactly one edge of~$\rho(v)$ with~$v$ as head and one with~$v$ as a tail by \cref{claim:boundary_linked_Menger_for_embeddings_in_discs_claim2} resulting in the desired~$\{E^1,E^2\}$-linkage of order~$\delta(F)$.
\end{proof}

We get an immediate corollary to \cref{lem:boundary_linked_Menger_for_embeddings_in_cylinder} for discs.

\begin{definition}
    Let~$G$ be Eulerian with beads~$W\subset V(G)$ and let~$(\Gamma,\nu,\omega)$ be a respective Eulerian embedding in a disc~$\Delta$ with cuff~$\zeta_1$ fixing~$W$. Let~$C \subseteq \Gamma$ be an internal circle and~$F \subset \Delta$ a cut-cycle satisfying~$\rho(F) = \rho(V(C))$. Let~$\Sigma \coloneqq \Delta \setminus \Delta(F)^\circ$ be the cylinder with cuffs~$\zeta_1$ and~$\zeta = F$. Let~$G^\star$ be obtained from~$G$ by down-stitching~$V(C)$ and let~$x_*$ be the down-stitch. Let~$(\Gamma^*,\nu^\star,\omega^\star)$ be the Eulerian embedding of~$G^\star$ fixing~$W \cup \{x_*\}$ in~$\Sigma$ such that $(\Gamma^\star,\nu^\star)$ agrees with~$(\Gamma,\nu)$ on~$\Sigma$, and additionally~$\omega^\star(w)=\zeta$ if and only if~$w = x_*$. 
    
    Then we define~$\kappa(\zeta_1,C;\Gamma) \coloneqq \kappa(\zeta_1,\zeta;\Gamma^\star)$.
\end{definition}
\begin{remark}
    Note that $\Delta \setminus \Delta(F)^\circ$ may not be the same set as $\overline{\Delta \setminus \Delta(F)}$ since the discs may share part of their boundary. In particular the latter may not be a cylinder.
\end{remark}

\begin{corollary}\label{cor:Menger_on_a_disc}
     Let~$G$ be Eulerian with beads~$W$. Let~$(\Gamma,\nu,\omega)$ be an Eulerian embedding of~$G$ in a disc~$\Delta$ with cuff~$\zeta_1$ fixing~$W$. Let~$C \subseteq \Gamma$ be an internal circle.
    Then~$\kappa(\zeta_1,C;\Gamma)$ is equal to the minimum of~$\delta(F)$ over all cut-cycles~$F \subset\Sigma$ that bound a disc~$\Delta(F) \subseteq \Delta$ such that~$\Delta(C) \subseteq \Delta(F)$ and~$\zeta_1 \subset \Delta \setminus \Delta(F)^\circ$.
\end{corollary}

Further there is a well-known and easy observation for linkages in Eulerian embedded digraphs in  a cylinder, namely no two paths of such a linkage ``cross''. See \cite{Frank_2path} for details.

\begin{observation}\label{obs:linkages_in_euler_embedded_cylinder_are_homotopic}
    Let~$(\Gamma,\nu,\omega)$ be an Eulerian embedding of a bead-rooted Eulerian digraph~$(G,\pi(W_\omega))$ in some surface~$\Sigma$. Let~$\Delta,\Delta' \subseteq \Sigma$ be discs such that~$\Delta' \subset \Delta$ is a strict subset. Let~$\LLL$ be a linkage in~$G$ such that every path in~$P$ has its endpoints in~$\Sigma \setminus \Delta$. Then any two paths~$P,P' \in \restr{\LLL}{\overline{\Delta \setminus \Delta'}}$ that both have one end on~$\bd(\Delta)$ and one end on~$\bd(\Delta')$ are homotopic in $\overline{\Delta \setminus \Delta'}$.
\end{observation}
In fact Frank et al. \cite{Frank_2path} prove that under mild restrictions the converse to \cref{obs:linkages_in_euler_embedded_cylinder_are_homotopic} is also true, i.e., they prove an analogue to the well-known Two-Path Theorem (see \cite{GMVI,KWT_flatwall}) for Eulerian digraphs. We direct the reader to \cite[Theorem 3.7]{Frank_2path} and \cite[Section 5.5]{EDP_Euler} for further details.

\subsubsection{Reduction to Linegraphs} While many of the results in the following sections can be directly proved using the Eulerianness of the digraph, we will derive part of them from known results of \cite{GMVII,GMVIII} on undirected graphs to shorten the exposition. Since the known results are about standard minors, vertex separations and vertex-disjoint paths we will need the following construction of a directed \emph{linegraph} of~$G$. 

\begin{definition}[Linegraph of~$G$ and~$G^+$]\label{def:linegraph}
Let~$G$ be a digraph and let $G^+$ be the digraph obtained from~$G$ by subdividing each edge once. 

The \emph{directed linegraph $\ell(G)$ of~$G$} is defined as $\ell(G) \coloneqq (E(G^+), \PPP, \head, \tail)$, where~$\PPP$ is the set of~$2$-paths in~$G^+$, and~$f\coloneqq (e_1,e_2) \in \PPP$ satisfies~$\head(f) = e_2 \in E(G^+)$ and~$\tail(f) = e_1$. Further, for any subdigraph~$H$ of~$G$ we write~$\ell(H)$ to mean the subdigraph of~$\ell(G)$ with~$V(\ell(H)) = E(H)$ and~$E(\ell(H)) = \{(e_1,e_2) \mid (e_1,e_2) \text{ is a two-path in } H\}$.

Then we define~$\ell^+(G) \coloneqq \ell(G^+)$.
\end{definition}
Note that neither~$\ell(G)$ nor~$\ell^+(G)$ have loops. 

An easy but useful observation is the following.
\begin{observation}\label{obs:circles_are_circles_in_linegraph}
Let~$G$ be an Eulerian digraph and~$H \subseteq G$. Then
\begin{itemize}
    \item[(i)]~$H$ is a cycle if and only if~$\ell(H)$ is a circle, and
    \item[(ii)]~$H$ is a path if and only if~$\ell(H)$ is a linear path.
\end{itemize}
Similarly, let~$H_\ell \subseteq \ell(G)$. Then
\begin{itemize}
    \item[(iii)]~$H_\ell$ is a circle if and only if there is a cycle~$H' \subseteq G$ with~$\ell(H') = H_\ell$ , and
    \item[(iv)]~$H_\ell$ is a linear path if and only if there is a path~$H' \subseteq G$ with~$\ell(H') = H_\ell$.
\end{itemize}
\end{observation}
\begin{remark}
    The observation is one of the reasons why paths and cycles in this exposition are defined the way they are.
\end{remark}

On one hand~$\ell^+(G)$ gives us control over the degrees of the vertices; we say that a (di)graph is \emph{cubic} if the underlying undirected graph has maximum degree~$3$.
\begin{observation}\label{obs:linegraph_is_cubic}
    Let~$G$ be a an Eulerian digraph of degree at most~$4$, then~$\ell^+(G)$ is cubic and~$\ell(G)$ has degree at most~$4$.
\end{observation}
\begin{proof}
    We discuss the case of~$\ell^+(G)$; the case for~$\ell(G)$ is analogous. Let~$e \in E(G^+)$, then~$e$ is adjacent to a degree-two vertex, say~$v=\head(e)$ is of degree~$2$. Then, by Eulerianness, there is~$e' \in E(G^+)$ adjacent to~$v$ with~$v = \tail(e')$ and~$(e,e') \in \PPP$ where~$\PPP$ is the set of~$2$-paths in~$G^+$. In particular, this is the only~$2$-path in~$\PPP$ starting with~$e$ and thus~$e$ is a vertex of out-degree~$1$ in~$\ell^+(G)$. Let~$w$ be the other endpoint of~$e$. Then~$\deg(w) \in \{2,4\}$. If~$\deg(w) = 2$ then argue as above, whence assume that~$\deg(w) = 4$ and~$w = \tail(e)$. Again by Eulerianness, let~$e_1,e_2$ be the two edges having~$w$ as a head. Then~$(e_1,e),(e_2,e) \in \PPP$ are the two distinct and only~$2$-paths with~$e$ as an end and thus~$e$---as a vertex of~$\ell^+(G)$---has in-degree~$2$.
\end{proof}

On the other hand~$\ell(G)$ and~$\ell^+(G)$ don't lose information regarding linkages.
\begin{observation}\label{obs:linkages_stay_same_in_linegraph}
    Let~$\LLL$ be a linkage in~$G$, then~$\LLL^+$, which is obtained by subdividing each path in~$\LLL$, is a linkage in~$G^+$. Let~$\QQQ^+$ be a linkage in~$G^+$ starting and ending in vertices of~$V(G)$, then~$\QQQ$, which is obtained by dissolving the vertices in~$V(G^+) \setminus V(G)$ is a linkage in~$G$.

    Let~$\LLL'$ be a linkage in~$G^+$. Then~$\ell(\LLL') \coloneqq \{P \mid P = \ell(L) \text{ for } L \in \LLL'\}$ is a set of vertex-disjoint linear paths in~$\ell^+(G)$. And let~$\QQQ'$ be a set of vertex-disjoint linear paths in~$\ell^+(G)$, then there exists a linkage~$\QQQ$ in~$G^+$ with~$\ell(\QQQ) = \QQQ'$.
\end{observation}
\begin{proof}
    The first part is obvious by construction. The second part follows by definition of~$\ell^+(G) = \ell(G^+)$, transforming edge-disjoint linkages into vertex-disjoint linkages where~$2$-paths are replaced by edges (and edges by single vertices, where paths consisting of a single edge are transformed into paths consisting of a single vertex); and vice-versa vertex-disjoint linkages are transformed into edge-disjoint linkages, where vertices become edges, and edges are replaced by~$2$-paths.
\end{proof}

Finally let us define embeddings of linegraphs of bead-rooted Eulerian digraphs. 

\begin{definition}[Induced linegraph~$\ell(\Gamma;\Sigma)$ and its Embedding]\label{def:linegraph_induced_by_bead_embedding}
    Let~$(G,\pi(W_\omega))$ be bead-rooted Eulerian digraph with beads~$W$ and let~$(\Gamma,\nu,\omega)$ be an Eulerian~$2$-cell embedding of~$(G,\pi(W_\omega))$ in some surface~$\Sigma$. 

    We define the \emph{linegraph~$\ell(\Gamma;\Sigma) \subseteq \ell(G)$ induced by~$\Gamma$} to be the graph with vertex set~$V(\ell(\Gamma;\Sigma)) \coloneqq \nu^{-1}(\Sigma) \cap E(G)$ and~$E(\ell(\Gamma;\Sigma)) \subseteq E(\ell(G))$ to be maximum with the property that~$(e_1,e_2) \in E(\ell(\Gamma;\Sigma))$ if and only if~$e_1,e_2 \in V(\ell(G;U))$ and~$\nu(v) \in \Sigma$ where~$v \in V(G)\setminus W$ is the head of~$e_1$ and tail of~$e_2$. 

    We define an \emph{embedding~$(\Gamma_\ell,\nu_\ell)$ of~$\ell(\Gamma;\Sigma)$  induced by~$(\Gamma,\nu,\omega)$} in the obvious way:
    \begin{itemize}
        \item for every~$e \in V(\ell(\Gamma;\Sigma))$ we set $\nu_\ell(e) \coloneqq \nu(e)$, and
        \item for every~$(e_1,e_2) \in E(\ell(\Gamma;\Sigma))$ there is a line~$I$ in~$\Gamma_\ell$ between~$\nu_\ell(e_1)$ and~$\nu_\ell(e_2)$ otherwise disjoint from~$\nu_\ell(V(\ell(\Gamma;\Sigma))$ parallel and very close to the respective line of~$\Gamma$ representing the path~$(e_1,e_2)$, such that~$I$ does not intersect~$U$ except for its endpoints. Then we draw~$\nu_\ell((e_1,e_2))$ on that line~$I$ very close to~$\nu(v)$ (that is in a face of~$\Gamma$ with~$e_1,e_2$ on its boundary).
    \end{itemize}
\end{definition}%
\begin{remark}
    Note that for degree-two vertices~$v$ we may set~$\nu_\ell((e_1,e_2)) = \nu(v)$ for the two-path~$(e_1,v,e_2) \subset G$, getting rid of a ``choice of a face'' to draw~$\nu_\ell((e_1,e_2))$ in.
\end{remark}

The following is easily verified; we leave the details to the reader.

\begin{observation}\label{obs:properties_of_linegraph_embedding}
    Let~$(G,\pi(W_\omega))$ be a bead-rooted Eulerian digraph and let~$(\Gamma,\nu,\omega)$ be an Eulerian~$2$-cell embedding of~$(G,\pi(W_\omega))$ in some surface~$\Sigma$. Let~$\ell(\Gamma;\Sigma)$ be the linegraph induced by~$\Gamma$ and~$(\Gamma_\ell,\nu_\ell)$ an induced embedding. Then
    \begin{itemize}
        \item[(i)]$(\Gamma_\ell,\nu_\ell)$ is an embedding, in particular every component of~$\Gamma_\ell \setminus \nu\big(V(\ell(\Gamma;\Sigma)) \cup E(\ell(\Gamma;\Sigma))\big)$ is homeomorphic to a line,
        \item[(ii)]$(\Gamma_\ell,\nu_\ell)$  does not admit boundary-edges, and every boundary-edge of~$(\Gamma,\nu,\omega)$ is a boundary-vertex of~$(\Gamma_\ell,\nu_\ell)$ and vice-versa. In particular~$\Abs{\nu_\ell^{-1}(\bd(\Sigma))} = \Abs{\nu^{-1}(\bd(\Sigma))}$,
        \item[(iii)] every boundary-vertex of~$(\Gamma_\ell,\nu_\ell)$ is either a sink (out-degree $0$) or a source (in-degree $0$),
        \item[(iv)] for every~$O$-arc~$F$ in~$(\Gamma_\ell,\nu_\ell)$ bounding a disc $\Delta(F) \subseteq \Sigma$  there is an equivalent $O$-arc~$F'$ in~$(\Gamma_\ell,\nu_\ell)$ such that~$F'$ is a cut-cycle in~$(\Gamma,\nu,\omega)$ with~$\alpha(F) = \rho(F')$ and such that~$G[\Delta(F);\Gamma_\ell] = \ell(\Gamma_{F'};\Delta(F'))$ where~$(\Gamma_{F'},\nu_{F'},\omega_{F'})$ is the induced embedding in~$\Delta(F')$ as in \cref{def:induced_embedding_of_bead-rooted_graph_in_cut-cycle}. In particular~$I(F) = X(F')$. And,
        \item[(v)] for every cut-cycle $F$ in~$(\Gamma,\nu,\omega)$ there is an equivalent cut-cycle~$F'$ such that~$F'$ is an~$O$-arc in~$(\Gamma_\ell,\nu_\ell)$ with~$\rho(F) = \alpha(F')$ and such that~$ \ell(\Gamma_{F};\Delta(F)) = G[\Delta(F');\Gamma_\ell]$ where~$(\Gamma_{F},\nu_{F},\omega_{F})$ is the induced embedding in~$\Delta(F')$ as in \cref{def:induced_embedding_of_bead-rooted_graph_in_cut-cycle}. In particular~$X(F) = I(F')$.
    \end{itemize}
\end{observation}%
\begin{proof}
    Claims~$(i),(ii),(iii)$ are purely topological and follow from the undirected case, the definition and the \cref{def:Euler_embedding_upto_beads} of bead-rooted Eulerian embeddings respectively.
    For claims~$(iv)$ and~$(v)$ note that~$F'$ can be chosen to be~$F$ after some slight adjustments to~$F$ that do not affect~$\alpha(F)$ and~$\rho(F)$ (and the respective intersections of~$F$ with the embeddings) for the respective claims, by slightly enlarging or shrinking the cycle. The rest follows from duality between edges in the graph and vertices in the linegraph together with the previous claims, \cref{obs:cut-cycle_induces_bead-rooted_subgraph} and \cref{obs:graphs_inside_circles_are_Eulerian}.
\end{proof}

Note that~$(i)$ of \cref{obs:properties_of_linegraph_embedding} implies that if~$\Sigma$ is a disc then~$\ell(\Gamma;\Sigma)$ is planar.
\label{sec:surface_defs}

\section{Strong Immersions for High Representativity}\label{sec:high-rep}

This section deals with the case that our embeddings are of high ``representativity'', a notion we will make precise shortly. In that case we are able to leverage results from \cite{GMVII} using the above defined linegraphs. Since the needed results we wish to transfer are stated for undirected graphs and vertex disjoint paths we need a few lemmas that help us switch the settings accordingly.

Adapting the notation in \cite{MoharT2001} to our setting, we define the \emph{vertex-face-graph}\Index{vertex-face-graph} of an embedded graph. 

\begin{definition}\label{def:vertex-face-graph}
  Let~$\Sigma$ be a surface, $G$ a graph and let $(U, \nu)$ be an embedding of $G$ in~$\Sigma$. Let $F \coloneqq F(\Sigma,U)$ be the set of faces of $(U, \nu)$. 
  
  The \emph{vertex-face-graph}\Index{vertex-face-graph} of $(U, \nu)$ is defined as the undirected graph $G^*$ with $V(G^*) \coloneqq V(G) \cup F$ and 
  \[
  E(G^*) \coloneqq \big\{ \{ v, f \} \sth v \in V(G), f \in F\text{ and } \nu(v) \in \bd(f) \big\}.
  \] 
\end{definition}

By construction, every $U$-normal simple curve $C \subseteq \Sigma$ defines a walk in the vertex-face-graph in the obvious way and conversely for every path~$P$ in the vertex-face-graph there is a simple $U$-normal curve $C$ which defines the path $P$. We say that $C$ \emph{induces} the path $P$.

The next two lemmas provide the technical core that will allow us to transfer undirected vertex disjoint paths and cycles to edge-disjoint directed paths and cycles. 

\begin{lemma}\label{lem:vf-path-euler-paths}
    Let $\Sigma$ be a surface. Let $(U, \nu)$ be a Eulerian~$2$-cell embedding of an Eulerian digraph $G$ in $\Sigma$. Let $G^*$ be the vertex-face-graph of $(U, \nu)$. 
    Let $P$ be a linear path in $G^*$ starting and ending at vertices $u, v \in V(G)$ and let $F(P) \subseteq F(\Sigma,U)$ be the set of faces occurring on $P$. Then there are linear paths $P_1, P_2$ such that $P_1$ links $u$ to $v$ and $P_2$ links $v$ to $u$ such that  $P_1, P_2$ are both contained in the subgraph of $G$ induced by the vertices and edges on the boundary of the faces in $F(P)$. 
\end{lemma}
\begin{proof}
    As $(U, \nu)$ is a Eulerian~$2$-cell embedding, by \cref{obs:faces_in_2-cell_Euler_embeddings_bounded_by_cycle} every internal face is bounded by a cycle. Furthermore, if $P$ contains a subpath $(f,v,f')$, for some $v \in V(G)$,  then the boundaries of $f$ and $f'$ share the vertex $v$ (they may also share edges). This implies that the union of the boundaries of the faces in $F(P)$ induces a strongly connected digraph (which may not be Eulerian). The result follows. 
\end{proof}

\begin{lemma}\label{lem:vf-path-euler-circle}
    Let $\Sigma$ be a surface. %
    Let $(U, \nu)$ be an Eulerian~$2$-cell embedding of an Eulerian digraph $G$ in $\Sigma$. Let $G^*$ be the vertex-face-graph of $(U, \nu)$.    
    Let $C, C'$ be $O$-arcs in $\Sigma$ bounding discs $\Delta(C), \Delta(C')$ in $\hat{\Sigma}$, respectively, such that $\Delta(C')\subseteq\Delta(C)^\circ$. Let $C^\star$ be a cycle in the vertex-face graph $G^*$ of $(U, \nu)$ such that all faces on $C^\star$ are contained in $\Delta(C)$ and have an empty intersection with the interior of $\Delta(C')$. 

    Then $G$ contains two circles $S_1, S_2$ in the union of the boundaries of the faces in $C^\star$ such that $S_1, S_2$ both bound a disc containing $\Delta(C')$ and $S_1$ and $S_2$ are oriented in opposite directions.
\end{lemma}

\begin{proof}
    Let $F$ be the faces appearing on $C^*$.
    We first observe that all faces $f \in F$ must be internal. By \cref{obs:faces_in_2-cell_Euler_embeddings_bounded_by_circle}, the $\ol(f)$ is a cycle for every internal face~$f \in F(\Sigma,U)$. 
    Therefore, as in the proof of \cref{lem:vf-path-euler-paths}, the union $U' \coloneqq \bigcup \{ \ol(f) \sth f \in F \}$  induces a strongly connected digraph $H \subseteq G$.  Thus we can choose circles $S_1, S_2$ with a non-empty intersection with the boundary of every face on $C^\star$ such that $S_1$ and $S_2$ are oriented in opposite directions. As both circles intersect every face in $F$ they must bound discs in $\hat\Sigma$ containing $\Delta(C')$ in their interior.
\end{proof}

Note that the circles $S_1$ and $S_2$ in the previous lemma may not be edge-disjoint as neighbouring faces in $C^\star$ may share an edge that is then used by both circles.
Recall the \cref{def:internal_tracing_clean} of internal and clean subsets.

\begin{definition}\label{def:representativity}%
    Let $\Sigma$ be a surface other than the plane, cylinder, or sphere. Let $\{\zeta_1, \ldots, \zeta_k\} = c(\Sigma)$ be the cuffs of $\Sigma$. 
    Let $\bar G \coloneqq (G, \pi(W_\omega))$  be a bead-rooted Eulerian digraph and let $\Pi \coloneqq (\Gamma, \nu, \omega)$ be an Eulerian embedding of $\bar G$  in $\Sigma$.  

    Let $\gamma \sth [0,1] \to \Sigma$ be a $\Gamma$-tracing curve disjoint from the boundary of $\Sigma$ except possibly for its endpoints. Let $Y \coloneqq \gamma([0,1])$. Then

    \begin{itemize}
    \item $\gamma$ is \emph{cuff surrounding}\Index{curve!cuff surrounding} if $\gamma$ is closed, $Y$ is internal and there is a cuff $C$ such that $\gamma$ bounds a disk $\Delta(\gamma)$ in $\hat\Sigma$ which contains $C$ but no other cuff in its interior and $\Abs{Y \cap \Gamma} < \Abs{C \cap \Gamma}$.
    \item $\gamma$ is \emph{cuff shortening} if there is a cuff $C$ of $\Sigma$ such that $\gamma(0), \gamma(1) \in C \setminus \Gamma$ and
    $|Y \cap \Gamma| < |C' \cap\Gamma|$ for some component $C'$  of $C \setminus \{ \gamma(0), \gamma(1) \}$.
    \item $\gamma$ is \emph{cuff connecting}\Index{curve!cuff connecting}
  if its endpoints are on different cuffs of $\Sigma$.%
    \item $\gamma$ is \emph{cuff grouping}\Index{curve!cuff grouping} if $\gamma$ is a closed curve that is not cuff surrounding and $\Sigma \setminus \gamma$ is separated into two surfaces $\Sigma_1, \Sigma_2$, both containing fewer cuffs than $\Sigma$.
     \item $\gamma$ is \emph{genus reducing}\Index{curve!genus reducing} if it is closed or clean and cutting $\Sigma$ along $F$ results in one or more surfaces of genus lower than the genus of $\Sigma$.
      \end{itemize}
    We call $\gamma$  \emph{reducing}\Index{reducing curve}\Index{curve!reducing} if it is genus reducing or cuff grouping and we call it \emph{cuff separating}\Index{cuff separating}\Index{curve!cuff separating} if it is cuff surrounding or cuff shortening.

    If $\Pi$ does not have a cuff separating curve then we define the \emph{representativity}\index{representativity} $\fw(\Pi)$ as the minimum number $k$ such that every reducing curve intersects  $\Gamma$ in at least $k$ vertices. Otherwise we define $\fw(\Gamma) \coloneqq 0$. 
\end{definition}

\begin{remark}
     As we are working with digraphs of maximum degree $4$, if there is a genus reducing cut-cycle that intersects $\Gamma$ in $k$ edges, then there exists a genus reducing $O$-arc that intersects $\Gamma$ in $k$ vertices and, conversely, if there is a genus reducing $O$-arc intersecting $\Gamma$ in $k$ vertices then there is a genus reducing $O$-trace intersecting $\Gamma$ in at most $2k$ edges.

    Therefore, if $\Pi$ is an Eulerian embedding of a bead-rooted graph $(G, \pi(W_\omega))$ such that $\Pi$ has no cuff separating curve, then up to a factor of $2$ the representativity of $\Pi$ is the same as the representativity of the underlying undirected graph. Thus, if $\fw(\Gamma)>0$ then this allows us to use results for embedded undirected graphs which require high representativity. 
\end{remark}

    Note here that if $\Pi$ is an embedding of a bead-rooted digraph of representativity at least $2$ then $\Pi$ must be a $2$-cell embedding. Otherwise there would be an \emph{essential}---not nullhomotopic---curve contained in a face and therefore the representativity would be $0$.

The following is the main result of this section. In his thesis \cite[Theorem 9.2]{Johnson2002}, Johnson established a similar result for digraphs without roots. 
As we do need the variant for rooted digraphs, we include a proof here. 

\begin{theorem}\label{thm:high-rep}%
    Let $\Sigma$ be a surface, possibly with boundary, other than the plane or sphere. Let $\xi_1, \dots, \xi_k$ be the cuffs of $\Sigma$. 
    Let $\bar G_1 \coloneqq (G_1, \pi_1(W_{\omega_1}))$  and $\bar G_2 \coloneqq (G_2, \pi_2(W_{\omega_2}))$ be bead-rooted Eulerian digraphs with beads $W_1,W_2$ respectively, and let $\Pi_1 \coloneqq (\Gamma_1, \nu_1, \omega_1)$ and $\Pi_2 \coloneqq (\Gamma_2, \nu_2, \omega_2)$ be Eulerian embeddings of $\bar G_1$ and $\bar G_2$ in $\Sigma$.  

    For $1 \leq i \leq k$ and $j=1,2$ let $F^j_i$ be the set of edges $e \in E(G_i)$ with $\nu_i(e) \in \xi_i$. Let $F_1 \coloneqq \bigcup \{ F^1_i \sth 1 \leq i \leq k \}$ and $F_2 \coloneqq \bigcup \{ F^2_i \sth 1 \leq  i \leq k \}$. 
    
    Suppose that there is a bijection $\beta \sth F_1 \rightarrow F_2$ between $F_1$ and $F_2$ such that for all $1 \leq i \leq k$ and $e \in F_1$ we have $e \in F^1_i$ if and only if $\beta(e) \in F^2_i$ and for all $e, e' \in F_1$, $\pi_1(e) < \pi_1(e')$ if and only if $\pi_2(\beta(e)) < \pi_2(\beta(e'))$. Finally,  for all $e \in F_1$, $e$ is an in-edge if and only if $\beta(e)$ is an out-edge.

    Then there is a constant $c = c(G_1, F_1, \Sigma)$ such that if $\fw(\Pi_2) > c$ then $(G_1, \pi_1(W_1))$ has a strong rooted immersion into~$(G_2, \pi_2(W_2))$.
\end{theorem}

To prove the theorem we need some preparation.
We  need Theorem (9.1) in \cite{GMVII}. To this extent let $\Sigma$ be a surface.
Let $G_1$ and $G_2$ be undirected graphs and let  $(U_1, \nu_1)$ and 
$(U_2, \nu_2)$ be embeddings of $G_1$ and $G_2$, respectively, in $\Sigma$  
such that $\nu_1(V(G_1)) \cap \bd(\Sigma) = \nu_2(V(G_2)) \cap \bd(\Sigma)$ and for each $v \in V(G_1) \cap \nu_1^{-1}(\bd(\Sigma))$ we have $\nu_1(v) = \nu_2(v)$.  

We say that $G_1$ is a \emph{boundary rooted minor} of $G_2$ if there is a minor model (see \cite{GMVII} for a definition) $\phi$ of $G_1$ in $G_2$ such that $\phi(v)$ contains $v$ for all $v \in V(G_1) \cap \nu_1^{-1}(\bd(\Sigma))$.

\begin{theorem}[$ ${\cite[(9.1)]{GMVII}}]\label{thm:GMVII-9-1}
   Let $\Sigma$ be a surface other than the sphere, disc, or cylinder.
   Let $G_1$ and $G_2$ be undirected graphs and let  $(U_1, \nu_1)$ and $(U_2, \nu_2)$ be embeddings of $G_1$ and $G_2$, respectively, in $\Sigma$ such that $\nu_1(V(G_1)) \cap \bd(\Sigma) = \nu_2(V(G_2)) \cap \bd(\Sigma)$ and $\nu_1(v) = \nu_2(v)$ for each $v \in V(G_1) \cap \nu_1^{-1}(\bd(\Sigma))$.

    There is a constant $k = k_{\ref{thm:GMVII-9-1}}>0$ depending on $\Sigma$, $G_1$, and $(U_1, \nu_1)$ such that if $(U_2, \nu_2)$ has representativity at least $k$ then $G_1$ is a boundary rooted minor of $G_2$. 
\end{theorem}

We also need the next lemma, see e.g.~\cite[Proposition 5.5.10]{MoharT2001} or \cite{GMVII}.
\begin{lemma}[$ ${\cite[Proposition 5.5.10]{MoharT2001}}]\label{lem:isolated}
    Let $(U, \nu)$ be an embedding of an undirected graph $G$ in a surface $\Sigma$ of representativity $k \geq 2$. 
     Let $v \in V(G)$ be a vertex. Then there are $k' \coloneqq \lfloor \frac k2 \rfloor$ pairwise vertex disjoint cycles $C_1, \dots, C_{k'}$ induced by $O$-arcs $F_1, \dots, F_{k'}$ such that each $F_i$ bounds a disc $\Delta(F_i) \subseteq \Sigma$ containing $\nu(v)$ in its interior and such that for $j < i$, $C_j$ is contained in the interior of $\Delta(F_i)$.
\end{lemma}

We are now ready to prove \cref{thm:high-rep}.

\begin{proof}[of \cref{thm:high-rep}]
    Let $\Sigma, \zeta_1, \dots, \zeta_k, \bar G_1, \bar G_2$,  and $\beta$ be as in the statement of the theorem. For convenience we set $W_1 := W_{\omega_1}$ and $W_2 := W_{\omega_2}$. 

    Note that for every surface $\Sigma'$ there is a number $g(\Sigma')$ such that  any set of at least $g$ curves in $\Sigma'$ between the same endpoints which are disjoint except for their endpoints contains at least two curves which are homotopic in $\Sigma'$. (This follows from the well-known fact that there are only finitely many non-homotopic curves with fixed endpoints in any fixed surface). Let $g \coloneqq g(\Sigma)$.  

    We define $c := 2\cdot g \cdot w + k\cdot |F_1| \cdot g +  k_{\ref{thm:GMVII-9-1}}$.
    
    Let $H$ be the undirected graph obtained from $G_1$ as follows. 
    Let $v \in V(G_1) \setminus W_1$. Without loss of generality we assume that $\deg(v) = 4$. The case for $\deg(v)$ is analogous. Let $e_1 = (u_1, v), e_2 = (v, u_2), e_3 = (u_3,v)$, and $e_4 = (v, u_4)$ be the edges incident to $v$. 
    We replace $v$ and its incident edges by the following gadget $H_v$ (see \cref{fig:high-rep} for an illustration).
    
    For $1 \leq i \leq 5$ let
         \[
         C_i \coloneqq (l^i_1, \dots, l^i_g, a^i_1, ..., a^i_5,  t^i_1, \dots, t^i_g, a^i_6, ..., a^i_{10}, r^i_1, \dots, r^i_g, a^i_{11}, ..., a^i_{15}, b^i_1, \dots, b^i_g, a^i_{16}, ..., a^i_{20})
         \]
        be vertex-disjoint cycles on $4g+20$ vertices each. 
        $H_v$ is obtained from $\bigcup \{ C_i \sth 1 \leq i \leq 5 \}$ by adding the following edges: 
        for all $1 \leq i \leq 4$ and all $1 \leq j \leq g$ we add the edges $\{ l^i_j, l^{i+1}_j \}, \{ t^i_j, t^{i+1}_j \}, \{ t^i_j, t^{i+1}_j \}, \{ t^i_j, t^{i+1}_j \}$. Furthermore, 
        for all $1 \leq i \leq 4$ and all $1 \leq j \leq 20$ we add the edges $\{ a^i_j, a^{i+1}_j \}$.
    
     We refer to the sets $\{ t^5_1, \dots, t^5_g \}$, $\{ r^5_1, \dots, r^5_g \}$, $\{ b^5_1, \dots, b^5_g \}$, and $\{ l^5_1, \dots, l^5_g \}$ as the ports of $H_v$ and refer to them as the top, right, bottom, and left port, respectively.

    Now the graph $H$ is obtained from $G_1$ by replacing each vertex $v \in V(G_1) \setminus W_1$ by the graph $H_v$ and each edge $(u, v) \in E(G_1)$ by $g$ edges connecting one port of $H_u$ and one port of $H_v$ such that the embedding of $G_1$ can be extended to an embedding of $H$ in the obvious way. 
    Finally, for every $v \in W_1$ we duplicate the incident edge $g$ times and connect it to the port of its neighbour in the obvious way.
    \begin{figure}
        \centering
        \includegraphics[height=8cm]{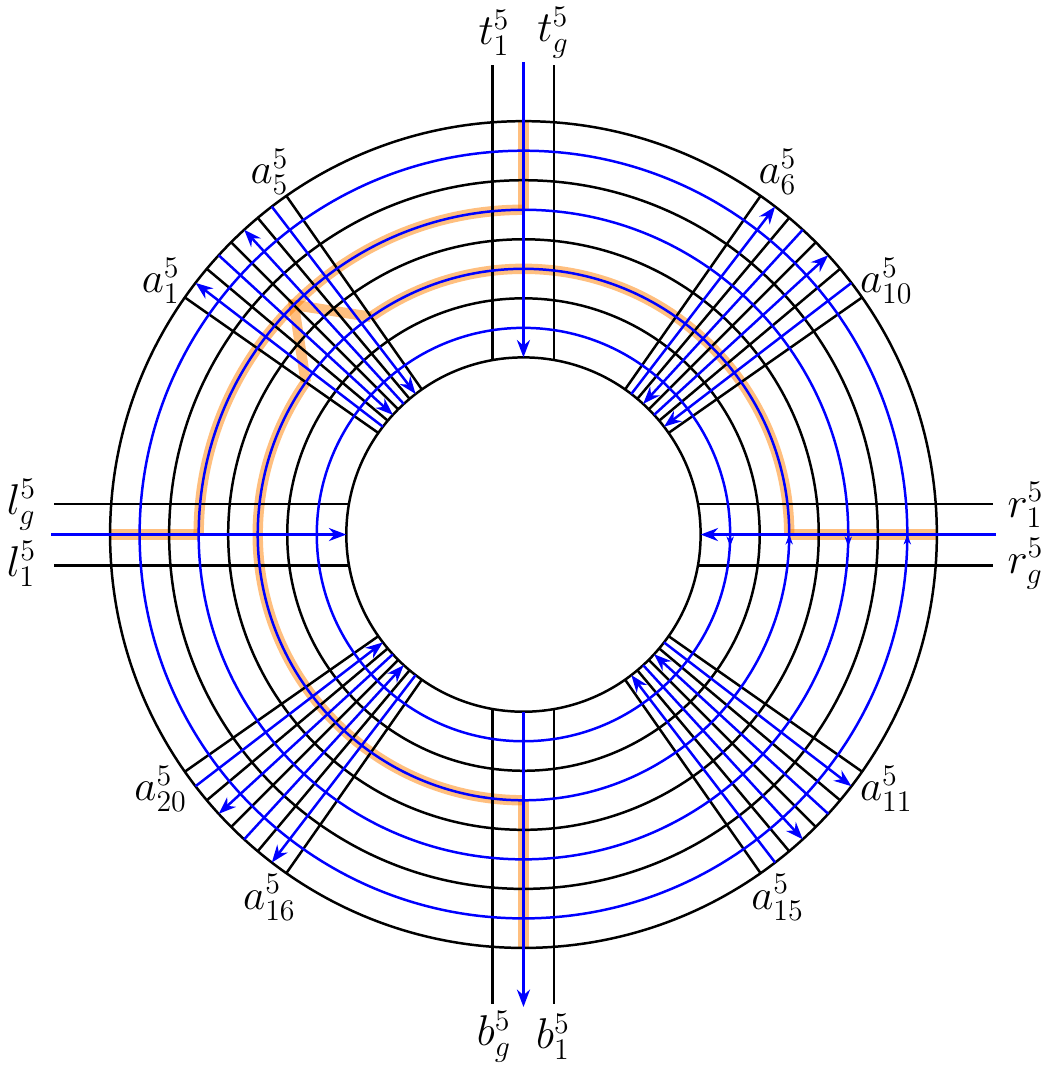}
        \caption{Gadget for the high representativity case. The $C_i$ are represented by the black circles. The paths $F_{u_1, v}, ...$ are marked in orange.}
        \label{fig:high-rep}
    \end{figure}
    
    We now consider $G_2$ and define two subgraphs $G_2^1$ and $G_2^2$ as follows. 
    Let $w = |W_2|$. %
    Let $w_i \coloneqq |F_i|$.
    We require that the representativity of $\Pi_2$ is at least $2\cdot g \cdot w$. Thus it follows from \cref{lem:isolated} that for each cuff $\zeta_i$ of $\Sigma$ there are $g\cdot w$
    pairwise vertex disjoint cycles $C_i^1 \dots C_i^{g \cdot w}$ such that each cycle bounds a disc $\Delta(C_i^j)$ in $\hat\Sigma$ that contains $\zeta_i$ in its interior and so that for $j < \ell$ the cycle $C_i^\ell$ is contained in the interior of $\Delta(C_i^j)$. 
    Furthermore, by definition of representativity, there are 
    $w_i$ pairwise edge-disjoint paths $P_i^1, \dots, P_i^{w_i}$ between $F_i$ %
    and $V(C_i^1)$.
    By the usual arguments as they are also applied in \cite{GMVII}, we may assume that 
    $P_i^j \cap C_i^{j'}$ is a linear path, for all $1 \leq j \leq w_i$ and $1 \leq j' \leq g \cdot w_i$. 
    It is easily seen that in $\bigcup_{1 \leq j \leq g\cdot w_i} C_i^j \cup \bigcup_{1 \leq j \leq w_i} P_i^j$ there is a cycle $C_i^\star$ of length at least $w_i\cdot g$ such that $C_i^\star$ bounds a disc $\Delta(C_i^\star)$ in $\hat\Sigma$ containing $\zeta_i$ in its interior and $w_i$ vertex disjoint paths between $W_i$ and $C_i^*$ (see \cref{fig:surface-high-reb-boundary}).%
    Furthermore, we can choose a set $S_i$ of $w_i\cdot g$ vertices on $C_i^\star$ and partition $S_i$ into $w_i$ parts of $g$ vertices each such that each $S_i$ induces a path in $C_i^\star$ not containing any vertex of $S_j$, $j\neq i$. Furthermore, this can be chosen such that there are $w_i$ pairwise vertex disjoint paths each linking one part of $S_j$ to a vertex in $W_i$.

    Let $G_2^1$ be the subgraph of $G_2$ induced by $\Delta(C_i^*)$ and let $G_2^2 \coloneqq G_2 - (V(G_2^1)\setminus V(C_i^\star))$.
    By construction, $G_2^1 \cap G_2^2 = C_i^\star$.

    Let $(U_2^2, \nu_2^2)$ be the embedding obtained from the restriction $\restr{(U_2, \nu_2)}{G_2^2}$ of the embedding of $G_2$ to $G_2^2$ by enlarging each cuff $\zeta_i$ so that it meets $(U_2^2, \nu_2^2)$ at $F_i$ (or, equivalently, move the edges in $F_i$ onto the cuff $\zeta_i$ in the obvious way). 

    Note that the representativity of the resulting undirected embedding $(U_2^2, \nu_2^2)$ is still at least $c' \coloneqq c - k\cdot |F| \cdot g$.
    We require that $c'$ is at least the constant $k_{\ref{thm:GMVII-9-1}}$ required for \cref{thm:GMVII-9-1}. Thus, by \cref{thm:GMVII-9-1} $H$ is a boundary rooted minor of $G_2^2$. Let $\phi$ be a boundary rooted model of $H$ in $G_2^2$. 

    We now construct a strong immersion of $H$ into $G_2^2$ as follows. 
    Each edge $(u, v) \in G_1$ is replaced in $H$ by $g$ paths between a port of $H_u$ and a port of $H_v$.  Let $P$ be such a path and let $P' \coloneqq \phi(P)$. It is easily seen that there is a $U_2$-normal curve between the endpoints of $P$ which intersects $U_2$ only at vertices of $\phi(P)$. 
    Thus, there must be two such curves $F^1_{u,v}$ and $F^2_{u,v}$ which are homotopic. By \cref{lem:vf-path-euler-paths}, there is a directed path $F_{u,v}$ from the inner port of $u$ to the inner port of $v$. 
    Similarly, in the model $\phi(H_v)$ of each gadget $H_v$ there are directed paths as indicated in \cref{fig:high-rep}. 
    We can now choose a vertex $v'$ in $H_v$ of degree $4$ and $4$ edge-disjoint paths connecting the four paths $F_{u_1,v}, ..., F_{v, u_4}$ to $v'$. We choose $v'$ as the image of $v$ in the immersion.

    What is left to do is to connect the vertices on $C_i^\star$ to the vertices of $G_1$ on the cuff $\zeta_i$, for each $1 \leq i \leq k$. 
    By construction, $G_2^1$ is embedded in the cylinder with boundary $C_1^\star$ and $\zeta_i$. By construction, if we choose one vertex $x_i^j$ from each part $S_i^j$ of $C_i^\star$ then there is a linkage of order $w_i$ in $G_2^1$ between $F_i$ and $\{ x_i^1, \ldots, x_i^{w_i} \}$.
    Thus we can apply \cref{lem:boundary_linked_Menger_for_embeddings_in_cylinder} to obtain a directed linkage of order $w_i$ linking $F_i$ to the set of endpoints of the paths representing the edges emerging from $C_i^*$. This concludes the proof. 
\end{proof}

\begin{figure}
    \centering
    \includegraphics[height=6cm]{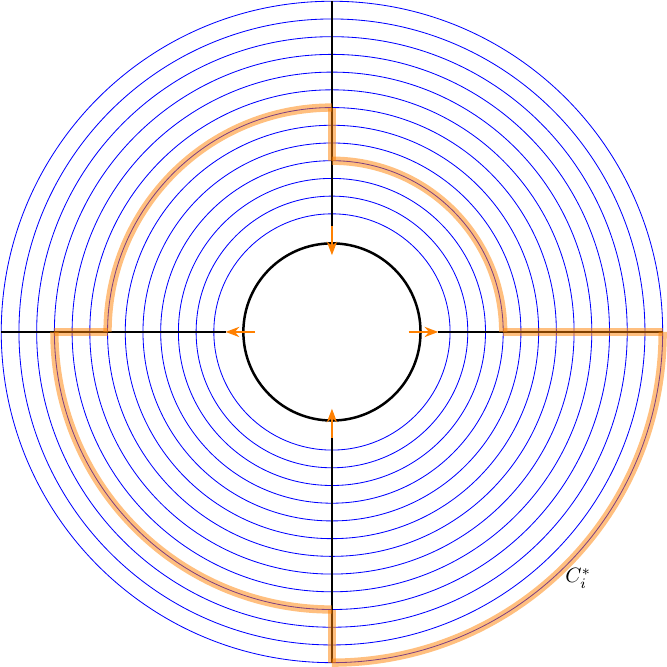}
    \caption{Construction of the cycle $C_i^*$ in \cref{thm:high-rep}}
    \label{fig:surface-high-reb-boundary}
\end{figure}

\section{Strong Immersions for the Disc} \label{sec:disc}
This section deals with the start of the grand induction over~$\Sigma$ and~$k$. In \cref{sec:high-rep} we discussed the case when our embeddings are of high representativity. If they do not have high representativity, we either find reducing curves along which we can cut to reduce the genus of~$\Sigma$ ultimately ending at the disc, or we find cuff separating curves along which we can cut out a cylinder of the surface reducing~$k$. In this section, we handle one of the base cases of the grand induction: the disc case.

\subsection{Euler-embedded Graphs in a Disc}
We start by proving the disc case of \cref{thm:wqo_on_surfaces_upto_beads}, i.e., using the notation of that theorem we prove the following. 

\begin{theorem}\label{thm:wqo:bead-root_for_disc}
 Let $\Delta$ be a disc and $k \in 2\N$. The class $\mathbf{G}(\Delta,k)$ is well-quasi-ordered by strong immersion.
\end{theorem}

The proof is by induction on~$k$. Note that for the disc case the partition $W_\omega$ of $W$ is always just $(W)$ and is thus uniquely defined; in particular $(G,\pi(W_\omega))$ can be written as $(G,\pi(W))$ and seen as a rooted Eulerian digraph in the usual sense, i.e., $\pi(W)$ is an order on $\rho(W)$ and the \cref{def:stitching_beadrooted} of stitching for bead-rooted graphs coincides with the standard \cref{def:stitching_std} of stitching and rooted immersions for bead-rooted graphs are simply rooted immersions. Thus we will write $(G,\pi(W))$ and vie bead-rooted graphs embedded in a disc as rooted graphs throughout this section.

\subsubsection{$\mathbf{G}(\Delta,0)$: Eulerian embedded Graphs without Roots} Recall that by \cref{obs:outline_of_plane_drawing_is_cycle} faces of Eulerian $2$-cell embeddings of graphs without cut-vertices are bounded by circles whence every edge is part of exactly two faces.

\begin{lemma}\label{thm:swirl_embedds_Eulerembeddable_graphs}
    Let $\Delta$ be a disc. There exists a function $f :\N \to \N$ such that the following holds. Let~$G \in \mathbf{G}(\Delta,0)$ be a connected Eulerian digraph of maximum degree $4$ and let~$(U,\nu)$ be an Eulerian embedding of~$G$ in~$\Delta$. Then $G \hookrightarrow \SSS_{f(n)}$ by strong immersion where $n \coloneqq \Abs{V(G)}$ .

\end{lemma} 
\begin{figure}
        \centering
        \includegraphics[height=6cm]{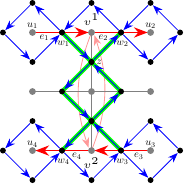}
        \caption{Illustration of the construction in the proof of \cref{thm:swirl_embedds_Eulerembeddable_graphs}.}
        \label{fig:swirl-emb}
    \end{figure}
\begin{proof}
    Let $G$ and $\Pi \coloneqq (U, \nu)$ be as in the statement of the lemma. Without loss of generality we may assume that $G$ is $4$-regular. 

    Let $G'$ be the Eulerian digraph obtained from $G$ by splitting each vertex as follows (see \cref{fig:swirl-emb} for an illustration of the construction). Let $v \in V(G)$ be a vertex with incident edges $e_1 = e_1(v), \dots, e_4 = e_4(v)$ ordered in clockwise orientation around $v$ given the embedding $\Pi$. As $\Pi$ is an Eulerian embedding, we may assume (up to a cyclic rotation) that $e_1$ and $e_3$ have $v$ as their head whereas $e_2$ and $e_4$ have $v$ as their tail. 
    We split $v$ into two new vertices $v^1, v^2$ with edges between them in both directions such that $v^1$ is incident to $e_1$ and $e_2$, and $v^2$ is incident to $e_3$ and $e_4$. See \cref{fig:swirl-emb} where the new edges are marked in light red.

    The Eulerian embedding $\Pi$ of $G$ can be modified to a Eulerian embedding $\Pi'$ of $G'$ in the obvious way. Note that in $\Pi'$, for all $v \in V(G)$, the edges $e_1(v)$, $e_2(v)$,  and the two new edges between $v^1, v^2$ appear in this order in clockwise orientation around $v^1$ and $e_3(v), e_4(v)$, and the two edges between $v^1, v^2$ appear in this order in clockwise orientation around $v^2$.
    
    Let $H$ be the undirected graph obtained from the underlying undirected graph of $G'$ by eliminating the duplicate edges between the copies $v^1$ and $v^2$ of vertices $v \in V(G)$. 

    Now $H$ is a planar graph of maximum degree $3$. Let $n \coloneqq |V(H)| = 2\cdot |V(G)|$. By \cite{RobertsonST1994}, $H$ is a minor of the undirected $12n \times 12n$ grid $\GGG \coloneqq \GGG_{12n\times 12n}$. As $H$ is a cubic graph, $H$ is a topological minor of $\GGG$. Let $\phi$ be a minor model of $H$ in $\GGG$ witnessing this. That is, $\phi$ maps each vertex $u \in V(H)$ to a vertex $\phi(u) \in V(\GGG)$ and each edge $e = \{ u,v \} \in E(H)$ to a linear path $\phi(e) \subseteq \GGG$ connecting $\phi(u)$ and $\phi(v)$. Furthermore, the images of distinct edges are vertex disjoint except possibly for a common endpoint. 

    Now consider the swirl $\SSS$ obtained from $\GGG$ as in \cref{def:swirl}. That is, $\SSS$ is obtained from $\GGG$ by replacing each vertex $v$ by a directed four-cycle $C_v$ oriented in anti-clockwise direction such that if $\{u,v\}$ is an edge in $\GGG$ then $C_v$ and $C_u$ share a vertex (see \cref{def:swirl} for details and \cref{fig:swirl} for an illustration). 

    By construction, every linear path $P \subseteq \GGG$ induces a sequence of four-cycles which forms a strongly connected subgraph of $\SSS$ in the obvious way.
    Let $v \in V(G)$ and let $x_1 = x_1(v) \coloneqq \phi(v^1)$ and $x_2 = x_2(v) \coloneqq \phi(v^2)$ using the above construction of $G'$ and $H$. 
     
    For $1 \leq i \leq 4$ let $P_i = P_i(v) \coloneqq \phi(e_i)$ and let $u_i = u_i(v)$ be the vertex on $P_i$ adjacent to $v$. Let $P_v \coloneqq \phi(\{ v^1, v^2 \})$. As there is an edge between $v$ and $u_i$ in $\GGG$, for each $1 \leq i \leq 4$, the cycles $C_{u_i}$ and $C_v$ share a vertex $w_i$. 
    Furthermore, let $y_1$ be the vertex on $P_v$ incident to $x_1$  and let $y_2$ be the vertex on $P_v$ incident to $x_2$. 
    Let $z_1 = z(v)$ be the vertex $C_{x_1}$ and $C_{y_1}$ have in common and let $z_2 = z_2(v)$ be the vertex $C_{y_2}$ and $C_{x_2}$ have in common (note that $z_1 = z_2$ is possible). 
    
    By construction of $\SSS$, each vertex was replaced by a four-cycle in anti-clockwise orientation. Thus, on $C_{v_1}$ the vertices $w_1, z_1, w_2$ appear in this order and on $C_{v_2}$ the vertices $w_3, z_2, w_4$ appear in this order. 
    This implies that in $\bigcup \{ C_a \sth a \in V(P_v) \}$ there are edge disjoint paths $Q_1(v), \dots, Q_4(v)$ such that $Q_1(v)$ and $Q_3(v)$ link $w_1$ and $w_3$ to $z_1$, respectively, and $Q_2(v)$ and $Q_4(v)$ link $z_1$ to $w_2$ and $w_4$, respectively. 
    
    We are now ready to construct the strong immersion $\iota$ of $G$ into $\SSS$. For each $v \in V(H)$, $\iota$ maps $v$ to $z_1(v)$. 
    Now let $e = (u,v) \in E(G)$ and let $1 \leq j_1, j_2
\leq 4$ be such that $e = e_{j_1}(u)$ and $e = e_{j_2}(v)$. 
Then we can choose a directed path $P_e$ in $\bigcup \{ C_a \sth a \in V(\phi(e)) \}$ from $w_{j_1}(u)$ to $w_{j_2}(v)$. $P_e$ together with $Q_{j_1}(u)$ and $Q_{j_2}(v)$ yields a path from $z_1(u)$ to $z_1(v)$ and for $e \not= e' \in E(G)$ these paths are edge-disjoint. 

Thus $\iota$ is a strong immersion of $G$ into $\SSS$. Setting $f(n) \coloneqq 12n$ completes the proof.

\end{proof}

We derive the base case for the inductive proof of \cref{thm:wqo:bead-root_for_disc}. 

\begin{theorem}\label{thm:wqo_planar}
        Let $\Delta$ be a disc. Then $\mathbf{G}(\Delta,0)$ is well-quasi-ordered by strong immersion.
\end{theorem}
\begin{proof}
    Assume the contrary and let $(G_i,\emptyset)_{i \in \N} \in \mathbf{G}(\Delta,0)$ be a bad sequence to the claim. Then for every $i \geq 1$ $G_i$ does not immerse $G_0$, and in particular by \cref{thm:swirl_embedds_Eulerembeddable_graphs} $G_i$ does not immerse an $f_{\ref{thm:swirl_embedds_Eulerembeddable_graphs}}(n)$-swirl for $n \coloneqq \Abs{V(G_0)}$. Thus by \cref{thm:4_reg_swirl} the carving width of $G_i$ is bounded for every $i \geq 1$ whence $(G_i,\emptyset)_{i \geq 1}$ is a sequence of ``root-less'' Eulerian digraphs of bounded carving-width. By \cref{thm:bounded_carving_with_non_labelled} there exist $1 \leq i < j$ such that $G_i \hookrightarrow G_j$; a contradiction.
\end{proof}

Although not needed we mention that the following stronger corollary follows analogously, again using that every digraph Euler-embeddable in the plane can be strongly immersed in a large enough swirl.

\begin{corollary}\label{cor:excluding_plane_graph_yields_chain}
    Let~$(G_i)_{i\in \N}$ be a sequence of Eulerian digraphs of maximum degree~$4$. If there exists a digraph~$H$ Euler-embeddable in the plane such that for every~$i \in \N$~$G_i$ does not immerse~$H$, then there exist~$j>i\geq 1$ such that~$G_i \hookrightarrow G_j$.
\end{corollary}

\subsubsection{$\mathbf{G}(\Delta,k)$: Eulerian Embedded Graphs with Roots} The induction step is split into two subcases: either we find a large swirl in our graph that is ``\emph{boundary-linked}'' in which case we use the structure to construct the required strong immersions similar to \cref{thm:swirl_embedds_Eulerembeddable_graphs}, or every large enough swirl can be cut off by some cut-cycle inducing a~$\leq k$-cut. We start with the latter.

While all the following results could be directly proved on Eulerian digraphs, we take a detour using the linegraph (see \cref{def:linegraph}) in order to derive the results from \cite{GMVIII}. We start by gathering the relevant results and definitions (all referenced definitions and results are transcribed to our setting). Whenever we talk about embedding of undirected graphs, we use the standard \cref{def:embedding} of embeddings, where we do not allow for boundary-edges.

\begin{definition}[$C$-ring, {\cite[Sec. 6]{GMVIII}}]\label{def:C-ring}
    Let~$G$ be an undirected graph embedded in a disc~$\Delta$ with cuff~$\zeta$, with embedding~$(U,\nu)$. Let~$C \subseteq G$ be a circle. 

    Define~$\lambda(C)$\Symbol{LAMBDAC@$\lambda(C)$} to be the maximum number of mutually vertex-disjoint paths of~$G$ between~$C$ and~$\zeta$. 
    Then a \emph{$C$-ring (in~$U$)} is an~$O$-arc~$F$ such that~$\alpha(F) = \lambda(C)$ and such that~$C$ is drawn inside~$\Delta(F)$, i.e,~$\nu(C) \subset \Delta(F)$. We call $F$ a \emph{minimal}~$C$-ring if there is no other~$C$-ring~$F'$ with~$I(F') \subsetneq I(F)$ being a strict subset.
\end{definition}
\begin{remark}
    Note that~$C$-rings exist since~$\lambda(C)$ equals the minimum of~$\Abs{U \cap F}$ over all~$O$-arcs~$F$ with~$E(C) \subseteq I(F)$ similar to \cref{cor:Menger_on_a_disc}.
\end{remark}

In a first step Robertson and Seymour \cite{GMVIII} proved the following topological result, which is used to separate rings of different circles.
\begin{lemma}[(6.5) in {\cite{GMVIII}}] \label{lem:6.5}
    Let~$(U,\nu)$ be an embedding of an undirected graph~$G$ in a disc~$\Delta$. Let~$C_1,C_2$ be circles in~$G$. Let~$F_i$ be a~$C_i$-ring (for~$i=1,2$) such that~$F_1$ and~$C_2$ bound disjoint open discs and~$C_2$ and~$F_1$ bound disjoint open discs. Then there is a~$C_1$-ring~$F_3$ and a~$C_2$-ring~$F_4$ such that~$F_1$ encloses~$F_3$,~$F_2$ encloses~$F_4$ and~$F_3,F_4$ bound disjoint open discs. 
\end{lemma}

They use the result to derive that minimal~$C$-rings are unique \cite{GMVIII}.

\begin{lemma}[(6.6) in {\cite{GMVIII}}] \label{lem:6.6}
    Let~$(U,\nu)$ be an embedding of an undirected graph~$G$ in a disc~$\Delta$ and let~$C \subseteq G$ be a circle. If~$F_1,F_2$ are minimal~$C$-rings then~$I(F_1) = I(F_2)$.
\end{lemma}

By \cref{lem:6.6} one may unambiguously define~$I(C) \coloneqq I(F_1)$ for a minimal~$C$-ring~$F_1$.

Recall that given an embedding of~$G$ in a disc~$\Delta$, every circle~$C$ in~$G$ bounds a unique disc~$\Delta(C) \subseteq \Delta$. 

\begin{definition}[Undirected $(r,s)$-nest and Boss, {\cite[Sec.6]{GMVIII}}]\label{def:boss}
    Let~$G$ be an undirected graph with embedding~$(U,\nu)$ in a disc~$\Delta$, and fix~$r,s \geq 0$. An \emph{undirected $(r,s)$-nest} in~$U$ is a sequence~$(C_1,\ldots,C_s)$ of vertex-disjoint undirected circles in~$G$ such that
    \begin{itemize}
        \item for~$1 \leq i < i' \leq s$, $\Delta(C_{i'}) \subseteq \Delta(C_i)$,
        \item there are~$r$ mutually vertex-disjoint linear paths~$P_1,\ldots,P_r$ of~$G$ between~$C_1$ and~$C_s$.
    \end{itemize}

    Let~$k \coloneqq |\nu(V(G)) \cap \bd(\Delta)|$. An undirected circle~$C$ in $G$ is a \emph{boss (of~$U$)} if~$\lambda(C) < k$ and there is an undirected~$(\lambda(C)+1,k)$-nest~$(C_1,\ldots,C_k)$ with~$C_k = C$. We may specify the witnessing linear paths writing~$(C_1,\ldots,C_s;P_1,\ldots,P_r)$.
\end{definition}
\begin{remark}
    Note that we slightly changed the definition for personal preferences: in \cite{GMVIII} they set~$\Delta(C_i) \subseteq \Delta(C_{i'})$ and~$C_1 = C$, i.e., the sequence~$(C_1,\ldots,C_s)$ induces a sequence of growing discs, while in this exposition the discs are shrinking.
\end{remark}

Clearly, all of the witnessing paths of the nest intersect all of the circles of the nest---Robertson and Seymour prove that one may choose the~$r$ vertex-disjoint paths to intersect each circle exactly in a sub-path---yielding the following, which we highlight as an observation.
\begin{observation}\label{obs:subseq_of_nest_is_nest_and_nice}
      Let~$G$ be an undirected graph with embedding~$(U,\nu)$ in a disc~$\Delta$ and let $r,s,t \in \N$. Let~$(C_1,\ldots,C_s)$ be an undirected~$(r,s)$-nest in~$U$. Then there exist~$r$ vertex-disjoint linear paths~$P_1,\ldots,P_r$ connecting~$C_1$ and~$C_s$ such that for every~$1 \leq i \leq r$ and every~$1 \leq j \leq s$,~$P_i \cap C_j$ is an induced linear subpath of~$P_i$ and~$C_j$ respectively. Further let~$1 \leq i_1 < \ldots < i_t \leq s$. Then~$(C_{i_1},\ldots,C_{i_t};P_1,\ldots,P_r)$ is an undirected~$(r,t)$-nest in~$U$.
\end{observation}

\begin{lemma}[(6.7) in {\cite{GMVIII}}]\label{lem:6.7}
    Let~$G$ be an undirected graph with embedding~$(U,\nu)$ in a disc~$\Delta$ and let~$C_1,C_2$ be bosses. Then either~$\Delta(C_1) \subseteq \Delta(C_2)$ or~$\Delta(C_2)\subseteq \Delta(C_1)$ or~$\Delta(C_1) \cap \Delta(C_2) = \emptyset$.
\end{lemma}

Robertson and Seymour \cite{GMVIII} use the previous lemmas to prove the existence of the following.

\begin{definition}[Foundation, {\cite[Sec. 6]{GMVIII}}]
    Let~$G$ be a graph,~$(U,\nu)$ an embedding of~$G$ in a disc~$\Delta$ and~$k \coloneqq \Abs{\nu(V(G)) \cap \bd(\Delta)}$. By \cref{lem:6.7} choose bosses~$C_1,\ldots,C_t$ of~$U$ such that
    \begin{enumerate}
        \item $I(C_1),\ldots,I(C_t)$ are mutually disjoint, and
    \item for every boss~$C$ of $U$,~$I(C)\subseteq I(C_i)$ for some~$1 \leq i \leq t$.
    \end{enumerate}
   For~$1 \leq i \leq t$ let~$F_i$ be a minimal~$C_i$-rings which by repeated application of \cref{lem:6.5} which can be chosen so that they bound mutually disjoint open discs and (by possibly shrinking the discs slightly)~$\Delta(F_i) \cap \Delta(F_j) \subset \nu(V(G))$ for distinct~$1 \leq i,j \leq t$. We call~$\{F_1,\ldots,F_t\}$ a \emph{foundation for~$U$}.
\end{definition}

\paragraph{Lifting Foundations from $\ell(G)$ to~$G$.} Having gathered the above results, we may use \cref{obs:properties_of_linegraph_embedding} to lift foundations of the underlying undirected graph of~$\ell(G)$ to~$G$ for an Euler-embeddable graph~$G$. To this extent we define \emph{$C$-sections} as an analogue to~$C$-rings in our setting. Note here that every embedding of~$\ell(G)$ is also an embedding of its underlying undirected graph. 

\begin{definition}[$C$-section]\label{def:C-section}
    Let~$(G,\pi(W))$ be a bead-rooted Eulerian digraph and let~$(\Gamma,\nu)$ be an embedding of~$(G,\pi(W))$ in a disc~$\Delta$ with cuff~$\zeta_1$. Let~$C \subset \Gamma$ be an internal circle. Using \cref{cor:Menger_on_a_disc} we define a~$C$-\emph{section} to be a cut-cycle~$F \subset \Delta$ such that~$\delta(F) = \kappa(\zeta_1,C;\Gamma)$ and~$\nu(C) \subset \Delta(F)$. We say that~$F$ is a \emph{minimal}~$C$-section (in~$\Gamma$) if there is no~$C$-section~$F' \subset \Delta$ such that~$X(F') \subsetneq X(F)$ is a strict subset.
\end{definition}

In particular note that for a~$C$-section it holds~$\delta(F) = \delta(X(F))$ and thus by \cref{lem:cut_cycles_are_alternating} we have the following.
\begin{observation}\label{obs:C-sec_is_alternating}
    If~$F$ is a~$C$-section then~$F$ is alternating.
\end{observation}

By \cref{obs:properties_of_linegraph_embedding} we immediatley derive an analogue to \cref{lem:6.5} for~$C$-sections.
\begin{corollary}\label{cor:disjoint-sections}
     Let~$(\Gamma,\nu)$ be an Eulerian embedding of a bead-rooted Eulerian digraph~$G$ in a disc~$\Delta$. Let~$C_1,C_2 \subset \Gamma$ be internal circles. Let~$F_i$ be a~$C_i$-section (for~$i=1,2$) such that~$F_1$ and~$C_2$ bound disjoint open discs and~$C_2$ and~$F_1$ bound disjoint open discs. Then there is a~$C_1$-section~$F_3$ and a~$C_2$-section~$F_4$ such that~$F_1$ encloses~$F_3$,~$F_2$ encloses~$F_4$ and~$F_3,F_4$ bound disjoint open discs. 
\end{corollary}
\begin{proof}
    This is immediate by switching to $\ell(\Gamma;\Delta)$ and~$(\Gamma_\ell,\nu_\ell)$, then applying (v) of \cref{obs:properties_of_linegraph_embedding} followed by \cref{lem:6.5} and finally (iv) of \cref{obs:properties_of_linegraph_embedding}.
\end{proof}

We prove a tight relation between~$C$-sections and~$\ell(C)$-rings, complementing (iv) and (v) of \cref{obs:properties_of_linegraph_embedding}.
\begin{lemma}\label{lem:rings_are_sections}
    Let~$(\Gamma,\nu)$ be an Eulerian embedding of a bead-rooted Eulerian graph~$(G,\pi(W))$ in a disc~$\Delta$ and let~$C \subset \Gamma$ be an internal circle. Let~$(\Gamma_\ell,\nu_\ell)$ be an embedding of the undirected underlying graph of $\ell(\Gamma;\Delta)$ induced by $\Gamma$, and let~$C_\ell \coloneqq \ell(C)$. 
    Let $F$ be a minimal~$C$-section, then there exists an equivalent minimal~$C$-section~$F'\subset \Delta$ in~$\Gamma$ such that~$F'$ is a minimal~$C_\ell$-ring in~$\Gamma_\ell$, and vice-versa.
\end{lemma}
\begin{proof}
    We prove one direction, where the other is analogous. Let~$F$ be a minimal~$C$-section in~$\Gamma$, then~$F$ is a cut-cycle and by (v) of \cref{obs:properties_of_linegraph_embedding} there exists an equivalent cut-cycle~$F'\subset \Delta$ in~$\Gamma$ such that~$F'$ is an~$O$-arc in~$\Gamma_\ell$. 
    \begin{claim}
        $F'$ is a minimal~$C$-section.
    \end{claim}
    \begin{claimproof}
        Since~$\delta(F') = \delta(F)$ as well as~$X(F) = X(F')$ using the equivalence of the cut-cycle, the claim follows by \cref{def:C-section} using that~$F$ is a minimal~$C$-section.
    \end{claimproof}
    Recall that by \cref{obs:circles_are_circles_in_linegraph}~$C_\ell$ is indeed a circle in~$\ell(\Gamma;\Delta)$ and thus an undirected circle in its underlying undirected graph. 
    \begin{claim}
    $F'$ is a minimal~$C_\ell$-ring.
    \end{claim}
\begin{claimproof}
     Assume towards a contradiction that~$F'$ is not a~$C_\ell$-ring. Then by \cref{def:C-ring} ~$\lambda(C_\ell) < \alpha(F')$, for it cannot be bigger. Thus there exists a~$C_\ell$-ring~$J' \subset \Delta$ with~$\alpha(J') < \alpha(F')$ and~$\nu(C_\ell) \subset \Delta(J')$ by the undirected well-known version of \cref{lem:boundary_linked_Menger_for_embeddings_in_cylinder}. Using~$(iv)$ of \cref{obs:properties_of_linegraph_embedding} we derive the existence of a cut-cycle~$J \subset \Delta$ with~$\delta(J) = \alpha(J') < \alpha(F') = \delta(F')$ and~$\nu(C) \subset \Delta(J)$ as a contradiction to~$F'$ being a~$C$-section.

    Finally to see that~$F'$ is minimal, suppose that there is an~$O$-arc~$F'' \subset \Delta$ with~$\alpha(F'') = \alpha(F')$ such that~$I(F'') \subsetneq I(F')$ is a strict subset. Then by~$(iv)$ of \cref{obs:properties_of_linegraph_embedding} we find a respective cut-cycle~$F''' \subset \Delta$ with~$\delta(F''') = \delta(F')$ and~$X(F''') \subset X(F')$ as a strict subset, contradicting the minimality of~$F'$ in~$\Gamma$. 
\end{claimproof}

    This concludes the proof.
\end{proof}

 The following is a consequence of \cref{lem:rings_are_sections}, and the dual result to \cref{lem:6.6}.
\begin{lemma}\label{lem:C-sections_unique}
    Let~$(\Gamma,\nu)$ be an Eulerian embedding of a bead-rooted Eulerian digraph~$(G,\pi(W))$ in a disc~$\Delta$ and let~$C \subset \Gamma$ be an internal circle. If~$F_1,F_2$ are minimal~$C$-sections then~$X(F_1) = X(F_2)$.
\end{lemma}
\begin{proof}
Let~$(\Gamma_i,\nu_i)$ be embeddings of the (up-stitched) bead-rooted graphs~$(G^{F_i},\pi(\overline{X(F_i)})$ in~$\Delta(F_i)$ as given by \cref{def:induced_embedding_of_bead-rooted_graph_in_cut-cycle} for~$i=1,2$. Since~$F_1,F_2$ are~$C$-sections we know that~$C \subset \Gamma_i$ is an internal cycle for~$i=1,2$.

Let~$(\Gamma_\ell,\nu_\ell)$ be the embedding of the underlying undirected graph of~$\ell(\Gamma;\Delta)$ induced by~$\Gamma$. Let~$C_\ell \coloneqq \ell(C) \subseteq \ell(\Gamma;\Delta)$. By~$(v)$ of \cref{obs:properties_of_linegraph_embedding} there exist~$O$-arcs~$F_1',F_2'$ in~$\ell(\Gamma;\Delta)$ with~$X(F_i) = I(F_i')$,~$\rho(F_i) = \alpha(F_i')$ and such that~$C_\ell\subset \ell(\Gamma;\Delta)[\Delta(F_i);\Gamma_\ell]$ for~$i=1,2$, where the latter is indeed a digraph by \cref{obs:graphs_inside_Oarcs}. 

By \cref{lem:rings_are_sections} we deduce that~ $F_1'$ and~$F_2'$ are minimal~$C_\ell$-rings, whence~$I(F_1') = I(F_2')$ and thus using~$X(F_i) = I(F_i')$ the claim follows.
\end{proof}

In particular, we may unambiguously define~$X(C) \coloneqq X(F)$ for any minimal~$C$-section~$F$.

\smallskip

While \cref{lem:rings_are_sections} lets it seem tempting to simply work with the ``undirected foundations'', we will need stronger assumptions on our foundations. In regard of \cref{obs:Foundation_pieces_induction_step} and recalling \cref{thm:wqo_bounded_carvingwidth_knitworks} it makes sense to impose that the links $\mathfrak{M}(\rho(X(F_i)))$ ``coming from'' the foundation are well-linked. 

We define the following.

\begin{definition}[Directed $(r,s)$-nest]
    Let~$(G,\pi(W))$ be a bead-rooted Eulerian digraph with embedding~$(\Gamma,\nu)$ in a disc~$\Delta$ with cuff~$\zeta_1$, and fix~$r,s \geq 0$. A directed \emph{$(r,s)$-nest} in~$\Gamma$ is a sequence~$(C_1,\ldots,C_s)$ of edge-disjoint internal circles of~$\Gamma$ such that
    \begin{itemize}
        \item for~$1 \leq i < i' \leq s$, $\Delta(C_{i'}) \subseteq \Delta(C_i)$,
        \item $C_i$ and~$C_{i+1}$ wind in opposite direction around~$\zeta_1$ (fixing an orientation of $\zeta_1)$,
        \item there are~$r$ mutually edge-disjoint paths~$(P_1,\ldots,P_r)$ of~$G$ between~$C_1$ and~$C_s$. 
    \end{itemize}

    Let~$k \coloneqq |\nu(E(G)) \cap \bd(\Delta)|$. A directed circle~$C$ in $G$ is a \emph{nucleus (of~$\Gamma$)} if~$\kappa(\zeta_1,C;\Gamma) < k$ and there is a directed~$(\kappa(\zeta_1,C;\Gamma)+1,k)$-nest~$(C_1,\ldots,C_k)$ with~$C_k = C$.
\end{definition}
\begin{remark}
   Usually directed nests have the additional assumptions that the paths are alternating~$(V(C_1),V(C_s))$ and~$(V(C_s),V(C_1))$-paths, which we immediately get due to the embedding restrictions and the Eulerianness of~$G$ (recall that any cut-cycle in a $2$-cell embedding is alternating \cref{lem:cut_cycles_are_alternating}).

   Note further that by \cref{thm:4_reg_swirl} large directed nests witness large swirls---for they wtiness large treewidth---and, in fact, it can be shown that the embedding restrictions yield a linear dependency. However, since we do not need this here, we omit the details.
\end{remark}

For simplicity will omit the word ``directed'' when talking about directed nests whenever it does not cause confusion. Similarly we may talk about undirected nests in a digraph $G$, meaning undirected nests in its underlying undirected graph.

Analogously to \cref{obs:subseq_of_nest_is_nest_and_nice} we have the following; the proof uses standard techniques and is straightforward and thus left to the reader (note that our nests use alternating cycles by definition and the linearity can be guaranteed using \cref{obs:linkage_gives_linear_linkage}).
\begin{observation}\label{obs:subseq_of_nest_is_nice_nest_dir}
      Let~$(G,\pi(W))$ be a bead-rooted Eulerian digraph with Eulerian embedding~$(\Gamma,\nu)$ in a disc~$\Delta$ and~$r,s,t \in \N$. Let $(C_1,\ldots,C_s)$ be a directed $(r,s)$-nest in $\Gamma$. Then there exist~$r$ mutually edge-disjoint paths~$P_1,\ldots,P_r$ such that for every~$1 \leq i \leq r$ and every~$1\leq j \leq s$, ~$P_i \cap C_j$ is a linear subpath in~$P_i$ and~$C_j$ respectively, and~$(C_1,\ldots,C_s;P_1,\ldots,P_r)$ is a directed~$(r,s)$-nest in~$\Gamma$. Further let~$1 \leq i_1 < \ldots < i_t \leq s$. Then~$(C_{i_1},\ldots,C_{i_t};P_1,\ldots,P_r)$ is a directed~$(r,t)$-nest in~$\Gamma$.
\end{observation}

Clearly every nucleus comes with a boss.
\begin{lemma}\label{lem:from_nucleus_to_boss}
    Let~$(G,\pi(W))$ be a bead-rooted Eulerian digraph with embedding~$(\Gamma,\nu)$ in a disc~$\Delta$. Let~$(\Gamma_\ell,\nu_\ell)$ be an embedding of the undirected underlying graph~$\ell(\Gamma;\Delta)$. Let~$C$ be a nucleus in~$\Gamma$, then omitting the directions,~$\ell(C)$ is a boss in~$\Gamma_{\ell}$.
\end{lemma}
\begin{proof}
    This follows at once from (v) of \cref{obs:properties_of_linegraph_embedding} together with \cref{lem:rings_are_sections}, noting that the directed nest witnessing that~$C$ is a nucleus transfers to an undirected nest witnessing that~$\ell(C)$ is a boss. 
\end{proof}

Combining the above results we derive the following analogue of \cref{lem:6.7} for nuclei.

\begin{corollary}\label{cor:nuclei_are_disjoint}
     Let~$(G,\pi(W))$ be a bead-rooted Eulerian digraph with embedding~$(\Gamma,\nu)$ in a disc~$\Delta$ and let~$C_1,C_2$ be nuclei. Then either~$X(C_1) \subseteq X(C_2)$ or~$X(C_2)\subseteq X(C_1)$ or~$X(C_1) \cap X(C_2) = \emptyset$.
\end{corollary}%
     
We conclude the following important result needed to strengthen our assumptions on nuclei.

\begin{lemma}\label{lem:nuclei_are_well-linked}
Let~$(\Gamma,\nu)$ be an Eulerian embedding of a bead-rooted Eulerian digraph $(G,\pi(W))$ in a disc~$\Delta$ with cuff~$\zeta_1$. Let~$k \coloneqq \Abs{\rho(W)}$. Let~$C_k$ be a directed circle in~$G$ such that~$\ell \coloneqq \kappa(\zeta_1,C_k;\Gamma) < k$. 

Let~$(C_1,\ldots,C_k,C_{k+1},\ldots,C_{3k+1})$ be a directed~$(\ell+1,3k+1)$-nest in~$\Gamma$.
Let~$F$ be a minimal~$C_k$-section. Then the link~$\mathfrak{M}\big(\rho(X(F))\big)$ is reliable and well-linked.
\end{lemma}
\begin{proof}
    Recall that for the disc we have $\rho(W) = \nu(E(G)) \cap \bd(\Delta)$ by \cref{def:bead-rooted_graph}.
    By \cref{obs:C-sec_is_alternating}~$F$ is alternating; let~$\pi_F \coloneqq (e_1,\ldots,e_\ell)$ be an alternating order of~$\rho(F) = \rho(X(F))$ with~$\ell \in 2\N$ by Eulerianness.

    \begin{claim}\label{claim:deeper_sections_well_connected}
     Let~$F'$ be a~$C_s$ section for any~$s\geq k$, then~$\delta(F') \geq \ell$. Furthermore~$F$ is a minimal~$C_s$ section.
    \end{claim}
    \begin{claimproof}
    The claim is clear for $s=k$ whence $s >k$. If~$F'$ is a~$C_k$-section then~$\delta(F') \geq \ell$ by assumption on the minimality of~$F$. Thus assume that~$F' \cap \Delta(C_k)^\circ \neq \emptyset$ where~$\Delta(C_k)^\circ$ denotes the interior of the disc. In particular~$F'$ is not a $C_k$-section which implies that either~$\overline{\Delta\setminus\Delta(C_k)} \cap F' \neq \emptyset$ or~$F' \subset \Delta(C_k)$. Note that $\Delta(C_s) \subset \Delta(F')$ by definition of $C_s$-section. For the latter, since $(C_k,\ldots,C_{3k+1})$ is a directed~$(\ell+1,2k+1)$-nest by \cref{obs:subseq_of_nest_is_nice_nest_dir}, a version of \cref{lem:boundary_linked_Menger_for_embeddings_in_cylinder} (taking~$C_k,C_{s}$ as cuffs) implies~$\delta(F') \geq \ell+1$. Analogously, we derive more generally that~$F' \not \subseteq \Delta(C_1)$.

    Thus~$F' \cap \Delta(C_1)^\circ \neq \emptyset$ and~$\overline{\Delta\setminus\Delta(C_k)} \cap F' \neq \emptyset$. But then~$F' \cap \nu(E(C_i)) \neq \emptyset$ for every~$1 \leq i \leq k$. since~$C_1, \ldots, C_k$ are edge-disjoint it follows~$\delta(F') \geq k \geq \ell+1$. Thus we proved that if~$F'$ is not a~$C_k$-section then~$\delta(F') \geq \ell+1$. But note that every~$C_k$-section~$F''$ is a~$C_s$-section---since $\Delta(C_s) \subseteq \Delta(C_k)$---witnessing~$\delta(F'') = \ell$. This together with \cref{lem:C-sections_unique} implies that~$F$ is a minimal~$C_s$-section.
    \end{claimproof}
    
 Recall that for a path~$L \in G$,~$V(L)$ denotes the set of internal vertices of the path (see \cref{def:paths}).
   By the \cref{claim:deeper_sections_well_connected} we derive that~$F$ is a minimal~$C_s$-section for every~$k \leq s \leq 3k+1$. In particular by \cref{lem:boundary_linked_Menger_for_embeddings_in_cylinder} there exists a~$\{\rho(F),\rho(C_{3k+1})\}$-linkage~$\LLL^\star$ of order~$\ell$ in~$G$ with~$\bigcup_{P \in \LLL}V(P) \subseteq X(F)$.  Using standard re-routing techniques, similar to \cref{obs:subseq_of_nest_is_nice_nest_dir}, we may assume that every path in~$\LLL^\star$ intersects the circles~$C_j$ for~$k \leq j \leq 3k+1$ exactly in linear subpaths. Hence, by \cref{obs:linkages_in_euler_embedded_cylinder_are_homotopic} the paths in~$\LLL^\star$ are pairwise homotopic in~$\Delta(F) \setminus \Delta(C_{3k+1})^\circ$ and moreover, setting $\Sigma_j \coloneqq \Delta(F)\setminus \Delta(C_j)^\circ$, the paths in~$\LLL_j^\star \coloneqq \restr{\LLL^\star}{\Sigma_j}$ are pairwise homotopic in~$\Sigma_j$ for every~$k \leq j \leq 3k+1$. 
   
   Finally let~$E_1^- \cup E_2^- = \rho^-(X(F))$ and~$E_1^+ \cup E_2^+ = \rho^+(X(F))$ be partitions such that~$\Abs{E_i^-} = \Abs{E_i^+}$ for~$i=1,2$. 
   \begin{claim}
       There is a~$\big(\rho^-(X(F)),\rho^+(X(F))\big)$-linkage~$\LLL$ in~$G$ such that~$\bigcup_{P \in \LLL}V(P) \subseteq X(F)$ and~$\tau(\LLL) = M_1 \cup M_2$ where~$M_i \in \operatorname{Match}(E_i^-,E_i^+)$ is a perfect matching for every~$i=1,2$.
   \end{claim}
   \begin{claimproof}
       The proof is by induction on~$k$. Recall that~$F$ is alternating with~$\pi_F=(e_1,\ldots,e_\ell)$ and~$\ell \in 2\N$. In particular if~$e_i \in \rho^-(X(F))$ then~$e_{i+1} \in \rho^+(X(F))$ for every~$1 \leq i \leq \ell$ where we set~$e_{\ell+1}\coloneqq e_1$ for notational convenience. For every~$1 \leq i \leq \ell$ let~$P_i \in \LLL^\star$ be the unique path with~$e_i$ as an end (either starting or ending in~$e_i$). 
       
       Since~$\Abs{E_i^-} = \Abs{E_i^+}$ for~$i=1,2$, the pigeonhole principle implies the existence of~$i \in \{1,2\}$ and~$1 \leq j \leq \ell$ such that~$e_j\in E_i^-$ and~$e_{j+1} \in E_i^+$ or~$e_{j} \in E_i^+$ and~$e_{j+1} \in E_i^-$. (To see this color the edges in~$E_1^- \cup E_1^+$ and~$E_2^- \cup E_2^+$ with two colors, then since~$\pi_F$ is alternating one color will appear twice consecutively; this is the desired pair.)
       
   Without loss of generality assume~$e_j\in E_i^-$ and~$e_{j+1} \in E_i^+$ for the other case is analogous using that the cycles~$(C_k,\ldots,C_{3k})$ have alternating orientation. Now~$P_j$ is a path starting in~$e_j$ and ending in an edge~$f_j \in E(C_{3k+1})$ and~$P_{j+1}$ is a path starting in~$f_{j+1} \in E(C_{3k+1})$ and ending in~$e_{j+1}$ where~$P_j,P_{j+1}$ are edge-disjoint and homotopic in~${\Delta(F) \setminus \Delta(C_{3k+1})^\circ}$. For~$q=1,2$ let~$P_j^q,P_{j+1}^q \in \LLL^*_{k+q}$ be subpaths of~$P_1,P_2$ respectively such that~$P_j^q,P_{j+1}^q$ have their ends~$e_j,e_j^q$ and $e_{j+1},e_{j+1}^q$ in~$\rho(\Delta(F))$ and~$E(C_{k+q})$ respectively. For~$q =1,2$ let~$Q_{k+q}$ be the subpath of~$C_{k+q}$ starting in~$e_j^q$ and ending in~$e_{j+1}^q$. Finally, since the paths in~$\LLL^\star_{k+1}$ as well as the paths in~$\LLL^\star_{k+2}$ are pairwise homotopic, and since~$C_{k+1},C_{k+2}$ have opposite direction either~$P_1' \coloneqq P_j^1\circ Q_{k+1} \circ P_{j+1}^1$ or~$P_2' \coloneqq P_j^2 \circ Q_{k+2} \circ P_{j+1}^2$ is an~$(e_j,e_{j+1})$-path edge-disjoint from~$\LLL^\star \setminus \{P_1,P_2\}$. Continue inductively by finding a new pair and restricting to~$(C_{k+3},\ldots,C_{3k})$.
    This construction results in the desired linkage, concluding the proof.
   \end{claimproof}

    Note that \cref{def:well-linked_links} asks for subsets~$E^- \subseteq \rho^-(X(F))$ and~$E^+ \subseteq \rho^+(X(F))$ with~$\Abs{E^-} = \Abs{E^+}$. But now any partitions of~$E^-,E^+$ satisfying the assumptions of \cref{def:well-linked_links} can be extended to partitions satisfying the assumptions of the above claim; this concludes the proof, where the reliability follows from the fact that every subset $\LLL'\subseteq \LLL^\star$ is again a linkage.
\end{proof}

In light of \cref{lem:nuclei_are_well-linked} we call a nucleus~$C$ \emph{bombastic} if there is a minimal~$C$-section~$F$ such that~$\mathfrak{M}\big(\rho(X(F))\big)$ is well-linked. We get a refined version of \cref{cor:nuclei_are_disjoint}.

\begin{corollary}\label{cor:bombastic_nuclei_are_disjoint}
    Let~$(G,\pi(W))$ be a bead-rooted Eulerian digraph with Eulerian embedding~$(\Gamma,\nu)$ in a disc~$\Delta$ and let~$C_1,C_2$ be bombastic nuclei in $\Gamma$. Then either~$X(C_1) \subseteq X(C_2)$ or~$X(C_2)\subseteq X(C_1)$ or~$X(C_1) \cap X(C_2) = \emptyset$.
\end{corollary}
\begin{proof}
    Since bombastic nuclei are nuclei, the claim follows from \cref{cor:nuclei_are_disjoint}.
\end{proof}

We gathered all the relevant results to define \emph{foundations} for bead-rooted Eulerian digraphs.

\begin{definition}[Foundation for~$(G,\pi(W))$]\label{def:foundation}
    Let~$(G,\pi(W))$ be a bead-rooted Eulerian digraph,~$(\Gamma,\nu)$ an Eulerian embedding of~$G$ in a disc~$\Delta$ and~$k \coloneqq \Abs{\rho(W)}$. By \cref{cor:bombastic_nuclei_are_disjoint} we may choose bombastic nuclei~$C_1,\ldots,C_t$ of~$\Gamma$ such that
    \begin{enumerate}
        \item $X(C_1),\ldots,X(C_t)$ are mutually disjoint, and
    \item for every bombastic nucleus~$C$ of $\Gamma$,~$X(C)\subseteq X(C_i)$ for some~$1 \leq i \leq t$.
    \end{enumerate}
   For~$1 \leq i \leq t$ let~$F_i$ be a minimal~$C_i$-section and, by repeated application of \cref{cor:disjoint-sections} we may assume that they bound mutually disjoint open discs and (by possibly shrinking the discs slightly)~$\Delta(F_i) \cap \Delta(F_j) \subset \nu(E(G))$ for distinct~$1 \leq i,j \leq t$. We call~$\{F_1,\ldots,F_t\}$ a \emph{foundation for~$\Gamma$}.
\end{definition}
Note that the foundation may be empty, for example if the graph does not admit bombastic nuclei (which is surely the case if its treewidth is way lower than its index, for large nests admit large treewidth).

\paragraph{From a foundation to a well-linked~$\Omega$-knitwork}

The following is straightforward by combining \cref{lem:vf-path-euler-paths} and \cref{lem:vf-path-euler-circle}.

\begin{lemma}\label{lem:from_undirected_to_directed_nest}
    Let~$(\Gamma,\nu)$ be an Eulerian embedding of a bead-rooted Eulerian digraph~$(G,\pi(W))$ in a disc~$\Delta$. Let~$F \subset \Delta$ be an $O$-arc in~$\Gamma$. Let~$r,s \in \N$ and let~$(C_1,\ldots,C_{2s+1})$ be an undirected~$(2r+1,2s+1)$ nest in the underlying undirected graph of~$G[\Delta(F);\Gamma]$. Then there is a directed~$(r,s)$-nest~$(C_1',\ldots,C_s')$ in~$G[\Delta(F);\Gamma]$ where $\Delta(C_1')\subseteq\Delta(C_1)$.
\end{lemma}
\begin{proof}
    Since~$F$ is an~$O$-arc,~$G[\Delta(F);\Gamma]$ is a graph by \cref{obs:graphs_inside_Oarcs}. Further, any two internally disjoint simple curves with common endpoints in a disc~$\Delta$ are homotopic. Now take any two disjoint paths~$P_1,P_2$ witnessing the undirected nest with ends in~$V(C_1)$ and~$V(C_{2s+1})$; in particular they are vertex-disjoint. Then~$\nu(P_1) \cup \nu(P_2)$ together with a segment of~$C_1$ and a segment of~$C_2$ connecting the endpoints of the paths form a disc, i.e., the paths are homotopic following \cref{def:homotopic} (where the circles here may be seen as the cuffs in that definition). Further, all of the circles of the undirected nest are pairwise vertex-disjoint. The claim follows by grouping the paths into ``consecutive'' pairs and applying \cref{lem:vf-path-euler-paths}---two vertex-disjoint homotopic paths witness a path in the vertex-face-graph---and \cref{lem:vf-path-euler-circle}. 
\end{proof}

Similarly, given an Eulerian digraph~$G$, an undirected nest in the linegraph $\ell(G)$ yields an undirected nest in~$G$. This follows from \cref{obs:linkages_stay_same_in_linegraph} and \cref{obs:circles_are_circles_in_linegraph} for undirected graphs; we leave the details of the proof to the reader and remark that for undirected graphs $G'$, the linegraph $\ell(G')$ is defined in the obvious way as usual.

\begin{observation}\label{obs:nest_in_linegraph_yield_directed_nest}
    Let $\Delta$ be a disc. Let~$(G,\pi(W))$ be a bead-rooted Eulerian digraph and let~$(\Gamma,\nu)$ be a respective embedding in~$\Delta$ inducing a linegraph~$\ell(\Gamma;\Delta)$ with respective embedding~$(\Gamma_\ell,\nu_\ell)$. Let~$r,s \in \N$. If there is an undirected~$(r,s)$-nest~$(C_1,\ldots,C_s)$ in~$\Gamma_\ell$ then there is an undirected~$(r,s)$-nest~$(C_1',\ldots,C_s')$ in~$\Gamma$ such that~$\ell(C_i') = C_i$ for every~$1 \leq i \leq s$.
\end{observation}

We combine \cref{lem:from_undirected_to_directed_nest} and \cref{obs:nest_in_linegraph_yield_directed_nest} to prove the following.

\begin{lemma}\label{lem:from_boss_to_nucleus}
   Let $\Delta$ be a disc with cuff $\zeta_1$. Let~$(\Gamma,\nu)$ be an Eulerian embedding of a bead-rooted Eulerian digraph~$(G,\pi(W))$ in~$\Delta$ inducing a linegraph~$\ell(\Gamma;\Delta)$ with respective embedding~$(\Gamma_\ell,\nu_\ell)$. Let $k \coloneqq \Abs{\rho(W)}$. Let~$r,s \in \N$ and let~$(C_1,\ldots,C_{3s+1})$ be an undirected~$(r,3s+1)$-nest in~$\Gamma_\ell$ and let~$C_s$ be a boss in~$\Gamma_\ell$. Then there is a directed nest~$(C_1',\ldots,C_s')$ in~$\Gamma$ such that~$\Delta(C_i') \subseteq \Delta(C_s)$ for every~$1 \leq i \leq s$ and such that~$C_s'$ is a nucleus of~$\Gamma$.
\end{lemma}
\begin{proof}
  Recall that~$k = \Abs{\rho(W)} = \Abs{\nu(E(G)) \cap \bd(\Delta)}$ by \cref{obs:bead-rooted_edges_on_cuffs}. Let~$F$ be a minimal~$C_s$-ring and let~$\ell \coloneqq \Abs{\alpha(F)}$. Combining the \cref{def:boss} of boss with $(ii)$ of \cref{obs:properties_of_linegraph_embedding} we deduce that $\ell < k$. By (iv) of the same observation, there is a cut-cycle~$F^\star$ in~$\Gamma$ with~$\delta(F^\star) = \alpha(F)$ and such that~$I(F) = X(F^\star)$ where moreover $\nu_\ell(C_s)\subset \Delta(F^*)$ and hence $\Delta(C_s) \subset \Delta(F^*)$ (note that $F^*$ is an $O$-arc and thus $\nu_\ell(E(C_s)) \subset \Delta(F)^* $. By \cref{obs:nest_in_linegraph_yield_directed_nest} we find a respective undirected nest~$(C_1'',\ldots,C_{3s+1}'')$ in~$G$ such that~$C_s = \ell(C_{s}'')$. By the above we derive that $\nu(C_s'') \subset \Delta(F^*)$.
  
  By \cref{obs:subseq_of_nest_is_nest_and_nice} $(C_s'',\ldots,C_{2s+1}'')$ is an undirected~$(r,2s+1)$ nest in~$\Gamma$. Finally by \cref{lem:from_undirected_to_directed_nest} there exists a directed~$(r,s)$-nest~$(C_1',\ldots,C_s')$ in~$G[\Delta(F);\Gamma]$ with $\Delta(C_1') \subseteq \Delta(C_s'')$ and hence $\Delta(C_s') \subseteq \Delta(C_s'')\subset \Delta(F^*)$. In particular we deduce that~$F^\star$ is a~$C_s'$-section: otherwise there is a cut-cycle $F'$ that is a $C_s'$-section of order $\delta(F') < \delta(F^*) = \ell$, and applying $(v)$ of \cref{obs:properties_of_linegraph_embedding} yields a contradiction to the choice of $F$.

  Now~$C_s'$ is clearly a nucleus in~$\Gamma$ for we have proven the existence of the directed nest. Again, if it were no nucleus, then any  minimal~$C_s'$-section~$F'$ with~$\Abs{\rho(F')} = \kappa(\zeta_1,C_{s'};\Gamma)$ satisfying $ \Abs{\rho(F')} \geq k$. But this is impossible since~$F^\star$ is a~$C_s'$-section of lower order.
\end{proof}

Combining the above and  \cref{obs:properties_of_linegraph_embedding}---note that since here $\Sigma$ is a disc this results in a rooted cut---we derive the following crucial observation that will allow us to apply induction on~$k$; recall the \cref{def:stitching_std} of stitches.

\begin{observation}\label{obs:Foundation_pieces_induction_step}
    Let~$\bar{G}=(G,\pi(W)) \in \mathbf{G}(\Delta,k)$ with Eulerian embedding~$(\Gamma,\nu)$ in~$\Delta$. Let~$(F_1,\ldots,F_t)$ be a foundation for~$\Gamma$ and fix an order~$\pi(X(F_i)) = \pi(\overline{X(F_i)})$ where $\overline{X(F_i)}$ is a rooted cut for every~$1 \leq i \leq t$. Fix~$i \in \{1,\ldots,k\}$. Then for~$(G_{\overline{X(F_i)}},\pi(X(F_i))) \coloneqq \stitch\big(\bar{G};\pi({X(F_i)})\big)$ it holds~$(G_{\overline{X(F_i)}},\pi(X(F_i))) \in \mathbf{G}(\Delta,k')$ for some~$k'<k$.
\end{observation}

It turns out that, similar to the undirected case, if the graph~$(G,\pi(W))$ immerses a large swirl and the immersed swirl can be ``cut off'' by a low-order cut-cycle given the embedding~$(\Gamma,\ell)$, then the embedding of the swirl is ``mostly'' contained in one of the bombastic nuclei of the foundation  as we will see later.

\begin{definition}[Boundary-linked nest]\label{def:boundary-linked_swirl}
    Let~$r,s \in \N$ and let~$(G,\pi(W))$ be a bead-rooted Eulerian digraph with an Eulerian embedding~$(\Gamma,\nu,\omega)$ in some surface~$\Sigma$. Let~$k \coloneqq \rho(W)$. Let~$(C_1,\ldots,C_s)$ be a directed~$(r,s)$-nest in~$\Gamma$. We call~$(C_1,\ldots,C_s)$ \emph{boundary-linked} if there is no~$C_s$-section~$F$ of order~$\Abs{\rho(F)} < k$.
\end{definition}

In particular, bombastic nuclei witness large nests that are \emph{not} boundary-linked, i.e., \emph{loosely-linked nests}. 
We continue with a definition of~$\Omega$-knitworks induced by bead-rooted Eulerian digraphs admitting a foundation.

\begin{definition}
    Let~$\bar G = (G,\pi(W)) \in \mathbf{G}(\Delta,k)$ be a bead-rooted Eulerian digraph with Eulerian embedding~$(\Gamma,\nu)$ in some disc~$\Delta$. Let~$\Omega$ be a well-quasi-order and let~$\dom(\mu) = \dom(\m) = \dom(\Phi) = \emptyset$. We call~$(\bar{G},\mu,\m,\Phi)$ the~$\Omega$-knitwork \emph{induced by~$\bar{G}$}.
\end{definition}
Recall \cref{obs:cut-cycle_induces_bead-rooted_subgraph}. The following definition uses the standard \cref{def:stitching_std} of stitching for rooted cuts (since $\Omega$-knitworks are defined on rooted and not bead-rooted digraphs).

\begin{definition}[$\Omega$-knitwork from a Foundation]\label{def:stitch_foundation}
    Let~$\bar{G} = (G,\pi(W)) \in \mathbf{G}(\Delta,k)$ be a bead-rooted Eulerian digraph with an Eulerian embedding~$(\Gamma,\nu)$ in some disc~$\Delta$ and let~$\FFF\coloneqq (F_1,\ldots,F_t)$ be a foundation for~$\Gamma$ for some $t \geq 1$. Let~$\Omega$ be a well-quasi-order. Let~$\tau_\FFF:\{1,\ldots,t\} \to V(\Omega)$ be some map and~$\Pi_\FFF \coloneqq\big(\pi(X(F_1)),\ldots,\pi(X(F_t))\big)$ a tuple of alternating orderings~$\pi(X(F_i))$ for every~$1 \leq i \leq t$, which exist by \cref{obs:C-sec_is_alternating}.
    
    Let~$\GGG_0 \coloneqq (\bar{G},\mu,\m,\Phi)$ be the~$\Omega$-knitwork induced by~$\bar{G}$. 
    We inductively define~$(\bar{G_i},\mu_i,\m_i,\Phi_i) \coloneqq \stitch\big(\GGG_{i-1};\pi(\overline{X(F_i}))\big)$ with down-stitch vertex~$s_i$ for every~$1 \leq i \leq t$. Finally, we define
    $\stitch(\GGG_0;(\FFF,\tau_\FFF, \Pi_\FFF)) \coloneqq (\bar G_t,\mu_t,\m_t,\Phi')$ \emph{with respective down-stitch vertices $(s_1,\ldots,s_t)$} where~$\Phi'(s_i) \coloneqq \tau(i)$ for every~$1 \leq i \leq t$ and is otherwise undefined.

    We further define $\stitch(\bar G; (\FFF,\Pi_{\FFF})) \coloneqq \bar G_t$ and call $(s_1,\ldots,s_t)$ the respective \emph{stitch vertices}.
\end{definition}
\begin{remark}
    Recall that for a foundation~$X(F_i) \cap X(F_j) = \emptyset$ for every~$1 \leq i,j \leq t$ whence the above consecutive stitches are well-defined.
\end{remark}

The following are straightforward from the \cref{def:stitching_knitwork,def:stitch_foundation} of stitches together with the \cref{def:foundation} of foundations.

\begin{observation}\label{obs:stitched_foundation}
    Using the notation of \cref{def:stitch_foundation}, given $(\bar G_t,\mu_t,\m_t,\Phi') =\stitch((\bar{G},\mu,\m,\Phi);(\FFF,\tau_\FFF,\Pi_\FFF))$, the following hold.
    \begin{enumerate}
        \item $\bar{G_t} = (G_t,\pi(W))$, in particular the roots remain the same,
        \item $\deg_{G_t}(s_i) < k$ for every $1 \leq i \leq t$,
        \item $\dom(\mu_t) = \dom(\m_t) = \{s_1,\ldots,s_t\}$ and $\dom(\Phi')=\{s_1,\ldots,s_t\}$,
        \item $\mu_t(s_i) = \pi_\FFF(\overline{X(F_i)})$ and~$\m_t(s_i) = \mathfrak{M}(\rho(X(F_i)))$ is reliable and well-linked for every~$1 \leq i \leq t$.
    \end{enumerate}
\end{observation}
\begin{remark}
     One easily verifies that the order in which the stitches are taken is irrelevant, and the resulting graphs and $\Omega$-knitworks are identical.
\end{remark}

We are finally able to prove the first big step towards \cref{thm:wqo:bead-root_for_disc}.

\begin{theorem}\label{thm:no_bad_seq_in_disc_loosely_linked}
    Let $\Delta$ be a disc. Let~$k\in \N$ be minimal such that $\mathbf{G}(\Delta,k)$ is not well-quasi-ordered by strong immersion. Let~$((G_i,\pi(W_i))_{i \in \N}$ be a bad sequence (with respect to strong immersion) of bead-rooted Eulerian digraphs where~$\bar G_i \in \mathbf{G}(\Delta,k)$ with respective Eulerian embedding~$(\Gamma_i,\nu_i,\omega_i)$ and foundations~$\FFF_i=(F_1^i\ldots,F_{t_i}^i)$ for some $t_i \in \N$ with respective $\Pi_{\FFF_i}$ for every~$i\in \N$. 
    Then for every function $h: \N \to \N$ there exists an infinite sequence $I \subseteq \N$  such that $\ebw{\stitch(\bar G_i; (\FFF_i,\Pi_{\FFF_i}))} \geq h(k)$ for every $i \in I$ .
\end{theorem}
\begin{proof}
    Note that $k\geq 1$ by \cref{thm:wqo_planar}.
    
    Assume the contrary towards a contradiction; let~$h:\N \to \N$ be the respective function. For every~$i \in \N$ fix some orderings~$\Pi_{\FFF_i} = \big(\pi(X(F_1^i),\ldots,\pi(X(F_{t_i}^i))\big)$, where $\overline{X(F_\ell^i)}$ is a rooted cut in $(G_i,\pi(W_i))$ for every $i \in \N$ and $1 \leq \ell \leq t_i$ by \cref{obs:cut-cycle_induces_bead-rooted_subgraph}.
    
    By assumption of the theorem, for every $k' <k$ the class $\mathbf{G}(\Delta,k')$ is well-quasi-ordered by strong immersion, and thus $\bigcup_{0 \leq \ell <k}\mathbf{G}(\Delta,\ell)$ is well-quasi-ordered by strong immersion as it is a finite union of well-quasi-ordered classes (every infinite sequence has an infinite subsequence where all elements come from the same class); denote the respective well-quasi-order by $\Omega=(V(\Omega),\preceq)$ for simplicity. In particular, for every $i\in \N$ and every $1 \leq j \leq t_i$ \cref{obs:Foundation_pieces_induction_step} implies that 
   $$(G_{\overline{X(F_j^i})},\pi(X(F_i^j)) \coloneqq \stitch(\bar{G_i};\pi({X(F_j^i)})) \in V(\Omega),$$ 
   call it $\eta^i_j \in V(\Omega)$ and define $\tau_i:\{1,\ldots,t_i\} \to V(\Omega)$ via $\tau_i(j) \coloneqq \eta_j^i$. As usual we fix~$\pi(X(F_j^i)) = \pi(\overline{X(F_j^i)})$.

    For every $i \in \N$ denote by $\GGG_i$ the $\Omega$-knitwork induced by $\bar G_i$. Define the $\Omega$-knitwork $$\GGG_{\FFF_i} = (\bar G_{\FFF_i}, \mu_{\FFF_i}, \m_{\FFF_i},\Phi_{\FFF_i})\coloneqq \stitch(\GGG_i; (\FFF_i,\tau_i, \Pi_{\FFF_i}))$$ with rooted Eulerian digraph $\GGG_{\FFF_i} = (G_{\FFF_i},\pi(W_i)$ and respective down-stitch vertices $(s_1^i\ldots,s_{t_i}^i)$ (see \cref{def:stitch_foundation} and \cref{obs:stitched_foundation}).

    Recall the \cref{def:types_of_linkages_on_a_cut} of types of feasible linkages. By \cref{def:stitch_foundation} and \cref{obs:stitched_foundation} the following hold:
    \begin{align}
        & \mu_{\FFF_p}(s_\ell^p) = \pi(X(F_\ell^p)) =  \pi(\overline{X(F_\ell^p)}),\label{claim:thm_no_bad_seq_in_loose_knitwork_properties_mu}\\
        & \m_{\FFF_p}(s_\ell^p) = \mathfrak{M}(\rho(X(F_\ell^p))),\label{claim:thm_no_bad_seq_in_loose_knitwork_properties_m} \\
        & \Phi_{\FFF_p}(s_\ell^p) = \eta_\ell^p,\label{claim:thm_no_bad_seq_in_loose_knitwork_properties_phi}
    \end{align}
    for every $1 \leq \ell \leq t_p$ and every $p\in \N$. 

    By \cref{def:foundation} of foundations, the $\Omega$-knitwork~$\GGG_{\FFF_i}$ is well-linked for every~$i \in \N$. Further, by our assumption we have $\ebw{G_{\FFF_i}} < h(k)$ for all but finitely many $i \in \N$; without loss of generality let it be all $i \in \N$ after possibly switching to the respective infinite subsequence. Thus, by \cref{thm:wqo_bounded_carvingwidth_knitworks} there exist~$i,j \in \N$ with~$i<j$ such that~$\gamma_0: \GGG_{\FFF_i} \hookrightarrow \GGG_{\FFF_j}$ is a strong~$\Omega$-knitwork immersion. We proceed with an inductive argument lifting $\gamma_0$ to a strong rooted immersion $\gamma:(G_i,\pi(W_i)) \hookrightarrow (G_j,\pi(W_j))$.
    \smallskip

    We start with a few observations for $\gamma_0$; recall \cref{def:matching_of_gamma}. Using the fact that~$\operatorname{dom}(\mu_{\GGG_{\FFF_\ell}}) = \{s_1^\ell,\ldots,s_{t_\ell}^\ell\}$ for every~$\ell \in \N$ by 3. of \cref{obs:stitched_foundation}, the \cref{def:knitwork_immersion} of strong~$\Omega$-knitwork immersion implies the following.
        \begin{enumerate}[i)]
            \item $\gamma_0$ induces an injection~$\{s_1^i,\ldots,s_{t_i}^i\} \to \{s_1^j,\ldots,s_{t_j}^j\}$; let~$t \coloneqq t_i$ and without loss of generality by possibly switching to an infinite subsequence, assume~$\gamma_0(s_\ell^i) = s_\ell^j$ for every~$1 \leq \ell \leq t$,
            \item $\gamma_0(V(G_{\FFF_i})) \cap \{s_1^j,\ldots,s_{t_j}^j\} = \{s_1^j,\ldots,s_t^j\}$,
            \item for~$1 \leq \ell \leq t$, $\Abs{\rho_{\GGG_{\FFF_i}}(s^i_\ell)} = \Abs{\rho_{\GGG_{\FFF_j}}(s^j_\ell)}$ and for~$\mu_{\GGG_{\FFF_i}}(s^i_\ell) = (e^i_1,\ldots,e^i_{k'})$ and~$\mu_{\GGG_{\FFF_j}}(s^j_\ell) = (e^j_1,\ldots,e^j_{k'})$ it holds that~$e^j_r$ is contained in~$\eta(e^i_r)$ for every~$1 \leq r \leq k'$ and respective~$k'<k$,
            \item for every~$t < \ell \leq t_j$,~$M_{\gamma_0}(s_\ell^j) \in \m_{\FFF_j}(s^j_\ell)$
            \item for every~$1 \leq \ell \leq t$,~$\Phi_{\GGG_{\FFF_i}}(s^i_\ell) \preceq \Phi_{\GGG_{\FFF_j}}(s^j_\ell)$.
        \end{enumerate}

Property iv) together with \cref{claim:thm_no_bad_seq_in_loose_knitwork_properties_m} (for $p=j$) implies the following (recall \cref{def:types_of_linkages_on_a_cut}).

\begin{claim}\label{claim:thm_no_bad_seq_in_loose_linkages_witnessing_immersion_for_leftovers}
    Let $t  < \ell \leq t_j$. There exists $\LLL_\ell^j \in \LLL(\rho(X(F_\ell^j)))$ such that $M_{\gamma_0}(s_\ell ^j) = \tau(\LLL_\ell ^j)$. 
\end{claim}

In turn \cref{claim:thm_no_bad_seq_in_loose_linkages_witnessing_immersion_for_leftovers} yields the following.

\begin{claim}\label{claim:thm_no_bad_seq_in_loose_induction_step_leftovers}
    Let $t  < \ell \leq t_j$. There exists a strong $\Omega$-knitwork immersion $\gamma': \GGG_{\FFF_i} \hookrightarrow \knit\big(\GGG_{\FFF_j}, \GGG_{\overline{X(F_\ell^j)}}; \pi(X(F_\ell ^j)) \big)$ such that for every $v \in V(G_{\FFF_i})$, $\gamma_0(v) = \gamma'(v)$ and for every $e \in E(G_{\FFF_i})$, $E(\gamma_0(e)) = E(\gamma'(e)) \setminus E(G_j[X(F_j)])$.
\end{claim}
\begin{claimproof}
By construction $G_{\overline{X(F_\ell^j)}} \cap G_{\FFF_j} = \rho(X(F_\ell^j))$. Let $G' \coloneqq \knit\big(G_{\FFF_j}, G_{\overline{X(F_\ell^j)}}; \pi(X(F_\ell ^j)) \big)$ be the respective graph obtained in the claim. By \cref{def:knitting_knitworks} of knitting, $V(G_{\FFF_j})\setminus \{s_\ell^j\} \subset V(G')$ as well as $E(G_{\FFF_j}) \subseteq E(G')$. In essence, the vertex $s_\ell^j$ was replaced with $G_{\overline{X(F_\ell^j)}}$.

 Let $M_{\gamma_0}(s_\ell ^j) = \{(e_z,e_z')\mid 1\leq z \leq r\}$ for respective $r \in \N$ and let $E(s_\ell^j) \coloneqq \{e_z,e_z' \mid 1 \leq z \le r\}$. By \cref{claim:thm_no_bad_seq_in_loose_linkages_witnessing_immersion_for_leftovers} there exists a linkage $\LLL_\ell^j \in \LLL(\rho(X(F_\ell^j)))$ in $G_{\overline{X(F_\ell^j)}}$ such that $\tau(\LLL_\ell ^j) = M_{\gamma_0}(s_\ell ^j)$. Thus for every $1 \leq z \leq r$ there exists an $(e_z,e_z')$-path $L_z$ in $\LLL_\ell ^j$. Let $e^z \in E(G_{\FFF_i})$ be the unique edge such that $(e_z,e_z') \subseteq \gamma_0(e^z)=(f_1,\ldots,e_z,e_z',\ldots,f_q)$ for some $q \in \N$. Let $P_z =(f_1,\ldots,e_z),P_z' =(e_z',\ldots,f_q)$, then $P_z,P_z'$ are paths in $G'$. Define $P^z \coloneqq P_z \circ L_z \circ P_z'$, then $P^z$ is a path in $G'$ by \cref{obs:concat_paths} for the three summands are pairwise edge-disjoint (this is clear for $P_z,P_z'$ and follows for $L_z$ by 2. of \cref{obs:knitting_fundamentals}). Finally define $\gamma'$ via
 \begin{align*}
     & \gamma'(v) \coloneqq \gamma_0(v), && \text{ for all } v \in V(G_{\FFF_i}),\\
     & \gamma'(e) \coloneqq \gamma_0(e), && \text{ for all } e \in E(G_{\FFF_i})\setminus E(s_\ell^j),\\
     & \gamma'(e^z) \coloneqq P^z,&& \text{ for every } 1 \leq z \leq r.
 \end{align*}
 Then $\gamma'$ is a strong $\Omega$-knitwork immersion, (the remaining verification is easy and left to the reader) satisfying the assumptions of the claim.
\end{claimproof}

Property v) together with the fact that $\big(G_{\overline{X(F_\ell^p})},\pi(X(F_\ell^p))\big) = \eta_\ell^p \in V(\Omega)$ for every $1 \leq \ell \leq t_p$ and $p \in \N$ implies the following.

\begin{claim}\label{claim:thm_no_bad_seq_in_loose_linkages_witnessing_immersion_for_pieces}
    For every~$1 \leq \ell \leq t$, there is a strong rooted immersion $$\eta_\ell:\big(G_{\overline{X(F^i_\ell)}},\pi(X(F^i_\ell))\big) \hookrightarrow \big(G_{\overline{X(F^j_\ell)}},\pi(X(F^j_\ell))\big).$$
\end{claim}

For every $1< \ell \leq t$, let $\GGG_{\overline{X(F_\ell^j)}}$ be the $\Omega$-knitwork induced by $ (G_{\overline{X(F_\ell^j)}},\pi(X(F_\ell^j))$; note that $\eta_\ell$ induces a respective strong $\Omega$-knitwork immersion. 

\begin{claim}\label{claim:thm_no_bad_seq_in_loose_induction_step_pieces}
    Let $1  < \ell \leq t$. There there exists a strong $\Omega$-knitwork immersion $\gamma': \knit\big(\GGG_{\FFF_i}, \GGG_{\overline{X(F_\ell^i)}}; \pi(X(F_\ell ^i)) \big) \hookrightarrow \knit\big(\GGG_{\FFF_j}, \GGG_{\overline{X(F_\ell^j)}}; \pi(X(F_\ell ^j)) \big)$ such that for every $v \in \overline{X(F_\ell^i)}$, $\gamma'(v) = \gamma_0(v)$, for every $v \in {X(F_\ell^i)}$, $\gamma'(v) = \eta_\ell(v)$ and for every $e \in E(G')$ with both endpoints in $\overline{X(F_\ell^i)}$, $\gamma'(e) = \gamma_0(e)$, for every $e \in E(G')$ with both endpoints in ${X(F_\ell^i)}$, $\gamma'(e) = \eta_\ell(e)$.
\end{claim}
\begin{claimproof}
    This follows immediately from \cref{thm:knitting_knitwork_immersion} with \cref{lem:stitch-and-knit} where $\gamma_d$ is $\gamma_0$ and $\gamma_u$ is $\eta_\ell$.
\end{claimproof}

Finally we inductively apply \cref{claim:thm_no_bad_seq_in_loose_induction_step_pieces} for every $1 \leq \ell \leq t$ and subsequently inductively apply \cref{claim:thm_no_bad_seq_in_loose_induction_step_leftovers} for every $t < \ell \leq t_j$ to conclude the following.

    \begin{claim}
        There is a strong rooted immersion $\gamma:(G_i,\pi(W_i)) \hookrightarrow (G_j,\pi(W_j))$.
    \end{claim}
    \begin{claimproof}
            By the iterative \cref{def:stitch_foundation} of $\GGG_{\FFF_i}$ via stitching, knitting back all the $1 \leq \ell \leq t$ pieces $(G_{\overline X(F_\ell^i)},\pi(X(F_\ell^i))$ iteratively, an iterative application of \cref{lem:stitch-and-knit} implies that the resulting graph is $\bar G_i$ again. Similarly for $\GGG_{\FFF_j}$, by the second iteration for all $t < \ell \leq t_j$ the resulting rooted Eulerian digraph is $\bar G_j$. Note that the order of taking stitches and knitting back the stitches does not matter since $X(F_\ell ^p) \cap X(F_{\ell'}^p) = \emptyset$ for all $1 \leq \ell\neq \ell' \leq t_p$ and $p \in \{i,j\}$.

            Finally, \cref{claim:thm_no_bad_seq_in_loose_induction_step_pieces} and \cref{claim:thm_no_bad_seq_in_loose_induction_step_leftovers} guarantee that the resulting map $\gamma$ is a strong rooted immersion.
    \end{claimproof}

    This concludes the proof of the theorem.
\end{proof}

\paragraph{Dealing with boundary-linked nests.}

We prove now, that the contrary to \cref{thm:no_bad_seq_in_disc_loosely_linked} must hold, refuting the existence of a bad sequence. To this extent we prove the following.

\begin{lemma}\label{lem:stitched_foundation_bounded_tw}
    Let $\Sigma$ be a surface and $k \in 2\N$. Let $\Omega$ be a well-quasi-order and $f: \N \to \N$ a function. There is a function $h_f:\N \to \N$ such that the following holds. Let~$\bar{G} = (G,\pi(W)) \in \mathbf{G}(\Delta,k)$ be a bead-rooted Eulerian digraph with an Eulerian embedding~$(\Gamma,\nu)$ in some disc~$\Delta$ and let~$\FFF\coloneqq (F_1,\ldots,F_t)$ be a foundation for~$\Gamma$ and~$\Pi_\FFF=\big(\pi(X(F_1)),\ldots,\pi(X(F_t))\big)$ be some choice of orderings for the respective cuts. Let $(G_\FFF,\pi(W)) = \stitch(\bar{G}; (\FFF,\Pi_\FFF))$ with respective stitch vertices $(s_1,\ldots,s_t)$.  Then $\ebw{G_\FFF} \geq h_f(k)$ if and only if $\Gamma$ immerses a boundary-linked directed $(f(k),f(k))$-nest.
\end{lemma}
\begin{proof}
   Choose $f(k)=3k+1$ and $g_f=f_{\ref{thm:4_reg_swirl}}$ and $h_f(k)=k \cdot g_f(k)$.
    Let $C_1,\ldots,C_t$ be the respective cycles of the foundation $\FFF = (F_1,\ldots,F_t)$, i.e., $F_i$ is a $C_i$-section for $1 \leq i \leq t$, in particular $\Delta(C_i) \cap \Delta(C_j) = \emptyset$ for distinct $1 \leq i,j \leq t$.
    We construct an auxiliary graph $G'$ from $G$ as follows. For each $1\leq i \leq t$, split off all the vertices embedded inside $\Delta(C_i)$ (which effectively boils down to deleting said vertices), then $C_i \subset G'$, and it bounds a face in the respective Eulerian embedding. Recall that by \cref{obs:immersion_robust_under_splitting_off} and \cref{obs:eulerian_emb_closed_under_splitting_and_immersion} $G' \hookrightarrow G$ and in particular $G'$ is Eulerian. 
    Since $F_i$ is a $C_i$-section we deduce that $G_\FFF^-\coloneqq G_\FFF-\{s_1,\ldots,s_t\} \subset G'$, where $G_\FFF-\{s_1,\ldots,s_t\}$ is the (possibly non-Eulerian) digraph obtained by deleting the vertices $s_1,\ldots,s_t$.

   \begin{claim}
       $\ebw{G'}< g_f(k)$.
   \end{claim}
   \begin{claimproof}
       Assume that $\ebw{G'} \geq g_f(k)$, then by \cref{thm:4_reg_swirl} this implies the existence of an $f(k)$-swirl and in particular a directed $(f(k),f(k))$-nest $(C_1',\ldots,C_{f(k)}')$ in $G'$ (using its Eulerian embedding). If the nest $(C_1',\ldots,C_{f(k)}')$ is boundary-linked we are done since $G' \subset G$, thus assume it is not. In particular then $C_{k}',\ldots,C_{k+1}'$ are nuclei by \cref{def:foundation} and since $f(k) = 3k+2$, we derive that $C_k',C_{k+1}'$ are both bombastic by \cref{lem:nuclei_are_well-linked}. By \cref{def:foundation} $C_k' = C_i$ and $C_{k+1}' = C_j$ for distinct $1 \leq i,j \leq t$. This is a contradiction to $\Delta(C_i) \cap \Delta(C_j) = \emptyset$ whence the claim follows. 
   \end{claimproof}

   In particular the claim implies that $\ebw{G_\FFF^-} \leq g_f(k)$ since carving width is closed under deleting vertices. In particular any carving $(T,\ell)$ of $G'$ is also a carving for $G_\FFF'$ by assigning $\ell(t_i) = \emptyset$ if $t_i \in V(G') \setminus V(G_\FFF')$; let $(T,\ell)$ be the respective carving witnessing width at most $ g_f(k)$.  
   
   We next prove that adding back $s_1,\ldots,s_t$ to the graph does not change the carving-width by too much. To this extent define the $\operatorname{span}(C_i) \subseteq T$ to be the minimum subtree of $T$ such that $V(C_i) \subseteq \ell(\operatorname{span}(C_i))$. 

   \begin{claim}\label{claim:cw_of_GF}
       Let $e \in E(T)$ and let $J(t) \subseteq \{1,\ldots,t\}$ be the set of indices with $e \in E(\operatorname{span}(C_j))$ if and only if $j \in J(t)$. Then $\Abs{J(t)} \leq g_f(k)$.
   \end{claim}
   \begin{claimproof}
       Since $C_1,\ldots,C_t$ are circles, they are in particular connected and thus for every $f \in  E(\operatorname{span}(C_j))$ at least two edges of $E(C_j)$ count towards $\w(f)$ for every $1 \leq j \leq t$. In particular then $\w(e) \geq \Abs{J(t)}$; the claim follows.
   \end{claimproof}

   Finally we may construct a carving $(T',\ell')$ for $G_\FFF$ from $(T,\ell)$ by adding vertices $x_1,\ldots,x_t$ to $T$ and for every $1 \leq i \leq t$ subdividing any edge $e_i \in E(\operatorname{span}(C_i))$ by introducing a new vertex $y_i$ and connecting $x_i$ to $y_i$. Then define $\ell'(x_i) = s_i$ for every $1 \leq i \leq t$ and otherwise equal to $\ell$. Since $\deg(s_i) < k$ for every $1 \leq i \leq t$, by \cref{claim:cw_of_GF} we deduce that $\w(e) < k\cdot g_f(k)$ for every $e \in E(T')$; the lemma follows.
\end{proof}

For the next theorem we need the following extension of \cref{thm:swirl_embedds_Eulerembeddable_graphs}.
\begin{figure}
    \centering
    \includegraphics[height=4cm]{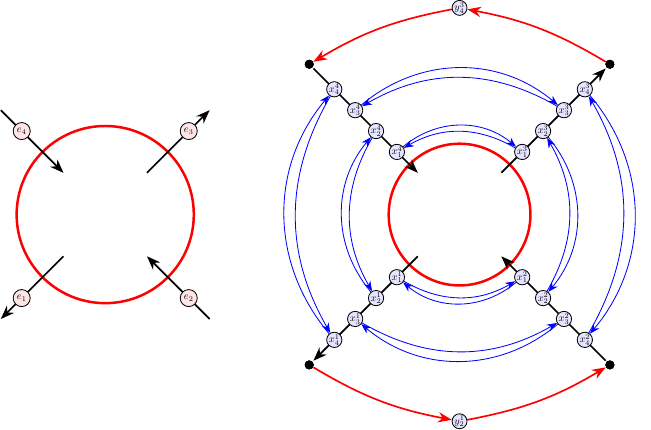}
    \caption{A schematic construction of $G^+$ in \cref{lem:disc-k-immerse-in-swirl} with respective embedding.}
    \label{fig:G^+_swirl_around_roots}
\end{figure}
\begin{lemma}\label{lem:disc-k-immerse-in-swirl}
    Let $\Delta$ be a disc and $k \in 2\N$. Let $(G,\pi(W)), (G',\pi(W')) \in \mathbf{G}(\Delta,k)$. Let $n \coloneqq \Abs{V(G)}$. Let $\pi(W)=(e_1,\ldots,e_k)$ and $\pi(W') = (f_1,\ldots,f_k)$ such that $e_i \in \rho^+(W) \iff f_i \in \rho^+(W')$ for every $1 \leq i \leq k$. Then there exists a function $f: \N \to \N$ such that the following holds.  If $(G',\pi(W'))$ contains a directed boundary-linked $(f(n),f(n))$-nest then $(G,\pi(W)) \hookrightarrow (G',\pi(W'))$.
\end{lemma}
\begin{proof}%
    Let $(\Gamma,\nu)$ and $(\Gamma',\nu')$ be Eulerian embeddings of $(G,\pi(W)), (G',\pi(W'))$ respectively. Let $g$ be as in \cref{thm:swirl_embedds_Eulerembeddable_graphs} and let $h(x) \coloneqq g(5x)$. 
    
    By \cref{lem:cuffs_are_alternating} cuffs are alternating and we assume without loss of generality that $e_1 \in \rho^-(W)$ and $f_1 \in \rho^-(W')$ are out-edges, $e_k \in \rho^+(W),f_k \in \rho^+(W')$ are in-edges and $e_i \in \rho^+(W) \iff e_{i+1} \in \rho^-(W)$.  The same holds for $f_i,f_{i+1}$ by the assumption of the theorem, for every $1 \leq i <k$. The claim follows analogously for different orderings by switching to respective bijections.

    Pair up the edges of $\pi(W)$ via the perfect matching $M = \{(e_i,e_{i+1}) \mid 1 \leq i \leq k, i \text{ odd }\}$, which exists since $k \in 2\N$. Construct a new graph $G^+$ from $G$ as follows. Subdivide each edge $e_i$ four times resulting in a path $(e_i,e_i^1,e_i^2,e_i^3,e_i^4)$ and let $x_i^j \coloneqq \tail(e_i^j)$ for $j \in \{1,2,3,4\}$ if $e_i \in \rho^-(W)$ (it is an out-edge)  and $(e_i^4,e_i^3,e_i^2,e_i^1,e_i)$ with $x_i^j \coloneqq \head(e_i^j)$ for $j \in \{1,2,3,4\}$ if $e_i \in \rho^+(W)$ (it is an in-edge). Finally remove all incidences with $W$ and for $(e_i,e_{i+1}) \in M$ introduce $x_i^{i+1}$ and let $\head(e_i) = x_i^{i+1}=\tail(e_{i+1})$. By construction $G^+$ is a graph (all edges have exactly two incidences) and it is Eulerian. Next for each $1 \leq i \leq k$ do the following: if $(e_i,e_{i+1}) \in M$, add four new edges $e_{i,(i+1)}^1,e_{(i+1),i}^1, e_{i,(i+1)}^3,e_{(i+1),i}^3$ to the graph such that $e_{i,(i+1)}^1=(x_i^1,x_{i+1}^1)$, $e_{(i+1),i}^1=(x_{i+1}^1,x_{i}^1)$ as well as $e_{i,(i+1)}^3=(x_i^3,x_{i+1}^3)$ and $e_{(i+1),i}^3=(x_{i+1}^3,x_{i}^3)$. If $(e_i,e_{(i+1)}) \notin M$ then add four new edges $e_{i,(i+1)}^2,e_{(i+1),i}^2, e_{i,(i+1)}^4,e_{(i+1),i}^4$ to the graph such that $e_{i,(i+1)}^2=(x_i^2,x_{i+1}^2)$, $e_{(i+1),i}^2=(x_{i+1}^2,x_{i}^2)$ as well as $e_{i,(i+1)}^4=(x_i^3,x_{i+1}^4)$ and $e_{(i+1),i}^4=(x_{i+1}^3,x_{i}^4)$. Additionally do the same for $(e_k,e_1)$, i.e., add four new edges $e_{k,1}^2,e_{1,k}^2, e_{k,1}^4,e_{1,k}^4$ to the graph such that $e_{k,1}^2=(x_k^2,x_1^2)$, $e_{1,k}^2=(x_{1}^2,x_{k}^2)$ as well as $e_{k,1}^4=(x_k^3,x_1^4)$ and $e_{1,k}^4=(x_{1}^3,x_{k}^4)$.  The resulting graph is of degree at most $4$, Eulerian, and Eulerian embeddable in the disc, i.e., $G^+ \in \mathbf{G}(\Delta,0)$. See \cref{fig:G^+_swirl_around_roots} for an illustration of the construction together with the resulting Eulerian embedding $(U,\nu)$ ``induced by $\Gamma$'' in the disc (that is $U$ agrees with $\Gamma$ where defined).

    We essentially replaced the root edges of $G$ by a small directed nest (or rather a swirl).

    Note that $\Abs{V(G^+)} \leq 5\Abs{V(G)}$ by construction. By our choice of $h$, \cref{thm:swirl_embedds_Eulerembeddable_graphs} implies that $G^+$ can be strongly immersed into an $h(s)$-swirl $\SSS$ with respective Eulerian embedding $\Pi =(U',\nu')$ in the disc $\Delta$ by some strong immersion $\gamma$ yielding an induced embedding $\Pi_\gamma=(U'_\gamma,\nu'_\gamma)$ (note that the induced embedding is an Eulerian embedding in $\Delta$ by \cref{obs:eulerian_emb_closed_under_splitting_and_immersion}). We refer to the graph $\gamma(G^+)$ as $G^+_\gamma$ for notational simplicity.
    
    Now, since the underlying undirected graph of the gadget we added to $G$ in order to construct $G^+$ is $3$-connected, by Whitney's Theorem \cite{Whitney1933}, the outermost circle $C^+$ of $G^+$---that is $C^+=\ol(f)$ for the outerface $f$ of $(U,\nu)$, which is a circle by \cref{obs:faces_in_2-cell_Euler_embeddings_bounded_by_circle}---is mapped to a circle $C^*$ in $\SSS$ (since the immersion is strong). Thus $C^*$ bounds a disc $\Delta(C^*) \subseteq \Delta$ and since the induced embedding $\Pi_\gamma$ is $2$-cell, \cref{obs:graphs_inside_circles_are_Eulerian} implies that $G^+_\gamma[\Delta(C^*);U_\gamma']$ is Eulerian. In particular, by the same observation since $G^+-C^+$ is not disconnected, either $G^+_\gamma[\Delta(C^*);U_\gamma'] \cap G_\gamma^+ = C^*$ or $G^+_\gamma[\Delta(C^*);U_\gamma'] \cap G_\gamma^+ = G_\gamma^+$. That is, informally speaking, either all of $G^+$ is mapped inside $\Delta(C^*)$ or outside of it.

    \begin{claim}
        We may assume that $G^+_\gamma[\Delta(C^*);U_\gamma'] \cap G_\gamma^+ = G_\gamma^+$.
    \end{claim}
    \begin{claimproof}
        Suppose that $G^+_\gamma[\Delta(C^*);U_\gamma'] \cap G_\gamma^+  = C^*$. 
        As $\SSS$ is planar, we can embed $\SSS$ into the sphere by an embedding $U''$ which extends $U'$ in the sense that every face of $U'$ other than the outer face is a face of $U''$. Now choose any face $f$ of $U'$ contained in the interior of $\Delta(C^*)$. Next remove an open disc in the interior of $f$ on the sphere and obtain a disc in which $\SSS$ is embedded by $U''$. But now, in $U''$ we have  $G^+_\gamma[\Delta(C^*);U''_\gamma] \cap G_\gamma^+ = G_\gamma^+$ as required.
    \end{claimproof}
 
    Finally choose $f=h(n)+2k$ and let $\ell = f(n)$.
    
    Let $(C_1,\ldots, C_\ell)$ be a directed boundary linked nest in $G'$ and let $P_1,\ldots,P_k$ be the respective°$\{\rho(W),E(C_\ell)\}$-paths where $P_i$ starts or ends in $f_i$ if $f_i$ is an in- or out-edge respectively for every $1 \leq i \leq k$. Note that as proved in \cite{EDP_Euler}, and as is easily verifiable due to the fact that the nest is already Eulerian embedded (and the cycles $C_1,\ldots,C_\ell$ are of alternating orientation), $G'$ immerses an $f(k)$-swirl $\SSS^*$ such that $\bigcup_{i=1}^k (C_i \cup P_i) \subseteq \SSS^*$. 

     Let $C_1^*, \ldots, C_{f(k)}^*$ be the alternating cycles of $\SSS^*$. As $f(k) = h(s)+2k$, $\SSS$ immerses into the sub-swirl of $\SSS^*$ induced by the cycles $C_{2k+1}^*, \dots, C_{f(k)}^*$. Furthermore, by flipping $\SSS^*$ if necessary, we may assume that the orientation of $C^*$ agrees with the orientation of the boundary paths $P_1, \dots, P_k$. Now, we can use the remaining $2k$ cycles $C_1^*, \dots, C_{2k}^*$ and the paths $P_1, \dots, P_k$ to connect the boundary edges $f_1, \dots, f_k$ to the edges $e_1, \dots, e_k$ of the immersion of $G$ in $\SSS^*$. This completes the proof. 
\end{proof}

The next theorem follows by combining \cref{lem:stitched_foundation_bounded_tw} and \cref{lem:disc-k-immerse-in-swirl} essentially refuting \cref{thm:no_bad_seq_in_disc_loosely_linked}. 

\begin{theorem}\label{thm:no_bad_seq_in_disc_highly_linked}
   Let $\Delta$ be a disc. Let~$k\in \N$ be minimal such that $\mathbf{G}(\Delta,k)$ is not well-quasi-ordered by strong immersion. Let~$((G_i,\pi(W_i))_{i \in \N}$ be a bad sequence (with respect to strong immersion) of bead-rooted Eulerian digraphs where~$\bar G_i \in \mathbf{G}(\Delta,k)$ with respective Eulerian embedding~$(\Gamma_i,\nu_i,\omega_i)$ and foundations~$\FFF_i=(F_1^i\ldots,F_{t_i}^i)$ with respective $\Pi_{\FFF_i}$ for every~$i\in \N$. 
    
    Then there exists a function $h: \N \to \N$ and an infinite sequence $I \subseteq \N$  such that $\ebw{\stitch(\bar G_i; (\FFF_i,\Pi_{\FFF_i}))} \leq h(k) $ for every $i \in I$.
\end{theorem}
\begin{proof}
    Let $n \coloneqq |V(G_1)|$. By \cref{lem:disc-k-immerse-in-swirl}, if there is $i > 1$ such that $(G_i, \pi(W_i))$ contains a boundary linked $(f(n), f(n))$-nest, then $(G_1, \pi(W_1)) \hookrightarrow (G_i, \pi_i(W_i))$. Thus, as $((G_i,\pi(W_i))_{i \in \N}$ is bad, there cannot be some $(G_i, \pi_i(W_i))$ containing a boundary linked $(f(n), f(n))$-nest. But then, \cref{lem:stitched_foundation_bounded_tw} implies the theorem. %
     
\end{proof}

As mentioned above we may now combine \cref{thm:no_bad_seq_in_disc_loosely_linked} and \cref{thm:no_bad_seq_in_disc_highly_linked} to prove the main \cref{thm:wqo:bead-root_for_disc} of this section.

\begin{proof}[of \cref{thm:wqo:bead-root_for_disc}]
    Assume the contrary and let $(\bar G_i)_{i \in \N}$ be a bad sequence with respect to strong immersion. For every $i \in \N$ let $(\Gamma_i,\nu_i)$ be respective embeddings of $\bar G_i$ in $\Delta$ with respective foundations $\FFF_i$ and orderings $\Pi_{\FFF_i}$. By \cref{thm:no_bad_seq_in_disc_highly_linked} there exists a function $h : \N \to \N$ and an infinite sequence $I \subseteq \N$ such that $\ebw{\stitch(\bar G_i; (\FFF_i,\Pi_{\FFF_i}))} \leq h(k)$. This is a contradiction to \cref{thm:no_bad_seq_in_disc_loosely_linked} refuting the existence of the bad sequence and thus concluding the proof.
\end{proof}

\section{Strong Immersions for the Cylinder} \label{sec:cylinder}

The following is the main theorem of this section.
\begin{theorem}\label{thm:wqo:bead-root_for_cylinder}
 Let $\Sigma$ be a cylinder and $k \in 2\N$. The class $\mathbf{G}(\Sigma,k)$ is well-quasi-ordered by strong immersion.
\end{theorem}

Note that if $k=0$ then $\mathbf{G}(\Sigma,0) \subseteq \mathbf{G}(\Delta,0)$ where $\Delta$ is a disc. In particular we have the following.
\begin{corollary}
    Let $\Sigma$ be a cylinder. The class $\mathbf{G}(\Sigma,0)$ is well-quasi-ordered by strong immersion.
\end{corollary}
\begin{proof}
    This is follow from \cref{thm:swirl_embedds_Eulerembeddable_graphs}.
\end{proof}

We start with a proof of the ``high representativity'' case for the cylinder as we will need to slightly adapt our framework when tackling the ``low representativity'' case.

Recall \cref{def:representativity} of cuff separating curves and note that for the cylinder these are cut-cycles. 

\begin{lemma}\label{lem:cyclinder_high_rep}
 Let $\Sigma$ be the cylinder with cuffs $\zeta_1,\zeta_2$. Let $k \in 2\N$. Then there is a function $f:\N \to \N$ such that the following holds. For $i=1,2$ let $(G_i,\pi_i(W_{\omega_i})) \in \mathbf{G}(\Sigma,k)$ with Eulerian $2$-cell embeddings $(\Gamma_i,\nu_i,\omega_i)$ and partitions $W_{\omega_i}=(W_1^i,W_2^i)$ such that $\Abs{\rho(W_j^1)} = \Abs{\rho(W_j^2)} \geq 1$ for $j \in \{1,2\}$. Let $n \coloneqq \Abs{V(G_1)}$. Let $k_1 \coloneqq |\rho(W_1^1)|$ and $k_2 \coloneqq |\rho(W_2^1)|$. 
    
    Suppose $G_2$ contains $\ell \coloneqq f(n)$ edge disjoint cycles $C_1, \dots, C_\ell$ such that each $C_i$ bounds a disc $\Delta(C_i) \in \hat\Sigma$ containing $\zeta_2$ in its interior
    and for $1 \leq i < j \leq \ell$, $C_j$ is contained in the interior of $\Delta(C_i)$ and $C_1, \ldots, C_\ell$ are oriented in alternating directions.      Furthermore, suppose that there are $\ell$ pairwise edge-disjoint paths between $C_1$ and $C_\ell$ contained in $\Delta(C_1) \setminus \Delta(C_\ell)^\circ$. Finally, suppose that there are $k_1$ edge-disjoint paths from $\rho(W_1^2)$ to $C_{2k_1}$ and $k_2$ edge-disjoint paths from $\rho(W^2_2)$ to $C_{\ell-2k_2}$.

    Then $(G_1,\pi_1(W_{\omega_1})) \hookrightarrow (G_2,\pi_2(W_{\omega_2}))$ by strong immersion. 
\end{lemma}
\begin{proof}
Essentially, we proceed exactly as in the proof of \cref{lem:disc-k-immerse-in-swirl} and replace in $G_1$ the set $W_1^1$ by the gadget depicted in \cref{fig:G^+_swirl_around_roots} and in the same way replace $W_2^1$ by such a gadget. Let $G^+$ be the resulting Eulerian digraph. As in \cref{lem:disc-k-immerse-in-swirl}, $G^+$ embeds into a large enough swirl $\SSS$. This swirl can be mapped to the swirl $\SSS^*$ in $G_2$ such that the gadgets face their respective cuff. We can then delete the gadgets and connect the roots on the cuffs to their corresponding edges as in \cref{lem:disc-k-immerse-in-swirl}.

    We now provide more details of this construction. 
    Let $g$ be the function defined in \cref{lem:disc-k-immerse-in-swirl}. We set $h(n) \coloneqq g(5n)$ and define $f(n) \coloneqq h(n) + 2(k_1+k_2)$. 
    Let $\rho(W_1^1) = (e_1, \dots, e_{k_1})$ be ordered by $\pi_1(W^1_1)$. Then the edges alternate between in- and out-edges. W.l.o.g. we assume that $e_1$ is an out-edge and therefore $e_{k_1}$ is an in-edge. (Note that $k_1, k_2$ must both be even due to the Eulerian-property of $G_1$.) If $e_i$ is an out-edge then let $u_i$ be the head of $e_i$  and otherwise let $u_i$ be the tail of $e_i$, for $1 \leq i 
    \leq k_i$. 
    
    We subdivide each $e_i$ four times. If $e_i$ is an out-edge then let $x^i_1, \dots, x^i_4$ be the new vertices introduced by subdividing $e_i$ in the order $x^i_1, ..., x^i_4$ when following $e_i$ from tail to head. 
    If $e_i$ is an in-edge we number the new vertices $x^i_1, \dots, x^i_4$ in the order $x^i_4, \dots, x^i_1$ when following $e_i$ from tail to head.
    
    For $1 \leq i \leq k_1$ with $i$ odd we add the edges $(x^i_1, x^{i+1}_1)$ and $(x^i_3,x^{i+1}_3)$ as well as $(x^{i+1}_1, x^i_1)$ and $(x^{i+1}_3,x^i_3)$.
    For $1 \leq i \leq k_1$ with $i$ even we add the edges $(x^i_2, x^{i+1}_2)$ and $(x^i_4,x^{i+1}_4)$ as well as $(x^{i+1}_4, x^i_4)$ and $(x^{i+1}_4,x^i_4)$, where we define $k_1+1$ as $1$. 
    Finally, for $1 \leq i \leq k_1$, $i$ odd, we add a new vertex $y^i_{i+1}$ and the edges $(u_i, y^i_{i+1})$ and $(y^i_{i+1}, u_{i+1})$. See \cref{fig:G^+_swirl_around_roots} for an illustration.
    
    The same construction we apply to the edges $f_1, \dots, f_{k_2} \in \rho(W^1_2)$.
        Let $C_1^*$ be the outer circle of the gadget for $W_1$, i.e. the circle containing $x^i_4$, for all $i$, and $y^i_{y+1}$ for all odd $i$. Likewise let $C_2^*$ be the corresponding circle in the gadget for $W_2$.
    
    Let $G^+$ be the resulting digraph. Then $|V(G^+)| \leq 5\cdot |V(G)|$. Thus, by \cref{lem:disc-k-immerse-in-swirl}, $G^+$ can be strongly immersed into a swirl $\SSS$ of order $h(n)$. For the embedding we assume that the underlying graph of $G^+$ is $3$-connected, for otherwise we can add $2$-circles to $G^+$ to ensure this property.
    This implies that we can choose the embedding (see the proof of \cref{lem:disc-k-immerse-in-swirl}) so that $C_1^*$ bounds a disc $\Delta(C_1^*)$ which contains $G^+$ and $C_2^*$ bounds a disc $\Delta(C_2^*)$ such that $\Delta(C_2^*) \cap G^+ = C_2^*$.
    
    As a last step we can now embed $\SSS$ into the swirl contained in $G'$ such that $C_1^*$ bounds a disc containing $\zeta_2$ and $G^+$ and $C_2^*$ bounds a disc $\Delta(C_2^*)$ containing $\zeta_2$ and $\Delta(C_2^*) \cap G^+ = C_2^*$. 
    Furthermore, as $G'$ contains a swirl $\SSS^*$ of order $f(n) = h(n) + 2k_1 + 2k_2$, we can choose this embedding so that the outermost $2k_1$ cycles of $\SSS^*$ bordering $\zeta_1$ are unused as are the $2k_2$ innermost cycles bordering $\zeta_2$. This allows us to use the boundary-linkedness of $\SSS^*$ to connect the edges on $\zeta_1$ to $e_1, \dots, e_{k_1}$ and the edges on $\zeta_2$ to the edges $f_1, \dots, f_{k_2}$.
\end{proof}

The previous lemma together with \cref{lem:boundary_linked_Menger_for_embeddings_in_cylinder} immediately implies the next result.

\begin{corollary}\label{cor:cyl-high-rep}
    Let $\Sigma$ be the cylinder with cuffs $\zeta_1,\zeta_2$. Let $k \in 2\N$. Then there is a function $f:\N \to \N$ such that the following holds. Let $(G_i,\pi_i(W_{\omega_i})) \in \mathbf{G}(\Sigma,k)$ with Eulerian $2$-cell embeddings $(\Gamma_i,\nu_i,\omega_i)$ and partitions $W_{\omega_i}=(W_1^i,W_2^i)$ such that $\Abs{\rho(W_j^1)} = \Abs{\rho(W_j^2)} \geq 1$ for $j=1,2$. Let $n \coloneqq \Abs{V(G_1)}$. If  there is no cuff connecting $\Gamma_2$-tracing curve $F$ with $\delta(F) < f(n)$ and no cuff separating cut-cycle $F'$ for $\zeta_1,\zeta_2$ with $\delta(F') < f(n)$, then $(G_1,\pi_1(W_{\omega_1})) \hookrightarrow (G_2,\pi_2(W_{\omega_2}))$ by strong immersion.
\end{corollary}

\subsection{Cutting and Knitting the Surfaces}
\label{sec:edge-rooted}

If \cref{lem:cyclinder_high_rep} does not apply, it turns out that we find ``offending curves'' in our embeddings pointing to a lack of generality in the drawing as described in \cref{def:representativity}. Thus we will inductively cut the surface and drawing along such offending curves. Unfortunately, the curves may not lower the genus of the surface: if we start with the cylinder for example, cuff separating curves cut the cylinder into two new cylinders, which is why we need to deal with the cylinder case separately prior to starting the general induction on the surface. Since we will need very similar results when cutting a general surface $\Sigma$, we present the required lemmas for cutting the general surfaces in this section.

\smallskip

Unfortunately, it turns out that cutting and knitting drawings for bead-rooted Eulerian digraphs requires a bit more work than for undirected graphs since we have stricter restrictions for our embeddings and graph class. For ease of argumentation we will change our framework and circumvent beads by switching to pseudo-Eulerian graphs rooted in edges, since we effectively only care about orders on the edges; in essence this new version is obtained by simply ``plugging'' all the incindences out of the bead-vertices and just keeping the resulting edges. (Note that if we cut our graph along a curve with both ends in a cuff that separates our surface, we may otherwise need to split beads and reassign them carefully to the new surfaces.) Besides these technicalities, the results presented are similar in flavor to the results presented in \cref{subsec:knitworks}. Thus, let us switch back to a version of \emph{rooted Eulerian digraphs}---not necessarily rooted in edges coming from an induced cut, see \cref{def:rooted_graph}---to simplify the arguments to come.

\begin{definition}[Eulerian Digraphs with defect $k$]
    Let $G$ be a digraph such that for every $v \in V(G)$ either $v$ is Eulerian and of degree $2$ or $4$, or $v$ is of degree $1$. We define $V^+ \subseteq V(G)$ to be the set of vertices of out-degree $1$ and in-degree $0$ and $V^- \subseteq V(G)$ to be the set of vertices of in-degree $1$ and out-degree $0$. Let $\Abs{V^+} = \Abs{V^-} = k \in \N$. Let $E^* \subset E(G)$ be the set of edges adjacent to vertices in $V^+ \cup V-$. Then we call $G$ \emph{Eulerian with defect $k$} and $E^*$ the \emph{root edges of $G$}. Let $\pi(E^*)$ be an ordering of $E^*$. Then we call $(G,\pi(E^*))$ an \emph{edge-rooted Eulerian digraph (with defect $k$)}.
\end{definition}
\begin{remark}
    Note that by definition the defect is always even. Further, whenever we use the term ``edge-rooted'' we implicitly talk about Eulerian digraphs with respective defect.
\end{remark}

An Eulerian digraph with defect $k$ is \emph{connected} if its underlying undirected graph is connected. As usual we will focus on connected graphs since the general case can be easily reduced it Thus we will tacitly assume Eulerian digraphs with defect $k$ (or edge-rooted Eulerian digraphs) to be connected unless stated otherwise.

We define strong immersion for edge-rooted Eulerian digraphs in the obvious way following \cref{def:rooted_immersion}.
\begin{definition}[Immersion of Edge-Rooted Graphs]\label{def:edge-rooted_immersion}
    Let~$(G_1,\pi(E_1))$ and~$(G_2,\pi(E_2))$ be edge-rooted Eulerian digraphs. We say that~$(G_2,\pi(E_2))$ \emph{(strongly) immerses}~$(G_1,\pi(E_1))$ if both have the same defect~$k \in 2\N$ and additionally there is a map $\gamma: V(G_1)\cup E(G_1) \to G_2$ satisfying the following:
    \begin{enumerate}
        \item $\gamma:G_1 \hookrightarrow G_2$ is a (strong) immersion, and
        \item given the natural bijection~$\eta: \rho_G(E_1) \to \rho_{G}(E_2)$ satisfying~$\eta(\pi(E_1)) = \pi(E_2)$, the path~$\gamma(e)$ for~$e \in E_1$ contains~$e' \in E_2$ if and only if~$\eta(e) = e'$. We say that \emph{$\gamma$ respects the order of the roots}.
    \end{enumerate}
    We call~$\gamma$ a \emph{(strong) edge-rooted immersion} or simply a (strong) immersion and write~$\gamma: (G_1,\pi(E_1)) \hookrightarrow (G_2,\pi(E_2))$.
\end{definition}

There is a canonical definition for embeddings of edge-rooted Eulerian digraphs following  \cref{def:Euler_embedding_upto_beads} of embeddings of bead-rooted Eulerian digraphs.

\begin{definition}[Embeddings fixing edges]\label{def:Euler_embedding_upto_root-edges}
  Let~$\Sigma$ be a surface and let $c(\Sigma) = \{\zeta_1,\ldots,\zeta_t\}$ be its set of cuffs for some~$t \geq 1$ and let~$\Delta_1,\ldots,\Delta_t$ be closed discs bounded by the respective cuffs in~$\hat{\Sigma} \setminus \Sigma$ and denote their interior by~$\Delta_1^\circ,\ldots,\Delta_t^\circ$ respectively. Let~$(G,\pi(E^*))$ be an edge-rooted Eulerian digraph and let $V^+ \subseteq V(G)$ be the set of vertices of out-degree $1$ and in-degree $0$ and $V^- \subseteq V(G)$ be the set of vertices of in-degree $1$ and out-degree $0$. Let~$(\hat \Gamma,\hat \nu)$ be an embedding of~$G$ in~$\hat{\Sigma}$ satisfying the following:

  \begin{itemize}
      \item[(i)] every~$v \in V(G)$ that is not of degree $1$ is Eulerian embedded in~$\Sigma$,
      \item[(ii)] every $e \in E(G) \setminus E^*$ is embedded in $\Sigma$,
      \item[(iii)] let~$W_i \coloneqq \hat\nu^{-1}(\Delta_i^\circ)$, then~$W_\omega=(W_1,\ldots,W_t)$ is a partition of~$V^+ \cap V-$,
      \item[(iv)] for every~$1 \leq i \leq t$ and every~$w \in W_i$ with adjacent edge $e$, it holds~$\hat\nu(e) \in \zeta_i$. Further for~$x \in \hat{\nu}^{-1}(\zeta_i)$ it holds~$x \in E^*$. In particular an element of $G$ is drawn on a cuff if and only if it is an edge in $E^*$.  
    \end{itemize}
      
  We define the map~$\omega: c(\Sigma) \to \bigcup_{1 \leq i \leq t} W_i$ via~$\omega(\zeta_i) = W_i$ and we let~$E^*_\omega \coloneqq(\nu^{-1}(\zeta_i),\ldots,\nu^{-1}(\zeta_t))$. Then, since $G$ is connected, $E^*_\omega$ is a partition of $E^*$. We call an edge of $E^*$ an \emph{in-edge} if its head is drawn in $\Sigma$ and an \emph{out-edge} if its tail is drawn in $\Sigma$. Furthermore, we  assume  that for each $1\leq j \leq t$, $\nu^{-1}(\zeta_i)$ admits equally many in- and out-edges.
  
  Finally let~$\Gamma \coloneqq \hat{\Gamma} \cap \Sigma$ and~$\nu \coloneqq \restr{\hat \nu}{\Gamma}$. Then we call~$(\Gamma,\nu,\omega)$\Symbol{GAMMA@$(\Gamma,\nu,\omega)$} an \emph{Eulerian embedding of~$(G,\pi(E^*))$ in~$\Sigma$}.\Index{Eulerian embedding of bead-rooted digraphs}
\end{definition}
\begin{remark}
    The reason why we did not include the restrictions on $E^*_\omega$ (equally many in and out edges for each part of the partition) as a fifth point is because the assumption follows from $(i)-(iv)$ under mild additional assumptions on the embedding as we will see shortly.
\end{remark}

In particular an Eulerian embedding of $(G,\pi(E^*))$ is Eulerian ``inside'' $\Sigma$ and draws $E^*$ on its cuffs. The remaining definitions for embeddings (such as faces, $2$-cell, $\Gamma$-tracing, cut-cycles, in-edges, out-edges, etc.) from \cref{sec:surface_defs} for general and Eulerian embeddings of bead-rooted digraphs extend verbatim.

We extend the \cref{def:stitching_std,def:stitching_beadrooted} of stitching to edge-rooted Eulerian digraphs as follows.
\begin{definition}\label{def:stitch_edge-rooted}
    Let $(G,\pi(E^*))$ be an edge-rooted Eulerian digraph. Let $w$ be a new element not in $V(G) \cup E(G)$. We define $G' \coloneqq \stitch(G;\pi(E))$ to be the graph obtained from $G$ by identifying all the vertices of degree $1$ to a single vertex $w$, i.e., for every edge $e \in E^*$ and $v \in V(G)$ of degree~$1$, if $(e,v) \in \operatorname{inc}_G$ then we add $(e,w)$ to $\operatorname{inc}_{G'}$ instead, and equivalently if $(v,e) \in \operatorname{inc}_G$ then $(w,e) \in \operatorname{inc}_{G'}$.  Then $\rho(w) = E^*$ by construction and we set $\pi(w) \coloneqq \pi(E^*)$.

    Similarly, given $E^*_\omega=(E_1,\ldots,E_t)$ together with induced orderings $\pi(E_1),\ldots,\pi(E_t) \subseteq \pi(E^*)$ and partition $W_\omega=(W_1\ldots,W_t)$ of the degree-one vertices, we define $(G,\pi(W_\omega)) \coloneqq \stitch(G; \pi(E^*_\omega))$ to be the graph obtained from $G$ by identifying all the vertices in $W_i$ to a single vertex $w_i$ for $1 \leq i \leq t$.
\end{definition}

In particular, by \cref{def:Euler_embedding_upto_root-edges}, $\stitch(G;\pi(E^*))$ is a rooted Eulerian digraph and $\stitch(G;\pi(E^*_\omega))$ is a bead-rooted Eulerian digraphs with $\leq t$ beads of order $\leq k$. 

On the flip side, every rooted and every bead-rooted Eulerian digraph $(G,\pi(W))$ and $(G,\pi(W_\omega))$ corresponds to a unique edge-rooted Eulerian digraph $(G',\pi(E^*))$ by setting $E^* \coloneqq \rho(W)$ and $\pi(E^*)$ (where $\pi(W_i) \subseteq \pi(E^*)$ for every $1 \leq i \leq t$), with a corresponding induced edge-rooted Eulerian embedding in the obvious way. 

\begin{definition} 
    Let $\Sigma$ be a surface and $k \in 2\N$ fixed. We define $\mathbf{G^*}(\Sigma,k)$ to be the class of edge-rooted Eulerian digraphs of defect $k$ that are Euler-embeddable in $\Sigma$.
\end{definition}

The following is imminent by \cref{def:rooted_immersion,def:edge-rooted_immersion}.
\begin{observation}\label{obs:edge-rooted_immersion_is_enough}
    Let $\Sigma$ be a surface and $k \in 2\N$. Let $(G_1,\pi(W_{\omega_1})), (G_2,\pi(W_{\omega_2})) \in \mathbf{G}(\Sigma,k)$ be bead-rooted Eulerian digraphs and let $(G_1',\pi(E_1^*)), (G_2',\pi(E_2^*))$ be the respective edge-rooted Eulerian digraphs. Then
    $$(G_1,\pi(W_1)) \hookrightarrow (G_2,\pi(W_2)) \iff (G_1',\pi(E_1^*)) \hookrightarrow (G_2',\pi(E_2^*)), $$
    by strong immersion and $(G_i,\pi(E_i^*)) \in \mathbf{G^*}(\Sigma,k)$.
\end{observation}

Similarly we have the following.
\begin{observation}\label{obs:edge-rooted_immersion_is_enough_dos}
     Let $(G_1,\pi(E_1^*)), (G_2,\pi(E_2^*)) \in \mathbf{G^*}(\Sigma,k)$ be edge-rooted Eulerian digraphs of defect $k \in 2\N$ with Eulerian embeddings $(\Gamma_i,\nu_i,\omega_i)$ in some surface $\Sigma$. Let $(G_i',\pi(w_i)) \coloneqq \stitch(G_i;\pi(E_i^*))$ as well as $(G_i'',\pi(W_{\omega_i})) \coloneqq \stitch(G_i; \pi({E_i^*}_{\omega_i}))$ for $i=1,2$ be their respective stitches. 
     Then 
     $$(G_1,\pi(E_1^*)) \hookrightarrow (G_2,\pi(E_2^*)) \iff (G_1',\pi(w_1)) \hookrightarrow (G_2',\pi(w_2)), $$
     as well as
     $$(G_1,\pi(E_1^*)) \hookrightarrow (G_2,\pi(E_2^*)) \iff (G_1'',\pi(W_{\omega_1})) \hookrightarrow (G_2'',\pi(W_{\omega_2}))$$
     by strong immersion.

     Moreover, $(G_i'',\pi(W_{\omega_i})) \in \mathbf{G}(\Sigma,k)$ for $i=1,2$.
\end{observation}
\begin{proof}
    The proof for the immersions is imminent from the \cref{def:rooted_immersion,def:edge-rooted_immersion}. The fact that $(G_i'',\pi(W_{\omega_i})) \in \mathbf{G}(\Sigma,k)$ is obvious since $(\Gamma_i,\nu_i,\omega_i)$ is a respective embedding for the bead-rooted graph in $\Sigma$.
\end{proof}

We conclude the following.

\begin{corollary}\label{cor:i_am_dying}
    Let $\Sigma$ be a surface and $k \in 2\N$. Then $\mathbf{G}(\Sigma,k)$ is well-quasi-ordered by strong rooted immersion if and only if $\mathbf{G^*}(\Sigma,k)$ is well-quasi-ordered by strong edge-rooted immersion.
\end{corollary}
\begin{proof}
    Applying \cref{obs:edge-rooted_immersion_is_enough} and \cref{obs:edge-rooted_immersion_is_enough_dos} implies the theorem.
\end{proof}
   
It turns out that edge-rooted Eulerian digraphs are easier to work with when inductively cutting the graphs along $\Gamma$-tracing curves, by loosing the need to carry beads. Note that \cref{cor:i_am_dying} implies that to prove \cref{thm:wqo_on_surfaces} it suffices to prove the following.

\begin{theorem}\label{thm:wqo_edge-rooted}
    Let $\Sigma$ be a surface and $k \in 2\N$. The class~$\mathbf{G^*}(\Sigma,k)$ is well-quasi-ordered by strong edge-rooted immersion.
\end{theorem}

The main result of the section is devoted to a proof of \cref{thm:wqo_edge-rooted} when $\Sigma$ is the cylinder. Since many of the tools needed for the proof can be stated for more general surfaces and will be needed in \cref{sec:la-grande-inductione}, we give them here in the more general form.

\smallskip

We extend the definition of alternating cut-cycles to alternating traced cuts. To this extent we define the following.

\begin{definition}\label{def:order_of_cedges_on_curve}
     Let $\Sigma$ be a surface and $(\Gamma,\nu,\omega)$ a $2$-cell embedding of an edge-rooted Eulerian digraph $(G,\pi(E^*))$ on $\Sigma$. Let $\gamma:[0,1] \to \Sigma$ be a $\Gamma$-tracing simple curve and let $F= \gamma([0,1])$. We define $\pi_\gamma(F) \coloneqq (e_1,\ldots,e_k)$ where $\{e_1,\ldots, e_k\} = \rho(F)$ and for $1 \leq i< j \leq k$ and $t_i,t_j \in [0,1]$ with $\gamma(t_i) = e_i, \gamma(t_j) = e_j$ it holds $t_i < t_j$. Whenever we give an ordering $\pi(F)$ we mean that $\pi(F) = \pi_{\gamma}(F)$ is obtained as above via some simple curve $\gamma$.
\end{definition}
\begin{remark}
    Note that $\pi(F)$ is for non-closed curves unambiguously defined up-to a reflection of the order, while for closed curves it is unambiguously defined up-to a reflection and a cyclic rotation. 
\end{remark}

Recall \cref{def:internal_tracing_clean}. 

\begin{definition}\label{def:left_and_right}
    Let $\Sigma$ be a surface and $(\Gamma,\nu,\omega)$ a $2$-cell embedding of an edge-rooted Eulerian digraph $(G,\pi(E^*))$ in $\Sigma$. Let $\gamma:[0,1] \to \Sigma$ be a $\Gamma$-tracing simple curve in $\Sigma$ (possibly with ends in the boundary) with $\gamma(0) \neq \gamma(1)$, and $\nu(\gamma(0)),\nu(\gamma(1)) \notin \Gamma$ and let $F= \gamma([0,1])$. Let $\pi_\gamma(F) = (e_1,\ldots,e_k)$. Let $\Delta \subset \Sigma$ be a disc containing $F$ such that $\{\gamma(0),\gamma(1)\} = \bd(\Delta)\cap F$. Let $\Delta_1,\Delta_2$ be an enumeration of the two components of $\Delta \setminus F$. We call $\Delta_1$ \emph{left of $F$} and \emph{$\Delta_2$ right of $F$} and call $\Delta$ a \emph{choice of left and right} for $F$.

  We call $F$ a \emph{cut-line} if there is a choice $\Delta$ of left and right for $F$ such that $\nu^{-1}(\Delta) \cap (E(G)\cup V(G)) = \rho(F)$, and for every $e \in \rho(F)$ it holds $\Abs{\Delta \cap \nu((\tail(e),e))} = 1$ and $\Abs{\Delta \cap \nu((e,\head(e)))} = 1$. Given such $\Delta$ we define \emph{left-to-right} and \emph{right-to-left} edges of $\rho(F)$ in the obvious way.
    
    We call a cut-line $F$ \emph{alternating (with respect to $\Delta$)} if $e_i$ is a left-to-right edge if and only if $e_{i+1}$ is a right-to-left edge for every $1 \leq i < k$.

    Similarly, if $\gamma$ is as above but closed with $\gamma(0) \notin \Gamma$, then (using notation as in \cref{def:order_of_cedges_on_curve}) $\gamma':[t_1,\ldots,t_k] \to \Sigma$ is a simple curve with $\gamma([0,1]) \cap \Gamma = \gamma'([t_1,\ldots,t_k]) \cap \Gamma$. We call $F$ a \emph{cut-line} if $\gamma'([t_1,t_k])$ is a cut-line and similarly we call $F$ \emph{alternating} if $\gamma'([t_1,\ldots,t_k])$ is alternating.
\end{definition}
\begin{remark}
Clearly if $F$ is a cut-line in a $2$-cell embedding, $\rho(F)$ does not contain a loop-edge of $G$ since then $\tail(e)=\head(e)$ and each incidence is cut twice. It is an easy exercise to check that whether a cut-line is alternating does not depend on a choice of left and right.
\end{remark}

Note that for alternating $F$ one easily verifies that for $2$-cell embeddings there is a choice $\Delta'$ of left and right such that for every edge $e \in \rho(F)$ $\tail(e) \in \Delta_1'$ if and only if $\head(e) \in \Delta_2'$, by slightly enlarging the above defined $\Delta$. 

\begin{observation}\label{obs:good-left-right-choice}
      Let $\Sigma$ be a surface and $(\Gamma,\nu,\omega)$ a $2$-cell embedding of an edge-rooted Eulerian digraph $(G,\pi(E^*))$ on $\Sigma$. Let $\gamma:[0,1] \to \Sigma$ be a $\Gamma$-tracing simple curve in $\Sigma$ (possibly with ends in the boundary) such that $F \coloneqq \gamma([0,1])$ is a cut-line. Then there is a choice $\Delta$ of left and right for $F$, such that for every edge $e \in \rho(F)$ we have $\head(e),\tail(e) \in \Delta$ and $\head(e) \in \Delta_1 \iff \tail(e) \in \Delta_2$ for the left and right components $\Delta_1,\Delta_2$ of $\Delta$.
\end{observation}
\begin{proof}
Since $F$ is a cut-line, $e\in \rho(F)$ is not a loop whence it has two distinct endpoints.
   Let $\Delta'$ be as in the \cref{def:left_and_right} of alternating with components $\Delta_1',\Delta_2'$ and segments $\ell_i \in \bd(\Delta_i') \setminus F$ for $i=1,2$. Now one easily verifies that by slightly extending $\Delta'$ by shifting the segments $\ell_1,\ell_2$ along the segments of $\Gamma$ corresponding to the incidences $\nu((v,e))$ and $\nu((e,v))$ respectively, we will eventually include the endpoints in the respective $\Delta_1',\Delta_2'$. This follows since the drawing is $2$-cell, and thus the two faces $f_1,f_2$ that contain the edge $e$ are both homeomorphic to a disc whence $F$ splits the faces into two discs. The lines $\ell_i$ are now homotopically shifted along the faces staying (mostly inside the respective discs) in the obvious way without crossing each other until they cover the respective vertex. 
\end{proof}

We refine the definition of $2$-cell.

\begin{definition}\label{def:nice}
    Let $\Sigma$ be a surface and let $(\Gamma,\nu,\omega)$ be an Eulerian embedding of an edge-rooted Eulerian digraph in $\Sigma$. We call the embedding \emph{nice} if it is $2$-cell and there is no curve $\gamma:[0,1] \to \Sigma$ such that $\gamma([0,1]) \cap \Gamma = \emptyset$ and $\gamma$ is cuff connecting.
\end{definition}
\begin{remark}
    Note that being $2$-cell implies that there exist no genus reducing curves $\gamma:[0,1] \to \Sigma$ such that $\gamma([0,1]) \cap \Gamma = \emptyset$. 
Further if an embedding is nice one easily verifies the existence of a choice of left and right in $\Delta$.%
\end{remark}

We have the following observation as an analogue to \cref{obs:outline_of_plane_drawing_is_cycle}, \cref{lem:cut_cycles_are_alternating} and \cref{lem:cuffs_are_alternating}; recall that paths are defined via sequences of edges.

\begin{lemma}\label{obs:edge-rooted_embedding_is_nice}
    Let $(\Gamma,\nu,\omega)$ be a nice embedding of an edge-rooted Eulerian digraph $(G,\pi(E^*))$ in a surface $\Sigma$ satisfying $(i)-(iv)$ of \cref{def:Euler_embedding_upto_root-edges}. Let $c(\Sigma) = \{\zeta_1,\ldots,\zeta_t\}$ be the cuffs of $\Sigma$ for some $t \in \N$. Then for each face $f \in F(\Sigma,\Gamma)$, $\ol(f)$ is a circle if $f$ is internal and $\ol(f)$ is a path with both ends in some same cuff if $f$ is a boundary face. 
    
    Further, each cut-cycle $F$ is alternating and $\Abs{\rho(F)} \in 2\N$, each cut-line $F$ is alternating (where $\Abs{\rho(F)} \in 2\N$ if $F$ is an $O$-trace) and each cuff $\zeta_i$ is an alternating cut-line with $\Abs{\zeta_i} \cap \Gamma \in 2\N$. 
    
    In particular, given the partition $E^*_\omega=(E_1,\ldots,E_t)$, $\Abs{E_i} \in 2\N$ and $E_i$ contains equally many in- and out-edges for each $1 \leq i \leq t$.
\end{lemma}
\begin{proof}
    The first part of the lemma is proved as usual, where it is noted that since the embedding is nice no face is adjacent to two cuffs for else we find a curve contradiction niceness.
    To see that every cut-cycle is alternating note that by the above every face it intersects is bounded by a circle or a directed path if it is adjacent to a cuff, and it intersects each face in an even number of edges. The claim follows as in \cref{lem:cut_cycles_are_alternating} (note that the embedding being nice is necessary here). In particular this implies that $\Abs{E_i} \in 2\N$ for every $1 \leq i \leq t$.
    
    Non-closed cut-lines can always be extended to cut-cycles since we have nice embeddings whence faces are discs and we can reroute the paths in ``parallel'', thus they are alternating by the above. If we have closed cut-lines, then again since we have nice embeddings, they intersect each face an even number of times and thus it follows that they are alternating similarly to \cref{lem:cut_cycles_are_alternating} using that the embedding is $2$-cell and every face is either bounded by a circle or a path. (This also follows from the fact that there is a cut-line that is not closed that intersects the same set of edges in the same order similar to the proof of \cref{def:left_and_right}). Since it is alternating, the number of edges is odd if and only if the first and the last edge visited by $F$---after fixing a curve $\gamma:[0,1]\to F$---have the same orientation. But since we have a $2$-cell embedding and the curve is closed the two edges lie on a common face bounding a disc which is either bounded by a circle or a path if it is a boundary face and thus this cannot be the case (similar to \cref{lem:cut_cycles_are_alternating}), whence again $\Abs{\rho(F)} \in 2\N$. 
    
\end{proof}
\begin{remark}
    Note that \cref{obs:edge-rooted_embedding_is_nice} solely depends on the fact the embedding $\Gamma$ inside $\Sigma$ is $2$-cell and every vertex drawn on $\Sigma$ is Eulerian.
\end{remark}

While \cref{obs:edge-rooted_embedding_is_nice} is a rather straightforward observation from the embedding restrictions, it allows us to cut our surfaces along cut-lines while staying Eulerian embedded as we will make precise.

\smallskip

We assume the reader to be familiar with ``cutting surfaces'', where by cutting a surface along a simple (closed) curve $\gamma:[0,1] \to \Sigma$ we either create two new surface $\Sigma_1,\Sigma_2 \subset \Sigma$ obtained via the closure of each of the two components of $\Sigma \setminus \gamma([0,1])$, i.e., $\Sigma_1 \cap \Sigma_2 = \gamma[0,1]$ in the obvious way. Or, $\Sigma \setminus \gamma([0,1])$ results in a single new surface $\Sigma'$ of lower genus with an additional cuff.
Similarly, if we cut a surface along a simple curve $\gamma:[0,1] \to \Sigma$ with two distinct ends on cuffs (possibly the same) we either get two new surfaces $\Sigma_1,\Sigma_2$ as the components of $\Sigma \setminus \gamma([0,1])$, or a single surface $\Sigma' = \Sigma \setminus \gamma([0,1])$ where $\gamma([0,1])$ merges two cuffs by ``slightly enlarging the cut in $\Sigma$'' if both its ends were on distinct cuffs. 

 \smallskip
We continue with the relevant definitions.
Recall \cref{def:representativity}.

\begin{definition}[Types of Cut-lines]\label{def:cutline_types}
     Let $(\Gamma,\nu,\omega)$ be a nice Eulerian embedding of an edge-rooted Eulerian digraph $(G,\pi(E^*))$ in some surface $\Sigma$. Let $F\subset \Sigma$ be a cut-line in $\Gamma$ with respective surjective simple (closed) curve $\gamma:[0,1] \to F$.
     We call $F$ \emph{cuff surrounding}, \emph{cuff shortening}, \emph{cuff connecting}, \emph{cuff grouping}, or \emph{genus reducing} if $\gamma$ is, respectively.
\end{definition}

It turns out that genus reducing curves fall again into two different types. If $F$ is a cut-line with both ends in a cuff we call it  \emph{cuff-based} (see \cref{fig:cuff-based} for a cuff-based curve that is not genus reducing), otherwise we call it \emph{free} (see \cref{fig:double-torus-cut}).

\begin{definition}
     Let $(\Gamma,\nu,\omega)$ be a nice embedding of an edge-rooted Eulerian digraph $(G,\pi(E^*))$ in some surface $\Sigma$. Let $F\subset \Sigma$ be a cut-line. Then we call $F$ \emph{separating} if $\Sigma \setminus F$ falls into two components. Otherwise we call it \emph{non-separating}.
\end{definition}
\begin{remark}
    Clearly cut-cycles are separating by definition.
\end{remark}

 The following is a well-known topological fact.
\begin{observation}\label{obs:types_of_cut-line_cuts}
    If $F$ is cuff surrounding, cuff grouping or cuff shortening then $F$ is separating.
    If $F$ is cuff connecting then it is non-separating.
    If $F$ is a (cuff based) genus reducing curve then it may be either separating or non-separating.
    If $F$ is cuff-based, then $F$ is either cuff shortening or genus reducing.
\end{observation}

We continue with an analysis of the different cases when cutting edge-rooted Eulerian embedded graphs along cut-lines;  recall that we write the concatenation of two orderings as $\pi \circ \pi'$.

\begin{figure}
    \centering
    \includegraphics[width=0.5\linewidth]{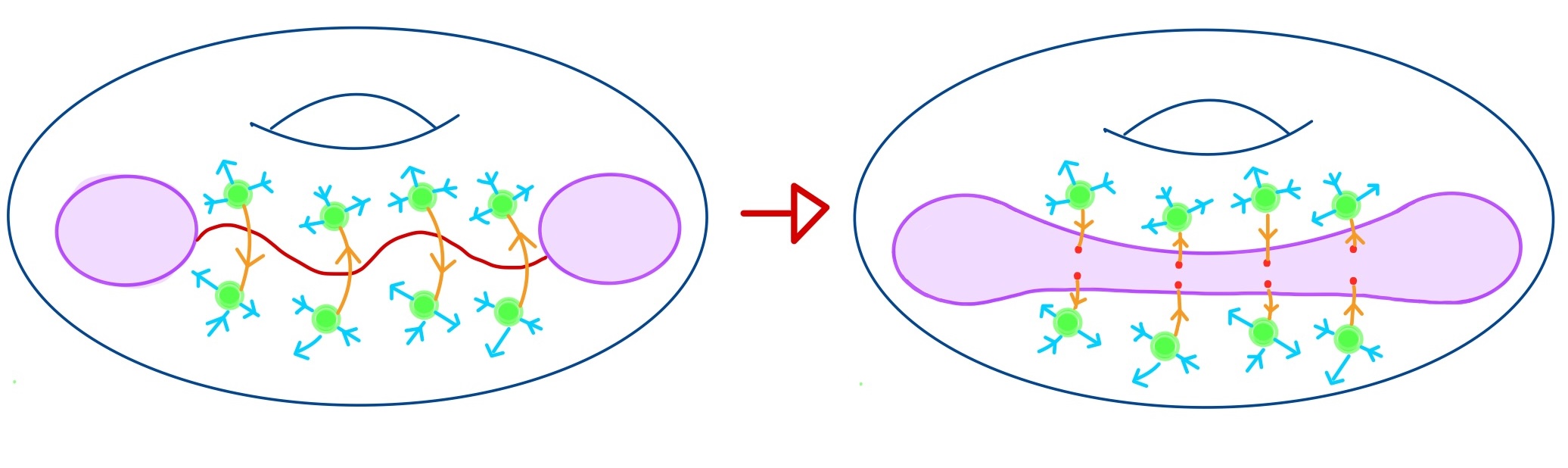}
    \caption{A schematic illustration of cutting a surface and the respective Eulerian embedding along a cuff connecting cut-line }
    \label{fig:cuff_conn_cut}
\end{figure}

\begin{definition}[Cutting a graph along non-separating cut-lines] \label{def:cutting_nonsep_lines}
    Let $(\Gamma,\nu,\omega)$ be a nice Eulerian embedding of an edge-rooted Eulerian digraph $(G,\pi(E^*))$ in some surface $\Sigma$ where $G=(V,E,\operatorname{inc})$. Let $F\subset \Sigma$ be a non-separating cut-line in $\Gamma$ with respective simple curve $\gamma: [0,1] \to F$. Let $\pi_\gamma(F) = (e_1,\ldots,e_k)$. Then by \cref{def:left_and_right} $F$ admits a left-to-right and right-to-left choice for each edge in $\rho(\Gamma)$.

    Let $e_i^+,e_i^-$ and $x_i^+,x_i^-$ be new elements for every $1 \leq i \leq k$, i.e., they are not in $V \cup E$. We define $G'=(V',E',\operatorname{inc}')$ as follows.
    \begin{align*}
        V' & \coloneqq V \cup \{x_i^+,x_i^-\mid 1 \leq i \leq k\},\\
        E' & \coloneqq E\setminus\rho(F) \cup \{e_i^+, e_i^- \mid 1 \leq i \leq k\},\\
        (e,v) \in \operatorname{inc}' &:\iff \begin{cases}
            (e,v) \in \operatorname{inc},& e\notin \rho(F),\\
            e=e_i^+, & (e_i,v) \in \operatorname{inc} , 1\leq i \leq k,\\ 
            e=e_i^-, v=x_i^- &, 1\leq i \leq k,\\ 
        \end{cases}\\
        (v,e) \in \operatorname{inc}' &:\iff \begin{cases}
             (v,e) \in \operatorname{inc},& e\notin \rho(F),\\
            e=e_i^-, & (v,e_i) \in \operatorname{inc} , 1\leq i \leq k,\\ 
            e=e_i^+, v=x_i^+ &, 1\leq i \leq k.\\ 
        \end{cases}\\
    \end{align*}

Let $E'_* \coloneqq E^* \cup \{e_i^+,e_i^- \mid 1 \leq i \leq k\}$. Finally set $\pi(E_*') \coloneqq \pi(E^*)\circ(e_1^+,e_1^-,\ldots,e_k^+,e_k^-)$. We define $\cut(G,\pi_\gamma(F)) \coloneqq (G',\pi(E_*'))$. 
\end{definition} %
\begin{remark}
Since $F$ is a cut-line, it is in particular $\Gamma$-tracing whence we do not ``cut'' through vertices and we do not ``touch'' edges, and by the left-to-right choice given by \cref{def:left_and_right} respecting the incidences of the the edges the above is a well defined graph, i.e., every edge has exactly one head and one tail (see \cref{fig:cuff_conn_cut} for an example).

Further, note that the definition formally depends on an embedding, but the embedding will always be clear when we are defining the graphs arising from cutting along cut-lines.
Intuitively this corresponds to the idea of cutting open along $F$ and splitting each edge on $\rho(F)$ into two edges, where $e_i^+$ become in-edges and $e_i^-$ become out-edges in the respective drawing.
\end{remark}

It is slightly tedious but easy to verify that we get Eulerian embeddings after cutting along cut-lines.

\begin{lemma}\label{lem:cutting_nonsep_embed}
    Let $\Sigma$ be a surface and $k \in 2\N$. Let $(G,\pi(E^*)) \in \mathbf{G^*}(\Sigma,k)$ with a nice embedding $(\Gamma,\nu,\omega)$. Let $F$ be a non-separating cut-line in $\Gamma$ with $\delta(F) = \ell \in \N$ and let $\Sigma'$ be obtained from $\Sigma$ by cutting along $F$. Then $\cut(G,\pi(E^*)) \in \mathbf{G^*}(\Sigma',k+2\ell)$ with a nice embedding $(\Gamma',\nu',\omega')$. Further it holds that $\Gamma'\subset \Gamma$ and $\nu'$ agrees with $\nu$ on $G \cap G'$.
\end{lemma}
\begin{proof}
      Since $F$ is a non-separating cut-line, $F$ is alternating by \cref{obs:edge-rooted_embedding_is_nice}. By \cref{obs:types_of_cut-line_cuts} $F$ is either (cuff-based) genus reducing or cuff connecting. We discuss the case where $F$ is cuff connecting for the other cases are analogous. Let $\gamma:[0,1]\to F$ be the respective simple curve for $F$ and fix $\pi_\gamma(F) = (e_1,\ldots,e_\ell)$. Let $(G',\pi(E_*')) \coloneqq \cut((G,\pi(E^*), \pi_\gamma(F))$; note that $\Abs{E_*'} = k+2\ell$ by \cref{def:cutting_nonsep_lines}.

     Assume that $F$ is cuff connecting, let $\zeta_1,\zeta_2 \in c(\Sigma)$ be the respective cuffs. By cutting $\Sigma$ along $F$ we obtain a connected surface $\Sigma' \subset \Sigma$ (by identifying the boundary of the new cuff respectively); let $\gamma(i) \in \zeta_i$ with $\gamma(i) \cap \Gamma = \emptyset$ for $i\in \{0,1\}$ and denote by $\zeta$ the cuff obtained by merging $\zeta_1$ and $\zeta_2$. That is, let $\zeta_i' = \zeta_i \setminus \gamma(i)$ for $i\in \{0,1\}$, then $\zeta_i'$ is a topologically open segment for $i=1,2$. Now $\zeta$ is obtained by slightly enlarging the cut $F$ resulting in two segments $\ell_1,\ell_2$ ``parallel'' to $F$ (using the left-right-definition) such that $\ell_1,\ell_2$ have no point in common and $\zeta= \zeta_1' \cup \ell_1 \cup \zeta_2' \cup \ell_2$ by identifying the ends of $\ell_i$ with one end of $\zeta_1'$ and $\zeta_2'$ respectively in the obvious way. More precisely let $\Delta \subset \Sigma$ be a choice of left and right for $F$ verifying that $F$ is alternating. In particular $F$ has both ends on $\Sigma$ and is otherwise disjoint from it and such that $\nu^{-1}(\Delta) \cap (E(G)\cup V(G)) = \rho(F)$. Let $\Delta_1,\Delta_2$ be the two components of $\Delta \setminus F$ marking left and right of $F$, let $\ell_i \coloneqq \bd(\Delta_i)$ for $i=1,2$. Further, there are  natural maps $\nu_i:\rho(F) \to \ell_i$ where $\nu_i(e_j) \in \ell_i$ such that $\nu_i(e_j) = \nu(e_j) \in F$ with the natural bijection using the homotopy between $F$ and $\ell_1$ in $\Delta_1$ for $1 \leq j \leq \ell$ and $i=1,2$. 
     
     We define an embedding $(\Gamma',\nu',\omega')$ for $(G',\pi(E_*'))$ as follows. Let $\omega'(\zeta) \coloneqq \omega(\zeta_1) \cup \omega(\zeta_2) \cup \{x_i^*, x_i^- \mid 1 \leq i \leq \ell\}$ and for the remaining cuffs of $\Sigma')$, i.e., $c(\Sigma') \cap c(\sigma)$, defined as $\omega$ for $\Sigma$. Let $\nu'$ agree with $\nu$ on $V(G)$---note that we do not embed $\{x_i^*, x_i^- \mid 1 \leq i \leq \ell\}$ on $\Sigma'$---and on $E(G) \setminus \rho(F)$ which can be achieved for $F \cap \Gamma = \nu(\rho(F))$. Let $e_i \in \rho(F)$ be left-to-right, then we define $\nu'(e_i^-) \coloneqq \nu_1(e_i)$ and $\nu'(e_i°-) = \nu_2(e_i)$ and vice versa if $e_i$ is right-to-left. 
         
     Since $F$ is alternating by \cref{obs:edge-rooted_embedding_is_nice} we derive that the newly defined cut-lines $\ell_1,\ell_2$ are alternating, since we draw in-edges and out-edges alternately on $\ell_i$ for $i=1,2$. 

     It is easily verified that every vertex of $V(G)$ is still Eulerian embedded in $\Sigma'$, for we did not alter the drawing for internal vertices locally, and we did not change their incidences in the \cref{def:cutting_nonsep_lines} up to a renaming of the edges in $\rho(F)$. Thus $(i)-(iv)$ of \cref{def:Euler_embedding_upto_root-edges} are satisfied, concluding the proof using \cref{obs:edge-rooted_embedding_is_nice}. See \cref{fig:cuff_conn_cut} for a schematic illustration.
         
\smallskip

    Finally note that if $F$ is a free or cuff-based genus reducing curve, the proof is the analogous. In the first case we create a new cuff $\zeta^*$, in the second case we enlarge a cuff $\zeta$. The embedding is obtained by duplicating all of the edges in $\rho(F)$ keeping the incidences of the edges on the respective left or right side of $F$ (following the \cref{def:left_and_right}) where incidences of the form $(e,v)$ are covered by new in-edges $(e_i^+,v)$ and incidences of the form $(v,e)$ are covered by new out-edges $(v,e_i^-)$. Again, note that locally around vertices drawn on $\Sigma$ we didn't alter the rotation of the incidences and thus the embedding remains Eulerian.

        Clearly the embeddings remain nice concluding the proof.
\end{proof}

Finally we have the following ``knitting Lemma" which guarantees us that we can lift immersions obtained after cutting back to the original surface. The lemma is in the spirit of and proved fairly similarly to \cref{thm:knitting_knitwork_immersion}.

\begin{lemma}[Cut-And-Knit Part 1.]\label{lem:knitting_nonsep_cutlines}
   Let $\Sigma$ be a surface not the sphere, disc or cylinder with $t \in \N$ cuffs and let $k \in 2\N$. Let $(G_i,\pi(E_i^*)) \in \mathbf{G}(\Sigma,k)$ be edge-rooted Eulerian digraphs with respective nice Eulerian embeddings $(\Gamma_i,\nu_i,\omega_i)$ for $i =1,2$.  For $i=1,2$ let $F_i \subset \Sigma$ be non-separating cutlines with alternating orderings $\pi(F_i)$ of $\rho(F_i)$ such that $\Sigma_1 \cong \Sigma_2$ for the surfaces resulting from cutting $\Sigma$ along $F_1,F_2$ respectively. 
   Let $(G_i',\pi(E_*^i) \coloneqq \cut((G_i,\pi(E_i^*)),\pi(F_i))$ $i=1,2$ respectively.

   If there exists a strong immersion $\gamma':(G_1',\pi(E_*^1)) \hookrightarrow (G_2',\pi(E_*^2))$, then there exists a strong immersion $\gamma:(G_1,\pi(E_1^*)) \hookrightarrow (G_2,\pi(E_2^*))$.
\end{lemma}
\begin{proof}
    Let $\pi(F_1) = (e_1\ldots,e_p)$ and $\pi(F_2) = (f_1,\ldots,f_q)$ for $p,q \in \N$. Let $F\subset \Sigma$ be a non-separating cut-line in $\Gamma$ with respective simple curve $\gamma: [0,1] \to F$. By \cref{def:cutting_nonsep_lines} we derive that $\pi(E_*^1) = \pi(E_1^*)\circ (e_1^+,e_1',\ldots,e_p^+,e_p^-)$ and $\pi(E_*^2) = \pi(E_2^*)\circ (f_1^+,f_1',\ldots,f_q^+,f_q^-)$. 
    By \cref{def:edge-rooted_immersion}, $\gamma'$ respects roots, and thus, since $\Abs{E_1^*} = \Abs{E_2^*} = k$ we derive that $p = q$ and $f_i^+ \in \gamma'(e_i^+)$ as well as $f_i^- \in \gamma'(e_i^-)$ for every $1 \leq i \leq p$. In particular the path $\gamma'(e_i^+)$ starts in $f_i^+$ and the path $\gamma'(e_i^-)$ ends in $f_i^-$ sice the edges are in- and out-edges respectively (note also that one of their endpoints has degree $1$ by construction).

Defining $\gamma$ on $V(G)$ via $\gamma(v) \coloneqq \gamma(v')$ and defining $\gamma(e) \coloneqq \gamma'(e)$ for every $e \notin \rho(F_1)$ and finally letting $\gamma(e_i) \coloneqq \gamma'(e_i^-) \circ \gamma'(e_i^+)$ by identifying $f_i^-$ and $f_i^+$ to $f_i$ for every $1 \leq i \leq k$ does the trick. One easily verifies that $\gamma$ is a strong immersion by construction and it is clearly rooted for $\gamma(e)=\gamma'(e)$ for every $e \in E^*$ where $\gamma'$ is a rooted strong immersion.
\end{proof}

We continue with a second Cutting and Knitting Lemma for separating cut-lines (and thus for cut-cycles) in the same spirit to cover the remaining types of curves in \cref{def:cutline_types}. (Note that this case is conceptually very similar to stitching rooted graphs along rooted cuts as discussed in \cref{subsec:knitworks}). The construction and respective lemmas are all fairly similar. Fortunately, for separating curves there is no need to keep track of left and right any more since by the \cref{def:left_and_right} of cut-lines, every edge that is part of $\rho(F)$ has its incidence split into both resulting surfaces.   

\begin{figure}
    \centering
    \includegraphics[width=0.4\linewidth]{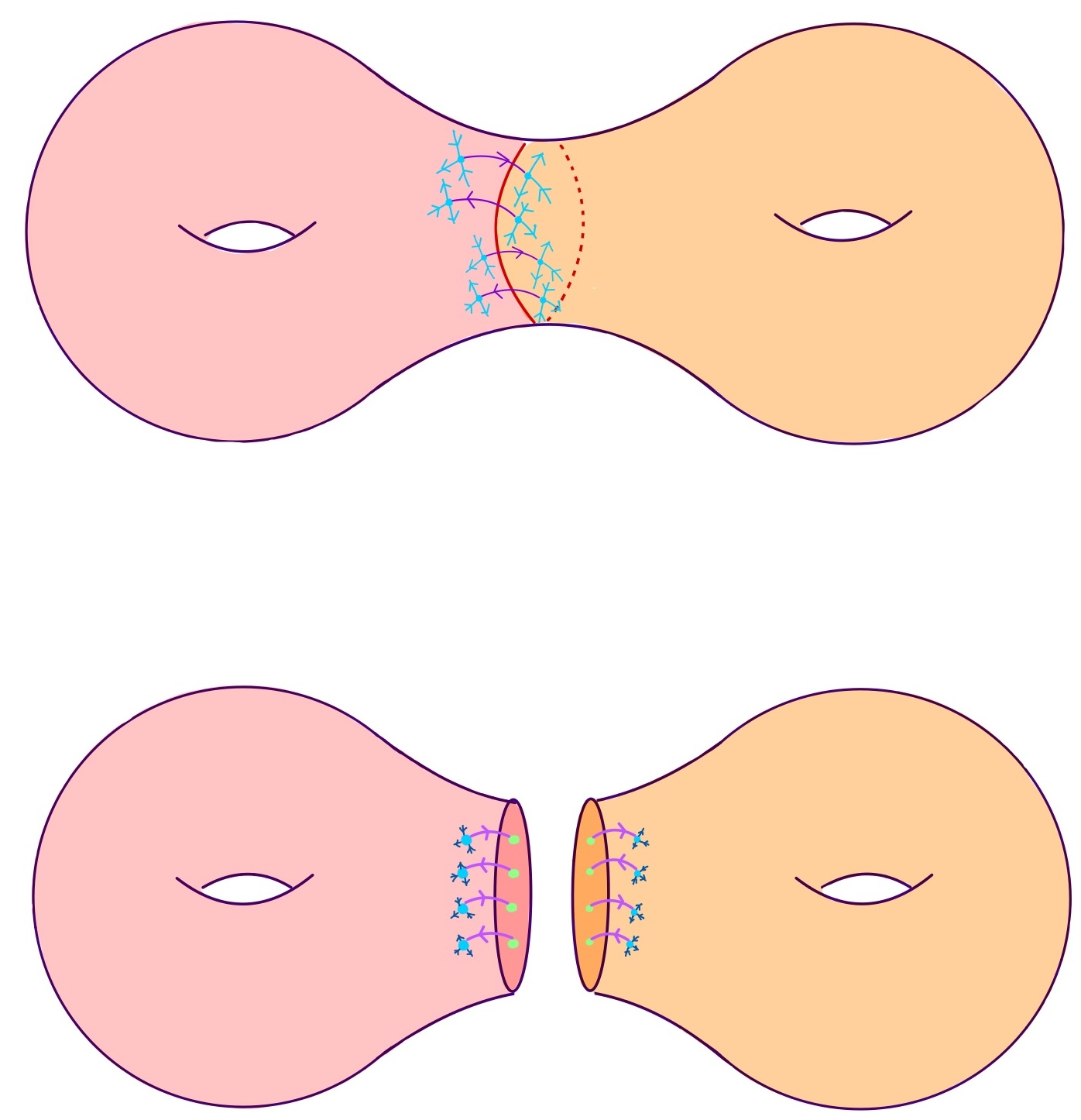}
    \caption{A schematic illustration of the surface and respective Eulerian embedding arising from cutting along a free separating (genus reducing) cut-line.}
    \label{fig:double-torus-cut}
\end{figure}
\begin{definition}[Cutting a graph along separating cut-lines] \label{def:cutting_sep_lines}
    Let $(\Gamma,\nu,\omega)$ be a nice embedding of an edge-rooted Eulerian digraph $(G,\pi(E^*))$ in some surface $\Sigma$ where $G=(V,E,\operatorname{inc})$. Let $F\subset \Sigma$ be a separating cut-line in $\Gamma$ with respective simple curve $\gamma: [0,1] \to F$. Let $\pi_\gamma(F) = (e_1,\ldots,e_k)$.

    Let $\Sigma_1,\Sigma_2$ be the two closed surfaces obtained from cutting $\Sigma$ along $F$. We set $X_i(F) \coloneqq \nu^{-1}(\Sigma_i) \cap V(G)$ and $E_i \coloneqq  \nu^{-1}(\Sigma_i) \cap E(G)$; note that $\rho(F) \subset E_i$ for $i=1,2$ and $E_1 \cap E_2 = \rho(F)$. Let $E_i^* \coloneqq E^* \cap E_i$ for $i=1,2$.

    Let $x_i^+,x_i^-$ be new elements for every $1 \leq i \leq k$, i.e., they are not in $V \cup E$.  We define $G_i=(V_i,E_i,\operatorname{inc}_i)$ such that $X_i(F) \subset V_i$ and, if for $e \in E_i$ and $v\in V_i$ it holds $(e,v) \in \operatorname{inc}$, then $(e,v) \in \operatorname{inc}_i$ and if $(v,e) \in \operatorname{inc}$ then $(v,e) \in \operatorname{inc}_i$ for $i =1,2$. Further for $\{p,q\}=\{1,2\}$ and for every $1 \leq i \leq k$ if $\head(e_i) \in \Sigma_p$ we let $(e_i,x_i^-) \in \operatorname{inc}_q$, and if $\tail(e_i) \in \Sigma_p$ we let $(x_i^+,e_i) \in \operatorname{inc}_q$.

    Let $E_*^i \coloneqq E_i^* \cup \rho(F)$ for $i=1,2$. And finally set $\pi(E_*^i) \coloneqq \pi(E_i^*)\circ (e_1, \ldots,e_k)$ for $i=1,2$. 
    
    We define $\cut(G,\pi(E^*)),\pi_\gamma(F)) \coloneqq \big((G_1,\pi(E_*^1)), (G_2,\pi(E_2^*))\big)$.
\end{definition}
\begin{remark}
    See \cref{fig:double-torus-cut} for a schematic illustration of cutting along a separating cut-line.
\end{remark}

Again we may verify that the resulting edge-rooted Eulerian digraphs from cutting along separating cut-lines are Eulerian embeddable in the respective surfaces by construction using \cref{obs:edge-rooted_embedding_is_nice}.

\begin{lemma}\label{lem:cutting_sep_embedding}
    Let $\Sigma$ be a surface and $k \in 2\N$. Let $(G,\pi(E^*)) \in \mathbf{G^*}(\Sigma,k)$ with a respective nice embedding $(\Gamma,\nu,\omega)$. Let $F$ be a separating cut-line in $\Gamma$ with $\delta(F) = \ell \in \N $ and let $\Sigma_1,\Sigma_2 \subset \Sigma$ be the surfaces obtained by cutting $\Sigma$ along $F$. Let $\cut(G,\pi(E^*)) = \big((G_1,\pi(E_*^1)), (G_2,\pi(E_2^*))\big)$ then 
    $(G_1,\pi(E_*^1))\in \mathbf{G^*}(\Sigma_1,\ell_1)$ and  $(G_2,\pi(E_*^2))\in \mathbf{G^*}(\Sigma_2,\ell_2)$ with nice embeddings $(\Gamma_i,\nu_i,\omega_i)$ for $i=1,2$, where $\ell_1,\ell_2 \in 2\N$ and $\ell_1 + \ell_2 = k+2\ell$. Further $\Gamma_i \subset \Gamma$ and $\nu_i$ agrees with $\nu$ on $G \cap G_i$ for $i=1,2$.
\end{lemma}
\begin{proof}
Since $F$ is a cut-line it is in particular $\Gamma$-tracing and $X_1(F) \cup X_2(F)$ is a partition of $V(G)$ by \cref{obs:edge-rooted_embedding_is_nice}. Note further that by the same \cref{obs:edge-rooted_embedding_is_nice} it suffices to prove $(i) - (iv)$ of \cref{def:Euler_embedding_upto_root-edges}, i.e., that there exist respective embeddings of $G_1,G_2$ where $\pi(E_*^1)$ and $\pi(E_2^*)$ are the only edges drawn on the cuffs of $\Sigma_1,\Sigma_2$ respectively where every vertex that is not of degree one is Eulerian embedded inside $\Sigma_1,\Sigma_2$. To see this, note that since all of the vertices $X_1(F)$ drawn inside $\Sigma_1$ induce a cut in $G$, and since $\stitch(G;\pi(E^*))$ is Eulerian, the cut has even order. Given the planar drawing of $G$ in $\Sigma$ each of these cut-edges must be drawn on $\bd(\Sigma_1)$, thus $\ell_1$ is even, and analogously for $\ell_2$. Thus $\ell_1,\ell_2 \in 2\N$, where $\ell_1 + \ell_2 = k+ 2\ell$ follows from the fact $\pi(E_*^i)$ are the only edges drawn on the cuffs for $i=1,2$ where $\Abs{E_*^1} + \Abs{E_*^2} = k + 2 \ell$ by definition.

    Thus we are left to prove the existence of the respective (nice) embeddings. This is fairly straightforward since $\Sigma_1,\Sigma_2$ partition the vertices of $\nu^{-1}(\Gamma)$ in a natural way without altering their local embedding. We discuss it briefly for $\Sigma_1$. By \cref{def:cutting_sep_lines} we may define $\nu_1:G \to \Sigma_1\subset \Sigma$ by letting $\restr{\nu_1}{X_1(F)} = \nu$ using the obvious injective inclusion map $\iota: \Sigma_1 \to \Sigma$ for $\Sigma_1 \subset \Sigma$. Similarly, since $E(G_1) = \nu_1^{-1}(\Sigma_1)\cap E(G)$ we may keep $\nu_1=\nu$ on $E(G_1)$. This completely defines $\nu_1$. We set $\omega_1$ to be defined on $c(\Sigma_1)$ by letting it agree with $\omega$ on $c(\Sigma_1) \cap c(\Sigma)$ and letting $c(\zeta) \coloneqq \rho(F) \cup \nu^{-1}(\ell_1)$. This completely defines the embedding and since we did not alter any of the incidences or rotations around vertices in $V(G_1)$ (note that we did not introduce any new edges), the embedding satisfies $(i) - (iv)$. This concludes the proof with analogous arguments for $\Sigma_2$. 
    
\end{proof}

\begin{figure}
\centering
\begin{subfigure}{0.4\textwidth}
\centering
    \includegraphics[width=0.8\linewidth]{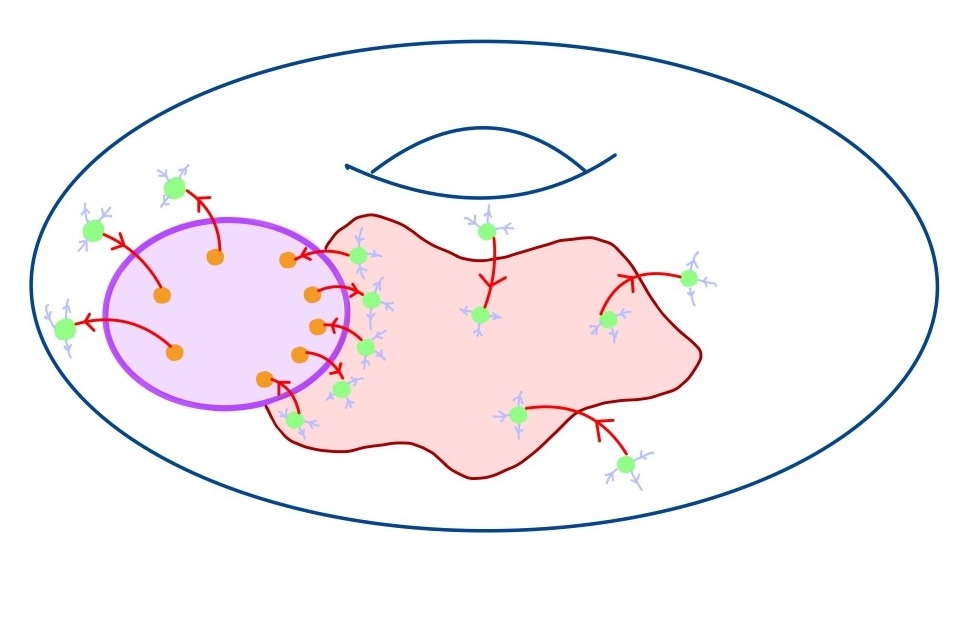}
    
    \
\end{subfigure}
\begin{subfigure}{0.4\textwidth}
\centering
    \includegraphics[width=0.9\linewidth]{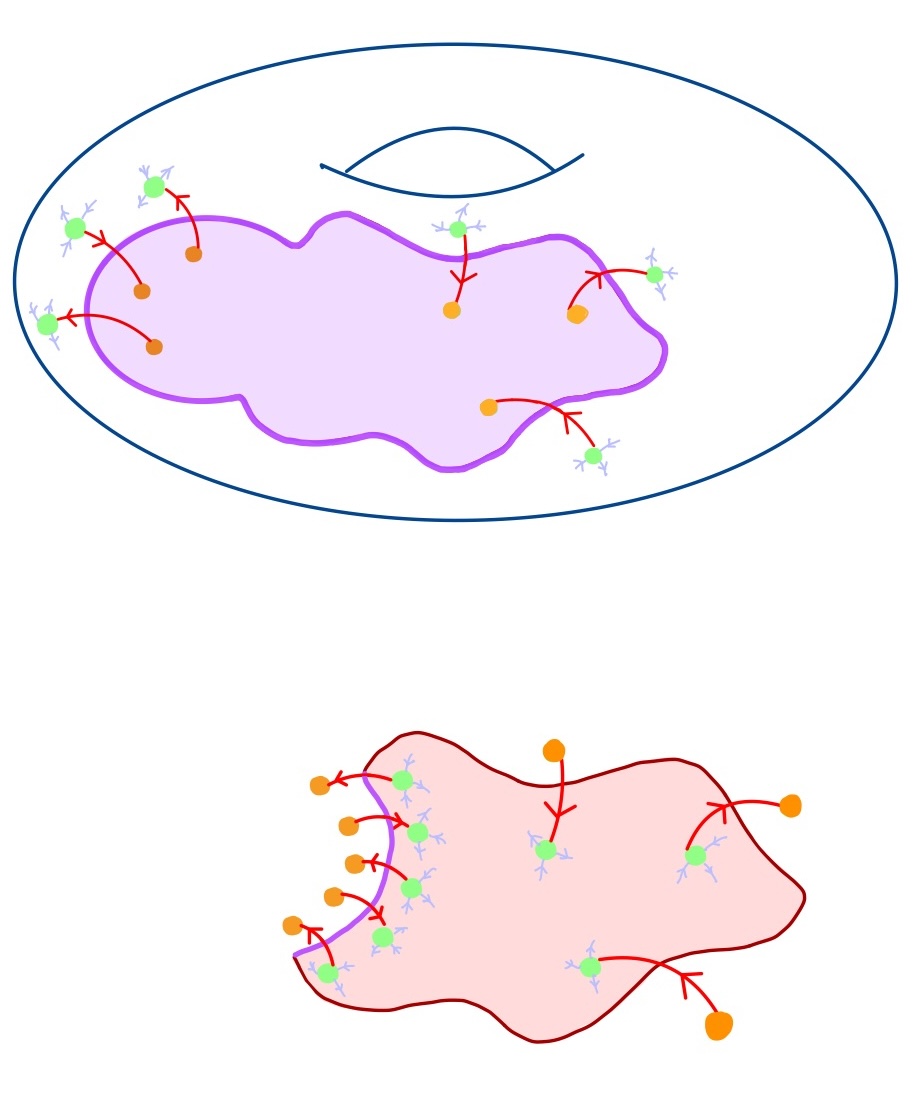}
    
\end{subfigure}
\caption{A schematic illustration of cutting a surface an d the respective Eulerian embedding along a cuff-base separating cut-line.}
\label{fig:cuff-based}
\end{figure}

We again prove a ``Knitting-Lemma'' allowing us to lift strong immersions.

\begin{lemma}[Cut-And-Knit Part 2.]\label{lem:knitting_separating_cutlines}
   Let $\Sigma$ be a surface not the sphere, disc or cylinder with $t \in \N$ cuffs and let $k \in 2\N$. Let $(G_i,\pi(E_i^*)) \in \mathbf{G}(\Sigma,k)$ be edge-rooted Eulerian digraphs with respective nice Eulerian embeddings $(\Gamma_i,\nu_i,\omega_i)$ for $i =1,2$.  For $i=1,2$ let $F_i \subset \Sigma$ be separating cutlines with alternating orderings $\pi(F_i)$ of $\rho(F_i)$. Let $\Sigma_1^i,\Sigma_2^i$ be the two surfaces resulting from $\Sigma$ by cutting along $F_i$ for $i=1,2$. Further assume that $\Sigma_1^i \cong \Sigma_2^i$ for $i=1,2$. Let $\big((G_1^i,\pi(E_1^i)), (G_2^i,\pi(E_2^i))\big) \coloneqq \cut((G_i,\pi(E_i^*)),\pi(F_i))$ for $i=1,2$ respectively. 
   
   If there exist strong immersions $\gamma_i:(G_1^i,\pi(E_1^i)) \hookrightarrow (G_2^i,\pi(E_2^i))$ for $i=1,2$, then there exists a strong immersion $\gamma:(G_1,\pi(E_1^*)) \hookrightarrow (G_2,\pi(E_2^*))$.

\end{lemma}
\begin{proof}
    This follows analogously to \cref{thm:knitting_knitwork_immersion}. To see this, let $E_{\omega_i}^* = (E_i^1,\ldots,E_i^t)$ and $W_{\omega_i}=(W_i^1,\ldots,W_i^t)$ be the respective partitions for $i=1,2$ as in \cref{def:Euler_embedding_upto_root-edges}. Recall \cref{def:stitch_edge-rooted} and let $H_j$ be the directed Eulerian digraph of $\stitch(G_j;\pi(E_{\omega_j}^*))$ and let $W_j \coloneqq \{w_1^j,\ldots,w_t^j\}$ for $j=1,2$. Let $\Omega$ be any well-quasi-order and define a reliable  well-linked $\Omega$-knitwork $\HHH_j=\big((H_j,\pi(W_j)),\mu_j,\m_j,\Phi_j)$ in the obvious way by choosing $\pi(W_j)$ arbitrary, setting $\dom(\mu_j) = W_j$ and letting $\mu_j(w_i) = \pi(E_j^i) \subseteq \pi(E_j^*)$, $\dom(\m_j) = \emptyset$ and fixing $x \in V(\Omega)$ and letting $\Phi_j(v)= x$ for every $v\in V(H_j)$ and $j=1,2$ and $1 \leq i \leq t$. One immediately verifies that $(G_1,\pi(E_1^*)) \hookrightarrow (G_2,\pi(E_2^*))$ if and only if $\HHH_1 \hookrightarrow \HHH_2$.
    
    Similarly, let $H_j^i$ denote the graph obtained from $G_j^i$ by stitching the respective partition induced by $E_j^i$ for $i=1,2$. That is, let $\omega_j^i$ be the respective maps of the induced embedding as given by \cref{lem:cutting_sep_embedding} for $j=1,2$ and $i=1,2$, and let $E_{\omega_j^i}=(E^j_{\ell^i_1},\ldots,E^j_{\ell^i_{p_i}},E_{i_*}^j)$ be the respective partitions for $E_j^i$, where $\{\ell_1^1,\ldots,\ell^1_{p_1}\},\{\ell_1^2,\ldots,\ell^2_{p_2}\}$ is a partition of $\{1,\ldots,t\}$, in particular $p_1 + p_2 = t$, and $E_{i_*}^j = \rho(F_j)$. Note that $H_j^1 \cap H_j^2 = \rho(F_j)$ by construction for $j=1,2$. Then $H_j^i$ is the graph of $\stitch(G_j^i;\pi(E_{\omega_j^i}))$. Let $W_j^i=(w^j_{\ell_1^i},\ldots,w^j_{\ell_{p_i}^i},w^j_{i_*})$ for $j=1,2$ and $i=1,2$.  Finally we define the reliable well-linked $\Omega$-knitworks $\HHH_j^i=\big((H_j,\pi(w^j_{i_*})),\mu_j^i,\m_j^i,\Phi_j^i)$ analogously to above by setting $\mu_j^i(w^j_{\ell_r^i}) = \pi(E^j_{\ell^i_r})\subseteq \pi(E_j^*)$ for $1 \leq r \leq p_i$ and $j=1,2$ as well as $i=1,2$, and $\dom(\m_j^i) = \emptyset$ and $\Phi_j^i(v)=x$ where defined.
    
    Then the claim follows by \cref{thm:knitting_knitwork_immersion}, noting that there are a rooted $k$-cuts $X_1,X_2$ in $\bar H_1,\\bar H_2$ respectively, such that $\rho(X_i) = \rho(F_i)$ and such that $\stitch(\HHH_j; \pi(X_j)) = \HHH_j^1$ and $\stitch(\HHH_j;\pi(\bar X_j)) = \HHH_j^2$ for $\pi(X_j) = \pi(F_j) = \pi(\bar X_j)$ and $j=1,2$. (See \cref{fig:cuff-based} and \cref{fig:double-torus-cut} for schematic illustrations of two cases).
\end{proof}

\subsection{Chopping up the Cylinder}
    Using the results of the previous section, the general idea is now as follows. Given a Eulerian embedding $(\Gamma,\nu,\omega)$ of an edge-rooted Eulerian digraph $(G,\pi(E^*))$ in a cylinder $\Sigma$ with cuffs $\zeta_1,\zeta_2$, we use  \cref{lem:boundary_linked_Menger_for_embeddings_in_cylinder} in order to inductively decompose $(G,\pi(E^*))$ along all cuff separating cut-cycles that have lower or equal order than the number of edges on the respective cuff. This results in a sequence $(F_1, \ldots,F_t)$ of cuff separating cut-cycles with $\zeta_1 \subset \Delta(F_i) \subseteq \Delta(F_j) \in \Sigma + \Delta_1$  for $i<j$ where $\Sigma + \Delta_1$ is the surface obtained by gluing the disc $\Delta_1$ to cap the cuff $\zeta_1$. Then every two consecutive cut-cycles induce a \emph{piece of $\Gamma$} that is again an edge-rooted Eulerian digraph Euler-embeddable in the cylinder. We prove that the pieces fall into different categories, where each category is well-quasi-ordered by strong immersion on its own. Finally we use arguments similar to  Higman's \cref{thm:higman} and in the proof of \cref{thm:wqo_bounded_carvingwidth_knitworks} to ``knit'' the strong immersions of the pieces back together. In essence, ``chopping'' is similar to ``carving'' (see \cref{def:carving}) where we will end up with a path of chopped pieces instead of a tree of carved pieces: Unfortunately the definition of carvings does not readily lift to paths, whence we introduce \emph{choppings}.

 Throughout this section let $\Sigma$ denote the cylinder with cuffs $\zeta_1,\zeta_2$ such that $\hat \Sigma = \Sigma_1 + \Delta_1 + \Delta_2$ where $\Delta_i$ are discs and $\hat\Sigma$ is obtained by gluing $\Delta_i$ to $\Sigma$ along the cuff $\zeta_i$ for $i=1,2$ as usual. Further, to simplify notation, given an edge-rooted Eulerian digraph $(G,\pi(E^*)) \in \mathbf{G^*}(\Sigma,k)$ with a respective Eulerian embedding $(\Gamma,\nu,\omega)$ and parirtion $E_\omega^*=(E_1,E_2)$ we will simply write $(G,\pi(E_1), \pi(E_2))$ where $(E_1,E_2) = E^*_\omega$ to simplify notation, where $\pi(E_1),\pi(E_2) \subset \pi(E^*)$ are the respective induced orderings. Since we always work with embeddings we will simply write $(G,\pi(E_1), \pi(E_2))$ for edge-rooted Eulerian digraphs to mean the obvious and write $(\Gamma,\nu)$ for the embedding where the definition of $\omega$ is clear from the respective partition $(E_1,E_2)$. Note that $\nu(E_i) \subset \zeta_i$ by definition and $\Abs{E_1},\Abs{E_2} \in 2\N$ by \cref{obs:edge-rooted_embedding_is_nice}; recall that $\zeta_i$ is an alternating cut-line by \cref{lem:cuffs_are_alternating}.
  Further we refine the definition of $\mathbf{G^*}(\Sigma,k)$ for the cylinder to match the new notation in the obvious way. 
  
 \begin{definition}
We define $\mathbf{G^*}(\Sigma,k_1,k_2)$ to be the class of edge-rooted Eulerian digraphs $(G,\pi(E_1),\pi(E_2))$ with Eulerian embeddings in the cylinder where $\Abs{E_i} = k_i \in 2\N$ for $i=1,2$.  
 \end{definition}

We have the following base case capturing a special case for ``low representativity''.
\begin{lemma}\label{lem:wqo_cylinder_no_small_cuts}
    Let $\Sigma$ be the cylinder with cuffs $\zeta_1,\zeta_2$. Let $k_1,k_2 \in 2\N$ and let $(G_i,\pi(E_1^i),\pi(E_2^i)) \in \mathbf{G}(\Sigma,k_1,k_2)$ with Eulerian $2$-cell embeddings $(\Gamma_i,\nu_i)$ for every $i \in \N$. Assume that $((G_i,\pi(E_1^i),\pi(E_2^i)))_{i \in \N}$ is a bad sequence with respect to strong immersion. 
    Then for every function $f:\N \to \N$ there is an infinite index set $I \subseteq \N$ such that $\Abs{E_1^i},\Abs{E_2^i} \geq 1$, and there is no cuff connecting cut-line $F_i$ in $\Gamma_i$ such that $\delta_{\Gamma_i}(F_i) < f(k_1 + k_2)$.
\end{lemma}
\begin{proof}
Let $k \coloneqq k_1 + k_2$.
    Assume the contrary and let $f:\N \to \N$ be a respective function. First assume that there is an infinite index $I \subseteq \N$ such that $E_1^i = \emptyset$ (the case for $E_2^i$ is analogous). Then $(\Gamma_i,\nu_i)$ induces a respective embedding in the disc $\Sigma + \Delta_1$ for every $i \in I$. Thus, $(G_i,\pi(E_2^i)) \in \mathbf{G^*}(\Delta,k_2)$ for some disc $\Delta$ for every $i \in I$; a contradiction to \cref{thm:wqo:bead-root_for_disc} together with \cref{cor:i_am_dying}.

    Thus assume that there is an infinite index $I \subseteq \N$ such that for every $i \in I$ there is a cuff connecting cut-line $F_i$ such that $\ell_i \coloneqq \delta_{\Gamma_i}(F_i) < f(k)$. In particular note that $f(k) \neq 0$, since otherwise we can cut $\Sigma$ along the respective curve without altering the embedding, resulting in $\Sigma'\cong \Delta$ yielding an embedding in the disc, a contradiction to \cref{thm:wqo:bead-root_for_disc} once more. Thus, we may assume that, by switching to a respective subsequence, for every $i \in I$ there is no cuff connecting curve of length $0$. In particular since the embeddings are assumed to be $2$-cell, the embedding is nice. 
    
    Let $i \in I$ be fixed and let $G_i = (V_i,E_i,\operatorname{inc}_i)$.  Recall that $F_i$ is alternating (by \cref{obs:edge-rooted_embedding_is_nice}) and let $\pi(F_i) = (e_1,\ldots,e_{\ell_i})$ be a respective ordering of $\rho(F_i)$, where $e_i$ is left-to-right if and only if $e_{i+1}$ is right-to-left for $1 \leq i < k$. Note that cutting $\Sigma$ along $F_i$ results in a new surface $\Sigma' \cong \Delta$. Let $(G_i',\pi(\tilde E_1^i), \pi(\tilde E_2^i))$ be the edge-rooted Eulerian digraphs resulting from $\cut\Big(\big(G_i,\pi(E_1^i),\pi(E_2^i)\big), \pi(F_i)\Big)$ as in \cref{def:cutting_sep_lines}. %

We have the following.
\begin{claim}
    There is an infinite index set $I' \subseteq I$ such that $\Abs{\tilde E_p^i} = \Abs{\tilde E_p^j}=\ell$ for $p =1,2$ and every $i,j \in I'$.
    
\end{claim}
\begin{claimproof}
 Since $\delta(F_i) \leq f(k)$ for every $i \in I$, by the pigeonhole principle there is an infinite index $I' \subseteq I$ such that $\delta(F_i) = \delta(F_j)$ for $i, j \in I'$. Since $\Abs{\tilde E_1^i} + \Abs{\tilde E_2^i} = k + 2\delta(F_i)$ by \cref{lem:cutting_nonsep_embed} for every $i \in I$, $I'$ satisfies the claim.
\end{claimproof}

Let $\ell \coloneqq \delta(F_i)$ for $i \in I'$. By \cref{lem:cutting_sep_embedding}  $(G_j',\pi(\tilde E_1^j),\pi(\tilde E_2^j)) \in \mathbf{G}(\Sigma',k + 2\ell)$ for every $j \in I'$. By \cref{thm:wqo:bead-root_for_disc} there are $i,j \in I'$ with $i<j $ such that $\gamma: (G_i',\pi(\tilde E_1^i),\pi(\tilde E_2^i))  \hookrightarrow (G_j',\pi(\tilde E_1^j),\pi(\tilde E_2^j)) $ is a strong immersion. By \cref{lem:knitting_nonsep_cutlines} then $(G_i,\pi( E_1^i),\pi( E_2^i)) \hookrightarrow (G_j,\pi( E_1^j),\pi( E_2^j))$; contradiction to the choice of bad sequence.
\end{proof}

The remainder of the section is devoted to chopping up embeddings on the cylinder along cut-cycles into pieces---using the Menger-type \cref{lem:boundary_linked_Menger_for_embeddings_in_cylinder} for the cylinder---that either satisfy \cref{cor:cyl-high-rep} or \cref{lem:wqo_cylinder_no_small_cuts}.

\begin{definition}[Choppings and Pieces of $\Gamma$]
        Let $(G,\pi(E_1),\pi(E_2))$ be an edge-rooted Eulerian digraph with Eulerian embedding $(\Gamma,\nu)$ in $\Sigma$. Let $\tau ,t \geq 1$.  Let $F_0\coloneqq \zeta_1, F_{t+1} \coloneqq \zeta_2$. 
        
        Let $\FFF \coloneqq (F_0,F_1,\ldots,F_t,F_{t+1})$ be a sequence of cut-cycles for $\Gamma$ such that $\delta(F_i) \leq \tau$ for every $0 \leq i \leq t+1$ and such that $\zeta_1 \subset \Delta(F_i) \subset \Delta(F_j)$ are strict subsets for every $1 \leq i < j\leq t$ and some $t\geq 1$, where $\Delta(F_i) \subset \Sigma +\Delta_1$ is the respective disc (note that this is well-defined for $F_0,F_{t+1}$). Then we call $\FFF$ \emph{a $\tau$-chopping of $\Gamma$ (of length $t$)}.

        Let $\Pi_\FFF \coloneqq (\pi(F_1),\ldots,\pi(F_t))$ where $\pi(F_i)$ are orderings of $\rho(F_i)$ for every $1 \leq i \leq t$ respectively. Let $\big((G_0,\pi(E_1^0),\pi(E_2^0)),(H_1,\pi(F_1^1),\pi(F_2^1))\big) \coloneqq \cut((G,\pi(E_1),\pi(E_2)), \pi(F_1))$. 
        
        Inductively we define $((G_i,\pi(E_1^i),\pi(E_2^i)),(H_{i+1},\pi(F_1^{i+1}),\pi(F_2^{i+1})) \coloneqq \cut((H_i,\pi(F_1^i),\pi(F_2^i)), \pi(F_{i+1}))$ for every $1 \leq i <t$. Finally, let $(G_t,\pi(E_1^t), \pi(E_2^t)) \coloneqq (H_t,\pi(F_1^t),\pi(F_2^t))$. 
        
        We call $(G_0, \pi(E_1^0),\pi(E_2^0)),\ldots,(G_t, \pi(E_1^t),\pi(E_2^t))$ \emph{the $\FFF$-pieces of $\Gamma$ (with respect to $\Pi_\FFF$)} and refer to a single piece $(G_i, \pi(E_1^i),\pi(E_2^i))$ as the \emph{$i$-th piece (of $\Gamma$ with respect to $\FFF$)} for $0 \leq i \leq t$.

        We call $(G,\pi(E_1),\pi(E_2))$ \emph{linked} if there is a $\{E_1,E_2\}$-linkage of order $\min(\Abs{E_1},\Abs{E_2})$. Analogously, for every $0 \leq i \leq t$ we call the $i$-th piece  \emph{linked} if there is a $\{\rho(F_i),\rho(F_{i+1})\}$-linkage in $(G_i, \pi(E_1^i),\pi(E_2^i))$ of order $\min(\delta(F_i),\delta(F_{i+1}))$. 
        
        We call the chopping \emph{$\ell$-chubby} if every piece is linked and further $\delta(F_i) = \ell$ for every $0 \leq i \leq t+1$.
\end{definition}
Henceforth, whenever we give a chopping $\FFF=(F_0,\ldots,F_{t+1})$ we implicitly assume $F_0 = \zeta_1$ and $F_{t+1} = \zeta_2$.

Next we define the relevant types of pieces suited for \cref{lem:cyclinder_high_rep} and \cref{lem:wqo_cylinder_no_small_cuts} respectively; we start with \emph{fat} pieces.

\begin{definition}[$\theta$-fat]\label{def:theta-fat}
     Let $(G,\pi(E_1),\pi(E_2))$ be an edge-rooted Eulerian digraph with Eulerian embedding $(\Gamma,\nu)$ in the cylinder $\Sigma$. Let $\theta \geq 1$.
     We call~$(G,\pi(E_1),\pi(E_2))$ \emph{$\theta$-fat} if it is linked, and there exists $\theta' \geq \theta$  and a linked $\theta'$-chopping $(F_0,F_1,\ldots,F_\theta,F_{\theta+1})$ of $\Gamma$ of length $\theta$ such that $F_i \cap F_j = \emptyset$ for any two distinct $1 \leq i,j \leq \theta$ and further $\delta(F_i) \geq \theta'$ for every $1 \leq i \leq \theta$. 
\end{definition}

We have the following reformulation of \cref{lem:cyclinder_high_rep}.
\begin{corollary}\label{lem:fat_are_good}
    Let $(G_i,\pi(E_1^i),\pi(E_2^i)) \in \mathbf{G^*}(\Sigma,k_1,k_2)$ for $k_1,k_2 \in 2\N$. Then there exists a function $f:\N\to \N$ such that the following holds. Let $n_1\coloneqq \Abs{V(G_1)}$. If $(G_2,\pi(E_1^2),\pi(E_2^2))$ is $f(n)$-fat, then $(G_1,\pi(E_1^1),\pi(E_2^1)) \hookrightarrow (G_2,\pi(E_1^2),\pi(E_2^2))$ by strong immersion.
\end{corollary}
\begin{proof}
    Let $f(k) \coloneqq f_{\ref{lem:cyclinder_high_rep}}(k) + 1$ then the claim follows by definition of $f(n)$-fatness. Note that since the cut-cycles are disjoint, \cref{lem:vf-path-euler-circle} provides the necessary alternating circles, where the existence of the disjoint paths follows from the linkedness of the chopping.
\end{proof}

We next deal with \emph{short} pieces.

\begin{definition}[$\theta$-short] \label{def:theta-short}
     Let $(G,\pi(E_1),\pi(E_2))$ be an edge-rooted Eulerian digraph with Eulerian embedding $(\Gamma,\nu)$ in a cylinder $\Sigma$. Let $\theta \geq 1$.
     We call~$(G,\pi(E_1),\pi(E_2))$ \emph{$\theta$-short} if it is linked, and there is a cuff connecting curve $F$ in $\Gamma$ such that $\delta(F) \leq \theta$.
\end{definition}

Again we have the following corollary due to \cref{lem:wqo_cylinder_no_small_cuts}; for $k_1,k_2,\theta \in \N$ define $\AAA(\theta,k_1,k_2) \subset \mathbf{G}(\Sigma,k_1,k_2)$ to be the maximal subclass such that every $(G,\pi(E_1),\pi(E_2)) \in \AAA(\theta,k_1,k_2)$ is $\theta$-short. 

\begin{corollary}\label{lem:short_are_good}
    Let $k_1,k_2 \in 2\N$ and $\theta \in \N$. Then the class $\AAA(\theta,k_1,k_2)$ is well-quasi-ordered by strong immersion, and every element of the class is $\theta$-short.
\end{corollary}
\begin{proof}
 The claim follows immediately from \cref{lem:wqo_cylinder_no_small_cuts} where the linkedness follows by \cref{def:theta-short}.
\end{proof}

Next we analyse the structure of embeddings that are not $\theta$-fat, proving that they decompose into a chopping of linked $\theta$-short pieces.

\begin{lemma}\label{lem:short_decomposition}
         Let $k_1,k_2 \in 2\N$. Let $(G,\pi(E_1),\pi(E_2)) \in \mathbf{G^*}(\Sigma,k_1,k_2)$ be an edge-rooted Eulerian digraph with Eulerian embedding $(\Gamma,\nu)$ in $\Sigma$. Let $\theta \geq 1$ and assume that $(G, \pi(E_1),\pi(E_2))$ is linked but not $\theta$-fat. Then there exist $t \in \N$ and a $\theta$-chopping $(F_0,F_1,\ldots,F_t,F_{t+1})$ of $\Gamma$ of length $t$ such that every piece is $\theta$-short.
\end{lemma}
\begin{proof}
    We construct the desired chopping iteratively as follows.
    Without loss of generality assume that $k_1 \leq k_2$. Since $(G,\pi(E_1),\pi(E_2))$ is linked either it is $\theta$-short, or we find some cut-cycle $F_1$ in $\Gamma$ with $\delta_1 = \delta(F_1) < \theta$ such that $\zeta_1 \in \Delta(F_1) \subset \Sigma + \Delta_1$ and $\delta_1 $ minimal. Note that if it is neither $\theta$-short, nor does there exist a respective cut-cycle $F_1$ of order $< \theta$ as claimed---whence $\kappa(\zeta_1,\zeta_2;\Gamma) \geq \theta$ by \cref{lem:boundary_linked_Menger_for_embeddings_in_cylinder}---$(G,\pi(E_1),\pi(E_2))$ is $\theta$-fat as a contradiction to our assumption. To see this note that it is linked by assumption, and since it is not $\theta$-short, there is no cuff connecting curve $F^*$ in $\Gamma$ with $\delta(F) \leq \theta$. The latter implies the existence of $(\theta+1)$ pairwise non-intersecting cut-cycles $F_1,\ldots,F_{\theta+1}$ with $\zeta_1 \subset \Delta(F_1) \subset \ldots \subset \Delta(F_{\theta+1}) \subset \Sigma + \Delta_1$---this is easily seen by induction---each satisfying $\delta(F_i) \geq \theta$ for $1 \leq i \leq \theta+1$. By \cref{lem:boundary_linked_Menger_for_embeddings_in_cylinder} then---since $\kappa(\zeta_1,\zeta_2;\Gamma) \geq \theta$---there is a $\theta$-linkage connecting $F_1$ to $F_{\theta+1}$. %
    
    Let $\pi(F_1)$ be an ordering of $\rho(F_1)$. Since $(G,\pi(E_1),\pi(E_2)$ is linked we derive $\delta_1 \geq k_1$. 
    
    Let $\big((G_1,\pi(E_1^1),\pi(E_2^1), (G_2,\pi(E_2^1),\pi(E_2^2))\big) = \cut(G,\pi(E_1),\pi(E_2)),\pi(F_1))$ with respective embeddings $(\Gamma_i,\nu_i)$ as in \cref{lem:cutting_sep_embedding} for $i=1,2$. 

    \begin{claim}\label{claim:linked_cut_linked}
        $(G_i,\pi(E_1^i),\pi(E_2^i))$ is linked.
    \end{claim}
    \begin{claimproof}
        Sine we chose $F_1$ such that $\delta_1$ is minimal, the claim follows by \cref{lem:boundary_linked_Menger_for_embeddings_in_cylinder}, since cutting $\Sigma$ along $F_1$ results again in a cylinder and $(\Gamma_i,\nu_i)$ are induced Eulerian embeddings in the cylinders $\Sigma \cap \Delta(F_1)$ and $\Sigma \setminus \Delta(F_1)^\circ$ respectively, such that $\Gamma_i$ agrees with $\Gamma$ for $i=1,2$. Thus if $(G_1,\pi(E_1^1),\pi(E_2^1))$  were not linked, say, we find a cut-cycle $F_1'$ in $\Gamma_1$ with $\delta(F_1') < \delta_1$. But then $F_1'$ is a cut-cycle in $\Gamma$ refuting our choice of $F_1$; contradiction.
    \end{claimproof}
    By \cref{claim:linked_cut_linked} both sides of the cut are again linked whence we may continue the construction by finding new cut-cycles $F_2^1,F_2^2$ in the respective pieces.

    Assume by recursion that we end up with a linked $\theta$-chopping $(\zeta_1,\ldots,F_{t_1},F_1)$ for $\Gamma_1$ and  $(F_1,\ldots,F_{t_2},\zeta_2)$ for $\Gamma_2$ where each piece is $\theta$-short.

    \begin{claim}\label{claim:linked_chopping_linked}
        $\FFF=(\zeta_1,\ldots, F_{t_1},F_1,\ldots,F_{t_2},\zeta_1)$ is a linked $\theta$-chopping of $\Gamma$.
    \end{claim}
    \begin{claimproof}
        By construction $\FFF$ is a $\theta$-chopping of $\Gamma$ since every cut-cycle in $\Gamma_1$ or $\Gamma_2$ is also a cut-cycle in $\Gamma$ of equal order. It is clearly linked since $G_1 \cap G_2 = \rho(F_1)$, and every piece is part of either $\Gamma_1 \subset \Gamma$ or $\Gamma_2 \subset \Gamma$ where they are linked by assumption.
    \end{claimproof}

   In particular all the pieces of $\FFF$ are still $\theta$-short (for we have effectively the same pieces); this concludes the proof.
\end{proof}

We will inductively prove, that the graphs that are not $\theta$-fat are well-quasi-ordered. The main ingredient is the following version of Higman's \cref{thm:higman} for $\ell$-chubby choppings.

\begin{lemma}\label{lem:knitting_chubby_choppings}
    Let $\ell \in 2\N$. Let $(G_i,\pi(E_1^i),\pi(E_2^i)) \in \mathbf{G^*}(\Sigma,\ell,\ell)$ be edge-rooted Eulerian digraphs with Eulerian embeddings $(\Gamma_i,\nu_i)$ in $\Sigma$ for $i=1,2$. Let $\FFF_i=(F_0^i,F_1^i,\ldots,F_{t_i}^i,F_{t_i+1}^i)$ be $\ell$-chubby choppings for $\Gamma_i$ of length $t_i \in \N$ for $i=1,2$ respectively. Then there exist orderings $\Pi_i=(\pi(F_0^i),\ldots,\pi(F_{t_i+1}^i))$ for $i=1,2$ such that the following holds.
    Let $0 \leq j_0 < \ldots <j_{t_1} \leq t_2+1$. If the $i$-th piece of $\FFF_1$ (with respect to $\Pi_1$) strongly immerses in the $j_{i}$-th piece of $\FFF_2$ (with respect to $\Pi_2$) for every $0 \leq i \leq t_1$, then $(G_1,\pi(E_1^1),\pi(E_1^2)) \hookrightarrow (G_2,\pi(E_1^2),\pi(E_2^2))$ by strong immersion.
\end{lemma}
\begin{proof}
    This follows by \cref{lem:knitting_separating_cutlines} similar to the proof of \cref{thm:bounded_carving_with_non_labelled}. To see this note that by definition of $\ell$-chubby chopping, there exists a $\{\rho(F_0^i),\rho(F_{t_{i+1}}^i)\}$-linkage~$\LLL_i=\{P_1^i,\ldots,P_\ell^i\}$ of order $\ell$ for $i=1,2$. Since $\delta(F_j^i) = \ell$ for every $0 \leq j \leq t_{i+1}^i$ there exists an ordering $\pi(F_j^i)=(e_1^{j,i},\ldots,e_\ell^{j,i})$ such that $e_p^{j,i} \in P_p^i$ for $i=1,2$. We claim that $\Pi_i = (\pi(F_0^i),\ldots,\pi(F_{t_{i+1}}^i))$ for $i=1,2$ does the trick. This follows easily by iterative application of \cref{lem:knitting_separating_cutlines}, noting that the linkage $\LLL_2$ can be used to extend the respective immersion of the pieces respecting the orderings.%
\end{proof}

In light of the above we define the following.

\begin{definition}\label{def:chubby_class}
    Let $\ell\in 2\N$, $\Sigma$ a cylinder and let $\AAA \subseteq \mathbf{G^*}(\Sigma,\ell,\ell)$ be a well-quasi-ordered class of edge-rooted Eulerian digraphs. We define $\operatorname{Chubby}_\ell(\AAA)\subseteq \bigcup_{2\ell_1,2\ell_2 \leq \ell}\mathbf{G^*}(\Sigma,2\ell_1,2\ell_2)$ to be the class of edge-rooted Eulerian digraphs $(G, \pi(E_1),\pi(E_2))$ with Eulerian embedding $(\Gamma,\nu)$ in $\Sigma$ such that there is an $\ell$-chopping $(F_0,F_1,\ldots,F_{t},F_{t+1})$ of $\Gamma$ for any length $t \in \N$ with $\delta(F_i) = \ell$ for every $1 \leq i \leq t$, and such that every piece is linked, the $i$-th piece is in $\AAA$ for every $1 \leq i \leq t-1$ and the first and last pieces are $\ell$-short.
\end{definition}

Using the newly introduced notation, the following is an easy extension to \cref{lem:knitting_chubby_choppings}.

\begin{lemma}\label{lem:wqo+chubby=wqo}
    Let $\AAA \subseteq \mathbf{G^*}(\Sigma,\ell,\ell)$ well-quasi-ordered such that every element in $\AAA$ is linked. Then $\operatorname{Chubby}_\ell(\AAA)$ is well-quasi-ordered and every element of $\operatorname{Chubby}_\ell(\AAA)$ is linked.
\end{lemma}
\begin{proof}
     The proof follows by \cref{lem:knitting_chubby_choppings} together with two applications of \cref{lem:knitting_separating_cutlines} for the first and last piece, noting that they are $\ell$-short whence we may use \cref{lem:short_are_good}. %
\end{proof}

Note that the choppings produced by \cref{lem:wqo+chubby=wqo} are not $\ell$-chubby, since the first and last cut-cycles may have a differing order, i.e., $\delta(F_0),\delta(F_{t_i+1})$ may be smaller than $\ell$. 

And by \cref{lem:short_are_good} we immediately derive the following.
\begin{corollary}
    Let $\ell \in 2\N$ and $\theta \in \N$. Then $\operatorname{\Chubby}_\ell(\AAA(\theta,\ell,\ell))$ is well-quasi-ordered and every element is linked; in particular the class is a subclass of $\mathbf{G^*}(\Sigma,\ell,\ell)$.
\end{corollary}

Finally we have the following.

\begin{corollary}\label{cor:thinning_sequences}
    Let $\ell \in 2\N$.  If $\AAA \subseteq \mathbf{G^*}(\ell,\ell)$ is well-quasi-ordered and every element in $\AAA$ is linked, then $\Chubby_{\ell-1}(\Chubby_\ell(\AAA))$ is well-quasi-ordered.
\end{corollary}
\begin{proof}
    This follows by iterative application of \cref{lem:wqo+chubby=wqo} and the \cref{def:chubby_class} of $\Chubby_\ell(\cdot)$.%
\end{proof}

Finally combining \cref{cor:thinning_sequences} and \cref{lem:wqo+chubby=wqo} we derive the following.
\begin{corollary}\label{superchubby_wqo}
    Let $\ell,k_1,k_2 \in 2\N$ and $\theta \geq 1$. The class $\Chubby_0(\Chubby_1(\ldots(\Chubby_\ell(\AAA(\theta,k_1,k_2))))$ is well-quasi-ordered by strong immersion.
\end{corollary}

We gathered all the results needed for the proof of \cref{thm:wqo:bead-root_for_cylinder}; recall that by \cref{cor:i_am_dying} it suffices to prove \cref{thm:wqo_edge-rooted} for the cylinder.

\begin{proof}[of \cref{thm:wqo:bead-root_for_cylinder}]
    The proof is by induction on $(k_1,k_2)$, where we assume it to be true for all $(k_1',k_2')$ with either $k_1' < k_1$ or $k_2' < k_2$ since it follows for $\mathbf{G^*}(0,k)$ and $\mathbf{G^*}(k,0))$ by \cref{thm:wqo:bead-root_for_disc}.
    
    Let $(G_i,\pi(E_1^i),\pi(E_2^i)) \in \mathbf{G^*}(k_1,k_2)$ with respective Eulerian embedding $(\Gamma_i,\nu_i)$ for every $i \in I$. We derive that there is an infinite $I\subseteq \N$ such that they are linked; if not there is an infinite sequence $I' \subseteq I$ and for every $i \in I'$ we find a a cut-cycle $F_i^*$ in $\Gamma_i$ along which we can cut the graph using \cref{def:cutting_sep_lines} resulting in two graphs $\big((G_i^1,\pi(\tilde E_i^1), \pi(\tilde E_i^2)),(G_i^2,\pi(\tilde F_i^1), \pi(\tilde F_i^2))) = \cut\big((G_i,\pi( E_i^1), \pi( E_i^2)), \pi(F_i^*)\big)$ given a respective ordering $\pi(F_i^*)$ of $\rho(F_i^*)$. The pieces are now in $\mathbf{G^*}(k_1^i,k_2^i)$ for either $k_1^i < k_1$ or $k_2^i < k_2$ for every $i \in \N$, and since $k_1 + k_2 $ is fixed and finite, by Higman's \cref{thm:higman} applied to the arising tuples and the pigeonhole principle, we are done by induction and \cref{lem:knitting_separating_cutlines}. %

    Now assume that the sequence is bad towards a contradiction. Let $\theta \coloneqq f_{\ref{lem:fat_are_good}}(\Abs{V(G_1)})$. By \cref{lem:fat_are_good} we derive that, for every $j \geq 2$, $(G_j,\pi(E_1^j),\pi(E_2^j))$ is not $f_{\ref{lem:fat_are_good}}(\Abs{V(G_1)})$-fat.

    By \cref{lem:short_decomposition} we derive that there are $\tau_j$-choppings $\FFF_j = (F_1,\ldots,F_{t_j})$ of $\Gamma_j$ for some $\tau_j \leq f_{\ref{lem:short_decomposition}}(\theta)$ and every $j \geq 2$ such that all the pieces are linked and $\theta$-short.  Using the pigeonhole principle we may switch to a subsequence $I' \subseteq I$ such that they are all $\ell$-choppings with linked $\theta$-short pieces for some fixed $\ell \leq f_{\ref{lem:short_decomposition}}(\theta)$. In particular then $(G_i,\pi(E_i^1),\pi(E_i^2))\in \Chubby_0(\Chubby_1(\ldots (\Chubby_\ell(\AAA(\theta,k_1,k_2))))$ and thus by \cref{superchubby_wqo} we derive that there are $i,j \in I'$ with $i<j$ and $(G_i,\pi(E_i^1),\pi(E_i^2)) \hookrightarrow (G_j,\pi(E_j^1),\pi(E_j^2))$ as a contradiction to the assumption. The proof follows.
\end{proof}

\section{La Grande Inductione}\label{sec:la-grande-inductione}

With \cref{sec:high-rep}, \cref{sec:disc} and \cref{sec:cylinder} we have finally gathered all the necessary results to prove the main \cref{thm:wqo_on_surfaces_upto_beads} of this paper; we restate it for completion.
\begin{equation*}
    \text{Let $\Sigma$ be a surface and $k \in 2\N$. The class~$\mathbf{G}(\Sigma,k)$ is well-quasi-ordered by rooted strong immersion.}
\end{equation*}

The remaining proof technique for the theorem is now fairly standard by induction on the genus, where either we can apply \cref{thm:high-rep} and are done, or we find tracing cuts (and cut-cycles), i.e., simple (closed) curves that intersect our graph only in edges, offending high representativity of the embeddings along which we can cut the surface to either reduce the genus, the number of cuffs or the number of root-edges and work by induction on either. We then apply \cref{lem:knitting_nonsep_cutlines} and \cref{lem:knitting_separating_cutlines} to knit the respective immersions back together. Note that most of the additional work has been done in \cref{sec:edge-rooted}; recall \cref{def:representativity}.

\begin{proof}[of \cref{thm:wqo_edge-rooted}]

    The proof is by induction on $\Sigma$ and subsequently $k \in 2\N$. For $\Sigma$ being the disc the claim follows for every $k \in 2\N$ by \cref{thm:wqo:bead-root_for_disc} together with \cref{cor:i_am_dying}. If $\Sigma$ is the cylinder, the claim follows for every $k \in 2\N$ by \cref{thm:wqo:bead-root_for_cylinder} and again \cref{cor:i_am_dying}. Thus, let $\Sigma$ be some surface not the disc, sphere or cylinder, and let $k \in 2\N$. Suppose further that we already established that $\mathbf{G^*}(\Sigma',k')$ is well-quasi-ordered by strong immersion for every surface $\Sigma'$ ``smaller'' than $\Sigma$---lower genus or less cuffs---and every $k' \in 2\N$, and if $\Sigma' \cong \Sigma$ then suppose we established it for every $0\leq k' < k$. Let $c(\Sigma)=\{\zeta_1,\ldots,\zeta_t\}$ be the set of cuffs of $\Sigma$.

    By \cref{cor:i_am_dying} it suffices to prove the theorem for edge-rooted Eulerian digraphs.
   Let $((G_i,\pi(E_i^*))_{i \in \N}$ be a sequence with $(G_i,\pi(E_i^*)) \in \mathbf{G^*}(\Sigma,k)$ and assume the $G_i$ to be weakly connected; otherwise it follows by an additional application of \cref{thm:higman}. For every $i \in \N$ let $(\Gamma_i,\nu_i,\omega_i)$ be an Eulerian embedding of $(G_i,\pi(E_i^*))$ in $\Sigma$ and let $W_{\omega_i} = (W_1^i,\ldots,W_t^i)$ be the respective partitions of degree one vertices in $G_i$ given by $\omega_i$. By repeated application of the pigeonhole principle we may assume that  $\omega_i(\zeta_\ell) \neq \emptyset$ for all $1 \leq \ell \leq t$; note that if there were an infinite index $I_\ell \subseteq \N$ with $\omega_i(\zeta_\ell) = \emptyset$ for any $1 \leq \ell \leq t$, then the edge-rooted graphs $(G_i,\pi(E_i^*))$ admit an Eulerian embedding in $\Sigma + \Delta_\ell$ obtained by gluing a disc to the cuff $\zeta_\ell$ whence the claim follows by induction on $\Sigma$.

   \begin{claim}
       Let $I \subseteq \N$ be infinite, then there exists $j \in I$ such that the embedding $(\Gamma_j,\nu_j,\omega_j)$ is nice.
   \end{claim}
   \begin{claimproof}
       Otherwise we find an infinite sub-sequence of not nice embeddings, in particular then either the embeddings are not $2$-cell or there is a cuff-connecting curve that does not intersect $\Gamma$. In either case we find a non-separating curve $\gamma_j:[0,1] \to \Sigma$ such that $\gamma_j([0,1]) \cap \Gamma_j = \emptyset$ and cutting along $\gamma_j$ results in a simpler surface $\Sigma'$: if the embedding is not $2$-cell there is a face that is not a disc and thus there is such a curve. Note that $\gamma_j$ is non-separating for else we decompose the graph into two parts refuting the graphs being connected (every cuff $\zeta_\ell$ witnesses at least one edge drawn on it since $\omega_i(\zeta_\ell) \neq \emptyset$ for $1 \leq \ell \leq t$).
   \end{claimproof}

   By the above we may assume our embeddings to be nice by moving to a respective subsequence. Again, by repeatedly applying the pigeonhole principle, there is an infinite index set $I \subseteq N$ such that $W_{\omega_i} = W_{\omega_j}$ (after possibly renaming the vertices) and $\Abs{\rho(W_i^\ell)} = \Abs{\rho(W_j^\ell)}$ for all $1 \leq \ell \leq t$ and all $i,j \in I$; let $I = \N$ for simplicity . 

    Let $n \coloneqq \Abs{V(G_1)}$. By \cref{obs:edge-rooted_immersion_is_enough_dos} we derive that $\stitch(G_i, \pi({E_i^*}_{\omega_i})) \in \mathbf{G}(\Sigma,k)$ and by \cref{thm:high-rep} there either exists $j \in \N$ with $\stitch(G_1, \pi({E_1^*}_{\omega_1})) \hookrightarrow \stitch(G_j, \pi({E_j^*}_{\omega_j}))$, in which case we are done using \cref{obs:edge-rooted_immersion_is_enough_dos}, or there exists a function $f:\N \to \N$ such that $\fw(\Gamma_i) \leq f(k)$ for all $i \in \N$. In particular, since we assume the former to not hold, \cref{def:representativity} implies the existence of either
    \begin{itemize}
        \item a cuff surrounding curve $\gamma:[0,1] \to \Sigma$,
        
        \item a cuff shortening curve $\theta:[0,1] \to \Sigma$,
        \item a cuff connecting curve $\xi: [0,1] \to \Sigma$,
        
        \item a cuff grouping curve $\lambda:[0,1] \to \Sigma$, or
         
        \item a genus reducing curve $\eta:[0,1] \to \Sigma$
    \end{itemize}
in $\Gamma_j$ of length $< f(k)$ for every $j > 1$ (note that the embeddings of the respective stitches agree with those of $(G_j, \pi(E_j^*))$ inside $\Sigma$). Note further that the offending curves are disjoint from $\bd(\Sigma)$ up-to possibly there endpoints (if they are cuff-based for example), and clearly by slightly moving the end-point to the left or right on the respective cuff, we can thus assure that $\chi([0,1]) \cap \bd(\Sigma) \cap \Gamma = \emptyset$, in particular $\chi$ is clean for every $\chi \in \{\gamma, \theta, \xi, \lambda, \eta)$.

We briefly discuss each of the cases separately for completion; recall \cref{obs:types_of_cut-line_cuts}. 
\begin{itemize}
    \item[$\gamma$] Let $\gamma:[0,1] \to \Sigma$ be a cuff surrounding curve for $\zeta_j \in c(\Sigma)$, i.e., $F \coloneqq \gamma([0,1])$ is a clean cuff surrounding cut-line (since the embedding is nice) with $F \cap \bd(\Sigma) \subseteq \zeta_j$ bounding a unique disc $\Delta(F) \subseteq \hat \Sigma$ such that $\ell \coloneqq \delta(F) < \Abs{\nu^{-1}(\zeta_j)}$ by \cref{def:representativity} of cuff surrounding curves.

    Let $\Sigma_1,\Sigma_2$ be the two surfaces obtained after cutting $F$, then $\Sigma_1$ is a cylinder and $\Sigma_2 \cong \Sigma$ say. Let $\cut(G_i,\pi(E_i^*)) = \big((G^i_1,\pi(E_*^1)), (G^i_2,\pi(E_2^*))\big)$, with $(G^i_1,\pi(E_*^1))\in \mathbf{G^*}(\Sigma_1,k_1)$ and  $(G^i_2,\pi(E_2^*))\in \mathbf{G^*}(\Sigma_2,k_2)$ for some $k_1,k_2 \in 2\N$ satisfying $k_1+k_2 = k + 2\ell$ by \cref{lem:cutting_sep_embedding}. Further, since $\ell < \Abs{\nu^{-1}(\zeta_j)}$, we derive that $k_1 \geq 1$ and $k_2 < k$. In particular induction is applicable to both, for the former one is embedded on a cylinder and the latter one admits less roots.

     \item[$\theta$] Let $\theta:[0,1] \to \Sigma$ be cuff shortening, let $F \coloneqq \theta([0,1])$ and let $\zeta_j \in c(\Sigma)$ be the cuff containing both endpoints of $F$; notably $F$ is a clean cut-line. By definition of cuff shortening curves, $F$ is homotopic in $\Sigma\setminus \{\theta(0),\theta(1)\}$ to one of the two segments. Let $\ell \coloneqq \delta(F)$.
     
     Let $\Sigma_1,\Sigma_2$ be the two surfaces obtained after cutting $F$, then $\Sigma_1$ is a disc and $\Sigma_2 \cong \Sigma$, say. Again, let $\cut(G_i,\pi(E_i^*)) = \big((G^i_1,\pi(E_*^1)), (G^i_2,\pi(E_2^*))\big)$, then by \cref{lem:cutting_sep_embedding} $(G^i_1,\pi(E_*^1))\in \mathbf{G^*}(\Sigma_1,k_1)$ and  $(G^i_2,\pi(E_2^*))\in \mathbf{G^*}(\Sigma_2,k_2)$ for $k_1,k_2 \in 2\N$ with $k_1+k_2 = k+2\ell$. Again by \cref{def:representativity} of cuff shortening curves we derive $k_2 < k$ and the induction is applicable to both pieces, for one is embedded in the disc and the other has less roots.

     \item[$\xi$] If $\xi$ is a cuff connecting curve, $F \coloneqq \xi([0,1])$ is a clean cuff connecting  cut-line with both its ends in different cuffs $\zeta_p,\zeta_q$ for $1 \leq p,q \leq t$ (notably not in edges $\nu(\pi(W))$) in which case cutting $\Sigma$ along $F$ results in a smaller surface $\Sigma'$. Let $\ell = \delta(F)$. Applying \cref{lem:cutting_nonsep_embed} we derive $\cut(G_i,\pi(E_i^*)) \in \mathbf{G^*}(\Sigma', k+2\ell)$; induction is applicable.
     
    \item[$\lambda,\eta$] Both curves can be either separating or non-separating, in either case all resulting surfaces are of lower genus by \cref{obs:types_of_cut-line_cuts} and the remainder follows as above using \cref{lem:cutting_nonsep_embed} and \cref{lem:cutting_nonsep_embed} where induction is a gain applicable.
\end{itemize}

Finally then, since there are up to isomorphism only finitely many surfaces $\Sigma'$ such that $\hat{\Sigma}'$ has a lower genus than $\Sigma$, and (up to isomorphism) only finitely many surfaces of equal genus but less cuffs, there is an index set $I \subseteq \N$ such that the cut-lines $F_i$ are of the same type for every $i \in \N$ and cutting $\Sigma$ along $F_i$ for every $i \in \N$ results in the same surface or the same pair of surfaces up to isomorphism. Similarly, we may use the pigeonhole principle to guarantee that the number of roots is again equal for the respective surfaces (or pairs), thus effectively either passing to one or two sequences of edge-rooted Eulerian digraphs of common defect Euler-embeddable on the same surface. The theorem now follows by induction, applying \cref{lem:knitting_nonsep_cutlines} and \cref{lem:knitting_separating_cutlines} respectively depending on the exact cut-line we cut the surface along.%
\end{proof}

\section*{Acknowledgments}
The authors thank Jil D. for the beautiful pictures she drew.

\bibliographystyle{alphaurl}

\end{document}